\documentclass[submission,copyright,creativecommons]{eptcs}

\providecommand{\event}{LINEARITY 2018} 
\usepackage{underscore}           

\usepackage{amssymb,amsmath,amsthm}
\usepackage{cmll}
\usepackage{txfonts}
\usepackage{graphicx}
\usepackage{stmaryrd}
\usepackage{todonotes}
\usepackage{mathpartir}
\usepackage{hyperref}
\usepackage{mdframed}
\usepackage[barr]{xy}
\usepackage{comment}
\usepackage{graphicx}
\usepackage[inline]{enumitem}
\usepackage{bbm}
\usepackage{array}

\usepackage{caption}
\newcommand{\NDdrule}[4][]{{\displaystyle\frac{\begin{array}{l}#2\end{array}}{#3}\quad\NDdrulename{#4}}}

\newcommand{\NDpremise}[1]{ #1 \\}
\newenvironment{NDdefnblock}[3][]{ \framebox{\mbox{#2}} \quad #3 \\[0pt]}{}

\newcommand{\NDnt}[1]{\mathit{#1}}
\newcommand{\NDmv}[1]{\mathit{#1}}

\newcommand{\NDsym}[1]{#1}

\newcommand{\NDdruleTXXidName}[0]{\NDdrulename{T\_id}}
\newcommand{\NDdruleTXXid}[1]{\NDdrule[#1]{%
}{
\NDmv{x}  \NDsym{:}  \NDnt{X}  \vdash_\mathcal{C}  \NDmv{x}  \NDsym{:}  \NDnt{X}}{%
{\NDdruleTXXidName}{}%
}}

\newcommand{\NDdruleTXXunitIName}[0]{\NDdrulename{T\_unitI}}
\newcommand{\NDdruleTXXunitI}[1]{\NDdrule[#1]{%
}{
 \cdot   \vdash_\mathcal{C}   \mathsf{triv}   \NDsym{:}   \mathsf{Unit} }{%
{\NDdruleTXXunitIName}{}%
}}

\newcommand{\NDdruleTXXunitEName}[0]{\NDdrulename{T\_unitE}}
\newcommand{\NDdruleTXXunitE}[1]{\NDdrule[#1]{%
\NDpremise{ \Phi  \vdash_\mathcal{C}  \NDnt{t_{{\mathrm{1}}}}  \NDsym{:}   \mathsf{Unit}   \quad  \Psi  \vdash_\mathcal{C}  \NDnt{t_{{\mathrm{2}}}}  \NDsym{:}  \NDnt{Y} }%
}{
\Phi  \NDsym{,}  \Psi  \vdash_\mathcal{C}   \mathsf{let}\, \NDnt{t_{{\mathrm{1}}}}  :   \mathsf{Unit}  \,\mathsf{be}\,  \mathsf{triv}  \,\mathsf{in}\, \NDnt{t_{{\mathrm{2}}}}   \NDsym{:}  \NDnt{Y}}{%
{\NDdruleTXXunitEName}{}%
}}

\newcommand{\NDdruleTXXtenIName}[0]{\NDdrulename{T\_tenI}}
\newcommand{\NDdruleTXXtenI}[1]{\NDdrule[#1]{%
\NDpremise{ \Phi  \vdash_\mathcal{C}  \NDnt{t_{{\mathrm{1}}}}  \NDsym{:}  \NDnt{X}  \quad  \Psi  \vdash_\mathcal{C}  \NDnt{t_{{\mathrm{2}}}}  \NDsym{:}  \NDnt{Y} }%
}{
\Phi  \NDsym{,}  \Psi  \vdash_\mathcal{C}  \NDnt{t_{{\mathrm{1}}}}  \otimes  \NDnt{t_{{\mathrm{2}}}}  \NDsym{:}  \NDnt{X}  \otimes  \NDnt{Y}}{%
{\NDdruleTXXtenIName}{}%
}}

\newcommand{\NDdruleTXXtenEName}[0]{\NDdrulename{T\_tenE}}
\newcommand{\NDdruleTXXtenE}[1]{\NDdrule[#1]{%
\NDpremise{ \Phi  \vdash_\mathcal{C}  \NDnt{t_{{\mathrm{1}}}}  \NDsym{:}  \NDnt{X}  \otimes  \NDnt{Y}  \quad  \Psi_{{\mathrm{1}}}  \NDsym{,}  \NDmv{x}  \NDsym{:}  \NDnt{X}  \NDsym{,}  \NDmv{y}  \NDsym{:}  \NDnt{Y}  \NDsym{,}  \Psi_{{\mathrm{2}}}  \vdash_\mathcal{C}  \NDnt{t_{{\mathrm{2}}}}  \NDsym{:}  \NDnt{Z} }%
}{
\Psi_{{\mathrm{1}}}  \NDsym{,}  \Phi  \NDsym{,}  \Psi_{{\mathrm{2}}}  \vdash_\mathcal{C}   \mathsf{let}\, \NDnt{t_{{\mathrm{1}}}}  :  \NDnt{X}  \otimes  \NDnt{Y} \,\mathsf{be}\, \NDmv{x}  \otimes  \NDmv{y} \,\mathsf{in}\, \NDnt{t_{{\mathrm{2}}}}   \NDsym{:}  \NDnt{Z}}{%
{\NDdruleTXXtenEName}{}%
}}

\newcommand{\NDdruleTXXimpIName}[0]{\NDdrulename{T\_impI}}
\newcommand{\NDdruleTXXimpI}[1]{\NDdrule[#1]{%
\NDpremise{\Phi  \NDsym{,}  \NDmv{x}  \NDsym{:}  \NDnt{X}  \vdash_\mathcal{C}  \NDnt{t}  \NDsym{:}  \NDnt{Y}}%
}{
\Phi  \vdash_\mathcal{C}   \lambda  \NDmv{x}  :  \NDnt{X} . \NDnt{t}   \NDsym{:}  \NDnt{X}  \multimap  \NDnt{Y}}{%
{\NDdruleTXXimpIName}{}%
}}

\newcommand{\NDdruleTXXimpEName}[0]{\NDdrulename{T\_impE}}
\newcommand{\NDdruleTXXimpE}[1]{\NDdrule[#1]{%
\NDpremise{ \Phi  \vdash_\mathcal{C}  \NDnt{t_{{\mathrm{1}}}}  \NDsym{:}  \NDnt{X}  \multimap  \NDnt{Y}  \quad  \Psi  \vdash_\mathcal{C}  \NDnt{t_{{\mathrm{2}}}}  \NDsym{:}  \NDnt{X} }%
}{
\Phi  \NDsym{,}  \Psi  \vdash_\mathcal{C}   \NDnt{t_{{\mathrm{1}}}}   \NDnt{t_{{\mathrm{2}}}}   \NDsym{:}  \NDnt{Y}}{%
{\NDdruleTXXimpEName}{}%
}}

\newcommand{\NDdruleTXXGIName}[0]{\NDdrulename{T\_GI}}
\newcommand{\NDdruleTXXGI}[1]{\NDdrule[#1]{%
\NDpremise{\Phi  \vdash_\mathcal{L}  \NDnt{s}  \NDsym{:}  \NDnt{A}}%
}{
\Phi  \vdash_\mathcal{C}   \mathsf{G}\, \NDnt{s}   \NDsym{:}   \mathsf{G} \NDnt{A} }{%
{\NDdruleTXXGIName}{}%
}}

\newcommand{\NDdruleTXXbetaName}[0]{\NDdrulename{T\_beta}}
\newcommand{\NDdruleTXXbeta}[1]{\NDdrule[#1]{%
\NDpremise{\Phi  \NDsym{,}  \NDmv{x}  \NDsym{:}  \NDnt{X}  \NDsym{,}  \NDmv{y}  \NDsym{:}  \NDnt{Y}  \NDsym{,}  \Psi  \vdash_\mathcal{C}  \NDnt{t}  \NDsym{:}  \NDnt{Z}}%
}{
\Phi  \NDsym{,}  \NDmv{z}  \NDsym{:}  \NDnt{Y}  \NDsym{,}  \NDmv{w}  \NDsym{:}  \NDnt{X}  \NDsym{,}  \Psi  \vdash_\mathcal{C}   \mathsf{ex}\, \NDmv{w} , \NDmv{z} \,\mathsf{with}\, \NDmv{x} , \NDmv{y} \,\mathsf{in}\, \NDnt{t}   \NDsym{:}  \NDnt{Z}}{%
{\NDdruleTXXbetaName}{}%
}}

\newcommand{\NDdruleTXXcutName}[0]{\NDdrulename{T\_cut}}
\newcommand{\NDdruleTXXcut}[1]{\NDdrule[#1]{%
\NDpremise{ \Phi  \vdash_\mathcal{C}  \NDnt{t_{{\mathrm{1}}}}  \NDsym{:}  \NDnt{X}  \quad  \Psi_{{\mathrm{1}}}  \NDsym{,}  \NDmv{x}  \NDsym{:}  \NDnt{X}  \NDsym{,}  \Psi_{{\mathrm{2}}}  \vdash_\mathcal{C}  \NDnt{t_{{\mathrm{2}}}}  \NDsym{:}  \NDnt{Y} }%
}{
\Psi_{{\mathrm{1}}}  \NDsym{,}  \Phi  \NDsym{,}  \Psi_{{\mathrm{2}}}  \vdash_\mathcal{C}  \NDsym{[}  \NDnt{t_{{\mathrm{1}}}}  \NDsym{/}  \NDmv{x}  \NDsym{]}  \NDnt{t_{{\mathrm{2}}}}  \NDsym{:}  \NDnt{Y}}{%
{\NDdruleTXXcutName}{}%
}}

\newcommand{\NDdruleSXXidName}[0]{\NDdrulename{S\_id}}
\newcommand{\NDdruleSXXid}[1]{\NDdrule[#1]{%
}{
\NDmv{x}  \NDsym{:}  \NDnt{A}  \vdash_\mathcal{L}  \NDmv{x}  \NDsym{:}  \NDnt{A}}{%
{\NDdruleSXXidName}{}%
}}

\newcommand{\NDdruleSXXunitIName}[0]{\NDdrulename{S\_unitI}}
\newcommand{\NDdruleSXXunitI}[1]{\NDdrule[#1]{%
}{
 \cdot   \vdash_\mathcal{L}   \mathsf{triv}   \NDsym{:}   \mathsf{Unit} }{%
{\NDdruleSXXunitIName}{}%
}}

\newcommand{\NDdruleSXXunitEOneName}[0]{\NDdrulename{S\_unitE1}}
\newcommand{\NDdruleSXXunitEOne}[1]{\NDdrule[#1]{%
\NDpremise{ \Phi  \vdash_\mathcal{C}  \NDnt{t}  \NDsym{:}   \mathsf{Unit}   \quad  \Gamma  \vdash_\mathcal{L}  \NDnt{s}  \NDsym{:}  \NDnt{A} }%
}{
\Phi  \NDsym{;}  \Gamma  \vdash_\mathcal{L}   \mathsf{let}\, \NDnt{t}  :   \mathsf{Unit}  \,\mathsf{be}\,  \mathsf{triv}  \,\mathsf{in}\, \NDnt{s}   \NDsym{:}  \NDnt{A}}{%
{\NDdruleSXXunitEOneName}{}%
}}

\newcommand{\NDdruleSXXunitETwoName}[0]{\NDdrulename{S\_unitE2}}
\newcommand{\NDdruleSXXunitETwo}[1]{\NDdrule[#1]{%
\NDpremise{ \Gamma  \vdash_\mathcal{L}  \NDnt{s_{{\mathrm{1}}}}  \NDsym{:}   \mathsf{Unit}   \quad  \Delta  \vdash_\mathcal{L}  \NDnt{s_{{\mathrm{2}}}}  \NDsym{:}  \NDnt{A} }%
}{
\Gamma  \NDsym{;}  \Delta  \vdash_\mathcal{L}   \mathsf{let}\, \NDnt{s_{{\mathrm{1}}}}  :   \mathsf{Unit}  \,\mathsf{be}\,  \mathsf{triv}  \,\mathsf{in}\, \NDnt{s_{{\mathrm{2}}}}   \NDsym{:}  \NDnt{A}}{%
{\NDdruleSXXunitETwoName}{}%
}}

\newcommand{\NDdruleSXXtenIName}[0]{\NDdrulename{S\_tenI}}
\newcommand{\NDdruleSXXtenI}[1]{\NDdrule[#1]{%
\NDpremise{ \Gamma  \vdash_\mathcal{L}  \NDnt{s_{{\mathrm{1}}}}  \NDsym{:}  \NDnt{A}  \quad  \Delta  \vdash_\mathcal{L}  \NDnt{s_{{\mathrm{2}}}}  \NDsym{:}  \NDnt{B} }%
}{
\Gamma  \NDsym{;}  \Delta  \vdash_\mathcal{L}  \NDnt{s_{{\mathrm{1}}}}  \triangleright  \NDnt{s_{{\mathrm{2}}}}  \NDsym{:}  \NDnt{A}  \triangleright  \NDnt{B}}{%
{\NDdruleSXXtenIName}{}%
}}

\newcommand{\NDdruleSXXtenEOneName}[0]{\NDdrulename{S\_tenE1}}
\newcommand{\NDdruleSXXtenEOne}[1]{\NDdrule[#1]{%
\NDpremise{ \Phi  \vdash_\mathcal{C}  \NDnt{t}  \NDsym{:}  \NDnt{X}  \otimes  \NDnt{Y}  \quad  \Gamma_{{\mathrm{1}}}  \NDsym{;}  \NDmv{x}  \NDsym{:}  \NDnt{X}  \NDsym{;}  \NDmv{y}  \NDsym{:}  \NDnt{Y}  \NDsym{;}  \Gamma_{{\mathrm{2}}}  \vdash_\mathcal{L}  \NDnt{s}  \NDsym{:}  \NDnt{A} }%
}{
\Gamma_{{\mathrm{1}}}  \NDsym{;}  \Phi  \NDsym{;}  \Gamma_{{\mathrm{2}}}  \vdash_\mathcal{L}   \mathsf{let}\, \NDnt{t}  :  \NDnt{X}  \otimes  \NDnt{Y} \,\mathsf{be}\, \NDmv{x}  \otimes  \NDmv{y} \,\mathsf{in}\, \NDnt{s}   \NDsym{:}  \NDnt{A}}{%
{\NDdruleSXXtenEOneName}{}%
}}

\newcommand{\NDdruleSXXtenETwoName}[0]{\NDdrulename{S\_tenE2}}
\newcommand{\NDdruleSXXtenETwo}[1]{\NDdrule[#1]{%
\NDpremise{ \Gamma  \vdash_\mathcal{L}  \NDnt{s_{{\mathrm{1}}}}  \NDsym{:}  \NDnt{A}  \triangleright  \NDnt{B}  \quad  \Delta_{{\mathrm{1}}}  \NDsym{;}  \NDmv{x}  \NDsym{:}  \NDnt{A}  \NDsym{;}  \NDmv{y}  \NDsym{:}  \NDnt{B}  \NDsym{;}  \Delta_{{\mathrm{2}}}  \vdash_\mathcal{L}  \NDnt{s_{{\mathrm{2}}}}  \NDsym{:}  \NDnt{C} }%
}{
\Delta_{{\mathrm{1}}}  \NDsym{;}  \Gamma  \NDsym{;}  \Delta_{{\mathrm{2}}}  \vdash_\mathcal{L}   \mathsf{let}\, \NDnt{s_{{\mathrm{1}}}}  :  \NDnt{A}  \triangleright  \NDnt{B} \,\mathsf{be}\, \NDmv{x}  \triangleright  \NDmv{y} \,\mathsf{in}\, \NDnt{s_{{\mathrm{2}}}}   \NDsym{:}  \NDnt{C}}{%
{\NDdruleSXXtenETwoName}{}%
}}

\newcommand{\NDdruleSXXimprIName}[0]{\NDdrulename{S\_imprI}}
\newcommand{\NDdruleSXXimprI}[1]{\NDdrule[#1]{%
\NDpremise{\Gamma  \NDsym{;}  \NDmv{x}  \NDsym{:}  \NDnt{A}  \vdash_\mathcal{L}  \NDnt{s}  \NDsym{:}  \NDnt{B}}%
}{
\Gamma  \vdash_\mathcal{L}   \lambda_r  \NDmv{x}  :  \NDnt{A} . \NDnt{s}   \NDsym{:}  \NDnt{A}  \rightharpoonup  \NDnt{B}}{%
{\NDdruleSXXimprIName}{}%
}}

\newcommand{\NDdruleSXXimprEName}[0]{\NDdrulename{S\_imprE}}
\newcommand{\NDdruleSXXimprE}[1]{\NDdrule[#1]{%
\NDpremise{ \Gamma  \vdash_\mathcal{L}  \NDnt{s_{{\mathrm{1}}}}  \NDsym{:}  \NDnt{A}  \rightharpoonup  \NDnt{B}  \quad  \Delta  \vdash_\mathcal{L}  \NDnt{s_{{\mathrm{2}}}}  \NDsym{:}  \NDnt{A} }%
}{
\Gamma  \NDsym{;}  \Delta  \vdash_\mathcal{L}   \mathsf{app}_r\, \NDnt{s_{{\mathrm{1}}}} \, \NDnt{s_{{\mathrm{2}}}}   \NDsym{:}  \NDnt{B}}{%
{\NDdruleSXXimprEName}{}%
}}

\newcommand{\NDdruleSXXimplIName}[0]{\NDdrulename{S\_implI}}
\newcommand{\NDdruleSXXimplI}[1]{\NDdrule[#1]{%
\NDpremise{\NDmv{x}  \NDsym{:}  \NDnt{A}  \NDsym{;}  \Gamma  \vdash_\mathcal{L}  \NDnt{s}  \NDsym{:}  \NDnt{B}}%
}{
\Gamma  \vdash_\mathcal{L}   \lambda_l  \NDmv{x}  :  \NDnt{A} . \NDnt{s}   \NDsym{:}  \NDnt{B}  \leftharpoonup  \NDnt{A}}{%
{\NDdruleSXXimplIName}{}%
}}

\newcommand{\NDdruleSXXimplEName}[0]{\NDdrulename{S\_implE}}
\newcommand{\NDdruleSXXimplE}[1]{\NDdrule[#1]{%
\NDpremise{ \Gamma  \vdash_\mathcal{L}  \NDnt{s_{{\mathrm{1}}}}  \NDsym{:}  \NDnt{B}  \leftharpoonup  \NDnt{A}  \quad  \Delta  \vdash_\mathcal{L}  \NDnt{s_{{\mathrm{2}}}}  \NDsym{:}  \NDnt{A} }%
}{
\Delta  \NDsym{;}  \Gamma  \vdash_\mathcal{L}   \mathsf{app}_l\, \NDnt{s_{{\mathrm{1}}}} \, \NDnt{s_{{\mathrm{2}}}}   \NDsym{:}  \NDnt{B}}{%
{\NDdruleSXXimplEName}{}%
}}

\newcommand{\NDdruleSXXFIName}[0]{\NDdrulename{S\_FI}}
\newcommand{\NDdruleSXXFI}[1]{\NDdrule[#1]{%
\NDpremise{\Phi  \vdash_\mathcal{C}  \NDnt{t}  \NDsym{:}  \NDnt{X}}%
}{
\Phi  \vdash_\mathcal{L}   \mathsf{F} \NDnt{t}   \NDsym{:}   \mathsf{F} \NDnt{X} }{%
{\NDdruleSXXFIName}{}%
}}

\newcommand{\NDdruleSXXFEName}[0]{\NDdrulename{S\_FE}}
\newcommand{\NDdruleSXXFE}[1]{\NDdrule[#1]{%
\NDpremise{ \Gamma  \vdash_\mathcal{L}  \NDnt{s_{{\mathrm{1}}}}  \NDsym{:}   \mathsf{F} \NDnt{X}   \quad  \Delta_{{\mathrm{1}}}  \NDsym{;}  \NDmv{x}  \NDsym{:}  \NDnt{X}  \NDsym{;}  \Delta_{{\mathrm{2}}}  \vdash_\mathcal{L}  \NDnt{s_{{\mathrm{2}}}}  \NDsym{:}  \NDnt{A} }%
}{
\Delta_{{\mathrm{1}}}  \NDsym{;}  \Gamma  \NDsym{;}  \Delta_{{\mathrm{2}}}  \vdash_\mathcal{L}   \mathsf{let}\, \NDnt{s_{{\mathrm{1}}}}  :   \mathsf{F} \NDnt{X}  \,\mathsf{be}\,  \mathsf{F}\, \NDmv{x}  \,\mathsf{in}\, \NDnt{s_{{\mathrm{2}}}}   \NDsym{:}  \NDnt{A}}{%
{\NDdruleSXXFEName}{}%
}}

\newcommand{\NDdruleSXXGEName}[0]{\NDdrulename{S\_GE}}
\newcommand{\NDdruleSXXGE}[1]{\NDdrule[#1]{%
\NDpremise{\Phi  \vdash_\mathcal{C}  \NDnt{t}  \NDsym{:}   \mathsf{G} \NDnt{A} }%
}{
\Phi  \vdash_\mathcal{L}   \mathsf{derelict}\, \NDnt{t}   \NDsym{:}  \NDnt{A}}{%
{\NDdruleSXXGEName}{}%
}}

\newcommand{\NDdruleSXXbetaName}[0]{\NDdrulename{S\_beta}}
\newcommand{\NDdruleSXXbeta}[1]{\NDdrule[#1]{%
\NDpremise{\Gamma  \NDsym{;}  \NDmv{x}  \NDsym{:}  \NDnt{X}  \NDsym{;}  \NDmv{y}  \NDsym{:}  \NDnt{Y}  \NDsym{;}  \Delta  \vdash_\mathcal{L}  \NDnt{s}  \NDsym{:}  \NDnt{A}}%
}{
\Gamma  \NDsym{;}  \NDmv{z}  \NDsym{:}  \NDnt{Y}  \NDsym{;}  \NDmv{w}  \NDsym{:}  \NDnt{X}  \NDsym{;}  \Delta  \vdash_\mathcal{L}   \mathsf{ex}\, \NDmv{w} , \NDmv{z} \,\mathsf{with}\, \NDmv{x} , \NDmv{y} \,\mathsf{in}\, \NDnt{s}   \NDsym{:}  \NDnt{A}}{%
{\NDdruleSXXbetaName}{}%
}}

\newcommand{\NDdruleSXXcutOneName}[0]{\NDdrulename{S\_cut1}}
\newcommand{\NDdruleSXXcutOne}[1]{\NDdrule[#1]{%
\NDpremise{ \Phi  \vdash_\mathcal{C}  \NDnt{t}  \NDsym{:}  \NDnt{X}  \quad  \Gamma_{{\mathrm{1}}}  \NDsym{;}  \NDmv{x}  \NDsym{:}  \NDnt{X}  \NDsym{;}  \Gamma_{{\mathrm{2}}}  \vdash_\mathcal{L}  \NDnt{s}  \NDsym{:}  \NDnt{A} }%
}{
\Gamma_{{\mathrm{1}}}  \NDsym{;}  \Phi  \NDsym{;}  \Gamma_{{\mathrm{1}}}  \vdash_\mathcal{L}  \NDsym{[}  \NDnt{t}  \NDsym{/}  \NDmv{x}  \NDsym{]}  \NDnt{s}  \NDsym{:}  \NDnt{A}}{%
{\NDdruleSXXcutOneName}{}%
}}

\newcommand{\NDdruleSXXcutTwoName}[0]{\NDdrulename{S\_cut2}}
\newcommand{\NDdruleSXXcutTwo}[1]{\NDdrule[#1]{%
\NDpremise{ \Gamma  \vdash_\mathcal{L}  \NDnt{s_{{\mathrm{1}}}}  \NDsym{:}  \NDnt{A}  \quad  \Delta_{{\mathrm{1}}}  \NDsym{;}  \NDmv{x}  \NDsym{:}  \NDnt{A}  \NDsym{;}  \Delta_{{\mathrm{2}}}  \vdash_\mathcal{L}  \NDnt{s_{{\mathrm{2}}}}  \NDsym{:}  \NDnt{B} }%
}{
\Delta_{{\mathrm{1}}}  \NDsym{;}  \Gamma  \NDsym{;}  \Delta_{{\mathrm{2}}}  \vdash_\mathcal{L}  \NDsym{[}  \NDnt{s_{{\mathrm{1}}}}  \NDsym{/}  \NDmv{x}  \NDsym{]}  \NDnt{s_{{\mathrm{2}}}}  \NDsym{:}  \NDnt{B}}{%
{\NDdruleSXXcutTwoName}{}%
}}

\newcommand{\NDdruleTbetaXXletUName}[0]{\NDdrulename{Tbeta\_letU}}
\newcommand{\NDdruleTbetaXXletU}[1]{\NDdrule[#1]{%
}{
 \mathsf{let}\,  \mathsf{triv}   :   \mathsf{Unit}  \,\mathsf{be}\,  \mathsf{triv}  \,\mathsf{in}\, \NDnt{t}   \leadsto_\beta  \NDnt{t}}{%
{\NDdruleTbetaXXletUName}{}%
}}

\newcommand{\NDdruleTbetaXXletTName}[0]{\NDdrulename{Tbeta\_letT}}
\newcommand{\NDdruleTbetaXXletT}[1]{\NDdrule[#1]{%
}{
 \mathsf{let}\, \NDnt{t_{{\mathrm{1}}}}  \otimes  \NDnt{t_{{\mathrm{2}}}}  :  \NDnt{X}  \otimes  \NDnt{Y} \,\mathsf{be}\, \NDmv{x}  \otimes  \NDmv{y} \,\mathsf{in}\, \NDnt{t_{{\mathrm{3}}}}   \leadsto_\beta  \NDsym{[}  \NDnt{t_{{\mathrm{1}}}}  \NDsym{/}  \NDmv{x}  \NDsym{]}  \NDsym{[}  \NDnt{t_{{\mathrm{2}}}}  \NDsym{/}  \NDmv{y}  \NDsym{]}  \NDnt{t_{{\mathrm{3}}}}}{%
{\NDdruleTbetaXXletTName}{}%
}}

\newcommand{\NDdruleTbetaXXlamName}[0]{\NDdrulename{Tbeta\_lam}}
\newcommand{\NDdruleTbetaXXlam}[1]{\NDdrule[#1]{%
}{
 \NDsym{(}   \lambda  \NDmv{x}  :  \NDnt{X} . \NDnt{t_{{\mathrm{1}}}}   \NDsym{)}   \NDnt{t_{{\mathrm{2}}}}   \leadsto_\beta  \NDsym{[}  \NDnt{t_{{\mathrm{2}}}}  \NDsym{/}  \NDmv{x}  \NDsym{]}  \NDnt{t_{{\mathrm{1}}}}}{%
{\NDdruleTbetaXXlamName}{}%
}}

\newcommand{\NDdruleTbetaXXappOneName}[0]{\NDdrulename{Tbeta\_app1}}

\newcommand{\NDdruleTbetaXXappTwoName}[0]{\NDdrulename{Tbeta\_app2}}

\newcommand{\NDdruleTbetaXXappLetName}[0]{\NDdrulename{Tbeta\_appLet}}

\newcommand{\NDdruleTbetaXXletLetName}[0]{\NDdrulename{Tbeta\_letLet}}

\newcommand{\NDdruleTbetaXXletAppName}[0]{\NDdrulename{Tbeta\_letApp}}

\newcommand{\NDdruleSbetaXXletUOneName}[0]{\NDdrulename{Sbeta\_letU1}}
\newcommand{\NDdruleSbetaXXletUOne}[1]{\NDdrule[#1]{%
}{
 \mathsf{let}\,  \mathsf{triv}   :   \mathsf{Unit}  \,\mathsf{be}\,  \mathsf{triv}  \,\mathsf{in}\, \NDnt{s}   \leadsto_\beta  \NDnt{s}}{%
{\NDdruleSbetaXXletUOneName}{}%
}}

\newcommand{\NDdruleSbetaXXletTOneName}[0]{\NDdrulename{Sbeta\_letT1}}
\newcommand{\NDdruleSbetaXXletTOne}[1]{\NDdrule[#1]{%
}{
 \mathsf{let}\, \NDnt{t_{{\mathrm{1}}}}  \otimes  \NDnt{t_{{\mathrm{2}}}}  :  \NDnt{X}  \otimes  \NDnt{Y} \,\mathsf{be}\, \NDmv{x}  \triangleright  \NDmv{y} \,\mathsf{in}\, \NDnt{s}   \leadsto_\beta  \NDsym{[}  \NDnt{t_{{\mathrm{1}}}}  \NDsym{/}  \NDmv{x}  \NDsym{]}  \NDsym{[}  \NDnt{t_{{\mathrm{2}}}}  \NDsym{/}  \NDmv{y}  \NDsym{]}  \NDnt{s}}{%
{\NDdruleSbetaXXletTOneName}{}%
}}

\newcommand{\NDdruleSbetaXXletTTwoName}[0]{\NDdrulename{Sbeta\_letT2}}
\newcommand{\NDdruleSbetaXXletTTwo}[1]{\NDdrule[#1]{%
}{
 \mathsf{let}\, \NDnt{s_{{\mathrm{1}}}}  \triangleright  \NDnt{s_{{\mathrm{2}}}}  :  \NDnt{A}  \triangleright  \NDnt{B} \,\mathsf{be}\, \NDmv{x}  \triangleright  \NDmv{y} \,\mathsf{in}\, \NDnt{s_{{\mathrm{3}}}}   \leadsto_\beta  \NDsym{[}  \NDnt{s_{{\mathrm{1}}}}  \NDsym{/}  \NDmv{x}  \NDsym{]}  \NDsym{[}  \NDnt{s_{{\mathrm{2}}}}  \NDsym{/}  \NDmv{y}  \NDsym{]}  \NDnt{s_{{\mathrm{3}}}}}{%
{\NDdruleSbetaXXletTTwoName}{}%
}}

\newcommand{\NDdruleSbetaXXletFName}[0]{\NDdrulename{Sbeta\_letF}}
\newcommand{\NDdruleSbetaXXletF}[1]{\NDdrule[#1]{%
}{
 \mathsf{let}\,  \mathsf{F} \NDnt{t}   :   \mathsf{F} \NDnt{X}  \,\mathsf{be}\,  \mathsf{F}\, \NDmv{x}  \,\mathsf{in}\, \NDnt{s}   \leadsto_\beta  \NDsym{[}  \NDnt{t}  \NDsym{/}  \NDmv{x}  \NDsym{]}  \NDnt{s}}{%
{\NDdruleSbetaXXletFName}{}%
}}

\newcommand{\NDdruleSbetaXXlamLName}[0]{\NDdrulename{Sbeta\_lamL}}
\newcommand{\NDdruleSbetaXXlamL}[1]{\NDdrule[#1]{%
}{
 \mathsf{app}_l\, \NDsym{(}   \lambda_l  \NDmv{x}  :  \NDnt{A} . \NDnt{s_{{\mathrm{1}}}}   \NDsym{)} \, \NDnt{s_{{\mathrm{2}}}}   \leadsto_\beta  \NDsym{[}  \NDnt{s_{{\mathrm{2}}}}  \NDsym{/}  \NDmv{x}  \NDsym{]}  \NDnt{s_{{\mathrm{1}}}}}{%
{\NDdruleSbetaXXlamLName}{}%
}}

\newcommand{\NDdruleSbetaXXlamRName}[0]{\NDdrulename{Sbeta\_lamR}}
\newcommand{\NDdruleSbetaXXlamR}[1]{\NDdrule[#1]{%
}{
 \mathsf{app}_r\, \NDsym{(}   \lambda_r  \NDmv{x}  :  \NDnt{A} . \NDnt{s_{{\mathrm{1}}}}   \NDsym{)} \, \NDnt{s_{{\mathrm{2}}}}   \leadsto_\beta  \NDsym{[}  \NDnt{s_{{\mathrm{2}}}}  \NDsym{/}  \NDmv{x}  \NDsym{]}  \NDnt{s_{{\mathrm{1}}}}}{%
{\NDdruleSbetaXXlamRName}{}%
}}

\newcommand{\NDdruleSbetaXXapplOneName}[0]{\NDdrulename{Sbeta\_appl1}}

\newcommand{\NDdruleSbetaXXapplTwoName}[0]{\NDdrulename{Sbeta\_appl2}}

\newcommand{\NDdruleSbetaXXapprOneName}[0]{\NDdrulename{Sbeta\_appr1}}

\newcommand{\NDdruleSbetaXXapprTwoName}[0]{\NDdrulename{Sbeta\_appr2}}

\newcommand{\NDdruleSbetaXXderelictName}[0]{\NDdrulename{Sbeta\_derelict}}
\newcommand{\NDdruleSbetaXXderelict}[1]{\NDdrule[#1]{%
}{
 \mathsf{derelict}\, \NDsym{(}   \mathsf{G}\, \NDnt{s}   \NDsym{)}   \leadsto_\beta  \NDnt{s}}{%
{\NDdruleSbetaXXderelictName}{}%
}}

\newcommand{\NDdruleSbetaXXapplLetName}[0]{\NDdrulename{Sbeta\_applLet}}

\newcommand{\NDdruleSbetaXXapprLetName}[0]{\NDdrulename{Sbeta\_apprLet}}

\newcommand{\NDdruleSbetaXXletLetName}[0]{\NDdrulename{Sbeta\_letLet}}

\newcommand{\NDdruleSbetaXXletApplName}[0]{\NDdrulename{Sbeta\_letAppl}}

\newcommand{\NDdruleSbetaXXletApprName}[0]{\NDdrulename{Sbeta\_letAppr}}

\newcommand{\NDdruleTcomXXunitEXXunitEName}[0]{\NDdrulename{Tcom\_unitE\_unitE}}
\newcommand{\NDdruleTcomXXunitEXXunitE}[1]{\NDdrule[#1]{%
}{
 \mathsf{let}\, \NDsym{(}   \mathsf{let}\, \NDnt{t_{{\mathrm{2}}}}  :   \mathsf{Unit}  \,\mathsf{be}\,  \mathsf{triv}  \,\mathsf{in}\, \NDnt{t_{{\mathrm{1}}}}   \NDsym{)}  :   \mathsf{Unit}  \,\mathsf{be}\,  \mathsf{triv}  \,\mathsf{in}\, \NDnt{t_{{\mathrm{3}}}}   \leadsto_\mathsf{c}   \mathsf{let}\, \NDnt{t_{{\mathrm{2}}}}  :   \mathsf{Unit}  \,\mathsf{be}\,  \mathsf{triv}  \,\mathsf{in}\, \NDsym{(}   \mathsf{let}\, \NDnt{t_{{\mathrm{1}}}}  :   \mathsf{Unit}  \,\mathsf{be}\,  \mathsf{triv}  \,\mathsf{in}\, \NDnt{t_{{\mathrm{3}}}}   \NDsym{)} }{%
{\NDdruleTcomXXunitEXXunitEName}{}%
}}

\newcommand{\NDdruleTcomXXunitEXXtenEName}[0]{\NDdrulename{Tcom\_unitE\_tenE}}
\newcommand{\NDdruleTcomXXunitEXXtenE}[1]{\NDdrule[#1]{%
}{
 \mathsf{let}\, \NDsym{(}   \mathsf{let}\, \NDnt{t_{{\mathrm{2}}}}  :   \mathsf{Unit}  \,\mathsf{be}\,  \mathsf{triv}  \,\mathsf{in}\, \NDnt{t_{{\mathrm{1}}}}   \NDsym{)}  :  \NDnt{X}  \otimes  \NDnt{Y} \,\mathsf{be}\, \NDmv{x}  \otimes  \NDmv{y} \,\mathsf{in}\, \NDnt{t_{{\mathrm{3}}}}   \leadsto_\mathsf{c}   \mathsf{let}\, \NDnt{t_{{\mathrm{2}}}}  :   \mathsf{Unit}  \,\mathsf{be}\,  \mathsf{triv}  \,\mathsf{in}\, \NDsym{(}   \mathsf{let}\, \NDnt{t_{{\mathrm{1}}}}  :  \NDnt{X}  \otimes  \NDnt{Y} \,\mathsf{be}\, \NDmv{x}  \otimes  \NDmv{y} \,\mathsf{in}\, \NDnt{t_{{\mathrm{3}}}}   \NDsym{)} }{%
{\NDdruleTcomXXunitEXXtenEName}{}%
}}

\newcommand{\NDdruleTcomXXunitEXXimpEName}[0]{\NDdrulename{Tcom\_unitE\_impE}}
\newcommand{\NDdruleTcomXXunitEXXimpE}[1]{\NDdrule[#1]{%
}{
 \NDsym{(}   \mathsf{let}\, \NDnt{t_{{\mathrm{2}}}}  :   \mathsf{Unit}  \,\mathsf{be}\,  \mathsf{triv}  \,\mathsf{in}\, \NDnt{t_{{\mathrm{1}}}}   \NDsym{)}   \NDnt{t_{{\mathrm{3}}}}   \leadsto_\mathsf{c}   \mathsf{let}\, \NDnt{t_{{\mathrm{2}}}}  :   \mathsf{Unit}  \,\mathsf{be}\,  \mathsf{triv}  \,\mathsf{in}\, \NDsym{(}   \NDnt{t_{{\mathrm{1}}}}   \NDnt{t_{{\mathrm{3}}}}   \NDsym{)} }{%
{\NDdruleTcomXXunitEXXimpEName}{}%
}}

\newcommand{\NDdruleTcomXXtenEXXunitEName}[0]{\NDdrulename{Tcom\_tenE\_unitE}}
\newcommand{\NDdruleTcomXXtenEXXunitE}[1]{\NDdrule[#1]{%
}{
 \mathsf{let}\, \NDsym{(}   \mathsf{let}\, \NDnt{t_{{\mathrm{2}}}}  :  \NDnt{X}  \otimes  \NDnt{Y} \,\mathsf{be}\, \NDmv{x}  \otimes  \NDmv{y} \,\mathsf{in}\, \NDnt{t_{{\mathrm{1}}}}   \NDsym{)}  :   \mathsf{Unit}  \,\mathsf{be}\,  \mathsf{triv}  \,\mathsf{in}\, \NDnt{t_{{\mathrm{3}}}}   \leadsto_\mathsf{c}   \mathsf{let}\, \NDnt{t_{{\mathrm{2}}}}  :  \NDnt{X}  \otimes  \NDnt{Y} \,\mathsf{be}\, \NDmv{x}  \otimes  \NDmv{y} \,\mathsf{in}\, \NDsym{(}   \mathsf{let}\, \NDnt{t_{{\mathrm{1}}}}  :   \mathsf{Unit}  \,\mathsf{be}\,  \mathsf{triv}  \,\mathsf{in}\, \NDnt{t_{{\mathrm{3}}}}   \NDsym{)} }{%
{\NDdruleTcomXXtenEXXunitEName}{}%
}}

\newcommand{\NDdruleTcomXXtenEXXtenEName}[0]{\NDdrulename{Tcom\_tenE\_tenE}}
\newcommand{\NDdruleTcomXXtenEXXtenE}[1]{\NDdrule[#1]{%
}{
 \mathsf{let}\, \NDsym{(}   \mathsf{let}\, \NDnt{t_{{\mathrm{2}}}}  :  \NDnt{X_{{\mathrm{2}}}}  \otimes  \NDnt{Y_{{\mathrm{2}}}} \,\mathsf{be}\, \NDmv{x}  \otimes  \NDmv{y} \,\mathsf{in}\, \NDnt{t_{{\mathrm{1}}}}   \NDsym{)}  :  \NDnt{X_{{\mathrm{1}}}}  \otimes  \NDnt{Y_{{\mathrm{1}}}} \,\mathsf{be}\, \NDmv{w}  \otimes  \NDmv{z} \,\mathsf{in}\, \NDnt{t_{{\mathrm{3}}}}   \leadsto_\mathsf{c}   \mathsf{let}\, \NDnt{t_{{\mathrm{2}}}}  :  \NDnt{X_{{\mathrm{2}}}}  \otimes  \NDnt{Y_{{\mathrm{2}}}} \,\mathsf{be}\, \NDmv{x}  \otimes  \NDmv{y} \,\mathsf{in}\, \NDsym{(}   \mathsf{let}\, \NDnt{t_{{\mathrm{1}}}}  :  \NDnt{X_{{\mathrm{1}}}}  \otimes  \NDnt{Y_{{\mathrm{1}}}} \,\mathsf{be}\, \NDmv{w}  \otimes  \NDmv{z} \,\mathsf{in}\, \NDnt{t_{{\mathrm{3}}}}   \NDsym{)} }{%
{\NDdruleTcomXXtenEXXtenEName}{}%
}}

\newcommand{\NDdruleTcomXXtenEXXimpEName}[0]{\NDdrulename{Tcom\_tenE\_impE}}
\newcommand{\NDdruleTcomXXtenEXXimpE}[1]{\NDdrule[#1]{%
}{
 \NDsym{(}   \mathsf{let}\, \NDnt{t_{{\mathrm{2}}}}  :  \NDnt{X_{{\mathrm{2}}}}  \otimes  \NDnt{Y_{{\mathrm{2}}}} \,\mathsf{be}\, \NDmv{x}  \otimes  \NDmv{y} \,\mathsf{in}\, \NDnt{t_{{\mathrm{1}}}}   \NDsym{)}   \NDnt{t_{{\mathrm{3}}}}   \leadsto_\mathsf{c}   \mathsf{let}\, \NDnt{t_{{\mathrm{2}}}}  :  \NDnt{X_{{\mathrm{2}}}}  \otimes  \NDnt{Y_{{\mathrm{2}}}} \,\mathsf{be}\, \NDmv{x}  \otimes  \NDmv{y} \,\mathsf{in}\, \NDsym{(}   \NDnt{t_{{\mathrm{1}}}}   \NDnt{t_{{\mathrm{3}}}}   \NDsym{)} }{%
{\NDdruleTcomXXtenEXXimpEName}{}%
}}

\newcommand{\NDdruleTcomXXimpEXXunitEName}[0]{\NDdrulename{Tcom\_impE\_unitE}}
\newcommand{\NDdruleTcomXXimpEXXunitE}[1]{\NDdrule[#1]{%
}{
 \mathsf{let}\, \NDsym{(}   \NDnt{t_{{\mathrm{1}}}}   \NDnt{t_{{\mathrm{2}}}}   \NDsym{)}  :   \mathsf{Unit}  \,\mathsf{be}\,  \mathsf{triv}  \,\mathsf{in}\, \NDnt{t_{{\mathrm{3}}}}   \leadsto_\mathsf{c}   \NDnt{t_{{\mathrm{1}}}}   \NDsym{(}   \mathsf{let}\, \NDnt{t_{{\mathrm{2}}}}  :   \mathsf{Unit}  \,\mathsf{be}\,  \mathsf{triv}  \,\mathsf{in}\, \NDnt{t_{{\mathrm{3}}}}   \NDsym{)} }{%
{\NDdruleTcomXXimpEXXunitEName}{}%
}}

\newcommand{\NDdruleScomXXunitEXXunitEName}[0]{\NDdrulename{Scom\_unitE\_unitE}}
\newcommand{\NDdruleScomXXunitEXXunitE}[1]{\NDdrule[#1]{%
}{
 \mathsf{let}\, \NDsym{(}   \mathsf{let}\, \NDnt{s_{{\mathrm{2}}}}  :   \mathsf{Unit}  \,\mathsf{be}\,  \mathsf{triv}  \,\mathsf{in}\, \NDnt{s_{{\mathrm{1}}}}   \NDsym{)}  :   \mathsf{Unit}  \,\mathsf{be}\,  \mathsf{triv}  \,\mathsf{in}\, \NDnt{s_{{\mathrm{3}}}}   \leadsto_\mathsf{c}   \mathsf{let}\, \NDnt{s_{{\mathrm{2}}}}  :   \mathsf{Unit}  \,\mathsf{be}\,  \mathsf{triv}  \,\mathsf{in}\, \NDsym{(}   \mathsf{let}\, \NDnt{s_{{\mathrm{1}}}}  :   \mathsf{Unit}  \,\mathsf{be}\,  \mathsf{triv}  \,\mathsf{in}\, \NDnt{s_{{\mathrm{3}}}}   \NDsym{)} }{%
{\NDdruleScomXXunitEXXunitEName}{}%
}}

\newcommand{\NDdruleScomXXunitETwoXXunitEName}[0]{\NDdrulename{Scom\_unitE2\_unitE}}
\newcommand{\NDdruleScomXXunitETwoXXunitE}[1]{\NDdrule[#1]{%
}{
 \mathsf{let}\, \NDsym{(}   \mathsf{let}\, \NDnt{t}  :   \mathsf{Unit}  \,\mathsf{be}\,  \mathsf{triv}  \,\mathsf{in}\, \NDnt{s_{{\mathrm{1}}}}   \NDsym{)}  :   \mathsf{Unit}  \,\mathsf{be}\,  \mathsf{triv}  \,\mathsf{in}\, \NDnt{s_{{\mathrm{2}}}}   \leadsto_\mathsf{c}   \mathsf{let}\, \NDnt{t}  :   \mathsf{Unit}  \,\mathsf{be}\,  \mathsf{triv}  \,\mathsf{in}\, \NDsym{(}   \mathsf{let}\, \NDnt{s_{{\mathrm{1}}}}  :   \mathsf{Unit}  \,\mathsf{be}\,  \mathsf{triv}  \,\mathsf{in}\, \NDnt{s_{{\mathrm{2}}}}   \NDsym{)} }{%
{\NDdruleScomXXunitETwoXXunitEName}{}%
}}

\newcommand{\NDdruleScomXXunitEXXimprEName}[0]{\NDdrulename{Scom\_unitE\_imprE}}
\newcommand{\NDdruleScomXXunitEXXimprE}[1]{\NDdrule[#1]{%
}{
 \mathsf{app}_r\, \NDsym{(}   \mathsf{let}\, \NDnt{s_{{\mathrm{2}}}}  :   \mathsf{Unit}  \,\mathsf{be}\,  \mathsf{triv}  \,\mathsf{in}\, \NDnt{s_{{\mathrm{1}}}}   \NDsym{)} \, \NDnt{s_{{\mathrm{3}}}}   \leadsto_\mathsf{c}   \mathsf{let}\, \NDnt{s_{{\mathrm{2}}}}  :   \mathsf{Unit}  \,\mathsf{be}\,  \mathsf{triv}  \,\mathsf{in}\, \NDsym{(}   \mathsf{app}_r\, \NDnt{s_{{\mathrm{1}}}} \, \NDnt{s_{{\mathrm{3}}}}   \NDsym{)} }{%
{\NDdruleScomXXunitEXXimprEName}{}%
}}

\newcommand{\NDdruleScomXXunitETwoXXimprEName}[0]{\NDdrulename{Scom\_unitE2\_imprE}}
\newcommand{\NDdruleScomXXunitETwoXXimprE}[1]{\NDdrule[#1]{%
}{
 \mathsf{app}_r\, \NDsym{(}   \mathsf{let}\, \NDnt{t}  :   \mathsf{Unit}  \,\mathsf{be}\,  \mathsf{triv}  \,\mathsf{in}\, \NDnt{s_{{\mathrm{1}}}}   \NDsym{)} \, \NDnt{s_{{\mathrm{2}}}}   \leadsto_\mathsf{c}   \mathsf{let}\, \NDnt{t}  :   \mathsf{Unit}  \,\mathsf{be}\,  \mathsf{triv}  \,\mathsf{in}\, \NDsym{(}   \mathsf{app}_r\, \NDnt{s_{{\mathrm{1}}}} \, \NDnt{s_{{\mathrm{2}}}}   \NDsym{)} }{%
{\NDdruleScomXXunitETwoXXimprEName}{}%
}}

\newcommand{\NDdruleScomXXunitEXXFEName}[0]{\NDdrulename{Scom\_unitE\_FE}}
\newcommand{\NDdruleScomXXunitEXXFE}[1]{\NDdrule[#1]{%
}{
 \mathsf{let}\, \NDsym{(}   \mathsf{let}\, \NDnt{s_{{\mathrm{2}}}}  :   \mathsf{Unit}  \,\mathsf{be}\,  \mathsf{triv}  \,\mathsf{in}\, \NDnt{s_{{\mathrm{1}}}}   \NDsym{)}  :   \mathsf{F} \NDnt{X}  \,\mathsf{be}\,  \mathsf{F}\, \NDmv{x}  \,\mathsf{in}\, \NDnt{s_{{\mathrm{3}}}}   \leadsto_\mathsf{c}   \mathsf{let}\, \NDnt{s_{{\mathrm{2}}}}  :   \mathsf{Unit}  \,\mathsf{be}\,  \mathsf{triv}  \,\mathsf{in}\, \NDsym{(}   \mathsf{let}\, \NDnt{s_{{\mathrm{1}}}}  :   \mathsf{F} \NDnt{X}  \,\mathsf{be}\,  \mathsf{F}\, \NDmv{x}  \,\mathsf{in}\, \NDnt{s_{{\mathrm{3}}}}   \NDsym{)} }{%
{\NDdruleScomXXunitEXXFEName}{}%
}}

\newcommand{\NDdruleScomXXunitETwoXXFEName}[0]{\NDdrulename{Scom\_unitE2\_FE}}
\newcommand{\NDdruleScomXXunitETwoXXFE}[1]{\NDdrule[#1]{%
}{
 \mathsf{let}\, \NDsym{(}   \mathsf{let}\, \NDnt{t}  :   \mathsf{Unit}  \,\mathsf{be}\,  \mathsf{triv}  \,\mathsf{in}\, \NDnt{s_{{\mathrm{1}}}}   \NDsym{)}  :   \mathsf{F} \NDnt{X}  \,\mathsf{be}\,  \mathsf{F}\, \NDmv{x}  \,\mathsf{in}\, \NDnt{s_{{\mathrm{2}}}}   \leadsto_\mathsf{c}   \mathsf{let}\, \NDnt{t}  :   \mathsf{Unit}  \,\mathsf{be}\,  \mathsf{triv}  \,\mathsf{in}\, \NDsym{(}   \mathsf{let}\, \NDnt{s_{{\mathrm{1}}}}  :   \mathsf{F} \NDnt{X}  \,\mathsf{be}\,  \mathsf{F}\, \NDmv{x}  \,\mathsf{in}\, \NDnt{s_{{\mathrm{2}}}}   \NDsym{)} }{%
{\NDdruleScomXXunitETwoXXFEName}{}%
}}

\newcommand{\NDdruleScomXXtenEXXunitEName}[0]{\NDdrulename{Scom\_tenE\_unitE}}
\newcommand{\NDdruleScomXXtenEXXunitE}[1]{\NDdrule[#1]{%
}{
 \mathsf{let}\, \NDsym{(}   \mathsf{let}\, \NDnt{s_{{\mathrm{2}}}}  :  \NDnt{A}  \triangleright  \NDnt{B} \,\mathsf{be}\, \NDmv{x}  \triangleright  \NDmv{y} \,\mathsf{in}\, \NDnt{s_{{\mathrm{1}}}}   \NDsym{)}  :   \mathsf{Unit}  \,\mathsf{be}\,  \mathsf{triv}  \,\mathsf{in}\, \NDnt{s_{{\mathrm{3}}}}   \leadsto_\mathsf{c}   \mathsf{let}\, \NDnt{s_{{\mathrm{2}}}}  :  \NDnt{A}  \triangleright  \NDnt{B} \,\mathsf{be}\, \NDmv{x}  \triangleright  \NDmv{y} \,\mathsf{in}\, \NDsym{(}   \mathsf{let}\, \NDnt{s_{{\mathrm{1}}}}  :   \mathsf{Unit}  \,\mathsf{be}\,  \mathsf{triv}  \,\mathsf{in}\, \NDnt{s_{{\mathrm{3}}}}   \NDsym{)} }{%
{\NDdruleScomXXtenEXXunitEName}{}%
}}

\newcommand{\NDdruleScomXXtenETwoXXunitEName}[0]{\NDdrulename{Scom\_tenE2\_unitE}}
\newcommand{\NDdruleScomXXtenETwoXXunitE}[1]{\NDdrule[#1]{%
}{
 \mathsf{let}\, \NDsym{(}   \mathsf{let}\, \NDnt{t}  :  \NDnt{X}  \otimes  \NDnt{Y} \,\mathsf{be}\, \NDmv{x}  \otimes  \NDmv{y} \,\mathsf{in}\, \NDnt{s_{{\mathrm{1}}}}   \NDsym{)}  :   \mathsf{Unit}  \,\mathsf{be}\,  \mathsf{triv}  \,\mathsf{in}\, \NDnt{s_{{\mathrm{2}}}}   \leadsto_\mathsf{c}   \mathsf{let}\, \NDnt{t}  :  \NDnt{X}  \otimes  \NDnt{Y} \,\mathsf{be}\, \NDmv{x}  \otimes  \NDmv{y} \,\mathsf{in}\, \NDsym{(}   \mathsf{let}\, \NDnt{s_{{\mathrm{1}}}}  :   \mathsf{Unit}  \,\mathsf{be}\,  \mathsf{triv}  \,\mathsf{in}\, \NDnt{s_{{\mathrm{2}}}}   \NDsym{)} }{%
{\NDdruleScomXXtenETwoXXunitEName}{}%
}}

\newcommand{\NDdruleScomXXtenEXXtenEName}[0]{\NDdrulename{Scom\_tenE\_tenE}}
\newcommand{\NDdruleScomXXtenEXXtenE}[1]{\NDdrule[#1]{%
}{
 \mathsf{let}\, \NDsym{(}   \mathsf{let}\, \NDnt{s_{{\mathrm{2}}}}  :  \NDnt{A_{{\mathrm{2}}}}  \triangleright  \NDnt{B_{{\mathrm{2}}}} \,\mathsf{be}\, \NDmv{x}  \triangleright  \NDmv{y} \,\mathsf{in}\, \NDnt{s_{{\mathrm{1}}}}   \NDsym{)}  :  \NDnt{A_{{\mathrm{1}}}}  \triangleright  \NDnt{B_{{\mathrm{1}}}} \,\mathsf{be}\, \NDmv{w}  \triangleright  \NDmv{z} \,\mathsf{in}\, \NDnt{s_{{\mathrm{3}}}}   \leadsto_\mathsf{c}   \mathsf{let}\, \NDnt{s_{{\mathrm{2}}}}  :  \NDnt{A_{{\mathrm{2}}}}  \triangleright  \NDnt{B_{{\mathrm{2}}}} \,\mathsf{be}\, \NDmv{x}  \triangleright  \NDmv{y} \,\mathsf{in}\, \NDsym{(}   \mathsf{let}\, \NDnt{s_{{\mathrm{1}}}}  :  \NDnt{A_{{\mathrm{1}}}}  \triangleright  \NDnt{B_{{\mathrm{1}}}} \,\mathsf{be}\, \NDmv{w}  \triangleright  \NDmv{z} \,\mathsf{in}\, \NDnt{s_{{\mathrm{3}}}}   \NDsym{)} }{%
{\NDdruleScomXXtenEXXtenEName}{}%
}}

\newcommand{\NDdruleScomXXtenETwoXXtenEName}[0]{\NDdrulename{Scom\_tenE2\_tenE}}
\newcommand{\NDdruleScomXXtenETwoXXtenE}[1]{\NDdrule[#1]{%
}{
 \mathsf{let}\, \NDsym{(}   \mathsf{let}\, \NDnt{t}  :  \NDnt{X}  \otimes  \NDnt{Y} \,\mathsf{be}\, \NDmv{x}  \otimes  \NDmv{y} \,\mathsf{in}\, \NDnt{s_{{\mathrm{1}}}}   \NDsym{)}  :  \NDnt{A_{{\mathrm{1}}}}  \triangleright  \NDnt{B_{{\mathrm{1}}}} \,\mathsf{be}\, \NDmv{w}  \triangleright  \NDmv{z} \,\mathsf{in}\, \NDnt{s_{{\mathrm{2}}}}   \leadsto_\mathsf{c}   \mathsf{let}\, \NDnt{t}  :  \NDnt{X}  \otimes  \NDnt{Y} \,\mathsf{be}\, \NDmv{x}  \otimes  \NDmv{y} \,\mathsf{in}\, \NDsym{(}   \mathsf{let}\, \NDnt{s_{{\mathrm{1}}}}  :  \NDnt{A_{{\mathrm{1}}}}  \triangleright  \NDnt{B_{{\mathrm{1}}}} \,\mathsf{be}\, \NDmv{w}  \triangleright  \NDmv{z} \,\mathsf{in}\, \NDnt{s_{{\mathrm{2}}}}   \NDsym{)} }{%
{\NDdruleScomXXtenETwoXXtenEName}{}%
}}

\newcommand{\NDdruleScomXXtenEXXimprEName}[0]{\NDdrulename{Scom\_tenE\_imprE}}
\newcommand{\NDdruleScomXXtenEXXimprE}[1]{\NDdrule[#1]{%
}{
 \mathsf{app}_r\, \NDsym{(}   \mathsf{let}\, \NDnt{s_{{\mathrm{2}}}}  :  \NDnt{A_{{\mathrm{2}}}}  \triangleright  \NDnt{B_{{\mathrm{2}}}} \,\mathsf{be}\, \NDmv{x}  \triangleright  \NDmv{y} \,\mathsf{in}\, \NDnt{s_{{\mathrm{1}}}}   \NDsym{)} \, \NDnt{s_{{\mathrm{3}}}}   \leadsto_\mathsf{c}   \mathsf{let}\, \NDnt{s_{{\mathrm{2}}}}  :  \NDnt{A_{{\mathrm{2}}}}  \triangleright  \NDnt{B_{{\mathrm{2}}}} \,\mathsf{be}\, \NDmv{x}  \triangleright  \NDmv{y} \,\mathsf{in}\, \NDsym{(}   \mathsf{app}_r\, \NDnt{s_{{\mathrm{1}}}} \, \NDnt{s_{{\mathrm{3}}}}   \NDsym{)} }{%
{\NDdruleScomXXtenEXXimprEName}{}%
}}

\newcommand{\NDdruleScomXXtenETwoXXimprEName}[0]{\NDdrulename{Scom\_tenE2\_imprE}}
\newcommand{\NDdruleScomXXtenETwoXXimprE}[1]{\NDdrule[#1]{%
}{
 \mathsf{app}_r\, \NDsym{(}   \mathsf{let}\, \NDnt{t}  :  \NDnt{X}  \otimes  \NDnt{Y} \,\mathsf{be}\, \NDmv{x}  \otimes  \NDmv{y} \,\mathsf{in}\, \NDnt{s_{{\mathrm{1}}}}   \NDsym{)} \, \NDnt{s_{{\mathrm{2}}}}   \leadsto_\mathsf{c}   \mathsf{let}\, \NDnt{t}  :  \NDnt{X}  \otimes  \NDnt{Y} \,\mathsf{be}\, \NDmv{x}  \otimes  \NDmv{y} \,\mathsf{in}\, \NDsym{(}   \mathsf{app}_r\, \NDnt{s_{{\mathrm{1}}}} \, \NDnt{s_{{\mathrm{2}}}}   \NDsym{)} }{%
{\NDdruleScomXXtenETwoXXimprEName}{}%
}}

\newcommand{\NDdruleScomXXtenEXXimplEName}[0]{\NDdrulename{Scom\_tenE\_implE}}
\newcommand{\NDdruleScomXXtenEXXimplE}[1]{\NDdrule[#1]{%
}{
 \mathsf{app}_l\, \NDsym{(}   \mathsf{let}\, \NDnt{s_{{\mathrm{2}}}}  :  \NDnt{A_{{\mathrm{2}}}}  \triangleright  \NDnt{B_{{\mathrm{2}}}} \,\mathsf{be}\, \NDmv{x}  \triangleright  \NDmv{y} \,\mathsf{in}\, \NDnt{s_{{\mathrm{1}}}}   \NDsym{)} \, \NDnt{s_{{\mathrm{3}}}}   \leadsto_\mathsf{c}   \mathsf{let}\, \NDnt{s_{{\mathrm{2}}}}  :  \NDnt{A_{{\mathrm{2}}}}  \triangleright  \NDnt{B_{{\mathrm{2}}}} \,\mathsf{be}\, \NDmv{x}  \triangleright  \NDmv{y} \,\mathsf{in}\, \NDsym{(}   \mathsf{app}_l\, \NDnt{s_{{\mathrm{1}}}} \, \NDnt{s_{{\mathrm{3}}}}   \NDsym{)} }{%
{\NDdruleScomXXtenEXXimplEName}{}%
}}

\newcommand{\NDdruleScomXXtenETwoXXimplEName}[0]{\NDdrulename{Scom\_tenE2\_implE}}
\newcommand{\NDdruleScomXXtenETwoXXimplE}[1]{\NDdrule[#1]{%
}{
 \mathsf{app}_l\, \NDsym{(}   \mathsf{let}\, \NDnt{t}  :  \NDnt{X}  \otimes  \NDnt{Y} \,\mathsf{be}\, \NDmv{x}  \otimes  \NDmv{y} \,\mathsf{in}\, \NDnt{s_{{\mathrm{1}}}}   \NDsym{)} \, \NDnt{s_{{\mathrm{2}}}}   \leadsto_\mathsf{c}   \mathsf{let}\, \NDnt{t}  :  \NDnt{X}  \otimes  \NDnt{Y} \,\mathsf{be}\, \NDmv{x}  \otimes  \NDmv{y} \,\mathsf{in}\, \NDsym{(}   \mathsf{app}_l\, \NDnt{s_{{\mathrm{1}}}} \, \NDnt{s_{{\mathrm{2}}}}   \NDsym{)} }{%
{\NDdruleScomXXtenETwoXXimplEName}{}%
}}

\newcommand{\NDdruleScomXXtenEXXFEName}[0]{\NDdrulename{Scom\_tenE\_FE}}
\newcommand{\NDdruleScomXXtenEXXFE}[1]{\NDdrule[#1]{%
}{
 \mathsf{let}\, \NDsym{(}   \mathsf{let}\, \NDnt{s_{{\mathrm{2}}}}  :  \NDnt{A}  \triangleright  \NDnt{B} \,\mathsf{be}\, \NDmv{x}  \triangleright  \NDmv{y} \,\mathsf{in}\, \NDnt{s_{{\mathrm{1}}}}   \NDsym{)}  :   \mathsf{F} \NDnt{X}  \,\mathsf{be}\,  \mathsf{F}\, \NDmv{z}  \,\mathsf{in}\, \NDnt{s_{{\mathrm{3}}}}   \leadsto_\mathsf{c}   \mathsf{let}\, \NDnt{s_{{\mathrm{2}}}}  :  \NDnt{A}  \triangleright  \NDnt{B} \,\mathsf{be}\, \NDmv{x}  \triangleright  \NDmv{y} \,\mathsf{in}\, \NDsym{(}   \mathsf{let}\, \NDnt{s_{{\mathrm{1}}}}  :   \mathsf{F} \NDnt{X}  \,\mathsf{be}\,  \mathsf{F}\, \NDmv{z}  \,\mathsf{in}\, \NDnt{s_{{\mathrm{3}}}}   \NDsym{)} }{%
{\NDdruleScomXXtenEXXFEName}{}%
}}

\newcommand{\NDdruleScomXXtenETwoXXFEName}[0]{\NDdrulename{Scom\_tenE2\_FE}}
\newcommand{\NDdruleScomXXtenETwoXXFE}[1]{\NDdrule[#1]{%
}{
 \mathsf{let}\, \NDsym{(}   \mathsf{let}\, \NDnt{t}  :  \NDnt{X}  \otimes  \NDnt{Y} \,\mathsf{be}\, \NDmv{x}  \otimes  \NDmv{y} \,\mathsf{in}\, \NDnt{s_{{\mathrm{1}}}}   \NDsym{)}  :   \mathsf{F} \NDnt{Z}  \,\mathsf{be}\,  \mathsf{F}\, \NDmv{z}  \,\mathsf{in}\, \NDnt{s_{{\mathrm{3}}}}   \leadsto_\mathsf{c}   \mathsf{let}\, \NDnt{t}  :  \NDnt{X}  \otimes  \NDnt{Y} \,\mathsf{be}\, \NDmv{x}  \otimes  \NDmv{y} \,\mathsf{in}\, \NDsym{(}   \mathsf{let}\, \NDnt{s_{{\mathrm{1}}}}  :   \mathsf{F} \NDnt{Z}  \,\mathsf{be}\,  \mathsf{F}\, \NDmv{z}  \,\mathsf{in}\, \NDnt{s_{{\mathrm{3}}}}   \NDsym{)} }{%
{\NDdruleScomXXtenETwoXXFEName}{}%
}}

\newcommand{\NDdruleScomXXFEXXunitEName}[0]{\NDdrulename{Scom\_FE\_unitE}}
\newcommand{\NDdruleScomXXFEXXunitE}[1]{\NDdrule[#1]{%
}{
 \mathsf{let}\, \NDsym{(}   \mathsf{let}\, \NDnt{s_{{\mathrm{2}}}}  :   \mathsf{F} \NDnt{X}  \,\mathsf{be}\,  \mathsf{F}\, \NDmv{x}  \,\mathsf{in}\, \NDnt{s_{{\mathrm{1}}}}   \NDsym{)}  :   \mathsf{Unit}  \,\mathsf{be}\,  \mathsf{triv}  \,\mathsf{in}\, \NDnt{s_{{\mathrm{3}}}}   \leadsto_\mathsf{c}   \mathsf{let}\, \NDnt{s_{{\mathrm{2}}}}  :   \mathsf{F} \NDnt{X}  \,\mathsf{be}\,  \mathsf{F}\, \NDmv{x}  \,\mathsf{in}\, \NDsym{(}   \mathsf{let}\, \NDnt{s_{{\mathrm{1}}}}  :   \mathsf{Unit}  \,\mathsf{be}\,  \mathsf{triv}  \,\mathsf{in}\, \NDnt{s_{{\mathrm{2}}}}   \NDsym{)} }{%
{\NDdruleScomXXFEXXunitEName}{}%
}}

\newcommand{\NDdruleScomXXFEXXtenEName}[0]{\NDdrulename{Scom\_FE\_tenE}}
\newcommand{\NDdruleScomXXFEXXtenE}[1]{\NDdrule[#1]{%
}{
 \mathsf{let}\, \NDsym{(}   \mathsf{let}\, \NDnt{s_{{\mathrm{2}}}}  :   \mathsf{F} \NDnt{X}  \,\mathsf{be}\,  \mathsf{F}\, \NDmv{x}  \,\mathsf{in}\, \NDnt{s_{{\mathrm{1}}}}   \NDsym{)}  :  \NDnt{A}  \triangleright  \NDnt{B} \,\mathsf{be}\, \NDmv{x}  \triangleright  \NDmv{y} \,\mathsf{in}\, \NDnt{s_{{\mathrm{3}}}}   \leadsto_\mathsf{c}   \mathsf{let}\, \NDnt{s_{{\mathrm{2}}}}  :   \mathsf{F} \NDnt{X}  \,\mathsf{be}\,  \mathsf{F}\, \NDmv{x}  \,\mathsf{in}\, \NDsym{(}   \mathsf{let}\, \NDnt{s_{{\mathrm{1}}}}  :  \NDnt{A}  \triangleright  \NDnt{B} \,\mathsf{be}\, \NDmv{x}  \triangleright  \NDmv{y} \,\mathsf{in}\, \NDnt{s_{{\mathrm{3}}}}   \NDsym{)} }{%
{\NDdruleScomXXFEXXtenEName}{}%
}}

\newcommand{\NDdruleScomXXFEXXimprEName}[0]{\NDdrulename{Scom\_FE\_imprE}}
\newcommand{\NDdruleScomXXFEXXimprE}[1]{\NDdrule[#1]{%
}{
 \mathsf{app}_r\, \NDsym{(}   \mathsf{let}\, \NDnt{s_{{\mathrm{2}}}}  :   \mathsf{F} \NDnt{X}  \,\mathsf{be}\,  \mathsf{F}\, \NDmv{x}  \,\mathsf{in}\, \NDnt{s_{{\mathrm{1}}}}   \NDsym{)} \, \NDnt{s_{{\mathrm{3}}}}   \leadsto_\mathsf{c}   \mathsf{let}\, \NDnt{s_{{\mathrm{2}}}}  :   \mathsf{F} \NDnt{X}  \,\mathsf{be}\,  \mathsf{F}\, \NDmv{x}  \,\mathsf{in}\, \NDsym{(}   \mathsf{app}_r\, \NDnt{s_{{\mathrm{1}}}} \, \NDnt{s_{{\mathrm{3}}}}   \NDsym{)} }{%
{\NDdruleScomXXFEXXimprEName}{}%
}}

\newcommand{\NDdruleScomXXFEXXimplEName}[0]{\NDdrulename{Scom\_FE\_implE}}
\newcommand{\NDdruleScomXXFEXXimplE}[1]{\NDdrule[#1]{%
}{
 \mathsf{app}_l\, \NDsym{(}   \mathsf{let}\, \NDnt{s_{{\mathrm{2}}}}  :   \mathsf{F} \NDnt{X}  \,\mathsf{be}\,  \mathsf{F}\, \NDmv{x}  \,\mathsf{in}\, \NDnt{s_{{\mathrm{1}}}}   \NDsym{)} \, \NDnt{s_{{\mathrm{3}}}}   \leadsto_\mathsf{c}   \mathsf{let}\, \NDnt{s_{{\mathrm{2}}}}  :   \mathsf{F} \NDnt{X}  \,\mathsf{be}\,  \mathsf{F}\, \NDmv{x}  \,\mathsf{in}\, \NDsym{(}   \mathsf{app}_l\, \NDnt{s_{{\mathrm{1}}}} \, \NDnt{s_{{\mathrm{3}}}}   \NDsym{)} }{%
{\NDdruleScomXXFEXXimplEName}{}%
}}

\newcommand{\NDdruleScomXXFEXXFEName}[0]{\NDdrulename{Scom\_FE\_FE}}
\newcommand{\NDdruleScomXXFEXXFE}[1]{\NDdrule[#1]{%
}{
 \mathsf{let}\, \NDsym{(}   \mathsf{let}\, \NDnt{s_{{\mathrm{2}}}}  :   \mathsf{F} \NDnt{X}  \,\mathsf{be}\,  \mathsf{F}\, \NDmv{x}  \,\mathsf{in}\, \NDnt{s_{{\mathrm{1}}}}   \NDsym{)}  :   \mathsf{F} \NDnt{Y}  \,\mathsf{be}\,  \mathsf{F}\, \NDmv{y}  \,\mathsf{in}\, \NDnt{s_{{\mathrm{3}}}}   \leadsto_\mathsf{c}   \mathsf{let}\, \NDnt{s_{{\mathrm{2}}}}  :   \mathsf{F} \NDnt{X}  \,\mathsf{be}\,  \mathsf{F}\, \NDmv{x}  \,\mathsf{in}\, \NDsym{(}   \mathsf{let}\, \NDnt{s_{{\mathrm{1}}}}  :   \mathsf{F} \NDnt{Y}  \,\mathsf{be}\,  \mathsf{F}\, \NDmv{y}  \,\mathsf{in}\, \NDnt{s_{{\mathrm{3}}}}   \NDsym{)} }{%
{\NDdruleScomXXFEXXFEName}{}%
}}

\renewcommand{\NDdrule}[4][]{{\displaystyle\frac{\begin{array}{l}#2\end{array}}{#3}\,#4}}
\renewcommand{\NDdruleTXXidName}{\mathcal{C}\text{-ax}}
\renewcommand{\NDdruleTXXunitIName}{\mathcal{C}\text{-}\mathsf{Unit}_I}
\renewcommand{\NDdruleTXXunitEName}{\mathcal{C}\text{-}\mathsf{Unit}_E}
\renewcommand{\NDdruleTXXtenIName}{\mathcal{C}\text{-}\otimes_I}
\renewcommand{\NDdruleTXXtenEName}{\mathcal{C}\text{-}\otimes_E}
\renewcommand{\NDdruleTXXimpIName}{\mathcal{C}\text{-}\multimap_I}
\renewcommand{\NDdruleTXXimpEName}{\mathcal{C}\text{-}\multimap_E}
\renewcommand{\NDdruleTXXGIName}{\mathcal{C}\text{-}\mathsf{G}_I}
\renewcommand{\NDdruleTXXbetaName}{\mathcal{C}\text{-}\mathsf{ex}}
\renewcommand{\NDdruleTXXcutName}{\mathcal{C}\text{-}\mathsf{Cut}}
\renewcommand{\NDdruleSXXidName}{\mathcal{L}\text{-ax}}
\renewcommand{\NDdruleSXXunitIName}{\mathcal{L}\text{-}\mathsf{Unit}_I}
\renewcommand{\NDdruleSXXunitEOneName}{\mathcal{LC}\text{-}\mathsf{Unit}_E}
\renewcommand{\NDdruleSXXunitETwoName}{\mathcal{L}\text{-}\mathsf{Unit}_E}
\renewcommand{\NDdruleSXXtenIName}{\mathcal{L}\text{-}\otimes_I}
\renewcommand{\NDdruleSXXtenEOneName}{\mathcal{LC}\text{-}\otimes_E}
\renewcommand{\NDdruleSXXtenETwoName}{\mathcal{L}\text{-}\otimes_E}
\renewcommand{\NDdruleSXXimprIName}{\mathcal{L}\text{-}\rightharpoonup_I}
\renewcommand{\NDdruleSXXimprEName}{\mathcal{L}\text{-}\rightharpoonup_E}
\renewcommand{\NDdruleSXXimplIName}{\mathcal{L}\text{-}\leftharpoonup_I}
\renewcommand{\NDdruleSXXimplEName}{\mathcal{L}\text{-}\leftharpoonup_E}
\renewcommand{\NDdruleSXXFIName}{\mathcal{L}\text{-}\mathsf{F}_I}
\renewcommand{\NDdruleSXXFEName}{\mathcal{L}\text{-}\mathsf{F}_E}
\renewcommand{\NDdruleSXXGEName}{\mathcal{L}\text{-}\mathsf{G}_E}
\renewcommand{\NDdruleSXXbetaName}{\mathcal{L}\text{-}\mathsf{ex}}
\renewcommand{\NDdruleSXXcutOneName}{\mathcal{LC}\text{-}\mathsf{Cut}}
\renewcommand{\NDdruleSXXcutTwoName}{\mathcal{L}\text{-}\mathsf{Cut}}

\newcommand{\SCdrule}[4][]{{\displaystyle\frac{\begin{array}{l}#2\end{array}}{#3}\quad\SCdrulename{#4}}}

\newcommand{\SCpremise}[1]{ #1 \\}
\newenvironment{SCdefnblock}[3][]{ \framebox{\mbox{#2}} \quad #3 \\[0pt]}{}

\newcommand{\SCnt}[1]{\mathit{#1}}

\newcommand{\SCsym}[1]{#1}

\newcommand{\SCdruleTXXaxName}[0]{\SCdrulename{T\_ax}}
\newcommand{\SCdruleTXXax}[1]{\SCdrule[#1]{%
}{
\SCnt{X}  \vdash_\mathcal{C}  \SCnt{X}}{%
{\SCdruleTXXaxName}{}%
}}

\newcommand{\SCdruleTXXunitLName}[0]{\SCdrulename{T\_unitL}}
\newcommand{\SCdruleTXXunitL}[1]{\SCdrule[#1]{%
\SCpremise{\Phi  \SCsym{,}  \Psi  \vdash_\mathcal{C}  \SCnt{X}}%
}{
\Phi  \SCsym{,}   \mathsf{Unit}   \SCsym{,}  \Psi  \vdash_\mathcal{C}  \SCnt{X}}{%
{\SCdruleTXXunitLName}{}%
}}

\newcommand{\SCdruleTXXunitRName}[0]{\SCdrulename{T\_unitR}}
\newcommand{\SCdruleTXXunitR}[1]{\SCdrule[#1]{%
}{
 \cdot   \vdash_\mathcal{C}   \mathsf{Unit} }{%
{\SCdruleTXXunitRName}{}%
}}

\newcommand{\SCdruleTXXtenLName}[0]{\SCdrulename{T\_tenL}}
\newcommand{\SCdruleTXXtenL}[1]{\SCdrule[#1]{%
\SCpremise{\Phi  \SCsym{,}  \SCnt{X}  \SCsym{,}  \SCnt{Y}  \SCsym{,}  \Psi  \vdash_\mathcal{C}  \SCnt{Z}}%
}{
\Phi  \SCsym{,}  \SCnt{X}  \otimes  \SCnt{Y}  \SCsym{,}  \Psi  \vdash_\mathcal{C}  \SCnt{Z}}{%
{\SCdruleTXXtenLName}{}%
}}

\newcommand{\SCdruleTXXtenRName}[0]{\SCdrulename{T\_tenR}}
\newcommand{\SCdruleTXXtenR}[1]{\SCdrule[#1]{%
\SCpremise{ \Phi  \vdash_\mathcal{C}  \SCnt{X}  \quad  \Psi  \vdash_\mathcal{C}  \SCnt{Y} }%
}{
\Phi  \SCsym{,}  \Psi  \vdash_\mathcal{C}  \SCnt{X}  \otimes  \SCnt{Y}}{%
{\SCdruleTXXtenRName}{}%
}}

\newcommand{\SCdruleTXXimpLName}[0]{\SCdrulename{T\_impL}}
\newcommand{\SCdruleTXXimpL}[1]{\SCdrule[#1]{%
\SCpremise{ \Phi  \vdash_\mathcal{C}  \SCnt{X}  \quad  \Psi_{{\mathrm{1}}}  \SCsym{,}  \SCnt{Y}  \SCsym{,}  \Psi_{{\mathrm{2}}}  \vdash_\mathcal{C}  \SCnt{Z} }%
}{
\Psi_{{\mathrm{1}}}  \SCsym{,}  \SCnt{X}  \multimap  \SCnt{Y}  \SCsym{,}  \Phi  \SCsym{,}  \Psi_{{\mathrm{2}}}  \vdash_\mathcal{C}  \SCnt{Z}}{%
{\SCdruleTXXimpLName}{}%
}}

\newcommand{\SCdruleTXXimpRName}[0]{\SCdrulename{T\_impR}}
\newcommand{\SCdruleTXXimpR}[1]{\SCdrule[#1]{%
\SCpremise{\Phi  \SCsym{,}  \SCnt{X}  \SCsym{,}  \Psi  \vdash_\mathcal{C}  \SCnt{Y}}%
}{
\Phi  \SCsym{,}  \Psi  \vdash_\mathcal{C}  \SCnt{X}  \multimap  \SCnt{Y}}{%
{\SCdruleTXXimpRName}{}%
}}

\newcommand{\SCdruleTXXGrName}[0]{\SCdrulename{T\_Gr}}
\newcommand{\SCdruleTXXGr}[1]{\SCdrule[#1]{%
\SCpremise{\Phi  \vdash_\mathcal{L}  \SCnt{A}}%
}{
\Phi  \vdash_\mathcal{C}   \mathsf{G} \SCnt{A} }{%
{\SCdruleTXXGrName}{}%
}}

\newcommand{\SCdruleTXXcutName}[0]{\SCdrulename{T\_cut}}
\newcommand{\SCdruleTXXcut}[1]{\SCdrule[#1]{%
\SCpremise{ \Phi  \vdash_\mathcal{C}  \SCnt{X}  \quad  \Psi_{{\mathrm{1}}}  \SCsym{,}  \SCnt{X}  \SCsym{,}  \Psi_{{\mathrm{2}}}  \vdash_\mathcal{C}  \SCnt{Y} }%
}{
\Psi_{{\mathrm{1}}}  \SCsym{,}  \Phi  \SCsym{,}  \Psi_{{\mathrm{2}}}  \vdash_\mathcal{C}  \SCnt{Y}}{%
{\SCdruleTXXcutName}{}%
}}

\newcommand{\SCdruleTXXexName}[0]{\SCdrulename{T\_ex}}
\newcommand{\SCdruleTXXex}[1]{\SCdrule[#1]{%
\SCpremise{\Phi  \SCsym{,}  \SCnt{X}  \SCsym{,}  \SCnt{Y}  \SCsym{,}  \Psi  \vdash_\mathcal{C}  \SCnt{Z}}%
}{
\Phi  \SCsym{,}  \SCnt{Y}  \SCsym{,}  \SCnt{X}  \SCsym{,}  \Psi  \vdash_\mathcal{C}  \SCnt{Z}}{%
{\SCdruleTXXexName}{}%
}}

\newcommand{\SCdruleSXXaxName}[0]{\SCdrulename{S\_ax}}
\newcommand{\SCdruleSXXax}[1]{\SCdrule[#1]{%
}{
\SCnt{A}  \vdash_\mathcal{L}  \SCnt{A}}{%
{\SCdruleSXXaxName}{}%
}}

\newcommand{\SCdruleSXXunitLOneName}[0]{\SCdrulename{S\_unitL1}}
\newcommand{\SCdruleSXXunitLOne}[1]{\SCdrule[#1]{%
\SCpremise{\Gamma  \SCsym{;}  \Delta  \vdash_\mathcal{L}  \SCnt{A}}%
}{
\Gamma  \SCsym{;}   \mathsf{Unit}   \SCsym{;}  \Delta  \vdash_\mathcal{L}  \SCnt{A}}{%
{\SCdruleSXXunitLOneName}{}%
}}

\newcommand{\SCdruleSXXunitLTwoName}[0]{\SCdrulename{S\_unitL2}}
\newcommand{\SCdruleSXXunitLTwo}[1]{\SCdrule[#1]{%
\SCpremise{\Gamma  \SCsym{;}  \Delta  \vdash_\mathcal{L}  \SCnt{A}}%
}{
\Gamma  \SCsym{;}   \mathsf{Unit}   \SCsym{;}  \Delta  \vdash_\mathcal{L}  \SCnt{A}}{%
{\SCdruleSXXunitLTwoName}{}%
}}

\newcommand{\SCdruleSXXunitRName}[0]{\SCdrulename{S\_unitR}}
\newcommand{\SCdruleSXXunitR}[1]{\SCdrule[#1]{%
}{
 \cdot   \vdash_\mathcal{L}   \mathsf{Unit} }{%
{\SCdruleSXXunitRName}{}%
}}

\newcommand{\SCdruleSXXexName}[0]{\SCdrulename{S\_ex}}
\newcommand{\SCdruleSXXex}[1]{\SCdrule[#1]{%
\SCpremise{\Gamma  \SCsym{;}  \SCnt{X}  \SCsym{;}  \SCnt{Y}  \SCsym{;}  \Delta  \vdash_\mathcal{L}  \SCnt{A}}%
}{
\Gamma  \SCsym{;}  \SCnt{Y}  \SCsym{;}  \SCnt{X}  \SCsym{;}  \Delta  \vdash_\mathcal{L}  \SCnt{A}}{%
{\SCdruleSXXexName}{}%
}}

\newcommand{\SCdruleSXXtenLOneName}[0]{\SCdrulename{S\_tenL1}}
\newcommand{\SCdruleSXXtenLOne}[1]{\SCdrule[#1]{%
\SCpremise{\Gamma  \SCsym{;}  \SCnt{X}  \SCsym{;}  \SCnt{Y}  \SCsym{;}  \Delta  \vdash_\mathcal{L}  \SCnt{A}}%
}{
\Gamma  \SCsym{;}  \SCnt{X}  \otimes  \SCnt{Y}  \SCsym{;}  \Delta  \vdash_\mathcal{L}  \SCnt{A}}{%
{\SCdruleSXXtenLOneName}{}%
}}

\newcommand{\SCdruleSXXtenLTwoName}[0]{\SCdrulename{S\_tenL2}}
\newcommand{\SCdruleSXXtenLTwo}[1]{\SCdrule[#1]{%
\SCpremise{\Gamma  \SCsym{;}  \SCnt{A}  \SCsym{;}  \SCnt{B}  \SCsym{;}  \Delta  \vdash_\mathcal{L}  \SCnt{C}}%
}{
\Gamma  \SCsym{;}  \SCnt{A}  \triangleright  \SCnt{B}  \SCsym{;}  \Delta  \vdash_\mathcal{L}  \SCnt{C}}{%
{\SCdruleSXXtenLTwoName}{}%
}}

\newcommand{\SCdruleSXXtenRName}[0]{\SCdrulename{S\_tenR}}
\newcommand{\SCdruleSXXtenR}[1]{\SCdrule[#1]{%
\SCpremise{ \Gamma  \vdash_\mathcal{L}  \SCnt{A}  \quad  \Delta  \vdash_\mathcal{L}  \SCnt{B} }%
}{
\Gamma  \SCsym{;}  \Delta  \vdash_\mathcal{L}  \SCnt{A}  \triangleright  \SCnt{B}}{%
{\SCdruleSXXtenRName}{}%
}}

\newcommand{\SCdruleSXXimpLName}[0]{\SCdrulename{S\_impL}}
\newcommand{\SCdruleSXXimpL}[1]{\SCdrule[#1]{%
\SCpremise{ \Phi  \vdash_\mathcal{C}  \SCnt{X}  \quad  \Gamma  \SCsym{;}  \SCnt{Y}  \SCsym{;}  \Delta  \vdash_\mathcal{L}  \SCnt{A} }%
}{
\Gamma  \SCsym{;}  \SCnt{X}  \multimap  \SCnt{Y}  \SCsym{;}  \Phi  \SCsym{;}  \Delta  \vdash_\mathcal{L}  \SCnt{A}}{%
{\SCdruleSXXimpLName}{}%
}}

\newcommand{\SCdruleSXXimprLName}[0]{\SCdrulename{S\_imprL}}
\newcommand{\SCdruleSXXimprL}[1]{\SCdrule[#1]{%
\SCpremise{ \Gamma  \vdash_\mathcal{L}  \SCnt{A}  \quad  \Delta_{{\mathrm{1}}}  \SCsym{;}  \SCnt{B}  \SCsym{;}  \Delta_{{\mathrm{2}}}  \vdash_\mathcal{L}  \SCnt{C} }%
}{
\Delta_{{\mathrm{1}}}  \SCsym{;}  \SCnt{A}  \rightharpoonup  \SCnt{B}  \SCsym{;}  \Gamma  \SCsym{;}  \Delta_{{\mathrm{2}}}  \vdash_\mathcal{L}  \SCnt{C}}{%
{\SCdruleSXXimprLName}{}%
}}

\newcommand{\SCdruleSXXimprRName}[0]{\SCdrulename{S\_imprR}}
\newcommand{\SCdruleSXXimprR}[1]{\SCdrule[#1]{%
\SCpremise{\Gamma  \SCsym{;}  \SCnt{A}  \vdash_\mathcal{L}  \SCnt{B}}%
}{
\Gamma  \vdash_\mathcal{L}  \SCnt{A}  \rightharpoonup  \SCnt{B}}{%
{\SCdruleSXXimprRName}{}%
}}

\newcommand{\SCdruleSXXimplLName}[0]{\SCdrulename{S\_implL}}
\newcommand{\SCdruleSXXimplL}[1]{\SCdrule[#1]{%
\SCpremise{ \Gamma  \vdash_\mathcal{L}  \SCnt{A}  \quad  \Delta_{{\mathrm{1}}}  \SCsym{;}  \SCnt{B}  \SCsym{;}  \Delta_{{\mathrm{2}}}  \vdash_\mathcal{L}  \SCnt{C} }%
}{
\Delta_{{\mathrm{1}}}  \SCsym{;}  \Gamma  \SCsym{;}  \SCnt{B}  \leftharpoonup  \SCnt{A}  \SCsym{;}  \Delta_{{\mathrm{2}}}  \vdash_\mathcal{L}  \SCnt{C}}{%
{\SCdruleSXXimplLName}{}%
}}

\newcommand{\SCdruleSXXimplRName}[0]{\SCdrulename{S\_implR}}
\newcommand{\SCdruleSXXimplR}[1]{\SCdrule[#1]{%
\SCpremise{\SCnt{A}  \SCsym{;}  \Gamma  \vdash_\mathcal{L}  \SCnt{B}}%
}{
\Gamma  \vdash_\mathcal{L}  \SCnt{B}  \leftharpoonup  \SCnt{A}}{%
{\SCdruleSXXimplRName}{}%
}}

\newcommand{\SCdruleSXXFlName}[0]{\SCdrulename{S\_Fl}}
\newcommand{\SCdruleSXXFl}[1]{\SCdrule[#1]{%
\SCpremise{\Gamma  \SCsym{;}  \SCnt{X}  \SCsym{;}  \Delta  \vdash_\mathcal{L}  \SCnt{A}}%
}{
\Gamma  \SCsym{;}   \mathsf{F} \SCnt{X}   \SCsym{;}  \Delta  \vdash_\mathcal{L}  \SCnt{A}}{%
{\SCdruleSXXFlName}{}%
}}

\newcommand{\SCdruleSXXFrName}[0]{\SCdrulename{S\_Fr}}
\newcommand{\SCdruleSXXFr}[1]{\SCdrule[#1]{%
\SCpremise{\Phi  \vdash_\mathcal{C}  \SCnt{X}}%
}{
\Phi  \vdash_\mathcal{L}   \mathsf{F} \SCnt{X} }{%
{\SCdruleSXXFrName}{}%
}}

\newcommand{\SCdruleSXXGlName}[0]{\SCdrulename{S\_Gl}}
\newcommand{\SCdruleSXXGl}[1]{\SCdrule[#1]{%
\SCpremise{\Gamma  \SCsym{;}  \SCnt{A}  \SCsym{;}  \Delta  \vdash_\mathcal{L}  \SCnt{B}}%
}{
\Gamma  \SCsym{;}   \mathsf{G} \SCnt{A}   \SCsym{;}  \Delta  \vdash_\mathcal{L}  \SCnt{B}}{%
{\SCdruleSXXGlName}{}%
}}

\newcommand{\SCdruleSXXcutOneName}[0]{\SCdrulename{S\_cut1}}
\newcommand{\SCdruleSXXcutOne}[1]{\SCdrule[#1]{%
\SCpremise{ \Phi  \vdash_\mathcal{C}  \SCnt{X}  \quad  \Delta_{{\mathrm{1}}}  \SCsym{;}  \SCnt{X}  \SCsym{;}  \Delta_{{\mathrm{2}}}  \vdash_\mathcal{L}  \SCnt{A} }%
}{
\Delta_{{\mathrm{1}}}  \SCsym{;}  \Phi  \SCsym{;}  \Delta_{{\mathrm{1}}}  \vdash_\mathcal{L}  \SCnt{A}}{%
{\SCdruleSXXcutOneName}{}%
}}

\newcommand{\SCdruleSXXcutTwoName}[0]{\SCdrulename{S\_cut2}}
\newcommand{\SCdruleSXXcutTwo}[1]{\SCdrule[#1]{%
\SCpremise{ \Gamma  \vdash_\mathcal{L}  \SCnt{A}  \quad  \Delta_{{\mathrm{1}}}  \SCsym{;}  \SCnt{A}  \SCsym{;}  \Delta_{{\mathrm{2}}}  \vdash_\mathcal{L}  \SCnt{B} }%
}{
\Delta_{{\mathrm{1}}}  \SCsym{;}  \Gamma  \SCsym{;}  \Delta_{{\mathrm{2}}}  \vdash_\mathcal{L}  \SCnt{B}}{%
{\SCdruleSXXcutTwoName}{}%
}}

\renewcommand{\SCdrule}[4][]{{\displaystyle\frac{\begin{array}{l}#2\end{array}}{#3}\,#4}}
\renewcommand{\SCdruleTXXaxName}{\mathcal{C}\text{-ax}}
\renewcommand{\SCdruleTXXunitLName}{\mathcal{C}\text{-}\mathsf{Unit}_L}
\renewcommand{\SCdruleTXXunitRName}{\mathcal{C}\text{-}\mathsf{Unit}_L}
\renewcommand{\SCdruleTXXtenLName}{\mathcal{C}\text{-}\otimes_L}
\renewcommand{\SCdruleTXXtenRName}{\mathcal{C}\text{-}\otimes_R}
\renewcommand{\SCdruleTXXimpLName}{\mathcal{C}\text{-}\multimap_L}
\renewcommand{\SCdruleTXXimpRName}{\mathcal{C}\text{-}\multimap_R}
\renewcommand{\SCdruleTXXGrName}{\mathcal{C}\text{-}\mathsf{G}_R}
\renewcommand{\SCdruleTXXcutName}{\mathcal{C}\text{-}\mathsf{Cut}}
\renewcommand{\SCdruleTXXexName}{\mathcal{C}\text{-}\mathsf{ex}}
\renewcommand{\SCdruleSXXaxName}{\mathcal{L}\text{-ax}}
\renewcommand{\SCdruleSXXunitLOneName}{\mathcal{LC}\text{-}\mathsf{Unit}_L}
\renewcommand{\SCdruleSXXunitLTwoName}{\mathcal{L}\text{-}\mathsf{Unit}_L}
\renewcommand{\SCdruleSXXunitRName}{\mathcal{L}\text{-}\mathsf{Unit}_R}
\renewcommand{\SCdruleSXXexName}{\mathcal{L}\text{-}\mathsf{ex}}
\renewcommand{\SCdruleSXXtenLOneName}{\mathcal{LC}\text{-}\otimes_L}
\renewcommand{\SCdruleSXXtenLTwoName}{\mathcal{L}\text{-}\otimes_L}
\renewcommand{\SCdruleSXXtenRName}{\mathcal{L}\text{-}\otimes_R}
\renewcommand{\SCdruleSXXimpLName}{\mathcal{L}\text{-}\multimap_L}
\renewcommand{\SCdruleSXXimprLName}{\mathcal{L}\text{-}\rightharpoonup_L}
\renewcommand{\SCdruleSXXimprRName}{\mathcal{L}\text{-}\rightharpoonup_R}
\renewcommand{\SCdruleSXXimplLName}{\mathcal{L}\text{-}\leftharpoonup_L}
\renewcommand{\SCdruleSXXimplRName}{\mathcal{L}\text{-}\leftharpoonup_R}
\renewcommand{\SCdruleSXXFlName}{\mathcal{L}\text{-}\mathsf{F}_L}
\renewcommand{\SCdruleSXXFrName}{\mathcal{L}\text{-}\mathsf{F}_R}
\renewcommand{\SCdruleSXXGlName}{\mathcal{L}\text{-}\mathsf{G}_L}
\renewcommand{\SCdruleSXXcutOneName}{\mathcal{LC}\text{-}\mathsf{Cut}}
\renewcommand{\SCdruleSXXcutTwoName}{\mathcal{L}\text{-}\mathsf{Cut}}

\newcommand{\Elledrule}[4][]{{\displaystyle\frac{\begin{array}{l}#2\end{array}}{#3}\quad\Elledrulename{#4}}}

\newcommand{\Ellepremise}[1]{ #1 \\}
\newenvironment{Elledefnblock}[3][]{ \framebox{\mbox{#2}} \quad #3 \\[0pt]}{}

\newcommand{\Ellent}[1]{\mathit{#1}}
\newcommand{\Ellemv}[1]{\mathit{#1}}

\newcommand{\Ellesym}[1]{#1}

\newcommand{\ElledruleTXXaxName}[0]{\Elledrulename{T\_ax}}

\newcommand{\ElledruleTXXunitLName}[0]{\Elledrulename{T\_unitL}}
\newcommand{\ElledruleTXXunitL}[1]{\Elledrule[#1]{%
\Ellepremise{\Phi  \Ellesym{,}  \Psi  \vdash_\mathcal{C}  \Ellent{t}  \Ellesym{:}  \Ellent{X}}%
}{
\Phi  \Ellesym{,}  \Ellemv{x}  \Ellesym{:}   \mathsf{Unit}   \Ellesym{,}  \Psi  \vdash_\mathcal{C}   \mathsf{let}\, \Ellemv{x}  :   \mathsf{Unit}  \,\mathsf{be}\,  \mathsf{triv}  \,\mathsf{in}\, \Ellent{t}   \Ellesym{:}  \Ellent{X}}{%
{\ElledruleTXXunitLName}{}%
}}

\newcommand{\ElledruleTXXtenLName}[0]{\Elledrulename{T\_tenL}}
\newcommand{\ElledruleTXXtenL}[1]{\Elledrule[#1]{%
\Ellepremise{\Phi  \Ellesym{,}  \Ellemv{x}  \Ellesym{:}  \Ellent{X}  \Ellesym{,}  \Ellemv{y}  \Ellesym{:}  \Ellent{Y}  \Ellesym{,}  \Psi  \vdash_\mathcal{C}  \Ellent{t}  \Ellesym{:}  \Ellent{Z}}%
}{
\Phi  \Ellesym{,}  \Ellemv{z}  \Ellesym{:}  \Ellent{X}  \otimes  \Ellent{Y}  \Ellesym{,}  \Psi  \vdash_\mathcal{C}   \mathsf{let}\, \Ellemv{z}  :  \Ellent{X}  \otimes  \Ellent{Y} \,\mathsf{be}\, \Ellemv{x}  \otimes  \Ellemv{y} \,\mathsf{in}\, \Ellent{t}   \Ellesym{:}  \Ellent{Z}}{%
{\ElledruleTXXtenLName}{}%
}}

\newcommand{\ElledruleTXXimpLName}[0]{\Elledrulename{T\_impL}}
\newcommand{\ElledruleTXXimpL}[1]{\Elledrule[#1]{%
\Ellepremise{ \Phi  \vdash_\mathcal{C}  \Ellent{t_{{\mathrm{1}}}}  \Ellesym{:}  \Ellent{X}  \quad  \Psi_{{\mathrm{1}}}  \Ellesym{,}  \Ellemv{x}  \Ellesym{:}  \Ellent{Y}  \Ellesym{,}  \Psi_{{\mathrm{2}}}  \vdash_\mathcal{C}  \Ellent{t_{{\mathrm{2}}}}  \Ellesym{:}  \Ellent{Z} }%
}{
\Psi_{{\mathrm{1}}}  \Ellesym{,}  \Ellemv{y}  \Ellesym{:}  \Ellent{X}  \multimap  \Ellent{Y}  \Ellesym{,}  \Phi  \Ellesym{,}  \Psi_{{\mathrm{2}}}  \vdash_\mathcal{C}  \Ellesym{[}   \Ellemv{y}   \Ellent{t_{{\mathrm{1}}}}   \Ellesym{/}  \Ellemv{x}  \Ellesym{]}  \Ellent{t_{{\mathrm{2}}}}  \Ellesym{:}  \Ellent{Z}}{%
{\ElledruleTXXimpLName}{}%
}}

\newcommand{\ElledruleTXXcutName}[0]{\Elledrulename{T\_cut}}

\newcommand{\ElledruleSXXaxName}[0]{\Elledrulename{S\_ax}}

\newcommand{\ElledruleSXXunitLOneName}[0]{\Elledrulename{S\_unitL1}}
\newcommand{\ElledruleSXXunitLOne}[1]{\Elledrule[#1]{%
\Ellepremise{\Gamma  \Ellesym{;}  \Delta  \vdash_\mathcal{L}  \Ellent{s}  \Ellesym{:}  \Ellent{A}}%
}{
\Gamma  \Ellesym{;}  \Ellemv{x}  \Ellesym{:}   \mathsf{Unit}   \Ellesym{;}  \Delta  \vdash_\mathcal{L}   \mathsf{let}\, \Ellemv{x}  :   \mathsf{Unit}  \,\mathsf{be}\,  \mathsf{triv}  \,\mathsf{in}\, \Ellent{s}   \Ellesym{:}  \Ellent{A}}{%
{\ElledruleSXXunitLOneName}{}%
}}

\newcommand{\ElledruleSXXunitLTwoName}[0]{\Elledrulename{S\_unitL2}}
\newcommand{\ElledruleSXXunitLTwo}[1]{\Elledrule[#1]{%
\Ellepremise{\Gamma  \Ellesym{;}  \Delta  \vdash_\mathcal{L}  \Ellent{s}  \Ellesym{:}  \Ellent{A}}%
}{
\Gamma  \Ellesym{;}  \Ellemv{x}  \Ellesym{:}   \mathsf{Unit}   \Ellesym{;}  \Delta  \vdash_\mathcal{L}   \mathsf{let}\, \Ellemv{x}  :   \mathsf{Unit}  \,\mathsf{be}\,  \mathsf{triv}  \,\mathsf{in}\, \Ellent{s}   \Ellesym{:}  \Ellent{A}}{%
{\ElledruleSXXunitLTwoName}{}%
}}

\newcommand{\ElledruleSXXtenLOneName}[0]{\Elledrulename{S\_tenL1}}
\newcommand{\ElledruleSXXtenLOne}[1]{\Elledrule[#1]{%
\Ellepremise{\Gamma  \Ellesym{;}  \Ellemv{x}  \Ellesym{:}  \Ellent{X}  \Ellesym{;}  \Ellemv{y}  \Ellesym{:}  \Ellent{Y}  \Ellesym{;}  \Delta  \vdash_\mathcal{L}  \Ellent{s}  \Ellesym{:}  \Ellent{A}}%
}{
\Gamma  \Ellesym{;}  \Ellemv{z}  \Ellesym{:}  \Ellent{X}  \otimes  \Ellent{Y}  \Ellesym{;}  \Delta  \vdash_\mathcal{L}   \mathsf{let}\, \Ellemv{z}  :  \Ellent{X}  \otimes  \Ellent{Y} \,\mathsf{be}\, \Ellemv{x}  \otimes  \Ellemv{y} \,\mathsf{in}\, \Ellent{s}   \Ellesym{:}  \Ellent{A}}{%
{\ElledruleSXXtenLOneName}{}%
}}

\newcommand{\ElledruleSXXtenLTwoName}[0]{\Elledrulename{S\_tenL2}}
\newcommand{\ElledruleSXXtenLTwo}[1]{\Elledrule[#1]{%
\Ellepremise{\Gamma  \Ellesym{;}  \Ellemv{x}  \Ellesym{:}  \Ellent{A}  \Ellesym{;}  \Ellemv{y}  \Ellesym{:}  \Ellent{B}  \Ellesym{;}  \Delta  \vdash_\mathcal{L}  \Ellent{s}  \Ellesym{:}  \Ellent{C}}%
}{
\Gamma  \Ellesym{;}  \Ellemv{z}  \Ellesym{:}  \Ellent{A}  \triangleright  \Ellent{B}  \Ellesym{;}  \Delta  \vdash_\mathcal{L}   \mathsf{let}\, \Ellemv{z}  :  \Ellent{A}  \triangleright  \Ellent{B} \,\mathsf{be}\, \Ellemv{x}  \triangleright  \Ellemv{y} \,\mathsf{in}\, \Ellent{s}   \Ellesym{:}  \Ellent{C}}{%
{\ElledruleSXXtenLTwoName}{}%
}}

\newcommand{\ElledruleSXXimpLName}[0]{\Elledrulename{S\_impL}}
\newcommand{\ElledruleSXXimpL}[1]{\Elledrule[#1]{%
\Ellepremise{ \Phi  \vdash_\mathcal{C}  \Ellent{t}  \Ellesym{:}  \Ellent{X}  \quad  \Gamma  \Ellesym{;}  \Ellemv{x}  \Ellesym{:}  \Ellent{Y}  \Ellesym{;}  \Delta  \vdash_\mathcal{L}  \Ellent{s}  \Ellesym{:}  \Ellent{A} }%
}{
\Gamma  \Ellesym{;}  \Ellemv{y}  \Ellesym{:}  \Ellent{X}  \multimap  \Ellent{Y}  \Ellesym{;}  \Phi  \Ellesym{;}  \Delta  \vdash_\mathcal{L}  \Ellesym{[}   \Ellemv{y}   \Ellent{t}   \Ellesym{/}  \Ellemv{x}  \Ellesym{]}  \Ellent{s}  \Ellesym{:}  \Ellent{A}}{%
{\ElledruleSXXimpLName}{}%
}}

\newcommand{\ElledruleSXXimprLName}[0]{\Elledrulename{S\_imprL}}
\newcommand{\ElledruleSXXimprL}[1]{\Elledrule[#1]{%
\Ellepremise{ \Gamma  \vdash_\mathcal{L}  \Ellent{s_{{\mathrm{1}}}}  \Ellesym{:}  \Ellent{A}  \quad  \Delta_{{\mathrm{1}}}  \Ellesym{;}  \Ellemv{x}  \Ellesym{:}  \Ellent{B}  \Ellesym{;}  \Delta_{{\mathrm{2}}}  \vdash_\mathcal{L}  \Ellent{s_{{\mathrm{2}}}}  \Ellesym{:}  \Ellent{C} }%
}{
\Delta_{{\mathrm{1}}}  \Ellesym{;}  \Ellemv{y}  \Ellesym{:}  \Ellent{A}  \rightharpoonup  \Ellent{B}  \Ellesym{;}  \Gamma  \Ellesym{;}  \Delta_{{\mathrm{2}}}  \vdash_\mathcal{L}  \Ellesym{[}   \mathsf{app}_r\, \Ellemv{y} \, \Ellent{s_{{\mathrm{1}}}}   \Ellesym{/}  \Ellemv{x}  \Ellesym{]}  \Ellent{s_{{\mathrm{2}}}}  \Ellesym{:}  \Ellent{C}}{%
{\ElledruleSXXimprLName}{}%
}}

\newcommand{\ElledruleSXXimplLName}[0]{\Elledrulename{S\_implL}}
\newcommand{\ElledruleSXXimplL}[1]{\Elledrule[#1]{%
\Ellepremise{ \Gamma  \vdash_\mathcal{L}  \Ellent{s_{{\mathrm{1}}}}  \Ellesym{:}  \Ellent{A}  \quad  \Delta_{{\mathrm{1}}}  \Ellesym{;}  \Ellemv{x}  \Ellesym{:}  \Ellent{B}  \Ellesym{;}  \Delta_{{\mathrm{2}}}  \vdash_\mathcal{L}  \Ellent{s_{{\mathrm{2}}}}  \Ellesym{:}  \Ellent{C} }%
}{
\Delta_{{\mathrm{1}}}  \Ellesym{;}  \Gamma  \Ellesym{;}  \Ellemv{y}  \Ellesym{:}  \Ellent{B}  \leftharpoonup  \Ellent{A}  \Ellesym{;}  \Delta_{{\mathrm{2}}}  \vdash_\mathcal{L}  \Ellesym{[}   \mathsf{app}_l\, \Ellemv{y} \, \Ellent{s_{{\mathrm{1}}}}   \Ellesym{/}  \Ellemv{x}  \Ellesym{]}  \Ellent{s_{{\mathrm{2}}}}  \Ellesym{:}  \Ellent{C}}{%
{\ElledruleSXXimplLName}{}%
}}

\newcommand{\ElledruleSXXFlName}[0]{\Elledrulename{S\_Fl}}
\newcommand{\ElledruleSXXFl}[1]{\Elledrule[#1]{%
\Ellepremise{\Gamma  \Ellesym{;}  \Ellemv{x}  \Ellesym{:}  \Ellent{X}  \Ellesym{;}  \Delta  \vdash_\mathcal{L}  \Ellent{s}  \Ellesym{:}  \Ellent{A}}%
}{
\Gamma  \Ellesym{;}  \Ellemv{y}  \Ellesym{:}   \mathsf{F} \Ellent{X}   \Ellesym{;}  \Delta  \vdash_\mathcal{L}   \mathsf{let}\, \Ellemv{y}  :   \mathsf{F} \Ellent{X}  \,\mathsf{be}\,  \mathsf{F}\, \Ellemv{x}  \,\mathsf{in}\, \Ellent{s}   \Ellesym{:}  \Ellent{A}}{%
{\ElledruleSXXFlName}{}%
}}

\newcommand{\ElledruleSXXGlName}[0]{\Elledrulename{S\_Gl}}
\newcommand{\ElledruleSXXGl}[1]{\Elledrule[#1]{%
\Ellepremise{\Gamma  \Ellesym{;}  \Ellemv{x}  \Ellesym{:}  \Ellent{A}  \Ellesym{;}  \Delta  \vdash_\mathcal{L}  \Ellent{s}  \Ellesym{:}  \Ellent{B}}%
}{
\Gamma  \Ellesym{;}  \Ellemv{y}  \Ellesym{:}   \mathsf{G} \Ellent{A}   \Ellesym{;}  \Delta  \vdash_\mathcal{L}   \mathsf{let}\, \Ellemv{y}  :   \mathsf{G} \Ellent{A}  \,\mathsf{be}\,  \mathsf{G}\, \Ellemv{x}  \,\mathsf{in}\, \Ellent{s}   \Ellesym{:}  \Ellent{B}}{%
{\ElledruleSXXGlName}{}%
}}

\newcommand{\ElledruleSXXcutOneName}[0]{\Elledrulename{S\_cut1}}

\newcommand{\ElledruleSXXcutTwoName}[0]{\Elledrulename{S\_cut2}}

\renewcommand{\Elledrule}[4][]{{\displaystyle\frac{\begin{array}{l}#2\end{array}}{#3}\,#4}}
\renewcommand{\ElledruleTXXaxName}{\mathcal{C}\text{-ax}}
\renewcommand{\ElledruleTXXunitLName}{\mathcal{C}\text{-}\mathsf{Unit}_L}

\renewcommand{\ElledruleTXXtenLName}{\mathcal{C}\text{-}\otimes_L}

\renewcommand{\ElledruleTXXimpLName}{\mathcal{C}\text{-}\multimap_L}

\renewcommand{\ElledruleTXXcutName}{\mathcal{C}\text{-}\mathsf{Cut}}

\renewcommand{\ElledruleSXXaxName}{\mathcal{L}\text{-ax}}
\renewcommand{\ElledruleSXXunitLOneName}{\mathcal{LC}\text{-}\mathsf{Unit}_L}
\renewcommand{\ElledruleSXXunitLTwoName}{\mathcal{L}\text{-}\mathsf{Unit}_L}

\renewcommand{\ElledruleSXXtenLOneName}{\mathcal{LC}\text{-}\otimes_L}
\renewcommand{\ElledruleSXXtenLTwoName}{\mathcal{L}\text{-}\otimes_L}

\renewcommand{\ElledruleSXXimpLName}{\mathcal{L}\text{-}\multimap_L}
\renewcommand{\ElledruleSXXimprLName}{\mathcal{L}\text{-}\rightharpoonup_L}

\renewcommand{\ElledruleSXXimplLName}{\mathcal{L}\text{-}\leftharpoonup_L}

\renewcommand{\ElledruleSXXFlName}{\mathcal{L}\text{-}\mathsf{F}_L}

\renewcommand{\ElledruleSXXGlName}{\mathcal{L}\text{-}\mathsf{G}_L}
\renewcommand{\ElledruleSXXcutOneName}{\mathcal{LC}\text{-}\mathsf{Cut}}
\renewcommand{\ElledruleSXXcutTwoName}{\mathcal{L}\text{-}\mathsf{Cut}}

\newenvironment{LNLdefnblock}[3][]{ \framebox{\mbox{#2}} \quad #3 \\[0pt]}{}

\newcommand{\Ldrule}[4][]{{\displaystyle\frac{\begin{array}{l}#2\end{array}}{#3}\quad\Ldrulename{#4}}}

\newcommand{\Lpremise}[1]{ #1 \\}
\newenvironment{Ldefnblock}[3][]{ \framebox{\mbox{#2}} \quad #3 \\[0pt]}{}

\newcommand{\Lnt}[1]{\mathit{#1}}

\newcommand{\Lsym}[1]{#1}

\newcommand{\Ldrulename}[1]{\textsc{#1}}

\newcommand{\LdruleErName}[0]{\Ldrulename{Er}}
\newcommand{\LdruleEr}[1]{\Ldrule[#1]{%
\Lpremise{  \kappa  \Gamma   \vdash  \Lnt{B} }%
}{
  \kappa  \Gamma   \vdash   \kappa  \Lnt{B}  }{%
{\LdruleErName}{}%
}}

\newcommand{\LdruleElName}[0]{\Ldrulename{El}}
\newcommand{\LdruleEl}[1]{\Ldrule[#1]{%
\Lpremise{ \Gamma_{{\mathrm{1}}}  \Lsym{,}  \Lnt{A}  \Lsym{,}  \Gamma_{{\mathrm{2}}}  \vdash  \Lnt{B} }%
}{
 \Gamma_{{\mathrm{1}}}  \Lsym{,}   \kappa  \Lnt{A}   \Lsym{,}  \Gamma_{{\mathrm{2}}}  \vdash  \Lnt{B} }{%
{\LdruleElName}{}%
}}

\newcommand{\LdruleEOneName}[0]{\Ldrulename{E1}}
\newcommand{\LdruleEOne}[1]{\Ldrule[#1]{%
\Lpremise{ \Gamma_{{\mathrm{1}}}  \Lsym{,}   \kappa  \Lnt{A}   \Lsym{,}  \Lnt{B}  \Lsym{,}  \Gamma_{{\mathrm{2}}}  \vdash  \Lnt{C} }%
}{
 \Gamma_{{\mathrm{1}}}  \Lsym{,}  \Lnt{B}  \Lsym{,}   \kappa  \Lnt{A}   \Lsym{,}  \Gamma_{{\mathrm{2}}}  \vdash  \Lnt{C} }{%
{\LdruleEOneName}{}%
}}

\newcommand{\LdruleETwoName}[0]{\Ldrulename{E2}}
\newcommand{\LdruleETwo}[1]{\Ldrule[#1]{%
\Lpremise{ \Gamma_{{\mathrm{1}}}  \Lsym{,}  \Lnt{A}  \Lsym{,}   \kappa  \Lnt{B}   \Lsym{,}  \Gamma_{{\mathrm{2}}}  \vdash  \Lnt{C} }%
}{
 \Gamma_{{\mathrm{1}}}  \Lsym{,}   \kappa  \Lnt{B}   \Lsym{,}  \Lnt{A}  \Lsym{,}  \Gamma_{{\mathrm{2}}}  \vdash  \Lnt{C} }{%
{\LdruleETwoName}{}%
}}

\renewcommand{\Ldrule}[4][]{{\displaystyle\frac{\begin{array}{l}#2\end{array}}{#3}\,\Ldrulename{#4}}}

\let\mto\to
\let\to\relax
\newcommand{\to}{\rightarrow}

\newcommand{\rto}{\leftharpoonup}
\newcommand{\lto}{\rightharpoonup}
\newcommand{\tri}{\triangleright}

\let\t\relax

\newcommand{\cat}[1]{\mathcal{#1}}
\newcommand{\func}[1]{\mathsf{#1}}

\newcommand{\Hom}[3]{\mathsf{Hom}_{\cat{#1}}(#2,#3)}

\newcommand{\Id}[0]{\mathsf{Id}}

\newcommand{\e}[1]{\mathsf{ex}_{#1}}

\newcommand{\m}[1]{\mathsf{m}_{#1}}
\newcommand{\n}[1]{\mathsf{n}_{#1}}

\newcommand{\p}[1]{\mathsf{p}_{#1}}

\newcommand{\t}[1]{\mathsf{t}_{#1}}

\newtheorem{theorem}{Theorem}
\newtheorem{lemma}[theorem]{Lemma}

\newtheorem{definition}[theorem]{Definition}

\title{On the Lambek Calculus with an Exchange Modality}

\author{Jiaming Jiang
\institute{Computer Science \\ North Carolina State University \\ Raleigh, North Carolina, USA}
\email{jjiang13@ncsu.edu}
\and
Harley Eades III
\institute{Computer Science \\ Augusta University \\ Augusta, Georgia, USA}
\email{harley.eades@gmail.com}
\and
Valeria de Paiva
\institute{Nuance Communications \\ Sunnyvale, California, USA}
\email{valeria.depaiva@gmail.com}
}
\def\titlerunning{On the Lambek Calculus with an Exchange Modality}
\def\authorrunning{J. Jiang, H. Eades III \& V. de Paiva}

\begin{document}
\maketitle 

\begin{abstract}
  In this paper we introduce Commutative/Non-Commutative Logic (CNC
  logic) and two categorical models for CNC logic.  This work
  abstracts Benton's Linear/Non-Linear Logic \cite{Benton:1994} by
  removing the existence of the exchange structural rule. One should
  view this logic as composed of two logics; one sitting to the left
  of the other.  On the left, there is intuitionistic linear logic,
  and on the right is a mixed commutative/non-commutative
  formalization of the Lambek calculus. Then both of these logics are
  connected via a pair of monoidal adjoint functors.  An exchange
  modality is then derivable within the logic using the adjunction
  between both sides.  Thus, the adjoint functors allow one to pull
  the exchange structural rule from the left side to the right side.
  We then give a categorical model in terms of a monoidal adjunction,
  and then a concrete model in terms of dialectica Lambek spaces.
\end{abstract}

\section{Introduction}
\label{sec:introduction}
Joachim Lambek first introduced the Syntactic Calculus, now known as
the Lambek Calculus, in 1958 \cite{Lambek1958}.  Since then the Lambek
Calculus has largely been motivated by providing an explanation of the
mathematics of sentence structure, and can be found at the core of
Categorical Grammar; a term first used in the title of Bar-Hillel,
Gaifman and Shamir (1960), but categorical grammar began with
Ajdukiewicz (1935) quite a few years earlier. At the end of the
eighties the Lambek calculus and other systems of categorical grammars
were taken up by computational linguists as exemplified by
\cite{oehrle2012categorial,moortgat1988categorial,Barry:1991:PFS:977180.977215,hepple1990grammar}.

It was computational linguists who posed the question of whether it is
possible to isolate exchange using a modality in the same way that the
of-course modality of linear logic, $!A$, isolates weakening and
contraction.  de Paiva and Eades \cite{dePaiva2018} propose one
solution to this problem by extending the Lambek calculus with the
modality characterized by the following sequent calculus inference
rules:
\[
\small
\begin{array}{ccccccccccccccccccccc}  
  \LdruleEr{} & & \LdruleEl{} & & \LdruleEOne{} & & \LdruleETwo{} 
\end{array}
\]
The thing to note is that the modality $\kappa A$ appears on only one
of the operands being exchanged.  That is, these rules along with
those for the tensor product allow one to prove that $\kappa A \otimes
B \multimap B \otimes \kappa A$ holds.  This is somewhat at odds with
algebraic intuition, and it is unclear how this modality could be
decomposed into adjoint functors in a linear/non-linear (LNL)
formalization of the Lambek calculus.

In this paper we show how to add an exchange modality, $eA$, where the
modality now occurs on both operands being exchanged. That is, one can
show that $eA \otimes eB \multimap eB \otimes eA$ holds.  We give a sequent
calculus and a LNL natural deduction formalization for the Lambek calculus
with this new modality, and two categorical models: a LNL model and a
concrete model in dialectica spaces. Thus giving a second solution to the
problem proposed above.

The Lambek Calculus also has the potential for many applications in
other areas of computer science, such as, modeling processes.  Linear
Logic has been at the forefront of the study of process calculi for
many years \cite{HONDA20102223,Pratt:1997,ABRAMSKY19945}. We can think
of the commutative tensor product of linear logic as a parallel
operator.  For example, given a process $A$ and a process $B$, then we
can form the process $A \otimes B$ which runs both processes in
parallel.  If we remove commutativity from the tensor product we
obtain a sequential composition instead of parallel composition.  That
is, the process $A \rhd B$ first runs process $A$ and then process $B$
in that order.  Vaughan Pratt has stated that , ``sequential
composition has no evident counterpart in type theory'' see page 11 of
\cite{Pratt:1997}.  We believe that the Lambek Calculus will lead to
filling this hole.  

\textbf{Acknowledgments.}  The first two authors were supported by NSF
award \#1565557.  We thank the anonymous reviewers for their helpful
feedback that made this a better paper.



\section{A Sequent Calculus Formalization of CNC Logic}
\label{sec:sequent-calc}
We now introduce Commutative/Non-commutative (CNC) logic in the form of a
sequent calculus. One should view this logic as composed of two logics; one
sitting to the left of the other. On the left, there is intuitionistic
linear logic, denoted by $\cat{C}$ and on the right is the Lambek calculus
denoted by $\cat{L}$. Then we connect these two systems by a pair of
monoidal adjoint functors $\cat{C} : \func{F} \dashv \func{G} : \cat{L}$.
Keeping this intuition in mind we now define the syntax for CNC logic.

\begin{definition}
  \label{def:Lambek-syntax}
  The following grammar describes the syntax of the sequent calculus of
  CNC logic:
  \begin{center}\vspace{-3px}\small
    \begin{math}
      \begin{array}{lll}
        \text{($\cat{C}$-Types)} & \SCnt{W},\SCnt{X},\SCnt{Y},\SCnt{Z} ::=  \mathsf{Unit}  \mid \SCnt{X}  \otimes  \SCnt{Y} \mid \SCnt{X}  \multimap  \SCnt{Y} \mid  \mathsf{G} \SCnt{A} \\
        \text{($\cat{L}$-Types)} & \SCnt{A},\SCnt{B},\SCnt{C},D ::=  \mathsf{Unit}  \mid \SCnt{A}  \triangleright  \SCnt{B} \mid \SCnt{A}  \rightharpoonup  \SCnt{B} \mid \SCnt{B}  \leftharpoonup  \SCnt{A} \mid  \mathsf{F} \SCnt{X} \\
        \text{($\cat{C}$-Contexts)} & \Phi,\Psi ::=  \cdot  \mid \SCnt{X} \mid \Phi  \SCsym{,}  \Psi\\
        \text{($\cat{L}$-Contexts)} & \Gamma,\Delta ::=  \cdot  \mid \SCnt{A} \mid \SCnt{X} \mid \Gamma  \SCsym{;}  \Delta\\
      \end{array}
    \end{math}
  \end{center}
\end{definition}

The syntax for $\cat{C}$-types are the standard types for intuitionistic
linear logic. We have a constant $ \mathsf{Unit} $, tensor product $\SCnt{X}  \otimes  \SCnt{Y}$,
and linear implication $\SCnt{X}  \multimap  \SCnt{Y}$, but just as in LNL logic we also have a
type $ \mathsf{G} \SCnt{A} $ where $\SCnt{A}$ is an $\cat{L}$-type; that is, a type from the
non-commutative side corresponding to the right-adjoint functor between
$\cat{L}$ and $\cat{C}$. This functor can be used to import types from the
non-commutative side into the commutative side. Now a sequent in the the
commutative side is denoted by $\Phi  \vdash_\mathcal{C}  \SCnt{X}$ where $\Phi$ is a
$\cat{C}$-context, which is a sequence of types $\SCnt{X}$.

The non-commutative side is a bit more interesting than the commutative side
just introduced. Sequents in the non-commutative side are denoted by
$\Gamma  \vdash_\mathcal{L}  \SCnt{A}$ where $\Gamma$ is now a $\cat{L}$-context. These contexts are
ordered sequences of types from \emph{both} sides denoted by $\SCnt{B}$ and
$\SCnt{X}$ respectively. Given two contexts $\Gamma$ and $\Delta$ we denote their
concatenation by $\Gamma  \SCsym{;}  \Delta$; we use a semicolon here to emphasize the fact
that the contexts are ordered.

The context consisting of hypotheses from both sides goes back to
Benton~\cite{Benton:1994} and is a property unique to adjoint logics such as
Benton's LNL logic and CNC logic. This is also a very useful property
because it allows one to make use of both sides within the Lambek calculus
without the need to annotate every formula with a modality.

The reader familiar with LNL logic will notice that our sequent,
$\Gamma  \vdash_\mathcal{L}  \SCnt{A}$, differs from Benton's. His is of the form $\Gamma  \SCsym{;}  \Delta  \vdash_\mathcal{L}  \SCnt{A}$,
where $\Gamma$ contains non-linear types, and $\Delta$ contains linear
formulas. Just as Benton remarks, the splitting of his contexts was a
presentational device. One should view his contexts as merged, and hence,
linear formulas were fully mixed with non-linear formulas. Now why did we
not use this presentational device? Because, when contexts from LNL logic
become out of order Benton could use the exchange rule to put them back in
order again, but we no longer have general exchange. Thus, we are not able
to keep the context organized in this way.

The syntax for $\cat{L}$-types are of the typical form for the Lambek
Calculus. We have two unit types $ \mathsf{Unit} $ (one for each side), a
non-commutative tensor product $\SCnt{A}  \triangleright  \SCnt{B}$, right implication $\SCnt{A}  \rightharpoonup  \SCnt{B}$, and left implication $\SCnt{B}  \leftharpoonup  \SCnt{A}$. In standard Lambek
Calculus \cite{Pentus1995}, $\SCnt{A}  \rightharpoonup  \SCnt{B}$ is written as $B / A$ and
$\SCnt{B}  \leftharpoonup  \SCnt{A}$ as $A \backslash B$. We use $\rightharpoonup$ and
$\leftharpoonup$ here instead to indicate they are two directions of
the linear implication $\multimap$.

The sequent calculus for CNC logic can be found in
Figure~\ref{fig:CNC-sequent-calculus}. We split the figure in two: the top
of the figure are the rules of intuitionistic linear logic whose sequents
are the $\mathcal{C}$-sequents denoted by $\Psi  \vdash_\mathcal{C}  \SCnt{X}$, and the bottom of
the figure are the rules for the mixed commutative/non-commutative Lambek
calculus whose sequents are the $\mathcal{L}$-sequents denoted by
$\Gamma  \vdash_\mathcal{L}  \SCnt{A}$, but the two halves are connected via the functor rules
$\SCdruleTXXGrName{}$, $\SCdruleSXXGlName{}$, $\SCdruleSXXFlName{}$, and
$\SCdruleSXXFrName{}$, and the rules $\SCdruleSXXunitLOneName{}$,
$\SCdruleSXXexName{}$, $\SCdruleSXXtenLOneName{}$, $\SCdruleSXXimpLName{}$,
$\NDdruleSXXcutOneName{}$.
\begin{figure}[!h]
  \footnotesize
  \resizebox{\columnwidth}{!}{
  \begin{tabular}{|c|}
    \hline\\
    \begin{mathpar}
    \SCdruleTXXax{} \and
    \SCdruleTXXunitL{} \and
    \SCdruleTXXunitR{} \and
    \SCdruleTXXtenL{} \and
    \SCdruleTXXtenR{} \and
    \SCdruleTXXimpL{} \and
    \SCdruleTXXimpR{} \and
    \SCdruleTXXGr{} \and
    \SCdruleTXXex{} \and
    \SCdruleTXXcut{}
    \end{mathpar}\\\\
    \hline
    \\[5px]
    \begin{mathpar}
    \SCdruleSXXax{} \and
    \SCdruleSXXunitLOne{} \and
    \SCdruleSXXunitLTwo{} \and
    \SCdruleSXXunitR{} \and
    \SCdruleSXXex{} \and
    \SCdruleSXXtenLOne{} \and
    \SCdruleSXXtenLTwo{} \and
    \SCdruleSXXtenR{} \and
    \SCdruleSXXimpL{} \and
    \SCdruleSXXimprL{} \and
    \SCdruleSXXimprR{} \and
    \SCdruleSXXimplL{} \and
    \SCdruleSXXimplR{} \and
    \SCdruleSXXFl{} \and
    \SCdruleSXXFr{} \and
    \SCdruleSXXGl{} \and
    \SCdruleSXXcutOne{} \and
    \SCdruleSXXcutTwo{} \and
    \end{mathpar}\\\\
    \hline
  \end{tabular}
  }
  \caption{Sequent Calculus for CNC Logic}
  \label{fig:CNC-sequent-calculus}
\end{figure}

We prove cut elimination for the sequent calculus. We define the
\textit{rank} $|X|$ (resp. $|A|$) of a commutative (resp. non-commutative)
formula to be the number of logical connectives in the proposition. For
instance, $|\SCnt{X}  \otimes  \SCnt{Y}| = |\SCnt{X}| + |\SCnt{Y}| + 1$. The \textit{cut rank}
$c(\Pi)$ of a proof $\Pi$ is one more than the maximum of the ranks of all
the cut formulae in $\Pi$, and $0$ if $\Pi$ is cut-free. Then the
\textit{depth} $d(\Pi)$ of a proof $\Pi$ is the length of the longest path
in the proof tree (so the depth of an axiom is $0$). The key to the proof
of cut elimination is the following lemma, which shows how to transform a
single cut, either by removing it or by replacing it with one or more
simpler cuts.
\begin{lemma}[Cut Reduction]
  \label{lem:cut-reduction}
  The cut-reduction steps are as follows:
  \begin{enumerate}
  \item If $\Pi_1$ is a proof of $\Phi  \vdash_\mathcal{C}  \SCnt{X}$ and $\Pi_2$ is a proof of
  $\Psi_{{\mathrm{1}}}  \SCsym{,}  \SCnt{X}  \SCsym{,}  \Psi_{{\mathrm{2}}}  \vdash_\mathcal{C}  \SCnt{Y}$ with $c(\Pi_1)$, $c(\Pi_2)\leq |X|$, then there exists
  a proof $\Pi$ of $\Psi_{{\mathrm{1}}}  \SCsym{,}  \Phi  \SCsym{,}  \Psi_{{\mathrm{2}}}  \vdash_\mathcal{C}  \SCnt{Y}$ with $c(\Pi)\leq |X|$.
  \item If $\Pi_1$ is a proof of $\Phi  \vdash_\mathcal{C}  \SCnt{X}$ and $\Pi_2$ is a proof of
  $\Gamma_{{\mathrm{1}}}  \SCsym{;}  \SCnt{X}  \SCsym{;}  \Gamma_{{\mathrm{2}}}  \vdash_\mathcal{L}  \SCnt{A}$ with $c(\Pi_1)$, $c(\Pi_2)\leq |X|$, then there
  exists a proof $\Pi$ of $\Gamma_{{\mathrm{1}}}  \SCsym{;}  \Phi  \SCsym{;}  \Gamma_{{\mathrm{2}}}  \vdash_\mathcal{L}  \SCnt{A}$ with $c(\Pi)\leq |X|$.
  \item If $\Pi_1$ is a proof of $\Gamma  \vdash_\mathcal{L}  \SCnt{A}$ and $\Pi_2$ is a proof of
  $\Delta_{{\mathrm{1}}}  \SCsym{;}  \SCnt{A}  \SCsym{;}  \Delta_{{\mathrm{2}}}  \vdash_\mathcal{L}  \SCnt{B}$ with $c(\Pi_1)$, $c(\Pi_2)\leq |A|$, then there
  exists a proof $\Pi$ of $\Delta_{{\mathrm{1}}}  \SCsym{;}  \Gamma  \SCsym{;}  \Delta_{{\mathrm{2}}}  \vdash_\mathcal{L}  \SCnt{B}$ with $c(\Pi)\leq |A|$.
  \end{enumerate}
\end{lemma}
\begin{proof}
  This proof is done case by case on the last step of $\Pi_1$ and
  $\Pi_2$ and by induction on $d(\Pi_1)$ and $d(\Pi_2)$, following
  \cite{Mellies:2009}. For instance, suppose $\Pi_1$ is a proof of
  $\Phi_{{\mathrm{1}}}  \SCsym{,}  \SCnt{X_{{\mathrm{2}}}}  \SCsym{,}  \SCnt{X_{{\mathrm{1}}}}  \SCsym{,}  \Phi_{{\mathrm{2}}}  \vdash_\mathcal{C}  \SCnt{Y}$ and $\Pi_2$ is a proof of $\Psi_{{\mathrm{1}}}  \SCsym{,}  \SCnt{Y}  \SCsym{,}  \Psi_{{\mathrm{2}}}  \vdash_\mathcal{C}  \SCnt{Z}$.  Consider the case where the last step in $\Pi_1$ uses
  the rule $\NDdruleTXXbetaName{}$. $\Pi_1$ can be depicted as
  follows, where the previous steps are denoted by $\pi$:
  \begin{center}
    \scriptsize
    $\Pi_1$:
    \begin{math}
      $$\mprset{flushleft}
      \inferrule* [right={\tiny $\SCdruleTXXexName$}] {
        {
          \begin{array}{c}
            \pi \\
                {\Phi_{{\mathrm{1}}}  \SCsym{,}  \SCnt{X_{{\mathrm{1}}}}  \SCsym{,}  \SCnt{X_{{\mathrm{2}}}}  \SCsym{,}  \Phi_{{\mathrm{2}}}  \vdash_\mathcal{C}  \SCnt{Y}}
          \end{array}
        }
      }{\Phi_{{\mathrm{1}}}  \SCsym{,}  \SCnt{X_{{\mathrm{2}}}}  \SCsym{,}  \SCnt{X_{{\mathrm{1}}}}  \SCsym{,}  \Phi_{{\mathrm{2}}}  \vdash_\mathcal{C}  \SCnt{Y}}
    \end{math}
  \end{center}
  By assumption, $c(\Pi_1),c(\Pi_2)\leq |Y|$. By induction on $\pi$ and
  $\Pi_2$, there is a proof $\Pi'$ for the sequent \\
  $\Psi_{{\mathrm{1}}}  \SCsym{,}  \Phi_{{\mathrm{1}}}  \SCsym{,}  \SCnt{X_{{\mathrm{1}}}}  \SCsym{,}  \SCnt{X_{{\mathrm{2}}}}  \SCsym{,}  \Phi_{{\mathrm{2}}}  \SCsym{,}  \Psi_{{\mathrm{2}}}  \vdash_\mathcal{C}  \SCnt{Z}$ s.t. $c(\Pi')\leq|Y|$. Therefore, the
  proof $\Pi$ can be constructed as follows, and $c(\Pi)=c(\Pi')\leq|Y|$.
  \begin{center}
    \scriptsize
    \begin{math}
      $$\mprset{flushleft}
      \inferrule* [right={\tiny $\SCdruleTXXexName$}] {
        {
          \begin{array}{c}
            \Pi' \\
                 {\Psi_{{\mathrm{1}}}  \SCsym{,}  \Phi_{{\mathrm{1}}}  \SCsym{,}  \SCnt{X_{{\mathrm{1}}}}  \SCsym{,}  \SCnt{X_{{\mathrm{2}}}}  \SCsym{,}  \Phi_{{\mathrm{2}}}  \SCsym{,}  \Psi_{{\mathrm{2}}}  \vdash_\mathcal{C}  \SCnt{Z}}
          \end{array}
        }
      }{\Psi_{{\mathrm{1}}}  \SCsym{,}  \Phi_{{\mathrm{1}}}  \SCsym{,}  \SCnt{X_{{\mathrm{2}}}}  \SCsym{,}  \SCnt{X_{{\mathrm{1}}}}  \SCsym{,}  \Phi_{{\mathrm{2}}}  \SCsym{,}  \Psi_{{\mathrm{2}}}  \vdash_\mathcal{C}  \SCnt{Z}}
    \end{math}
  \end{center}
  The full proof can be found in Appendix~\ref{app:cut-reduction}.
\end{proof}
\noindent
Then we have the following lemma.

\begin{lemma}
  \label{lem:less-cut-rank}
  Let $\Pi$ be a proof of a sequent $\Phi  \vdash_\mathcal{C}  \SCnt{X}$ or $\Gamma  \vdash_\mathcal{L}  \SCnt{A}$ s.t.
  $c(\Pi)>0$. Then there is a proof $\Pi'$ of the same sequent with
  $c(\Pi')<c(\Pi)$.
\end{lemma}
\begin{proof}
  We prove the lemma by induction on $d(\Pi)$. We denote the proof $\Pi$ by 
  $\pi+r$, where $r$ is the last inference of $\Pi$ and $\pi$ denotes the
  rest of the proof. If $r$ is not a cut, then by induction hypothesis on
  $\pi$, there is a proof $\pi'$ s.t. $c(\pi')<c(\pi)$ and $\Pi'=\pi'+r$.
  Otherwise, we assume $r$ is a cut on a formula $Y$. If $c(\Pi)>|X|+1$,
  then there is a cut on $|Y|$ in $\pi$ with $|Y|>|X|$. So we can apply
  the induction hypothesis on $\pi$ to get $\Pi'$ with $c(\Pi')<c(\Pi)$. The
  last case to consider is when $c(\Pi)=|X|+1$ (note that $c(\Pi)$ cannot be
  less than $|X|+1$). In this case, $\Pi$ is in the form of
  \begin{center}
    \scriptsize
    \begin{math}
      $$\mprset{flushleft}
      \inferrule* [right={\tiny $\SCdruleTXXcutName$}] {
        {
          \begin{array}{cc}
            \Pi_1 & \Pi_2 \\
            {\Phi  \vdash_\mathcal{C}  \SCnt{X}} & {\Psi_{{\mathrm{1}}}  \SCsym{,}  \SCnt{X}  \SCsym{,}  \Psi_{{\mathrm{2}}}  \vdash_\mathcal{C}  \SCnt{Y}}
          \end{array}
        }
      }{\Psi_{{\mathrm{1}}}  \SCsym{,}  \Phi  \SCsym{,}  \Psi_{{\mathrm{2}}}  \vdash_\mathcal{C}  \SCnt{Y}}
    \end{math}
    \qquad\qquad
    or,
    \qquad\qquad
    \begin{math}
      $$\mprset{flushleft}
      \inferrule* [right={\tiny $\SCdruleSXXcutOneName$}] {
        {
          \begin{array}{cc}
            \Pi_1 & \Pi_2 \\
            {\Phi  \vdash_\mathcal{C}  \SCnt{X}} & {\Gamma_{{\mathrm{1}}}  \SCsym{;}  \SCnt{X}  \SCsym{;}  \Gamma_{{\mathrm{2}}}  \vdash_\mathcal{L}  \SCnt{A}}
          \end{array}
        }
      }{\Gamma_{{\mathrm{1}}}  \SCsym{;}  \Phi  \SCsym{;}  \Gamma_{{\mathrm{2}}}  \vdash_\mathcal{L}  \SCnt{A}}
    \end{math}
  \end{center}
  By assumption, $c(\Pi_1),c(\Pi_2)\leq |X|+1$. By induction, we can
  construct $c(\Pi_1')$ proving $\Phi  \vdash_\mathcal{C}  \SCnt{X}$ and $c(\Pi_2')$ proving
  $\Psi_{{\mathrm{1}}}  \SCsym{,}  \SCnt{X}  \SCsym{,}  \Psi_{{\mathrm{2}}}  \vdash_\mathcal{C}  \SCnt{Y}$ (or $<\Gamma_{{\mathrm{1}}}  \SCsym{;}  \SCnt{X}  \SCsym{;}  \Gamma_{{\mathrm{2}}}  \vdash_\mathcal{L}  \SCnt{A}$) with
  $c(\Pi_1'), c(\Pi_2')\leq |X|$. Then by Lemma~\ref{lem:cut-reduction}, we
  can construct $\Pi'$ proving $\Psi_{{\mathrm{1}}}  \SCsym{,}  \Phi  \SCsym{,}  \Psi_{{\mathrm{2}}}  \vdash_\mathcal{C}  \SCnt{Y}$ (or
  $\Gamma_{{\mathrm{1}}}  \SCsym{;}  \Phi  \SCsym{;}  \Gamma_{{\mathrm{2}}}  \vdash_\mathcal{L}  \SCnt{A}$) with $c(\Pi')\leq |X|$. \\
  The case where the last inference is a cut on a formula $A$ is similar as
  when it is a cut on $X$.
\end{proof}
\noindent
By induction on $c(\Pi)$ and Lemma~\ref{lem:less-cut-rank}, the cut
elimination theorem follows immediately.
\begin{theorem}[Cut Elimination]
  Let $\Pi$ be a proof of a sequent $\Phi  \vdash_\mathcal{C}  \SCnt{X}$ or $\Gamma  \vdash_\mathcal{L}  \SCnt{A}$ s.t.
  $c(\Pi)>0$. Then there is an algorithm which yields a cut-free proof
  $\Pi'$ of the same sequent.
\end{theorem}


\section{A Type Theoretic Formalization of CNC Logic}
\label{sec:the-lambek-calculus}
Similar as the sequent calculus, the term assignment for CNC logic is
also composed of two logics; intuitionistic linear logic on the left,
denoted by $\cat{C}$, and the Lambek calculus on the right, denoted by
$\cat{L}$. The syntax for types and contexts we use in the term
assignment is the same as in the sequent calculus. The rest of the
syntax for the term assignment is defined as follows.
\begin{definition}
  \label{def:Lambek-syntax}
  The following grammar describes the syntax of the term assignment of the
  CNC logic:
  \begin{center}\vspace{-3px}\small
    \begin{math}
      \begin{array}{lll}        
        \text{($\cat{C}$-Terms)} & \NDnt{t} ::= \NDmv{x} \mid  \mathsf{triv}  \mid \NDnt{t_{{\mathrm{1}}}}  \otimes  \NDnt{t_{{\mathrm{2}}}} \mid  \mathsf{let}\, \NDnt{t_{{\mathrm{1}}}}  :  \NDnt{X} \,\mathsf{be}\, \NDnt{q} \,\mathsf{in}\, \NDnt{t_{{\mathrm{2}}}}  \mid  \lambda  \NDmv{x}  :  \NDnt{X} . \NDnt{t}  \mid  \NDnt{t_{{\mathrm{1}}}}   \NDnt{t_{{\mathrm{2}}}}  \mid  \mathsf{ex}\, \NDnt{t_{{\mathrm{1}}}} , \NDnt{t_{{\mathrm{2}}}} \,\mathsf{with}\, \NDmv{x_{{\mathrm{1}}}} , \NDmv{x_{{\mathrm{2}}}} \,\mathsf{in}\, \NDnt{t_{{\mathrm{3}}}}  \mid  \mathsf{G}\, \NDnt{s} \\
        \text{($\cat{L}$-Terms)} & \NDnt{s} ::= \NDmv{x} \mid  \mathsf{triv}  \mid \NDnt{s_{{\mathrm{1}}}}  \triangleright  \NDnt{s_{{\mathrm{2}}}} \mid  \mathsf{let}\, \NDnt{s_{{\mathrm{1}}}}  :  \NDnt{A} \,\mathsf{be}\, \NDnt{p} \,\mathsf{in}\, \NDnt{s_{{\mathrm{2}}}}  \mid  \mathsf{let}\, \NDnt{t}  :  \NDnt{X} \,\mathsf{be}\, \NDnt{q} \,\mathsf{in}\, \NDnt{s}  \mid  \lambda_l  \NDmv{x}  :  \NDnt{A} . \NDnt{s}  \mid  \lambda_r  \NDmv{x}  :  \NDnt{A} . \NDnt{s}  \\
        & \,\,\,\,\,\,\,\,\,\mid  \mathsf{app}_l\, \NDnt{s_{{\mathrm{1}}}} \, \NDnt{s_{{\mathrm{2}}}}  \mid  \mathsf{app}_r\, \NDnt{s_{{\mathrm{1}}}} \, \NDnt{s_{{\mathrm{2}}}}  \mid  \mathsf{F} \NDnt{t} \\        
        \text{($\cat{C}$-Patterns)} & \NDnt{q} ::=  \mathsf{triv}  \mid \NDmv{x} \mid \NDnt{q_{{\mathrm{1}}}}  \otimes  \NDnt{q_{{\mathrm{2}}}} \mid  \mathsf{G}\, \NDnt{p} \\
        \text{($\cat{L}$-Patterns)} & \NDnt{p} ::=  \mathsf{triv}  \mid \NDmv{x} \mid \NDnt{p_{{\mathrm{1}}}}  \triangleright  \NDnt{p_{{\mathrm{2}}}} \mid  \mathsf{F}\, \NDnt{q} \\        
        \text{($\cat{C}$-Typing Judgment)} & \Phi  \vdash_\mathcal{C}  \NDnt{t}  \NDsym{:}  \NDnt{X}\\
        \text{($\cat{L}$-Typing Judgment)} & \Gamma  \vdash_\mathcal{L}  \NDnt{s}  \NDsym{:}  \NDnt{A}\\
      \end{array}
    \end{math}
  \end{center}
\end{definition}

Now $\cat{C}$-typing judgments are denoted by $\Psi  \vdash_\mathcal{C}  \NDnt{t}  \NDsym{:}  \NDnt{X}$ where
$\Psi$ is a sequence of pairs of variables and their types, denoted by
$\NDmv{x}  \NDsym{:}  \NDnt{X}$, $\NDnt{t}$ is a $\cat{C}$-term, and $\NDnt{X}$ is a $\cat{C}$-type.  
The $\cat{C}$-terms are all standard, but $ \mathsf{G}\, \NDnt{s} $ corresponds to the
morphism part of the right-adjoint of the adjunction between both logics,
and $ \mathsf{ex}\, \NDnt{t_{{\mathrm{1}}}} , \NDnt{t_{{\mathrm{2}}}} \,\mathsf{with}\, \NDmv{x_{{\mathrm{1}}}} , \NDmv{x_{{\mathrm{2}}}} \,\mathsf{in}\, \NDnt{t_{{\mathrm{3}}}} $ is the introduction form for the
structural rule exchange.

The $\cat{L}$-typing judgment has the form $\Gamma  \vdash_\mathcal{L}  \NDnt{s}  \NDsym{:}  \NDnt{A}$ where $\Gamma$
is now a $\cat{L}$-context, denoted by $\Gamma$ or $\Delta$. These contexts
are ordered sequences of pairs of free variables with their types from
\emph{both} sides denoted by $\NDmv{x}  \NDsym{:}  \NDnt{B}$ and $\NDmv{x}  \NDsym{:}  \NDnt{X}$ respectively.
Finally, the term $\NDnt{s}$ is a $\cat{L}$-term, and $\NDnt{A}$ is a
$\cat{L}$-type.  Given two typing contexts $\Gamma$ and $\Delta$ we denote
their concatenation by $\Gamma  \NDsym{;}  \Delta$; we use a semicolon here to emphasize the
fact that the contexts are ordered. $\cat{L}$-terms correspond to
introduction and elimination forms for each of the previous types. For
example, $\NDnt{s_{{\mathrm{1}}}}  \triangleright  \NDnt{s_{{\mathrm{2}}}}$ introduces a tensor, and
$ \mathsf{let}\, \NDnt{s_{{\mathrm{1}}}}  :  \NDnt{A}  \triangleright  \NDnt{B} \,\mathsf{be}\, \NDmv{x}  \triangleright  \NDmv{y} \,\mathsf{in}\, \NDnt{s_{{\mathrm{2}}}} $ eliminates a tensor.

The typing rules for CNC logic can be found in
Figure~\ref{fig:CNC-typing-rules}.
\begin{figure}
  \footnotesize
  \resizebox{\columnwidth}{!}{
  \begin{tabular}{|c|}
    \hline\\
      \begin{mathpar}
      \NDdruleTXXid{} \and
      \NDdruleTXXunitI{} \and
      \NDdruleTXXunitE{} \and
      \NDdruleTXXtenI{} \and
      \NDdruleTXXtenE{} \and
      \NDdruleTXXimpI{} \and
      \NDdruleTXXimpE{} \and
      \NDdruleTXXGI{} \and
      \NDdruleTXXbeta{} \and
      \NDdruleTXXcut{}      
      \end{mathpar}
      \\
      \\
      \hline
      \\[5px]
    \begin{mathpar}
      \NDdruleSXXid{} \and
      \NDdruleSXXunitI{} \and
      \NDdruleSXXunitETwo{} \and
      \NDdruleSXXunitEOne{} \and
      \NDdruleSXXtenI{} \and
      \NDdruleSXXtenETwo{} \and
      \NDdruleSXXtenEOne{} \and
      \NDdruleSXXimprI{} \and
      \NDdruleSXXimprE{} \and
      \NDdruleSXXimplI{} \and
      \NDdruleSXXimplE{} \and
      \NDdruleSXXFI{} \and
      \NDdruleSXXFE{} \and
      \NDdruleSXXGE{} \and
      \NDdruleSXXbeta{} \and
      \NDdruleSXXcutTwo{} \and
      \NDdruleSXXcutOne{}
    \end{mathpar}\\\\
    \hline
  \end{tabular}
  }
  \caption{Typing Rules for CNC Logic}
  \label{fig:CNC-typing-rules}
\end{figure}
We split the figure in two: the top of the figure are the rules of
intuitionistic linear logic whose judgment is the $\mathcal{C}$-typing
judgment denoted by $\Psi  \vdash_\mathcal{C}  \NDnt{t}  \NDsym{:}  \NDnt{X}$, and the bottom of the figure
are the rules for the mixed commutative/non-commutative Lambek
calculus whose judgment is the $\mathcal{L}$-judgment denoted by
$\Gamma  \vdash_\mathcal{L}  \NDnt{s}  \NDsym{:}  \NDnt{A}$, and the two halves are connected via the rules 
rules $\NDdruleTXXGIName{}$, $\NDdruleSXXGEName{}$,
$\NDdruleSXXFIName{}$, and $\NDdruleSXXFEName{}$,
$\NDdruleSXXunitEOneName{}$, $\NDdruleSXXtenEOneName{}$, and
$\NDdruleSXXcutOneName{}$.

The one step $\beta$-reduction rules are listed in
Figure~\ref{fig:CNC-beta-reductions}. Similarly to the typing rules,
the figure is split in two: the top lists the rules of the
intuitionistic linear logic, and the bottom are those of the mixed
commutative/non-commutative Lambek calculus. 
\renewcommand{\NDdruleTbetaXXletUName}{}
\renewcommand{\NDdruleTbetaXXletTName}{}
\renewcommand{\NDdruleTbetaXXlamName}{}
\renewcommand{\NDdruleTbetaXXappOneName}{}
\renewcommand{\NDdruleTbetaXXappTwoName}{}
\renewcommand{\NDdruleTbetaXXappLetName}{}
\renewcommand{\NDdruleTbetaXXletLetName}{}
\renewcommand{\NDdruleTbetaXXletAppName}{}
\renewcommand{\NDdruleSbetaXXletUOneName}{}
\renewcommand{\NDdruleSbetaXXletTOneName}{}
\renewcommand{\NDdruleSbetaXXletTTwoName}{}
\renewcommand{\NDdruleSbetaXXletFName}{}
\renewcommand{\NDdruleSbetaXXlamLName}{}
\renewcommand{\NDdruleSbetaXXlamRName}{}
\renewcommand{\NDdruleSbetaXXapplOneName}{}
\renewcommand{\NDdruleSbetaXXapplTwoName}{}
\renewcommand{\NDdruleSbetaXXapprOneName}{}
\renewcommand{\NDdruleSbetaXXapprTwoName}{}
\renewcommand{\NDdruleSbetaXXderelictName}{}
\renewcommand{\NDdruleSbetaXXapplLetName}{}
\renewcommand{\NDdruleSbetaXXapprLetName}{}
\renewcommand{\NDdruleSbetaXXletLetName}{}
\renewcommand{\NDdruleSbetaXXletApplName}{}
\renewcommand{\NDdruleSbetaXXletApprName}{}
\renewcommand{\NDdruleTcomXXunitEXXunitEName}{}
\renewcommand{\NDdruleTcomXXunitEXXtenEName}{}
\renewcommand{\NDdruleTcomXXunitEXXimpEName}{}
\renewcommand{\NDdruleTcomXXtenEXXunitEName}{}
\renewcommand{\NDdruleTcomXXtenEXXtenEName}{}
\renewcommand{\NDdruleTcomXXtenEXXimpEName}{}
\renewcommand{\NDdruleTcomXXimpEXXunitEName}{}
\renewcommand{\NDdruleScomXXunitEXXunitEName}{}
\renewcommand{\NDdruleScomXXunitETwoXXunitEName}{}
\renewcommand{\NDdruleScomXXunitEXXimprEName}{}
\renewcommand{\NDdruleScomXXunitETwoXXimprEName}{}
\renewcommand{\NDdruleScomXXunitEXXFEName}{}
\renewcommand{\NDdruleScomXXunitETwoXXFEName}{}
\renewcommand{\NDdruleScomXXtenEXXunitEName}{}
\renewcommand{\NDdruleScomXXtenETwoXXunitEName}{}
\renewcommand{\NDdruleScomXXtenEXXtenEName}{}
\renewcommand{\NDdruleScomXXtenETwoXXtenEName}{}
\renewcommand{\NDdruleScomXXtenEXXimprEName}{}
\renewcommand{\NDdruleScomXXtenETwoXXimprEName}{}
\renewcommand{\NDdruleScomXXtenEXXimplEName}{}
\renewcommand{\NDdruleScomXXtenETwoXXimplEName}{}
\renewcommand{\NDdruleScomXXtenEXXFEName}{}
\renewcommand{\NDdruleScomXXtenETwoXXFEName}{}
\renewcommand{\NDdruleScomXXFEXXunitEName}{}
\renewcommand{\NDdruleScomXXFEXXtenEName}{}
\renewcommand{\NDdruleScomXXFEXXimprEName}{}
\renewcommand{\NDdruleScomXXFEXXimplEName}{}
\renewcommand{\NDdruleScomXXFEXXFEName}{}
\begin{figure}[!h]
  \footnotesize
  \resizebox{\columnwidth}{!}{
  \begin{tabular}{|c|}
    \hline\\
      \begin{mathpar}
      \NDdruleTbetaXXletU{} \and
      \NDdruleTbetaXXletT{} \and
      \NDdruleTbetaXXlam{}
      \end{mathpar}
      \\
      \\
      \hline
      \\
    \begin{mathpar}
      \NDdruleSbetaXXletUOne{} \and
      \NDdruleSbetaXXletTOne{} \and
      \NDdruleSbetaXXletTTwo{} \and
      \NDdruleSbetaXXletF{} \and
      \NDdruleSbetaXXlamL{} \and
      \NDdruleSbetaXXlamR{} \and
      \NDdruleSbetaXXderelict{}
    \end{mathpar}\\\\
    \hline
  \end{tabular}
  }
  \caption{$\beta$-reductions for CNC Logic}
  \label{fig:CNC-beta-reductions}
\end{figure}

The commuting conversions can be found in
Figures~\ref{fig:CNC-commutating-conversions-intuitionistic}-\ref{fig:CNC-commutating-conversions-both}. We
divide the rules into three parts due to the length. The first part,
Figure~\ref{fig:CNC-commutating-conversions-intuitionistic}, includes
the rules for the intuitionistic linear logic. The second,
Figure~\ref{fig:CNC-commutating-conversions-mixed}, includes the rules
for the commutative/non-commutative Lambek calculus. The third,
Figure~\ref{fig:CNC-commutating-conversions-both}, includes the mixed
rules $\NDdruleSXXunitEOneName{}$ and $\NDdruleSXXtenEOneName{}$.

\begin{figure}[!h]
  \footnotesize
  \resizebox{\columnwidth}{!}{
  \begin{tabular}{|c|}
    \hline\\
    \begin{mathpar}
      \NDdruleTcomXXunitEXXunitE{} \and
      \NDdruleTcomXXunitEXXtenE{} \and
      \NDdruleTcomXXunitEXXimpE{} \and
      \NDdruleTcomXXtenEXXunitE{} \and
      \NDdruleTcomXXtenEXXtenE{} \and
      \NDdruleTcomXXtenEXXimpE{} \and
      \NDdruleTcomXXimpEXXunitE{}
    \end{mathpar}
    \\
    \\
    \hline
  \end{tabular}
  }
  \caption{Commuting Conversions: Intuitionistic Linear Logic}
  \label{fig:CNC-commutating-conversions-intuitionistic}
\end{figure}
\begin{figure}[!h]
  \footnotesize
  \resizebox{\columnwidth}{!}{
  \begin{tabular}{|c|}
    \hline\\
    \begin{mathpar}
      \NDdruleScomXXunitEXXunitE{} \and
      \NDdruleScomXXunitEXXimprE{} \and
      \NDdruleScomXXunitEXXFE{} \and
      \NDdruleScomXXtenEXXunitE{} \and
      \NDdruleScomXXtenEXXtenE{} \and
      \NDdruleScomXXtenEXXimprE{} \and
      \NDdruleScomXXtenEXXimplE{} \and
      \NDdruleScomXXtenEXXFE{} \and
      \NDdruleScomXXFEXXunitE{} \and
      \NDdruleScomXXFEXXtenE{} \and
      \NDdruleScomXXFEXXimprE{} \and
      \NDdruleScomXXFEXXimplE{} \and
      \NDdruleScomXXFEXXFE{}
    \end{mathpar}\\\\
    \hline
  \end{tabular}
  }
  \caption{Commuting Conversions: Commutative/Non-commutative Lambek Calculus}
  \label{fig:CNC-commutating-conversions-mixed}
\end{figure}
\begin{figure}[!h]
  \footnotesize
  \resizebox{\columnwidth}{!}{
  \begin{tabular}{|c|}
    \hline\\
    \begin{mathpar}
      \NDdruleScomXXunitETwoXXunitE{} \and
      \NDdruleScomXXunitETwoXXimprE{} \and
      \NDdruleScomXXunitETwoXXFE{} \and
      \NDdruleScomXXtenETwoXXunitE{} \and
      \NDdruleScomXXtenETwoXXtenE{} \and
      \NDdruleScomXXtenETwoXXimprE{} \and
      \NDdruleScomXXtenETwoXXimplE{} \and
      \NDdruleScomXXtenETwoXXFE{}
    \end{mathpar}\\\\
    \hline
  \end{tabular}
  }
  \caption{Commuting Conversions: Mixed Rules}
  \label{fig:CNC-commutating-conversions-both}
\end{figure}

We also proved that the sequent calculus formalization given in
Figure~\ref{fig:CNC-sequent-calculus} is equivalent to the typing rules (or
else called the natural deduction formalization) given in
Figure~\ref{fig:CNC-typing-rules} are equivalent, as stated in the following
theorem.
\begin{theorem}
  \label{thm:sc-nd-equiv}
  The sequent calculus ($\mathit{SC}$) and natural deduction ($\mathit{ND}$)
  formalizations for CNC logic are equivalent in the sense that there are
  two mappings $N:\mathit{SC}\rightarrow\mathit{ND}$ and
  $S:\mathit{ND}\rightarrow\mathit{SC}$ that map each rule in $\mathit{SC}$
  to a proof in $\mathit{ND}$, and each rule in $\mathit{ND}$ to a proof
  in $\mathit{SC}$, respectively.
\end{theorem}
\begin{proof}
  The proof is done case by case on each rule in the sequence calculus and
  natural deduction formalizations. It is obvious that the axioms in one
  formalization can be mapped to the axioms in the other. The introduction
  rules in $\mathit{ND}$ are mapped to the right rules in $\mathit{SC}$, and
  vice versa. The elimination rules and lefts rules are mapped to each other
  with some fiddling. For instance, the elimination rule for the
  non-commutative tensor is mapped to the following proof in $\mathit{SC}$:
  \begin{center}
    \scriptsize
    \begin{math}
      $$\mprset{flushleft}
      \inferrule* [right={\scriptsize $\ElledruleTXXcutName$}] {
        {\Phi  \vdash_\mathcal{C}  \NDnt{t_{{\mathrm{1}}}}  \NDsym{:}  \NDnt{X}  \otimes  \NDnt{Y}} \\
        $$\mprset{flushleft}
        \inferrule* [right={\scriptsize $\ElledruleTXXtenLName$}] {
          {\Psi_{{\mathrm{1}}}  \NDsym{,}  \NDmv{x}  \NDsym{:}  \NDnt{X}  \NDsym{,}  \NDmv{y}  \NDsym{:}  \NDnt{Y}  \NDsym{,}  \Psi_{{\mathrm{2}}}  \vdash_\mathcal{C}  \NDnt{t_{{\mathrm{2}}}}  \NDsym{:}  \NDnt{Z}}
        }{\Psi_{{\mathrm{1}}}  \NDsym{,}  \NDmv{z}  \NDsym{:}  \NDnt{X}  \otimes  \NDnt{Y}  \NDsym{,}  \Psi_{{\mathrm{2}}}  \vdash_\mathcal{C}   \mathsf{let}\, \NDmv{z}  :  \NDnt{X}  \otimes  \NDnt{Y} \,\mathsf{be}\, \NDmv{x}  \otimes  \NDmv{y} \,\mathsf{in}\, \NDnt{t_{{\mathrm{2}}}}   \NDsym{:}  \NDnt{Z}}
      }{\Psi_{{\mathrm{1}}}  \NDsym{,}  \Phi  \NDsym{,}  \Psi_{{\mathrm{2}}}  \vdash_\mathcal{C}  \NDsym{[}  \NDnt{t_{{\mathrm{1}}}}  \NDsym{/}  \NDmv{z}  \NDsym{]}  \NDsym{(}   \mathsf{let}\, \NDmv{z}  :  \NDnt{X}  \otimes  \NDnt{Y} \,\mathsf{be}\, \NDmv{x}  \otimes  \NDmv{y} \,\mathsf{in}\, \NDnt{t_{{\mathrm{2}}}}   \NDsym{)}  \NDsym{:}  \NDnt{Z}}
    \end{math}
  \end{center}
  The full proof is in Appendix~\ref{app:sc-nd-equiv}.
\end{proof}




\section{An Adjoint Model}
\label{sec:adjoint-model}
In this section we introduce Lambek Adjoint Models (LAMs). Benton's
LNL model consists of a symmetric monoidal adjunction
$F:\cat{C}\dashv\cat{L}:G$ between a Cartesian closed category
$\cat{C}$ and a symmetric monoidal closed category $\cat{L}$. LAM
consists of a monoidal adjunction between a symmetric monoidal closed
category and a Lambek category.
\begin{definition}
  \label{def:lambek-category}
  A \textbf{Lambek category} is a monoidal category $(\cat{L},\tri,I',\alpha',\lambda',\rho')$
  with two functors $- \rightharpoonup - : \cat{L}^{\mathsf{op}} \times \cat{L} \mto \cat{L}$ and
  $- \leftharpoonup - : \cat{L} \times \cat{L}^{\mathsf{op}} \mto \cat{L}$ such that the following
  two natural bijections hold:
  \[
  \begin{array}{lllll}
    \Hom{L}{A \tri B}{C} \cong \Hom{L}{A}{B \rightharpoonup C} & \quad &
    \Hom{L}{A \tri B}{C} \cong \Hom{L}{B}{C \leftharpoonup A}\\
  \end{array}
  \]  
\end{definition}
\noindent
Lambek categories are also known as monoidal bi-closed categories.

\begin{definition}
  A \textbf{Lambek Adjoint Model (LAM)}, $(\cat{C},\cat{L},F,G,\eta,\varepsilon)$, consists of
  \begin{itemize}
  \item a symmetric monoidal closed category $(\cat{C},\otimes,I,\alpha,\lambda,\rho)$;
  \item a Lambek category $(\cat{L},\tri,I',\alpha',\lambda',\rho')$;
  \item a monoidal adjunction $F:\cat{C}\dashv\cat{L}:G$ with unit $\eta:\Id_{\cat{C}} \rightarrow GF$ and
        counit $\varepsilon:FG\rightarrow \Id_\cat{L}$, where $(F:\cat{C}\rightarrow\cat{L}, m)$
        and $(G:\cat{L}\rightarrow\cat{C}, n)$ are monoidal functors.
  \end{itemize}
\end{definition}
\noindent
Following the tradition, we use letters $X$, $Y$, $Z$ for objects in
$\cat{C}$ and $A$, $B$, $C$ for objects in $\cat{L}$. The rest of this
section proves essential properties of any LAM.

\textbf{An isomorphism}. Let $(\cat{C},\cat{L},F,G,\eta,\varepsilon)$
be a LAM, where $(F,\m{})$ and $(G,\n{})$ are monoidal
functors. Similarly as in Benton's LNL model, $\m{X,Y} : FX \tri FY
\mto F(X \otimes Y)$ are components of a natural isomorphism, and
$\m{I} : I' \mto FI$ is an isomorphism. This is essential for modeling
certain rules of CNC logic, such as tensor elimination in natural
deduction.  We define the inverses of $\m{X,Y}:FX\tri FY\rightarrow
F(X\otimes Y)$ and $\m{I}:I'\rightarrow FI$ as:
\vspace{-1em}
\begin{mathpar}
\footnotesize
\bfig
  \morphism<1000,0>[\p{X,Y}:F(X\otimes Y)`F(GFX\otimes GFY);F(\eta_X\otimes\eta_Y)]
  \morphism(1000,0)<900,0>[F(GFX\otimes GFY)`FG(FX\tri FY);F\n{FX,FY}]
  \morphism(1900,0)<750,0>[FG(FX\tri FY)`FX\tri FY;\varepsilon_{FX\tri FX}]
\efig
\end{mathpar}
\vspace{-1.6em}
\begin{mathpar}
\footnotesize
\bfig
  \morphism<500,0>[\p{I}:FI`FGI';F\n{I'}]
  \morphism(500,0)<400,0>[FGI'`I';\varepsilon_{I'}]
\efig
\end{mathpar}
Due to \cite{kelly1974doctrinal}, it can be easily shown that $\m{I}$ is an
isomorphism with inverse, and that $\m{X,Y}$ are components of a natural
isomorphism with inverses $\p{X,Y}$.

\textbf{Strong non-commutative monad}. Next we show that the monad on
$\cat{C}$ in LAM is strong but non-commutative. In Benton's LNL model,
the monad on the Cartesian closed category is commutative, but later
Benton and Wadler \cite{Benton:1996} wonder, is it possible to model
non-commutative monads using adjoint models similar to LNL models? The
following shows that LAMs correspond to strong non-commutative
monoidal monads.
\begin{lemma}
\label{lem:monoidal-monad}
The monad induced by any LAM, $GF : \cat{C} \mto \cat{C}$, is monoidal.
\end{lemma}
\begin{proof}
  The proof is done by checking the conditions for a functor being monoidal.
  The detail of the proof is in Appendix~\ref{app:monoidal-monad}.
\end{proof}
\noindent
However, the monad is not symmetric because the following diagram does
not commute.
\begin{mathpar}
\bfig
  \ptriangle/->`->`/<900,400>[
    GFX\otimes GFY`GFY\otimes GFX`G(FX\tri FY);\e{GFX,GFY}`\n{FX,FY}`]
  \morphism(900,400)<900,0>[GFY\otimes GFX`G(FY\tri FX);\n{FY,FX}]
  \dtriangle(900,0)/`->`->/<900,400>[
    G(FY\tri FX)`GF(X\otimes Y)`GF(Y\otimes X);`G\m{Y,X}`GF\e{X,Y}]
  \morphism|b|<900,0>[G(FX\tri FY)`GF(X\otimes Y);G\m{X,Y}]
\efig
\end{mathpar}
Commutativity fails, because the functors defining the monad are not
symmetric monoidal, but only monoidal. This means that the diagram
\[
\bfig
\square<700,400>[
  FA\otimes'FB`FB\otimes'FA`F(A\otimes B)`F(B\otimes A);
  \e{FA,FB}`\m{A,B}`\m{B,A}`F\e{A,B}]
\efig
\]
does not hold for $G$ nor $F$.  However, we can prove the monad is
strong.
\begin{lemma}
  \label{lem:strong-monad}
  The monad, $GF : \cat{C} \mto \cat{C}$, on the symmetric monoidal
  closed category in LAM is strong.
\end{lemma}
\begin{proof}
The proof is done by first defining a natural transformation $\tau$, called
the \textbf{tensorial strength}, with components
$\tau_{A,B}:A\tri TB\rightarrow T(A\tri B)$, and then proving the
commutativity of several diagrams through diagram chasing. The formal
definition for a strong monad and the full proof are in
Appendix~\ref{app:strong-monad}.
\end{proof}
\noindent
Finally, we obtain the non-commutativity of the monad induced by some 
LAM as follows.
\begin{lemma}[Due to Kock~\cite{kock1972strong}]
\label{lem:monad-com-iff-sym}
  Let $\cat{M}$ be a symmetric monoidal category and $T$ be a strong monad on $\cat{M}$. Then
  $T$ is commutative iff it is symmetric monoidal.
\end{lemma}

\begin{theorem}
  There exists a LAM whose monad, $GF : \cat{C} \mto \cat{C}$, on the SMCC
  in the LAM is strong but non-commutative.
\end{theorem}
\begin{proof}
  This proof follows from Lemma~\ref{lem:strong-monad} and
  Lemma~\ref{lem:monad-com-iff-sym}.
\end{proof}

\textbf{Comonad for exchange}.  We conclude this section by showing
that the comonad induced by some LAM is monoidal and extends $\cat{L}$
with exchange. The former is proved in \cite{kelly1974doctrinal}. The
latter is shown by proving that its corresponding co-Eilenberg-Moore
category is symmetric monoidal.

\begin{theorem}
  \label{thm:em-exchange}
  Given a LAM $(\cat{C},\cat{L},F,G,\eta,\varepsilon)$ and the comonad
  $FG : \cat{L} \mto \cat{L}$, the co-Eilenberg-Moore category
  $\cat{L}^{FG}$ has an exchange natural transformation $\e{A,B}^{FG}:A\tri
  B\rightarrow B\tri A$, and $\e{A,B}^{FG}$ is a symmetry, i.e.,
  $\e{A,B}^{FG} \circ \e{A,B}^{FG} = id_A$.
\end{theorem}
\begin{proof}
  The natural transformation $\e{A,B}^{FG}:A\tri B\rightarrow B\tri A$ is defied
  as follows:
  $$\bfig
    \morphism<600,0>[A\tri B`FGA\tri FGB;h_A\tri h_B]
    \morphism(600,0)<800,0>[FGA\tri FGB`F(GA\otimes GB);\m{GA,GB}]
    \morphism(1400,0)<800,0>[F(GA\otimes GB)`F(GB\otimes GA);F\e{GA,GB}]
    \morphism(2200,0)<700,0>[F(GB\otimes GA)`FG(B\tri A);F\n{B,A}]
    \morphism(2900,0)<500,0>[FG(B\tri A)`B\tri A;\varepsilon_{B\tri A}]
  \efig$$
  in which $\e{}$ is the exchange for $\cat{C}$. Then $\e{}^{FG}$ is a
  natural transformation because the following diagrams commute for
  morphisms $f:A\rightarrow A'$ and $g:B\rightarrow B'$:
  \begin{mathpar}
  \bfig
    \square|almb|<700,400>[
      A\tri B`FGA\tri FGB`A'\tri B'`FGA'\tri FGB';
      h_A\tri h_B`f\tri g`FGf\tri FGg`h_{A'}\tri h_{B'}]
    \square(700,0)|ammb|/->``->`->/<800,400>[
      FGA\tri FGB`F(GA\otimes GB)`FGA'\tri FGB'`F(GA'\otimes GB');
      \m{GA,GB}``F(Gf\otimes Gg)`\m{GA',GB'}]
    \square(1500,0)|ammb|/->``->`->/<800,400>[
      F(GA\otimes GB)`F(GB\otimes GA)`F(GA'\otimes GB')`F(GB'\otimes GA');
      F\e{A,B}``F(Gg\otimes Gf)`F\e{A',B'}]
    \square(2300,0)|ammb|/->``->`->/<800,400>[
      F(GB\otimes GA)`FG(B\tri A)`F(GB'\otimes GA')`FG(B'\tri A');
      F\n{B,A}``FG(g\tri f)`F\n{B',A'}]
    \square(3100,0)|amrb|/->``->`->/<600,400>[
      FG(B\tri A)`B\tri A`FG(B'\tri A')`B'\tri A';
      \varepsilon_{B\tri A}``g\tri f`\varepsilon_{B'\tri A'}]
  \efig
  \end{mathpar}
  $\e{A,B}^{FG}$ is a symmetry because the following diagrams commute:
  \begin{mathpar}
  \bfig
    \ptriangle|amm|/->`=`->/<600,600>[
      A\tri B`FGA\tri FGB`A\tri B;h_A\tri h_B``\varepsilon_A\tri\varepsilon_B]
    \morphism(600,0)|b|<-600,0>[FG(A\tri B)`A\tri B;\varepsilon_{A\tri B}]
    \morphism(600,600)<700,0>[FGA\tri FGB`F(GA\otimes GB);\m{GA,GB}]
    \morphism(1300,0)|b|<-700,0>[F(GA\otimes GB)`FG(A\tri B);F\n{A,B}]
    \square(1300,0)/->`=`=`<-/<800,600>[
      F(GA\otimes GB)`F(GB\otimes GA)`F(GA\otimes GB)`F(GB\otimes GA);F\e{A,B}```F\e{B,A}]
    \qtriangle(2100,0)|amr|/->``->/<1400,600>[F(GB\otimes GA)`FG(B\tri A)`B\tri A;F\n{B,A}``\varepsilon_{B\tri A}]
    \morphism(2900,0)|b|<-800,0>[FGB\tri FGA`F(GB\otimes GA);\m{GB,GA}]
    \btriangle(2900,0)|mmb|/<-`=`<-/<600,400>[
      B\tri A`FGB\tri FGA`B\tri A;
      \varepsilon_B\tri\varepsilon_A``h_A\tri h_A])
  \efig
  \end{mathpar}
\end{proof}




\section{A Model in Dialectica Spaces}
\label{sec:a-model-in-dialectica-spaces}
\newcommand{\Set}{\mathsf{Set}}
\newcommand{\Dial}[2]{\mathsf{Dial}_{#1}(#2)}

In this section we give a different categorical model in terms of
dialectica categories; which are a sound and complete categorical
model of the Lambek Calculus as was shown by de Paiva and Eades
\cite{dePaiva2018}. This section is largely the same as the
corresponding section de Paiva and Eades give, but with some
modifications to their definition of biclosed posets with exchange
(see Definition~\ref{def:biclosed-exchange}).  However, we try to make
this section as self contained as possible.

Dialectica categories were first introduced by de Paiva as a
categorification of G\"odel's Dialectica interpretation
\cite{depaiva1990}.  Dialectica categories were one of the first sound
categorical models of intuitionistic linear logic with linear
modalities.  We show in this section that they can be adapted to
become a sound and complete model for CNC logic, with both the
exchange and of-course modalities.  Due to the complexities of working
with dialectica categories we have formally verified\footnote{The
  complete formalization can be found online at
  {\tiny \url{https://bit.ly/2TpoyWU.}}}
this section in the proof assistant Agda~\cite{bove2009}.

First, we define the notion of a biclosed poset.  These are used to
control the definition of morphisms in the dialectica model.
\begin{definition}
  \label{def:biclosed-poset}
  Suppose $(M, \leq, \circ, e)$ is an ordered non-commutative monoid.
  If there exists a largest $x \in M$ such that $a \circ x \leq b$ for
  any $a, b \in M$, then we denote $x$ by $a \lto b$ and called it
  the \textbf{left-pseudocomplement} of $a$ w.r.t $b$.  Additionally,
  if there exists a largest $x \in M$ such that $x \circ a \leq b$ for
  any $a, b \in M$, then we denote $x$ by $b \rto a$ and called it
  the \textbf{right-pseudocomplement} of $a$ w.r.t $b$.

  A \textbf{biclosed poset}, $(M, \leq, \circ, e, \lto, \rto)$, is an
  ordered non-commutative monoid, $(M, \leq, \circ, e)$, such that $a
  \lto b$ and $b \rto a$ exist for any $a,b \in M$.
\end{definition}
Now using the previous definition we define dialectica Lambek spaces.
\begin{definition}
  \label{def:dialectica-lambek-spaces}
  Suppose $(M, \leq, \circ, e, \lto, \rto)$ is a biclosed poset. Then
  we define the category of \textbf{dialectica Lambek spaces},
  $\mathsf{Dial}_M(\Set)$, as follows:
  \begin{itemize}
  \item[-] objects, or dialectica Lambek spaces, are triples $(U, X,
    \alpha)$ where $U$ and $X$ are sets, and $\alpha : U \times X \mto
    M$ is a generalized relation over $M$, and

  \item[-] maps that are pairs $(f, F) : (U , X, \alpha) \mto (V , Y ,
    \beta)$ where $f : U \mto V$, and $F : Y \mto X$ are functions
    such that the weak adjointness condition
    $\forall u \in U.\forall y \in Y. \alpha(u , F(y)) \leq \beta(f(u), y)$
    holds.
  \end{itemize}
\end{definition}
Notice that the biclosed poset is used here as the target of the
relations in objects, but also as providing the order  relation in the weak adjoint condition on morphisms.  This will allow the structure of the biclosed
poset to lift up into $\Dial{M}{\Set}$.

We will show that $\Dial{M}{\Set}$ is a model of the Lambek Calculus
with modalities.  First, we must show that $\Dial{M}{\Set}$ is
monoidal biclosed.
\begin{definition}
  \label{def:dial-monoidal-structure}
  Suppose $(U, X, \alpha)$ and $(V, Y, \beta)$ are two objects of
  $\Dial{M}{\Set}$. Then their tensor product is defined as follows:
  \[ \small
  (U, X, \alpha) \rhd (V, Y, \beta) = (U \times V, (V \to X) \times (U \to Y), \alpha \rhd \beta)
  \]
  where $- \to -$ is the function space from $\Set$, and $(\alpha
  \rhd \beta)((u, v), (f, g)) = \alpha(u, f(v)) \circ \beta(g(u), v)$.

  \ \\
  \noindent
  The unit of the above tensor product is defined as follows:
  \[ \small
  I = (\top , \top , \iota)
  \]
  where $\top$ is the initial object in $\Set$, and $\iota(*,*) = e$.
\end{definition}

\noindent
It follows from de Paiva and Eades \cite{dePaiva2018} that this does
indeed define a monoidal tensor product, but take note of the fact
that this tensor product is indeed non-commutative, because the
non-commutative multiplication of the biclosed poset is used to define
the relation of the tensor product.

The tensor product has two right adjoints making $\Dial{M}{\Set}$
biclosed.
\begin{definition}
  \label{def:dial-is-biclosed}
  Suppose $(U, X, \alpha)$ and $(V, Y, \beta)$ are two objects of
  $\Dial{M}{\Set}$. Then two internal-homs can be defined as follows:
  \[ \small
  \begin{array}{lll}
    (U, X, \alpha) \lto (V, Y, \beta) = ((U \to V) \times (Y \to X), U \times Y, \alpha \lto \beta)\\
    (V, Y, \beta) \rto (U, X, \alpha) = ((U \to V) \times (Y \to X), U \times Y, \alpha \rto \beta)\\
  \end{array}
  \]
\end{definition}
\noindent
It is straightforward to show that the typical bijections defining the
corresponding adjunctions hold; see de Paiva and Eades for the details
\cite{dePaiva2018}.

We now extend $\Dial{M}{\Set}$ with two modalities: the usual
modality, of-course, denoted $!A$, and the exchange modality denoted
$\xi A$.  However, we must first extended biclosed posets to
include an exchange operation.
\begin{definition}
  \label{def:biclosed-exchange}
  A \textbf{biclosed poset with exchange} is a biclosed poset $(M,
  \leq, \circ, e, \lto, \rto)$ equipped with an unary operation
  $\xi : M \to M$ satisfying the following:
  \[ \small
  \setlength{\arraycolsep}{4px}
  \begin{array}{lll}
    \begin{array}{lll}
    \text{(Compatibility)} & a \leq b \text{ implies } \xi a \leq \xi b \text{ for all } a,b,c \in M\\
    \text{(Minimality)} & \xi a \leq a \text{ for all } a \in M\\    
  \end{array}
  &
  \begin{array}{lll}
    \text{(Duplication)} & \xi a \leq \xi\xi a \text{ for all } a \in M\\
    \text{(Exchange)} & (\xi a \circ \xi b) \leq (\xi b \circ \xi a) \text{ for all } a, b \in M\\
  \end{array}
  \end{array}
  \]
\end{definition}
\noindent
This definition is where the construction given here departs from the
definition of biclosed posets with exchange given by de Paiva and
Eades \cite{dePaiva2018}.

We can now define the two modalities in $\Dial{M}{\Set}$ where $M$ is
a biclosed poset with exchange.
\begin{definition}
  \label{def:modalities-dial}
  Suppose $(U, X, \alpha)$ is an object of $\Dial{M}{\Set}$ where $M$
  is a biclosed poset with exchange. Then the \textbf{of-course} and
  \textbf{exchange} modalities can be defined as 
  $! (U, X, \alpha) = (U, U \to X^*, !\alpha)$ and
  $\xi (U, X, \alpha) = (U, X, \xi \alpha)$
  where $X^*$ is the free commutative monoid on $X$, $(!\alpha)(u, f)
  = \alpha(u, x_1) \circ \cdots \circ \alpha(u, x_i)$ for $f(u) =
  (x_1, \ldots, x_i)$, and $(\xi \alpha)(u, x) = \xi (\alpha(u,
  x))$.
\end{definition}
This definition highlights a fundamental difference between the two
modalities.  The definition of the exchange modality relies on an
extension of biclosed posets with essentially the exchange modality in
the category of posets.  However, the of-course modality is defined by
the structure already present in $\Dial{M}{\Set}$, specifically, the
structure of $\Set$.

Both of the modalities have the structure of a comonad.  That is,
there are monoidal natural transformations $\varepsilon_! : !A \mto
A$, $\varepsilon_\xi : \xi A \mto A$, $\delta_! : !A \mto !!A$,
and $\delta_\xi : \xi A \mto \xi\xi A$ which satisfy the
appropriate diagrams; see the formalization for the full
proofs. Furthermore, these comonads come equipped with arrows $w : !A
\mto I$, $d : !A \mto !A \otimes !A$, $\e{A,B} : \xi A \otimes \xi B \mto \xi B
\otimes \xi A$. 

Finally, using the fact that $\Dial{M}{\Set}$ for any biclosed poset
is essentially a non-commutative formalization of Bierman's linear
categories \cite{Bierman:1994} we can use Benton's construction of an
LNL model from a linear category to obtain a LAM model, and hence,
obtain the following.
\begin{theorem}
  \label{theorem:sound-dial-exchange-!}
  Suppose $M$ is a biclosed poset with exchange.  Then
  $\Dial{M}{\Set}$ is a sound model for CNC logic.
\end{theorem}


\section{Future Work}
\label{sec:future-work}
We introduce the idea above of having a modality for exchange, but
what about individual modalities for weakening and contraction?
Indeed it is possible to give modalities for these structural rules as
well using adjoint models.  Now that we have each structural rule
isolated into their own modality is it possible to put them together
to form new modalities that combine structural rules?  The answer to
this question has already been shown to be positive, at least for
weakening and contraction, by Melli{\'e}s~\cite{Mellies:2004}, but we
plan to extend this line of work to include exchange.

The monads induced by the adjunction in CNC logic is non-commutative,
but Benton and Wadler show that the monads induced by the adjunction
in LNL logic \cite{Benton:1996} are commutative.  Using the extension
of Melli{\'e}s' work we mention above would allow us to combine both
CNC logic with LNL logic, and then be able to support both commutative
monads as well as non-commutative monads.  We plan on exploring this
in the future.

Hasegawa~\cite{EPTCS238.6} studies the linear of-course modality,
$!A$, as a comonad induced by an adjunction between a Cartesian closed
category a (non-symmetric) monoidal category.  The results here
generalizes his by generalizing the Cartesian closed category to a
symmetric monoidal closed category.  However, his approach focuses on
the comonad rather than the adjunctions.  It would be interesting to
do the same for LAM as well.


\bibliographystyle{plainurl}
\bibliography{ref}

\appendix
\label{sec:appendix}
\section{Proof For Lemma~\ref{lem:cut-reduction}}
\label{app:cut-reduction}

\subsection{Commuting Conversion Cut vs. Cut}

\subsubsection{$\SCdruleTXXcutName$ vs. $\SCdruleTXXcutName$}
\begin{itemize}
\item Case 1:
      \begin{center}
        \scriptsize
        \begin{math}
          \begin{array}{c}
            \Pi_1 \\
            {\Phi  \vdash_\mathcal{C}  \SCnt{X}}
          \end{array}
        \end{math}
        \qquad\qquad
        $\Pi_2:$
        \begin{math}
          $$\mprset{flushleft}
          \inferrule* [right={\tiny cut}] {
            {
              \begin{array}{cc}
                \pi_1 & \pi_2 \\
                {\Psi_{{\mathrm{2}}}  \SCsym{,}  \SCnt{X}  \SCsym{,}  \Psi_{{\mathrm{3}}}  \vdash_\mathcal{C}  \SCnt{Y}} & {\Psi_{{\mathrm{1}}}  \SCsym{,}  \SCnt{Y}  \SCsym{,}  \Psi_{{\mathrm{4}}}  \vdash_\mathcal{C}  \SCnt{Z}}
              \end{array}
            }
          }{\Psi_{{\mathrm{1}}}  \SCsym{,}  \Psi_{{\mathrm{2}}}  \SCsym{,}  \SCnt{X}  \SCsym{,}  \Psi_{{\mathrm{3}}}  \SCsym{,}  \Psi_{{\mathrm{4}}}  \vdash_\mathcal{C}  \SCnt{Z}}
        \end{math}
      \end{center}
      By assumption, $c(\Pi_1),c(\Pi_2)\leq |X|$. Therefore, $c(\pi_1)$,
      $c(\pi_2)\leq |X|$. Since $Y$ is the cut formula on $\pi_1$ and
      $\pi_2$, we have $|Y|+1\leq|X|$. By induction on $\Pi_1$ and $\pi_1$
      there exists a proof $\Pi'$ for sequent $\Psi_{{\mathrm{2}}}  \SCsym{,}  \Phi  \SCsym{,}  \Psi_{{\mathrm{3}}}  \vdash_\mathcal{C}  \SCnt{Y}$ s.t.
      $c(\Pi')\leq|X|$. So $\Pi$ can be constructed as follows, with
      $c(\Pi)\leq max\{c(\Pi'),c(\pi_2),|Y|+1\}\leq |X|$.
      \begin{center}
        \scriptsize
        \begin{math}
          $$\mprset{flushleft}
          \inferrule* [right={\tiny cut}] {
            {
              \begin{array}{cc}
                \Pi' & \pi_2 \\
                {\Psi_{{\mathrm{2}}}  \SCsym{,}  \Phi  \SCsym{,}  \Psi_{{\mathrm{3}}}  \vdash_\mathcal{C}  \SCnt{Y}} & {\Psi_{{\mathrm{1}}}  \SCsym{,}  \SCnt{Y}  \SCsym{,}  \Psi_{{\mathrm{4}}}  \vdash_\mathcal{C}  \SCnt{Z}}
              \end{array}
            }
          }{\Psi_{{\mathrm{1}}}  \SCsym{,}  \Psi_{{\mathrm{2}}}  \SCsym{,}  \Phi  \SCsym{,}  \Psi_{{\mathrm{3}}}  \SCsym{,}  \Psi_{{\mathrm{4}}}  \vdash_\mathcal{C}  \SCnt{Z}}
        \end{math}
      \end{center}

\item Case 2:
      \begin{center}
        \scriptsize
        $\Pi_1$:
        \begin{math}
          $$\mprset{flushleft}
          \inferrule* [right={\tiny cut}] {
            {
              \begin{array}{cc}
                \pi_1 & \pi_2 \\
                {\Phi  \vdash_\mathcal{C}  \SCnt{X}} & {\Psi_{{\mathrm{2}}}  \SCsym{,}  \SCnt{X}  \SCsym{,}  \Psi_{{\mathrm{3}}}  \vdash_\mathcal{C}  \SCnt{Y}}
              \end{array}
            }
          }{\Psi_{{\mathrm{2}}}  \SCsym{,}  \Phi  \SCsym{,}  \Psi_{{\mathrm{3}}}  \vdash_\mathcal{C}  \SCnt{Y}}
        \end{math}
        \qquad\qquad
        \begin{math}
          \begin{array}{c}
            \Pi_2 \\
            {\Psi_{{\mathrm{1}}}  \SCsym{,}  \SCnt{Y}  \SCsym{,}  \Psi_{{\mathrm{4}}}  \vdash_\mathcal{C}  \SCnt{Z}}
          \end{array}
        \end{math}
      \end{center}
      By assumption, $c(\Pi_1),c(\Pi_2)\leq |Y|$. Since the cut rank of the last cut in
      $\Pi_1$ is $|X|+1$, then $|X|+1\leq |Y|$. By induction on $\Pi_1$ and $\Pi_2$, there is
      a proof $\Pi'$ for sequent $\Psi_{{\mathrm{1}}}  \SCsym{,}  \Psi_{{\mathrm{2}}}  \SCsym{,}  \SCnt{X}  \SCsym{,}  \Psi_{{\mathrm{3}}}  \SCsym{,}  \Psi_{{\mathrm{4}}}  \vdash_\mathcal{C}  \SCnt{Z}$ s.t. $c(\Pi')\leq|Y|$.
      Therefore, the proof $\Pi$ can be constructed as follows, and
      $c(\Pi)\leq max\{c(\pi_1),c(\Pi'),|X|+1\}\leq |Y|$.
      \begin{center}
        \scriptsize
        \begin{math}
          $$\mprset{flushleft}
          \inferrule* [right={\tiny cut}] {
            {
              \begin{array}{cc}
                \pi_1 & \Pi' \\
                {\Phi  \vdash_\mathcal{C}  \SCnt{X}} & {\Psi_{{\mathrm{1}}}  \SCsym{,}  \Psi_{{\mathrm{2}}}  \SCsym{,}  \SCnt{X}  \SCsym{,}  \Psi_{{\mathrm{3}}}  \SCsym{,}  \Psi_{{\mathrm{4}}}  \vdash_\mathcal{C}  \SCnt{Z}}
              \end{array}
            }
          }{\Psi_{{\mathrm{1}}}  \SCsym{,}  \Psi_{{\mathrm{2}}}  \SCsym{,}  \Phi  \SCsym{,}  \Psi_{{\mathrm{3}}}  \SCsym{,}  \Psi_{{\mathrm{4}}}  \vdash_\mathcal{C}  \SCnt{Z}}
        \end{math}
      \end{center}
\end{itemize}

\subsubsection{$\SCdruleTXXcutName$ vs. $\SCdruleSXXcutOneName$}
\begin{itemize}
\item Case 1:
      \begin{center}
        \scriptsize
        \begin{math}
          \begin{array}{c}
            \Pi_1 \\
            {\Phi  \vdash_\mathcal{C}  \SCnt{X}}
          \end{array}
        \end{math}
        \qquad\qquad
        $\Pi_2:$
        \begin{math}
          $$\mprset{flushleft}
          \inferrule* [right={\tiny cut1}] {
            {
              \begin{array}{cc}
                \pi_2 & \pi_3 \\
                {\Psi_{{\mathrm{1}}}  \SCsym{,}  \SCnt{X}  \SCsym{,}  \Psi_{{\mathrm{2}}}  \vdash_\mathcal{C}  \SCnt{Y}} & {\Gamma_{{\mathrm{1}}}  \SCsym{;}  \SCnt{Y}  \SCsym{;}  \Gamma_{{\mathrm{2}}}  \vdash_\mathcal{L}  \SCnt{A}}
              \end{array}
            }
          }{\Gamma_{{\mathrm{1}}}  \SCsym{;}  \Psi_{{\mathrm{1}}}  \SCsym{;}  \SCnt{X}  \SCsym{;}  \Psi_{{\mathrm{2}}}  \SCsym{;}  \Gamma_{{\mathrm{2}}}  \vdash_\mathcal{L}  \SCnt{A}}
        \end{math}
      \end{center}
      By assumption, $c(\Pi_1),c(\Pi_2)\leq |X|$. Therefore, $c(\pi_1)$,
      $c(\pi_2)\leq |X|$. Since $Y$ is the cut formula on $\pi_1$ and
      $\pi_2$, we have $|Y|+1\leq|X|$. By induction on $\Pi_1$ and $\pi_1$,
      there exists a proof $\Pi'$ for sequent $\Psi_{{\mathrm{1}}}  \SCsym{,}  \Phi  \SCsym{,}  \Psi_{{\mathrm{2}}}  \vdash_\mathcal{C}  \SCnt{Y}$ s.t.
      $c(\Pi')\leq|X|$. So $\Pi$ can be constructed as follows, with
      $c(\Pi)\leq max\{c(\Pi'),c(\pi_2),|Y|+1\}\leq |X|$.
      \begin{center}
        \scriptsize
        \begin{math}
          $$\mprset{flushleft}
          \inferrule* [right={\tiny cut1}] {
            {
              \begin{array}{cc}
                \Pi' & \pi_2 \\
                {\Psi_{{\mathrm{1}}}  \SCsym{,}  \Phi  \SCsym{,}  \Psi_{{\mathrm{2}}}  \vdash_\mathcal{C}  \SCnt{Y}} & {\Gamma_{{\mathrm{1}}}  \SCsym{;}  \SCnt{Y}  \SCsym{;}  \Gamma_{{\mathrm{2}}}  \vdash_\mathcal{L}  \SCnt{A}}
              \end{array}
            }
          }{\Gamma_{{\mathrm{1}}}  \SCsym{;}  \Psi_{{\mathrm{1}}}  \SCsym{;}  \Phi  \SCsym{;}  \Psi_{{\mathrm{2}}}  \SCsym{;}  \Gamma_{{\mathrm{2}}}  \vdash_\mathcal{L}  \SCnt{A}}
        \end{math}
      \end{center}

\item Case 2:
      \begin{center}
        \scriptsize
        $\Pi_1$:
        \begin{math}
          $$\mprset{flushleft}
          \inferrule* [right={\tiny cut}] {
            {
              \begin{array}{cc}
                \pi_1 & \pi_2 \\
                {\Phi  \vdash_\mathcal{C}  \SCnt{X}} & {\Psi_{{\mathrm{1}}}  \SCsym{,}  \SCnt{X}  \SCsym{,}  \Psi_{{\mathrm{2}}}  \vdash_\mathcal{C}  \SCnt{Y}}
              \end{array}
            }
          }{\Psi_{{\mathrm{1}}}  \SCsym{,}  \Phi  \SCsym{,}  \Psi_{{\mathrm{2}}}  \vdash_\mathcal{C}  \SCnt{Y}}
        \end{math}
        \qquad\qquad
        \begin{math}
          \begin{array}{c}
            \Pi_2 \\
            {\Gamma_{{\mathrm{1}}}  \SCsym{;}  \SCnt{Y}  \SCsym{;}  \Gamma_{{\mathrm{2}}}  \vdash_\mathcal{L}  \SCnt{A}}
          \end{array}
        \end{math}
      \end{center}
      By assumption, $c(\Pi_1),c(\Pi_2)\leq |Y|$. Similar as above,
      $|X|+1\leq |Y|$ and there is a proof $\Pi'$ constructed from $\pi_2$
      and $\Pi_2$ for sequent $\Gamma_{{\mathrm{1}}}  \SCsym{;}  \Psi_{{\mathrm{1}}}  \SCsym{;}  \SCnt{X}  \SCsym{;}  \Psi_{{\mathrm{2}}}  \SCsym{;}  \Gamma_{{\mathrm{2}}}  \vdash_\mathcal{L}  \SCnt{A}$ s.t.
      $c(\Pi')\leq|Y|$. Therefore, the proof $\Pi$ can be constructed as
      follows, and $c(\Pi)\leq max\{c(\pi_1),c(\Pi'),|X|+1\}\leq |Y|$.
      \begin{center}
        \scriptsize
        \begin{math}
          $$\mprset{flushleft}
          \inferrule* [right={\tiny cut}] {
            {
              \begin{array}{cc}
                \pi_1 & \Pi'\\
                {\Phi  \vdash_\mathcal{C}  \SCnt{X}} & {\Gamma_{{\mathrm{1}}}  \SCsym{;}  \Psi_{{\mathrm{1}}}  \SCsym{;}  \SCnt{X}  \SCsym{;}  \Psi_{{\mathrm{2}}}  \SCsym{;}  \Gamma_{{\mathrm{2}}}  \vdash_\mathcal{L}  \SCnt{A}}
              \end{array}
            }
          }{\Gamma_{{\mathrm{1}}}  \SCsym{;}  \Psi_{{\mathrm{1}}}  \SCsym{;}  \Phi  \SCsym{;}  \Psi_{{\mathrm{2}}}  \SCsym{;}  \Gamma_{{\mathrm{2}}}  \vdash_\mathcal{L}  \SCnt{A}}
        \end{math}
      \end{center}
\end{itemize}

\subsubsection{$\SCdruleSXXcutOneName$ vs. $\SCdruleSXXcutTwoName$}
\begin{itemize}
\item Case 1:
      \begin{center}
        \scriptsize
        \begin{math}
          \begin{array}{c}
            \Pi_1 \\
            {\Phi  \vdash_\mathcal{C}  \SCnt{X}}
          \end{array}
        \end{math}
        \qquad\qquad
        $\Pi_2:$
        \begin{math}
          $$\mprset{flushleft}
          \inferrule* [right={\tiny cut2}] {
            {
              \begin{array}{cc}
                \pi_1 & \pi_2 \\
                {\Gamma_{{\mathrm{2}}}  \SCsym{;}  \SCnt{X}  \SCsym{;}  \Gamma_{{\mathrm{3}}}  \vdash_\mathcal{L}  \SCnt{A}} & {\Gamma_{{\mathrm{1}}}  \SCsym{;}  \SCnt{A}  \SCsym{;}  \Gamma_{{\mathrm{4}}}  \vdash_\mathcal{L}  \SCnt{B}}
              \end{array}
            }
          }{\Gamma_{{\mathrm{1}}}  \SCsym{;}  \Gamma_{{\mathrm{2}}}  \SCsym{;}  \SCnt{X}  \SCsym{;}  \Gamma_{{\mathrm{3}}}  \SCsym{;}  \Gamma_{{\mathrm{4}}}  \vdash_\mathcal{L}  \SCnt{B}}
        \end{math}
      \end{center}
      By assumption, $c(\Pi_1),c(\Pi_2)\leq |X|$. Therefore, $c(\pi_1)$,
      $c(\pi_2)\leq |X|$. Since $A$ is the cut formula on $\pi_1$ and
      $\pi_2$, we have $|A|+1\leq|X|$. By induction on $\Pi_1$ and $\pi_1$,
      there exists a proof $\Pi'$ for sequent $\Gamma_{{\mathrm{2}}}  \SCsym{;}  \Phi  \SCsym{;}  \Gamma_{{\mathrm{3}}}  \vdash_\mathcal{L}  \SCnt{A}$ s.t.
      $c(\Pi')\leq|X|$. So $\Pi$ can be constructed as follows, with
      $c(\Pi)\leq max\{c(\Pi'),c(\pi_2),|A|+1\}\leq |X|$.
      \begin{center}
        \scriptsize
        \begin{math}
          $$\mprset{flushleft}
          \inferrule* [right={\tiny cut2}] {
            {
              \begin{array}{cc}
                \Pi' & \pi_2 \\
                {\Gamma_{{\mathrm{2}}}  \SCsym{;}  \Phi  \SCsym{;}  \Gamma_{{\mathrm{3}}}  \vdash_\mathcal{L}  \SCnt{A}} & {\Gamma_{{\mathrm{1}}}  \SCsym{;}  \SCnt{A}  \SCsym{;}  \Gamma_{{\mathrm{4}}}  \vdash_\mathcal{L}  \SCnt{B}}
              \end{array}
            }
          }{\Gamma_{{\mathrm{1}}}  \SCsym{;}  \Gamma_{{\mathrm{2}}}  \SCsym{;}  \Phi  \SCsym{;}  \Gamma_{{\mathrm{3}}}  \SCsym{;}  \Gamma_{{\mathrm{4}}}  \vdash_\mathcal{L}  \SCnt{B}}
        \end{math}
      \end{center}

\item Case 2:
      \begin{center}
        \scriptsize
        $\Pi_1$:
        \begin{math}
          $$\mprset{flushleft}
          \inferrule* [right={\tiny cut}] {
            {
              \begin{array}{cc}
                \pi_1 & \pi_2 \\
                {\Phi  \vdash_\mathcal{C}  \SCnt{X}} & {\Gamma_{{\mathrm{2}}}  \SCsym{;}  \SCnt{X}  \SCsym{;}  \Gamma_{{\mathrm{3}}}  \vdash_\mathcal{L}  \SCnt{A}}
              \end{array}
            }
          }{\Gamma_{{\mathrm{2}}}  \SCsym{;}  \Phi  \SCsym{;}  \Gamma_{{\mathrm{3}}}  \vdash_\mathcal{L}  \SCnt{A}}
        \end{math}
        \qquad\qquad
        \begin{math}
          \begin{array}{c}
            \Pi_2 \\
            {\Gamma_{{\mathrm{1}}}  \SCsym{;}  \SCnt{A}  \SCsym{;}  \Gamma_{{\mathrm{4}}}  \vdash_\mathcal{L}  \SCnt{B}}
          \end{array}
        \end{math}
      \end{center}
      By assumption, $c(\Pi_1),c(\Pi_2)\leq |A|$. Similar as above,
      $|X|+1\leq |A|$ and there is a proof $\Pi'$ constructed from'
      $\pi_2$ and $\Pi_2$ for sequent $\Gamma_{{\mathrm{1}}}  \SCsym{;}  \Gamma_{{\mathrm{2}}}  \SCsym{;}  \SCnt{X}  \SCsym{;}  \Gamma_{{\mathrm{3}}}  \SCsym{;}  \Gamma_{{\mathrm{4}}}  \vdash_\mathcal{L}  \SCnt{B}$ s.t.
      $c(\Pi')\leq|A|$. Therefore, the proof $\Pi$ can be constructed as
      follows, and $c(\Pi)\leq max\{c(\pi_1),c(\Pi'),|X|+1\}\leq |A|$.
      \begin{center}
        \scriptsize
        \begin{math}
          $$\mprset{flushleft}
          \inferrule* [right={\tiny cut}] {
            {
              \begin{array}{cc}
                \pi_1  & \Pi' \\
                {\Phi  \vdash_\mathcal{C}  \SCnt{X}} & {\Gamma_{{\mathrm{1}}}  \SCsym{;}  \Gamma_{{\mathrm{2}}}  \SCsym{;}  \SCnt{X}  \SCsym{;}  \Gamma_{{\mathrm{3}}}  \SCsym{;}  \Gamma_{{\mathrm{4}}}  \vdash_\mathcal{L}  \SCnt{B}}
              \end{array}
            }
          }{\Gamma_{{\mathrm{1}}}  \SCsym{;}  \Gamma_{{\mathrm{2}}}  \SCsym{;}  \Phi  \SCsym{;}  \Gamma_{{\mathrm{3}}}  \SCsym{;}  \Gamma_{{\mathrm{4}}}  \vdash_\mathcal{L}  \SCnt{B}}
        \end{math}
      \end{center}
\end{itemize}

\subsubsection{$\SCdruleSXXcutTwoName$ vs. $\SCdruleSXXcutTwoName$}
\begin{itemize}
\item Case 1:
      \begin{center}
        \scriptsize
        \begin{math}
          \begin{array}{c}
            \Pi_1 \\
            {\Gamma  \vdash_\mathcal{L}  \SCnt{A}}
          \end{array}
        \end{math}
        \qquad\qquad
        $\Pi_2:$
        \begin{math}
          $$\mprset{flushleft}
          \inferrule* [right={\tiny cut2}] {
            {
              \begin{array}{cc}
                \pi_1 & \pi_2 \\
                {\Delta_{{\mathrm{2}}}  \SCsym{;}  \SCnt{A}  \SCsym{;}  \Delta_{{\mathrm{3}}}  \vdash_\mathcal{L}  \SCnt{B}} & {\Delta_{{\mathrm{1}}}  \SCsym{;}  \SCnt{B}  \SCsym{;}  \Delta_{{\mathrm{4}}}  \vdash_\mathcal{L}  \SCnt{C}}
              \end{array}
            }
          }{\Delta_{{\mathrm{1}}}  \SCsym{;}  \Delta_{{\mathrm{2}}}  \SCsym{;}  \SCnt{A}  \SCsym{;}  \Delta_{{\mathrm{3}}}  \SCsym{;}  \Delta_{{\mathrm{4}}}  \vdash_\mathcal{L}  \SCnt{C}}
        \end{math}
      \end{center}
      By assumption, $c(\Pi_1),c(\Pi_2)\leq |A|$. Therefore, $c(\pi_1)$,
      $c(\pi_2)\leq |A|$. Since $B$ is the cut formula on $\pi_1$ and
      $\pi_3$, we have $|B|+1\leq|A|$. By induction on $\Pi_1$ and
      $\pi_1$, there exists a proof $\Pi'$ for sequent
      $\Delta_{{\mathrm{2}}}  \SCsym{;}  \Gamma  \SCsym{;}  \Delta_{{\mathrm{3}}}  \vdash_\mathcal{L}  \SCnt{B}$ s.t. $c(\Pi')\leq|A|$. So $\Pi$ can be
      constructed as follows,  with
      $c(\Pi)\leq max\{c(\Pi'),c(\pi_2),|B|+1\}\leq |A|$.
      \begin{center}
        \scriptsize
        \begin{math}
          $$\mprset{flushleft}
          \inferrule* [right={\tiny cut}] {
            {
              \begin{array}{cc}
                \Pi' & \pi_2 \\
                {\Delta_{{\mathrm{2}}}  \SCsym{;}  \Gamma  \SCsym{;}  \Delta_{{\mathrm{3}}}  \vdash_\mathcal{L}  \SCnt{B}} & {\Delta_{{\mathrm{1}}}  \SCsym{;}  \SCnt{B}  \SCsym{;}  \Delta_{{\mathrm{4}}}  \vdash_\mathcal{L}  \SCnt{C}}
              \end{array}
            }
          }{\Delta_{{\mathrm{1}}}  \SCsym{;}  \Delta_{{\mathrm{2}}}  \SCsym{;}  \Gamma  \SCsym{;}  \Delta_{{\mathrm{3}}}  \SCsym{;}  \Delta_{{\mathrm{4}}}  \vdash_\mathcal{L}  \SCnt{C}}
        \end{math}
      \end{center}

\item Case 2:
      \begin{center}
        \scriptsize
        $\Pi_1$:
        \begin{math}
          $$\mprset{flushleft}
          \inferrule* [right={\tiny cut}] {
            {
              \begin{array}{cc}
                \pi_1 & \pi_2 \\
                {\Delta  \vdash_\mathcal{L}  \SCnt{A}} & {\Delta_{{\mathrm{2}}}  \SCsym{;}  \SCnt{A}  \SCsym{;}  \Delta_{{\mathrm{3}}}  \vdash_\mathcal{L}  \SCnt{B}}
              \end{array}
            }
          }{\Delta_{{\mathrm{2}}}  \SCsym{;}  \Delta  \SCsym{;}  \Delta_{{\mathrm{3}}}  \vdash_\mathcal{L}  \SCnt{A}}
        \end{math}
        \qquad\qquad
        \begin{math}
          \begin{array}{c}
            \Pi_2 \\
            {\Delta_{{\mathrm{1}}}  \SCsym{;}  \SCnt{B}  \SCsym{;}  \Delta_{{\mathrm{4}}}  \vdash_\mathcal{L}  \SCnt{C}}
          \end{array}
        \end{math}
      \end{center}
      By assumption, $c(\Pi_1),c(\Pi_2)\leq |B|$. Similar as above,
      $|A|+1\leq |B|$ and there is a proof $\Pi'$ constructed from $\pi_2$ 
      and $\Pi_2$ for sequent $\Delta_{{\mathrm{1}}}  \SCsym{;}  \Delta_{{\mathrm{2}}}  \SCsym{;}  \SCnt{A}  \SCsym{;}  \Delta_{{\mathrm{3}}}  \SCsym{;}  \Delta_{{\mathrm{4}}}  \vdash_\mathcal{L}  \SCnt{C}$ s.t.
      $c(\Pi')\leq|A|$. Therefore, the proof $\Pi$ can be constructed as
      follows, and $c(\Pi)\leq max\{c(\pi_1),c(\Pi'),|A|+1\}\leq |B|$.
      \begin{center}
        \scriptsize
        \begin{math}
          $$\mprset{flushleft}
          \inferrule* [right={\tiny cut}] {
            {
              \begin{array}{cc}
                \pi_1 & \Pi' \\
                {\Gamma  \vdash_\mathcal{L}  \SCnt{A}} & {\Delta_{{\mathrm{1}}}  \SCsym{;}  \Delta_{{\mathrm{2}}}  \SCsym{;}  \SCnt{A}  \SCsym{;}  \Delta_{{\mathrm{3}}}  \SCsym{;}  \Delta_{{\mathrm{4}}}  \vdash_\mathcal{L}  \SCnt{C}}
              \end{array}
            }
          }{\Delta_{{\mathrm{1}}}  \SCsym{;}  \Delta_{{\mathrm{2}}}  \SCsym{;}  \Gamma  \SCsym{;}  \Delta_{{\mathrm{3}}}  \SCsym{;}  \Delta_{{\mathrm{4}}}  \vdash_\mathcal{L}  \SCnt{C}}
        \end{math}
      \end{center}

\end{itemize}

\subsection{The Axiom Steps}

\subsubsection{$\SCdruleTXXaxName$}
\begin{itemize}
\item Case 1:
      \begin{center}
        \scriptsize
        $\Pi_1$:
        \begin{math}
          $$\mprset{flushleft}
          \inferrule* [right={\tiny ax}] {
            \,
          }{\SCnt{X}  \vdash_\mathcal{C}  \SCnt{X}}
        \end{math}
        \qquad\qquad
        \begin{math}
          \begin{array}{c}
            \Pi_2 \\
            {\Phi_{{\mathrm{1}}}  \SCsym{,}  \SCnt{X}  \SCsym{,}  \Phi_{{\mathrm{2}}}  \vdash_\mathcal{C}  \SCnt{Y}}
          \end{array}
        \end{math}
      \end{center}
      By assumption, $c(\Pi_1),c(\Pi_2)\leq |X|$. The proof $\Pi$ is the
      same as $\Pi_2$.

\item Case 2:
      \begin{center}
        \scriptsize
        $\Pi_1$:
        \begin{math}
          \begin{array}{c}
            \Pi_1 \\
            {\Phi  \vdash_\mathcal{C}  \SCnt{X}}
          \end{array}
        \end{math}
        \qquad\qquad
        $\Pi_2$:
        \begin{math}
          $$\mprset{flushleft}
          \inferrule* [right={\tiny ax}] {
            \,
          }{\SCnt{X}  \vdash_\mathcal{C}  \SCnt{X}}
        \end{math}
      \end{center}
      By assumption, $c(\Pi_1),c(\Pi_2)\leq |X|$. The proof $\Pi$ is the
      same as $\Pi_1$.

\item Case 3:
      \begin{center}
        \scriptsize
        $\Pi_1$:
        \begin{math}
          $$\mprset{flushleft}
          \inferrule* [right={\tiny ax}] {
            \,
          }{\SCnt{X}  \vdash_\mathcal{C}  \SCnt{X}}
        \end{math}
        \qquad\qquad
        \begin{math}
          \begin{array}{c}
            \Pi_2 \\
            {\Gamma_{{\mathrm{1}}}  \SCsym{;}  \SCnt{X}  \SCsym{;}  \Gamma_{{\mathrm{2}}}  \vdash_\mathcal{L}  \SCnt{A}}
          \end{array}
        \end{math}
      \end{center}
      By assumption, $c(\Pi_1),c(\Pi_2)\leq |X|$. The proof $\Pi$ is the
      same as $\Pi_2$.
\end{itemize}

\subsubsection{$\SCdruleTXXaxName$}
\begin{itemize}
\item Case 1:
      \begin{center}
        \scriptsize
        $\Pi_1$:
        \begin{math}
          $$\mprset{flushleft}
          \inferrule* [right={\tiny ax}] {
            \,
          }{\SCnt{A}  \vdash_\mathcal{L}  \SCnt{A}}
        \end{math}
        \qquad\qquad
        \begin{math}
          \begin{array}{c}
            \Pi_2 \\
            {\Gamma_{{\mathrm{1}}}  \SCsym{;}  \SCnt{A}  \SCsym{;}  \Gamma_{{\mathrm{2}}}  \vdash_\mathcal{L}  \SCnt{B}}
          \end{array}
        \end{math}
      \end{center}
      By assumption, $c(\Pi_1),c(\Pi_2)\leq |A|$. The proof $\Pi$ is the
      same as $\Pi_2$.

\item Case 2:
      \begin{center}
        \scriptsize
        $\Pi_1$:
        \begin{math}
          \begin{array}{c}
            \Pi_1 \\
            {\Delta  \vdash_\mathcal{L}  \SCnt{A}}
          \end{array}
        \end{math}
        \qquad\qquad
        $\Pi_2$:
        \begin{math}
          $$\mprset{flushleft}
          \inferrule* [right={\tiny ax}] {
            \,
          }{\SCnt{A}  \vdash_\mathcal{L}  \SCnt{A}}
        \end{math}
      \end{center}
      By assumption, $c(\Pi_1),c(\Pi_2)\leq |X|$. The proof $\Pi$ is the
      same as $\Pi_1$.
\end{itemize}

\subsection{The Exchange Steps}

\subsubsection{$\SCdruleTXXexName$}

\begin{itemize}
\item Case 1:
      \begin{center}
        \scriptsize
        \begin{math}
          \begin{array}{c}
            \Pi_1 \\
            {\Psi  \vdash_\mathcal{C}  \SCnt{X_{{\mathrm{1}}}}}
          \end{array}
        \end{math}
        \qquad\qquad
        $\Pi_2$:
        \begin{math}
          $$\mprset{flushleft}
          \inferrule* [right={\tiny ex}] {
            {
              \begin{array}{c}
                \pi \\
                {\Phi_{{\mathrm{1}}}  \SCsym{,}  \SCnt{X_{{\mathrm{1}}}}  \SCsym{,}  \SCnt{X_{{\mathrm{2}}}}  \SCsym{,}  \Phi_{{\mathrm{2}}}  \vdash_\mathcal{C}  \SCnt{Y}}
              \end{array}
            }
          }{\Phi_{{\mathrm{1}}}  \SCsym{,}  \SCnt{X_{{\mathrm{2}}}}  \SCsym{,}  \SCnt{X_{{\mathrm{1}}}}  \SCsym{,}  \Phi_{{\mathrm{2}}}  \vdash_\mathcal{C}  \SCnt{Y}}
        \end{math}
      \end{center}
      By assumption, $c(\Pi_1),c(\Pi_2)\leq |X_1|$. By induction on $\pi$
      and $\Pi_1$, there is a proof $\Pi'$ for sequent
      $\Phi_{{\mathrm{1}}}  \SCsym{,}  \Psi  \SCsym{,}  \SCnt{X_{{\mathrm{2}}}}  \SCsym{,}  \Phi_{{\mathrm{2}}}  \vdash_\mathcal{C}  \SCnt{Y}$ s.t. $c(\Pi')\leq|X_1|$. Therefore, the
      proof $\Pi$ can be constructed as follows, and
      $c(\Pi)=c(\Pi')\leq|X_1|$.
      \begin{center}
        \scriptsize
        \begin{math}
          $$\mprset{flushleft}
          \inferrule* [right={\tiny series of ex}] {
            {
              \begin{array}{c}
                \Pi' \\
                {\Phi_{{\mathrm{1}}}  \SCsym{,}  \Psi  \SCsym{,}  \SCnt{X_{{\mathrm{2}}}}  \SCsym{,}  \Phi_{{\mathrm{2}}}  \vdash_\mathcal{C}  \SCnt{Y}}
              \end{array}
            }
          }{\Phi_{{\mathrm{1}}}  \SCsym{,}  \SCnt{X_{{\mathrm{2}}}}  \SCsym{,}  \Psi  \SCsym{,}  \Phi_{{\mathrm{2}}}  \vdash_\mathcal{C}  \SCnt{Y}}
        \end{math}
      \end{center}

\item Case 2:
      \begin{center}
        \scriptsize
        \begin{math}
          \begin{array}{c}
            \Pi_1 \\
            {\Psi  \vdash_\mathcal{C}  \SCnt{X_{{\mathrm{2}}}}}
          \end{array}
        \end{math}
        \qquad\qquad
        $\Pi_2$:
        \begin{math}
          $$\mprset{flushleft}
          \inferrule* [right={\tiny ex}] {
            {
              \begin{array}{c}
                \pi \\
                {\Phi_{{\mathrm{1}}}  \SCsym{,}  \SCnt{X_{{\mathrm{1}}}}  \SCsym{,}  \SCnt{X_{{\mathrm{2}}}}  \SCsym{,}  \Phi_{{\mathrm{2}}}  \vdash_\mathcal{C}  \SCnt{Y}}
              \end{array}
            }
          }{\Phi_{{\mathrm{1}}}  \SCsym{,}  \SCnt{X_{{\mathrm{2}}}}  \SCsym{,}  \SCnt{X_{{\mathrm{1}}}}  \SCsym{,}  \Phi_{{\mathrm{2}}}  \vdash_\mathcal{C}  \SCnt{Y}}
        \end{math}
      \end{center}
      By assumption, $c(\Pi_1),c(\Pi_2)\leq |X_2|$. By induction on $\pi$
      and $\Pi_1$, there is a proof $\Pi'$ for sequent
      $\Phi_{{\mathrm{1}}}  \SCsym{,}  \SCnt{X_{{\mathrm{1}}}}  \SCsym{,}  \Psi  \SCsym{,}  \Phi_{{\mathrm{2}}}  \vdash_\mathcal{C}  \SCnt{Y}$ s.t. $c(\Pi')\leq|X_2|$. Therefore, the
      proof $\Pi$ can be constructed as follows, and
      $c(\Pi)=c(\Pi')\leq|X_2|$.
      \begin{center}
        \scriptsize
        \begin{math}
          $$\mprset{flushleft}
          \inferrule* [right={\tiny series of ex}] {
            {
              \begin{array}{c}
                \Pi' \\
                {\Phi_{{\mathrm{1}}}  \SCsym{,}  \SCnt{X_{{\mathrm{1}}}}  \SCsym{,}  \Psi  \SCsym{,}  \Phi_{{\mathrm{2}}}  \vdash_\mathcal{C}  \SCnt{Y}}
              \end{array}
            }
          }{\Phi_{{\mathrm{1}}}  \SCsym{,}  \Psi  \SCsym{,}  \SCnt{X_{{\mathrm{1}}}}  \SCsym{,}  \Phi_{{\mathrm{2}}}  \vdash_\mathcal{C}  \SCnt{Y}}
        \end{math}
      \end{center}
\end{itemize}

\subsubsection{$\SCdruleSXXexName$}
\begin{itemize}
\item Case 1:
      \begin{center}
        \scriptsize
        \begin{math}
          \begin{array}{c}
            \Pi_1 \\
            {\Psi  \vdash_\mathcal{C}  \SCnt{X_{{\mathrm{1}}}}}
          \end{array}
        \end{math}
        \qquad\qquad
        $\Pi_2$:
        \begin{math}
          $$\mprset{flushleft}
          \inferrule* [right={\tiny ex}] {
            {
              \begin{array}{c}
                \pi \\
                {\Delta_{{\mathrm{1}}}  \SCsym{;}  \SCnt{X_{{\mathrm{1}}}}  \SCsym{;}  \SCnt{X_{{\mathrm{2}}}}  \SCsym{;}  \Delta_{{\mathrm{2}}}  \vdash_\mathcal{L}  \SCnt{A}}
              \end{array}
            }
          }{\Delta_{{\mathrm{1}}}  \SCsym{;}  \SCnt{X_{{\mathrm{2}}}}  \SCsym{;}  \SCnt{X_{{\mathrm{1}}}}  \SCsym{;}  \Delta_{{\mathrm{2}}}  \vdash_\mathcal{L}  \SCnt{A}}
        \end{math}
      \end{center}
      By assumption, $c(\Pi_1),c(\Pi_2)\leq |X_1|$. By induction on $\pi$
      and $\Pi_1$, there is a proof $\Pi'$ for sequent
      $\Delta_{{\mathrm{1}}}  \SCsym{;}  \Psi  \SCsym{;}  \SCnt{X_{{\mathrm{2}}}}  \SCsym{;}  \Delta_{{\mathrm{2}}}  \vdash_\mathcal{L}  \SCnt{A}$ s.t. $c(\Pi')\leq|X_1|$. Therefore, the
      proof $\Pi$ can be constructed as follows, and
      $c(\Pi)=c(\Pi')\leq|X_1|$.
      \begin{center}
        \scriptsize
        \begin{math}
          $$\mprset{flushleft}
          \inferrule* [right={\tiny series of ex}] {
            {
              \begin{array}{c}
                \Pi' \\
                {\Delta_{{\mathrm{1}}}  \SCsym{;}  \Psi  \SCsym{;}  \SCnt{X_{{\mathrm{2}}}}  \SCsym{;}  \Delta_{{\mathrm{2}}}  \vdash_\mathcal{L}  \SCnt{A}}
              \end{array}
            }
          }{\Delta_{{\mathrm{1}}}  \SCsym{;}  \SCnt{X_{{\mathrm{2}}}}  \SCsym{;}  \Psi  \SCsym{;}  \Delta_{{\mathrm{2}}}  \vdash_\mathcal{L}  \SCnt{A}}
        \end{math}
      \end{center}

\item Case 2:
      \begin{center}
        \scriptsize
        \begin{math}
          \begin{array}{c}
            \Pi_1 \\
            {\Psi  \vdash_\mathcal{C}  \SCnt{X_{{\mathrm{2}}}}}
          \end{array}
        \end{math}
        \qquad\qquad
        $\Pi_2$:
        \begin{math}
          $$\mprset{flushleft}
          \inferrule* [right={\tiny ex}] {
            {
              \begin{array}{c}
                \pi \\
                {\Delta_{{\mathrm{1}}}  \SCsym{;}  \SCnt{X_{{\mathrm{1}}}}  \SCsym{;}  \SCnt{X_{{\mathrm{2}}}}  \SCsym{;}  \Delta_{{\mathrm{2}}}  \vdash_\mathcal{L}  \SCnt{A}}
              \end{array}
            }
          }{\Delta_{{\mathrm{1}}}  \SCsym{;}  \SCnt{X_{{\mathrm{2}}}}  \SCsym{;}  \SCnt{X_{{\mathrm{1}}}}  \SCsym{;}  \Delta_{{\mathrm{2}}}  \vdash_\mathcal{L}  \SCnt{A}}
        \end{math}
      \end{center}
      By assumption, $c(\Pi_1),c(\Pi_2)\leq |X_2|$. By induction on $\pi$
      and $\Pi_1$, there is a proof $\Pi'$ for sequent
      $\Delta_{{\mathrm{1}}}  \SCsym{;}  \SCnt{X_{{\mathrm{1}}}}  \SCsym{;}  \Psi  \SCsym{;}  \Delta_{{\mathrm{2}}}  \vdash_\mathcal{L}  \SCnt{A}$ s.t. $c(\Pi')\leq|X_2|$. Therefore, the
      proof $\Pi$ can be constructed as follows, and
      $c(\Pi)=c(\Pi')\leq|X_2|$.
      \begin{center}
        \scriptsize
        \begin{math}
          $$\mprset{flushleft}
          \inferrule* [right={\tiny series of ex}] {
            {
              \begin{array}{c}
                \Pi' \\
                {\Delta_{{\mathrm{1}}}  \SCsym{;}  \SCnt{X_{{\mathrm{1}}}}  \SCsym{;}  \Psi  \SCsym{;}  \Phi_{{\mathrm{2}}}  \vdash_\mathcal{L}  \SCnt{A}}
              \end{array}
            }
          }{\Phi_{{\mathrm{1}}}  \SCsym{;}  \Psi  \SCsym{;}  \SCnt{X_{{\mathrm{1}}}}  \SCsym{;}  \Phi_{{\mathrm{2}}}  \vdash_\mathcal{L}  \SCnt{A}}
        \end{math}
      \end{center}
\end{itemize}

\subsection{Principal Formula vs. Principal Formula} 

\subsubsection{The Commutative Tensor Product $\otimes$}
\begin{center}
  \scriptsize
  $\Pi_1:$
  \begin{math}
    $$\mprset{flushleft}
    \inferrule* [right={\tiny tenR}] {
      {
        \begin{array}{cc}
          \pi_1 & \pi_2 \\
          {\Phi_{{\mathrm{1}}}  \vdash_\mathcal{C}  \SCnt{X}} & {\Phi_{{\mathrm{2}}}  \vdash_\mathcal{C}  \SCnt{Y}}
        \end{array}
      }
    }{\Phi_{{\mathrm{1}}}  \SCsym{,}  \Phi_{{\mathrm{2}}}  \vdash_\mathcal{C}  \SCnt{X}  \otimes  \SCnt{Y}}
  \end{math}
  \qquad\qquad
  $\Pi_2:$
  \begin{math}
    $$\mprset{flushleft}
    \inferrule* [right={\tiny tenL}] {
      {
        \begin{array}{c}
          \pi_3 \\
          {\Psi_{{\mathrm{1}}}  \SCsym{,}  \SCnt{X}  \SCsym{,}  \SCnt{Y}  \SCsym{,}  \Psi_{{\mathrm{2}}}  \vdash_\mathcal{C}  \SCnt{Z}}
        \end{array}
      }
    }{\Psi_{{\mathrm{1}}}  \SCsym{,}  \SCnt{X}  \otimes  \SCnt{Y}  \SCsym{,}  \Psi_{{\mathrm{2}}}  \vdash_\mathcal{C}  \SCnt{Z}}
  \end{math}
\end{center}
By assumption, $c(\Pi_1),c(\Pi_2)\leq |\SCnt{X}  \otimes  \SCnt{Y}| = |X|+|Y|+1$. The proof
$\Pi$ can be constructed as follows, and
$c(\Pi)\leq max\{c(\pi_1),c(\pi_2),c(\pi_3),|X|+1,|Y|+1\}\leq |X|+|Y|+1 = |\SCnt{X}  \otimes  \SCnt{Y}|$.
\begin{center}
  \scriptsize
  \begin{math}
    $$\mprset{flushleft}
    \inferrule* [right={\tiny cut}] {
      {
        \begin{array}{c}
          \pi_1 \\
          {\Phi_{{\mathrm{1}}}  \vdash_\mathcal{C}  \SCnt{X}}
        \end{array}
      }
      $$\mprset{flushleft}
      \inferrule* [right={\tiny cut}] {
      {
        \begin{array}{cc}
          \pi_2 & \pi_3 \\
          {\Phi_{{\mathrm{2}}}  \vdash_\mathcal{C}  \SCnt{Y}} & {\Psi_{{\mathrm{1}}}  \SCsym{,}  \SCnt{X}  \SCsym{,}  \SCnt{Y}  \SCsym{,}  \Psi_{{\mathrm{2}}}  \vdash_\mathcal{C}  \SCnt{Z}}
        \end{array}
      }
      }{\Psi_{{\mathrm{1}}}  \SCsym{,}  \SCnt{X}  \SCsym{,}  \Phi_{{\mathrm{2}}}  \SCsym{,}  \Psi_{{\mathrm{2}}}  \vdash_\mathcal{C}  \SCnt{Z}}
    }{\Psi_{{\mathrm{1}}}  \SCsym{,}  \Phi_{{\mathrm{1}}}  \SCsym{,}  \Phi_{{\mathrm{2}}}  \SCsym{,}  \Psi_{{\mathrm{2}}}  \vdash_\mathcal{C}  \SCnt{Z}}
  \end{math}
\end{center}

\subsubsection{The Non-commutative Tensor Product $\tri$}
\begin{center}
  \scriptsize
  $\Pi_1:$
  \begin{math}
    $$\mprset{flushleft}
    \inferrule* [right={\tiny tenR}] {
      {
        \begin{array}{cc}
          \pi_1 & \pi_2 \\
          {\Gamma_{{\mathrm{1}}}  \vdash_\mathcal{L}  \SCnt{A}} & {\Gamma_{{\mathrm{2}}}  \vdash_\mathcal{L}  \SCnt{B}}
        \end{array}
      }
    }{\Gamma_{{\mathrm{1}}}  \SCsym{;}  \Gamma_{{\mathrm{2}}}  \vdash_\mathcal{L}  \SCnt{A}  \triangleright  \SCnt{B}}
  \end{math}
  \qquad\qquad
  $\Pi_2:$
  \begin{math}
    $$\mprset{flushleft}
    \inferrule* [right={\tiny tenL1}] {
      {
        \begin{array}{c}
          \pi_3 \\
          {\Delta_{{\mathrm{1}}}  \SCsym{;}  \SCnt{A}  \SCsym{;}  \SCnt{B}  \SCsym{;}  \Delta_{{\mathrm{2}}}  \vdash_\mathcal{L}  \SCnt{C}}
        \end{array}
      }
    }{\Delta_{{\mathrm{1}}}  \SCsym{;}  \SCnt{A}  \triangleright  \SCnt{B}  \SCsym{;}  \Delta_{{\mathrm{2}}}  \vdash_\mathcal{L}  \SCnt{C}}
  \end{math}
\end{center}
By assumption, $c(\Pi_1),c(\Pi_2)\leq |\SCnt{A}  \triangleright  \SCnt{B}| = |X|+|Y|+1$. The proof
$\Pi$ can be constructed as follows, and
$c(\Pi)\leq max\{c(\pi_1),c(\pi_2),c(\pi_3),|A|+1,|B|+1\}\leq |A|+|B|+1 = |\SCnt{A}  \triangleright  \SCnt{B}|$.
\begin{center}
  \scriptsize
  \begin{math}
    $$\mprset{flushleft}
    \inferrule* [right={\tiny cut2}] {
      {
        \begin{array}{c}
          \pi_1 \\
          {\Gamma_{{\mathrm{1}}}  \vdash_\mathcal{L}  \SCnt{A}}
        \end{array}
      }
      $$\mprset{flushleft}
      \inferrule* [right={\tiny cut2}] {
      {
        \begin{array}{cc}
          \pi_2 & \pi_3 \\
          {\Gamma_{{\mathrm{2}}}  \vdash_\mathcal{L}  \SCnt{B}} & {\Delta_{{\mathrm{1}}}  \SCsym{;}  \SCnt{A}  \SCsym{;}  \SCnt{B}  \SCsym{;}  \Delta_{{\mathrm{2}}}  \vdash_\mathcal{L}  \SCnt{C}}
        \end{array}
      }
      }{\Delta_{{\mathrm{1}}}  \SCsym{;}  \SCnt{A}  \SCsym{;}  \Gamma_{{\mathrm{2}}}  \SCsym{;}  \Delta_{{\mathrm{2}}}  \vdash_\mathcal{L}  \SCnt{C}}
    }{\Delta_{{\mathrm{1}}}  \SCsym{;}  \Gamma_{{\mathrm{1}}}  \SCsym{;}  \Gamma_{{\mathrm{2}}}  \SCsym{;}  \Psi_{{\mathrm{2}}}  \vdash_\mathcal{L}  \SCnt{C}}
  \end{math}
\end{center}

\subsubsection{The Commutative Implication $\multimap$}
\begin{center}
  \scriptsize
  $\Pi_1:$
  \begin{math}
    $$\mprset{flushleft}
    \inferrule* [right={\tiny tenR}] {
      {
        \begin{array}{c}
          \pi_1 \\
          {\Phi_{{\mathrm{1}}}  \SCsym{,}  \SCnt{X}  \vdash_\mathcal{C}  \SCnt{Y}}
        \end{array}
      }
    }{\Phi_{{\mathrm{1}}}  \vdash_\mathcal{C}  \SCnt{X}  \multimap  \SCnt{Y}}
  \end{math}
  \qquad\qquad
  $\Pi_2:$
  \begin{math}
    $$\mprset{flushleft}
    \inferrule* [right={\tiny tenL}] {
      {
        \begin{array}{cc}
          \pi_2 & \pi_3 \\
          {\Phi_{{\mathrm{2}}}  \vdash_\mathcal{C}  \SCnt{X}} & {\Psi_{{\mathrm{1}}}  \SCsym{,}  \SCnt{Y}  \SCsym{,}  \Psi_{{\mathrm{2}}}  \vdash_\mathcal{C}  \SCnt{Z}}
        \end{array}
      }
    }{\Psi_{{\mathrm{1}}}  \SCsym{,}  \SCnt{X}  \multimap  \SCnt{Y}  \SCsym{,}  \Phi  \SCsym{,}  \Psi_{{\mathrm{2}}}  \vdash_\mathcal{C}  \SCnt{Z}}
  \end{math}
\end{center}
By assumption, $c(\Pi_1),c(\Pi_2)\leq |\SCnt{X}  \multimap  \SCnt{Y}| = |X|+|Y|+1$. The proof 
$\Pi$ is constructed as follows
$c(\Pi)\leq max\{c(\pi_1),c(\pi_2),c(\pi_3),|X|+1,|Y|+1\}\leq |X|+|Y|+1 = |\SCnt{X}  \multimap  \SCnt{Y}|$.
\begin{center}
  \scriptsize
  \begin{math}
    $$\mprset{flushleft}
    \inferrule* [right={\tiny tenR}] {
      $$\mprset{flushleft}
      \inferrule* [right={\tiny tenR}] {
        {
          \begin{array}{cc}
            \pi_1 & \pi_2 \\
            {\Phi_{{\mathrm{1}}}  \SCsym{,}  \SCnt{X}  \vdash_\mathcal{C}  \SCnt{Y}} & {\Phi_{{\mathrm{2}}}  \vdash_\mathcal{C}  \SCnt{X}}
          \end{array}
        }
      }{\Phi_{{\mathrm{1}}}  \SCsym{,}  \Phi_{{\mathrm{2}}}  \vdash_\mathcal{C}  \SCnt{Y}} \\
       {
         \begin{array}{c}
           \pi_3 \\
           {\Psi_{{\mathrm{1}}}  \SCsym{,}  \SCnt{Y}  \SCsym{,}  \Psi_{{\mathrm{2}}}  \vdash_\mathcal{C}  \SCnt{Z}}
         \end{array}
       }
    }{\Psi_{{\mathrm{1}}}  \SCsym{,}  \Phi_{{\mathrm{1}}}  \SCsym{,}  \Phi_{{\mathrm{2}}}  \SCsym{,}  \Psi_{{\mathrm{2}}}  \vdash_\mathcal{C}  \SCnt{Z}}
  \end{math}
\end{center}

\subsubsection{The Non-commutative Right Implication $\lto$}
\begin{center}
  \scriptsize
  $\Pi_1:$
  \begin{math}
    $$\mprset{flushleft}
    \inferrule* [right={\tiny imprR}] {
      {
        \begin{array}{c}
          \pi_1 \\
          {\Gamma  \SCsym{;}  \SCnt{A}  \vdash_\mathcal{L}  \SCnt{B}}
        \end{array}
      }
    }{\Gamma  \vdash_\mathcal{L}  \SCnt{A}  \rightharpoonup  \SCnt{B}}
  \end{math}
  \qquad\qquad
  $\Pi_2:$
  \begin{math}
    $$\mprset{flushleft}
    \inferrule* [right={\tiny imprL}] {
      {
        \begin{array}{cc}
          \pi_2 & \pi_3 \\
          {\Delta_{{\mathrm{1}}}  \vdash_\mathcal{L}  \SCnt{A}} & {\Delta_{{\mathrm{2}}}  \SCsym{;}  \SCnt{B}  \vdash_\mathcal{L}  \SCnt{C}}
        \end{array}
      }
    }{\Delta_{{\mathrm{2}}}  \SCsym{;}  \SCnt{A}  \rightharpoonup  \SCnt{B}  \SCsym{;}  \Delta_{{\mathrm{1}}}  \vdash_\mathcal{L}  \SCnt{C}}
  \end{math}
\end{center}
By assumption, $c(\Pi_1),c(\Pi_2)\leq |\SCnt{A}  \rightharpoonup  \SCnt{B}| = |A|+|B|+1$. The proof
$\Pi$ is constructed as follows, and
$c(\Pi)\leq max\{c(\pi_1),c(\pi_2),c(\pi_3),|A|+1,|B|+1\}\leq |A|+|B|+1 = |\SCnt{A}  \rightharpoonup  \SCnt{B}|$.
\begin{center}
  \scriptsize
  \begin{math}
    $$\mprset{flushleft}
    \inferrule* [right={\tiny cut2}] {
      $$\mprset{flushleft}
      \inferrule* [right={\tiny cut2}] {
        {
          \begin{array}{cc}
            \pi_1 & \pi_2 \\
            {\Gamma  \SCsym{;}  \SCnt{A}  \vdash_\mathcal{L}  \SCnt{B}} & {\Delta_{{\mathrm{1}}}  \vdash_\mathcal{L}  \SCnt{A}}
          \end{array}
        }
      }{\Gamma  \SCsym{;}  \Delta_{{\mathrm{1}}}  \vdash_\mathcal{L}  \SCnt{B}}
       {
         \begin{array}{c}
           \pi_3 \\
           {\Delta_{{\mathrm{2}}}  \SCsym{;}  \SCnt{B}  \vdash_\mathcal{L}  \SCnt{C}}
         \end{array}
       }
    }{\Delta_{{\mathrm{2}}}  \SCsym{;}  \Gamma  \SCsym{;}  \Delta_{{\mathrm{1}}}  \vdash_\mathcal{L}  \SCnt{C}}
  \end{math}
\end{center}

\subsubsection{The Non-commutative Left Implication $\rto$}
\begin{center}
  \scriptsize
  $\Pi_1:$
  \begin{math}
    $$\mprset{flushleft}
    \inferrule* [right={\tiny implR}] {
      {
        \begin{array}{c}
          \pi_1 \\
          {\SCnt{A}  \SCsym{;}  \Gamma  \vdash_\mathcal{L}  \SCnt{B}}
        \end{array}
      }
    }{\Gamma  \vdash_\mathcal{L}  \SCnt{B}  \leftharpoonup  \SCnt{A}}
  \end{math}
  \qquad\qquad
  $\Pi_2:$
  \begin{math}
    $$\mprset{flushleft}
    \inferrule* [right={\tiny implL}] {
      {
        \begin{array}{cc}
          \pi_2 & \pi_3 \\
          {\Delta_{{\mathrm{1}}}  \vdash_\mathcal{L}  \SCnt{A}} & {\SCnt{B}  \SCsym{;}  \Delta_{{\mathrm{2}}}  \vdash_\mathcal{L}  \SCnt{C}}
        \end{array}
      }
    }{\Delta_{{\mathrm{1}}}  \SCsym{;}  \SCnt{B}  \leftharpoonup  \SCnt{A}  \SCsym{;}  \Delta_{{\mathrm{2}}}  \vdash_\mathcal{L}  \SCnt{C}}
  \end{math}
\end{center}
By assumption, $c(\Pi_1),c(\Pi_2)\leq |\SCnt{B}  \leftharpoonup  \SCnt{A}| = |A|+|B|+1$. The
proof $\Pi$ is constructed as follows, and
$c(\Pi)\leq max\{c(\pi_1),c(\pi_2),c(\pi_3),|A|+1,|B|+1\}\leq |A|+|B|+1 = |\SCnt{B}  \leftharpoonup  \SCnt{A}|$.
\begin{center}
  \scriptsize
  \begin{math}
    $$\mprset{flushleft}
    \inferrule* [right={\tiny cut1}] {
      $$\mprset{flushleft}
      \inferrule* [right={\tiny cut2}] {
        {
          \begin{array}{cc}
            \pi_1 & \pi_2 \\
            {\SCnt{A}  \SCsym{;}  \Gamma  \vdash_\mathcal{L}  \SCnt{B}} & {\Delta_{{\mathrm{1}}}  \vdash_\mathcal{L}  \SCnt{A}}
          \end{array}
        }
      }{\Delta_{{\mathrm{1}}}  \SCsym{;}  \Gamma  \vdash_\mathcal{L}  \SCnt{B}}
       {
         \begin{array}{c}
           \pi_3 \\
           {\SCnt{B}  \SCsym{;}  \Delta_{{\mathrm{2}}}  \vdash_\mathcal{L}  \SCnt{C}}
         \end{array}
       }
    }{\Delta_{{\mathrm{1}}}  \SCsym{;}  \Gamma  \SCsym{;}  \Delta_{{\mathrm{2}}}  \vdash_\mathcal{L}  \SCnt{C}}
  \end{math}
\end{center}

\subsubsection{The Commutative Unit $ \mathsf{Unit} $}
\begin{itemize}
\item Case 1:
      \begin{center}
        \scriptsize
        $\Pi_1:$
        \begin{math}
          $$\mprset{flushleft}
          \inferrule* [right={\tiny unitR}] {
            \,
          }{ \cdot   \vdash_\mathcal{C}   \mathsf{Unit} }
        \end{math}
        \qquad\qquad
        $\Pi_2:$
        \begin{math}
          $$\mprset{flushleft}
          \inferrule* [right={\tiny unitL}] {
            {
              \begin{array}{c}
                \pi \\
                {\Phi  \SCsym{,}  \Psi  \vdash_\mathcal{C}  \SCnt{X}}
              \end{array}
            }
          }{\Phi  \SCsym{,}   \mathsf{Unit}   \SCsym{,}  \Psi  \vdash_\mathcal{C}  \SCnt{X}}
        \end{math}
      \end{center}
      By assumption, $c(\Pi_1),c(\Pi_2)\leq | \mathsf{Unit} |$. The proof $\Pi$
      is the subproof $\pi$ in $\Pi_2$ for sequent $\Phi  \vdash_\mathcal{C}  \SCnt{X}$. So
      $c(\Pi)=c(\Pi_2)\leq | \mathsf{Unit} |$.

\item Case 2:
      \begin{center}
        \scriptsize
        $\Pi_1:$
        \begin{math}
          $$\mprset{flushleft}
          \inferrule* [right={\tiny unitR}] {
            \,
          }{ \cdot   \vdash_\mathcal{C}   \mathsf{Unit} }
        \end{math}
        \qquad\qquad
        $\Pi_2:$
        \begin{math}
          $$\mprset{flushleft}
          \inferrule* [right={\tiny unitL1}] {
            {
              \begin{array}{c}
                \pi \\
                {\Gamma  \SCsym{;}  \Delta  \vdash_\mathcal{L}  \SCnt{A}}
              \end{array}
            }
          }{\Gamma  \SCsym{;}   \mathsf{Unit}   \SCsym{;}  \Delta  \vdash_\mathcal{L}  \SCnt{A}}
        \end{math}
      \end{center}
      Similar as above, $\Pi$ is $\pi$.
\end{itemize}

\subsubsection{The Non-commutative Unit $ \mathsf{Unit} $}
\begin{center}
  \scriptsize
  $\Pi_1:$
  \begin{math}
    $$\mprset{flushleft}
    \inferrule* [right={\tiny unitR}] {
      \,
    }{ \cdot   \vdash_\mathcal{L}   \mathsf{Unit} }
  \end{math}
  \qquad\qquad
  $\Pi_2:$
  \begin{math}
    $$\mprset{flushleft}
    \inferrule* [right={\tiny unitL2}] {
      {
        \begin{array}{c}
          \pi \\
          {\Gamma  \SCsym{;}  \Delta  \vdash_\mathcal{L}  \SCnt{A}}
        \end{array}
      }
    }{\Gamma  \SCsym{;}   \mathsf{Unit}   \SCsym{;}  \Delta  \vdash_\mathcal{L}  \SCnt{A}}
  \end{math}
\end{center}
By assumption, $c(\Pi_1),c(\Pi_2)\leq | \mathsf{Unit} |$. The proof $\Pi$ is the
subproof $\pi$ in $\Pi_2$ for sequent $\Delta  \vdash_\mathcal{L}  \SCnt{A}$. So
$c(\Pi)=c(\Pi_2)\leq | \mathsf{Unit} |$.

\subsubsection{The Functor $F$}
\begin{center}
  \scriptsize
  $\Pi_1:$
  \begin{math}
    $$\mprset{flushleft}
    \inferrule* [right={\tiny FR}] {
      {
        \begin{array}{c}
          \pi_1 \\
          {\Phi  \vdash_\mathcal{C}  \SCnt{X}}
        \end{array}
      }
    }{\Phi  \vdash_\mathcal{L}   \mathsf{F} \SCnt{X} }
  \end{math}
  \qquad\qquad
  $\Pi_2:$
  \begin{math}
    $$\mprset{flushleft}
    \inferrule* [right={\tiny FL}] {
      {
        \begin{array}{c}
          \pi_2 \\
          {\Gamma  \SCsym{;}  \SCnt{X}  \SCsym{;}  \Delta  \vdash_\mathcal{L}  \SCnt{A}}
        \end{array}
      }
    }{\Gamma  \SCsym{;}   \mathsf{F} \SCnt{X}   \SCsym{;}  \Delta  \vdash_\mathcal{L}  \SCnt{A}}
  \end{math}
\end{center}
By assumption, $c(\Pi_1),c(\Pi_2)\leq | \mathsf{F} \SCnt{X} | = |X|+1$. The proof
$\Pi$ is constructed as follows, and \\
$c(\Pi)\leq max\{c(\pi_1),c(\pi_2),|X|+1\}\leq | \mathsf{F} \SCnt{X} |$.
\begin{center}
  \scriptsize
  \begin{math}
    $$\mprset{flushleft}
    \inferrule* [right={\tiny cut2}] {
      {
        \begin{array}{cc}
          \pi_1 & \pi_2 \\
          {\Phi  \vdash_\mathcal{C}  \SCnt{X}} & {\Gamma  \SCsym{;}  \SCnt{X}  \SCsym{;}  \Delta  \vdash_\mathcal{L}  \SCnt{A}}
        \end{array}
      }
    }{\Gamma  \SCsym{;}  \Phi  \SCsym{;}  \Delta  \vdash_\mathcal{L}  \SCnt{A}}
  \end{math}
\end{center}

\subsubsection{The Functor $G$}
\begin{center}
  \scriptsize
  $\Pi_1:$
  \begin{math}
    $$\mprset{flushleft}
    \inferrule* [right={\tiny GR}] {
      {
        \begin{array}{c}
          \pi_1 \\
          {\Phi  \vdash_\mathcal{L}  \SCnt{A}}
        \end{array}
      }
    }{\Phi  \vdash_\mathcal{C}   \mathsf{G} \SCnt{A} }
  \end{math}
  \qquad\qquad
  $\Pi_2:$
  \begin{math}
    $$\mprset{flushleft}
    \inferrule* [right={\tiny GL}] {
      {
        \begin{array}{c}
          \pi_2 \\
          {\Gamma  \SCsym{;}  \SCnt{A}  \SCsym{;}  \Delta  \vdash_\mathcal{L}  \SCnt{B}}
        \end{array}
      }
    }{\Gamma  \SCsym{;}   \mathsf{G} \SCnt{A}   \SCsym{;}  \Delta  \vdash_\mathcal{L}  \SCnt{B}}
  \end{math}
\end{center}
By assumption, $c(\Pi_1),c(\Pi_2)\leq | \mathsf{G} \SCnt{A} | = |A|+1$. The proof $\Pi$ 
is constructed as follows, and \\
$c(\Pi)\leq max\{c(\pi_1),c(\pi_2),|A|+1\}\leq | \mathsf{G} \SCnt{A} |$.
\begin{center}
  \scriptsize
  \begin{math}
    $$\mprset{flushleft}
    \inferrule* [right={\tiny GL}] {
      {
        \begin{array}{cc}
          \pi_1 & \pi_2 \\
          {\Phi  \vdash_\mathcal{L}  \SCnt{A}} & {\Gamma  \SCsym{;}  \SCnt{A}  \SCsym{;}  \Delta  \vdash_\mathcal{L}  \SCnt{B}}
        \end{array}
      }
    }{\Gamma  \SCsym{;}  \Phi  \SCsym{;}  \Delta  \vdash_\mathcal{L}  \SCnt{B}}
  \end{math}
\end{center}

\subsection{Secondary Conclusion}

\subsubsection{Left introduction of the commutative implication $\multimap$}
\begin{itemize}
\item Case 1:
      \begin{center}
        \scriptsize
        $\Pi_1$:
        \begin{math}
          $$\mprset{flushleft}
          \inferrule* [right={\tiny impL}] {
            {
              \begin{array}{cc}
                \pi_1 & \pi_2 \\
                {\Phi_{{\mathrm{1}}}  \vdash_\mathcal{C}  \SCnt{X_{{\mathrm{1}}}}} & {\Phi_{{\mathrm{2}}}  \SCsym{,}  \SCnt{X_{{\mathrm{2}}}}  \SCsym{,}  \Phi_{{\mathrm{3}}}  \vdash_\mathcal{C}  \SCnt{Y}}
              \end{array}
            }
          }{\Phi_{{\mathrm{2}}}  \SCsym{,}  \SCnt{X_{{\mathrm{1}}}}  \multimap  \SCnt{X_{{\mathrm{2}}}}  \SCsym{,}  \Phi_{{\mathrm{1}}}  \SCsym{,}  \Phi_{{\mathrm{3}}}  \vdash_\mathcal{C}  \SCnt{Y}}
        \end{math}
        \qquad\qquad
        \begin{math}
          \begin{array}{c}
            \Pi_2 \\
            {\Psi_{{\mathrm{1}}}  \SCsym{,}  \SCnt{Y}  \SCsym{,}  \Psi_{{\mathrm{2}}}  \vdash_\mathcal{C}  \SCnt{Z}}
          \end{array}
        \end{math}
      \end{center}
      By assumption, $c(\Pi_1),c(\Pi_2)\leq |Y|$. By induction, there is a
      proof $\Pi'$ from $\pi_2$ and $\Pi_2$ for sequent
      $\Psi_{{\mathrm{1}}}  \SCsym{,}  \Phi_{{\mathrm{2}}}  \SCsym{,}  \SCnt{X_{{\mathrm{2}}}}  \SCsym{,}  \Phi_{{\mathrm{3}}}  \SCsym{,}  \Psi_{{\mathrm{2}}}  \vdash_\mathcal{C}  \SCnt{Z}$ s.t. $c(\Pi')\leq |Y|$. Therefore,
      the proof $\Pi$ can be constructed as follows with $c(\Pi)\leq |Y|$.
      \begin{center}
        \scriptsize
        \begin{math}
          $$\mprset{flushleft}
          \inferrule* [right={\tiny impL}] {
            {
              \begin{array}{c}
                \pi_1 \\
                {\Phi_{{\mathrm{1}}}  \vdash_\mathcal{C}  \SCnt{X_{{\mathrm{1}}}}}
              \end{array}
            }
            $$\mprset{flushleft}
            \inferrule* [right={\tiny cut}] {
              {
                \begin{array}{cc}
                  \pi_2 & \Pi_2 \\
                  {\Phi_{{\mathrm{2}}}  \SCsym{,}  \SCnt{X_{{\mathrm{2}}}}  \SCsym{,}  \Phi_{{\mathrm{3}}}  \vdash_\mathcal{C}  \SCnt{Y}} & {\Psi_{{\mathrm{1}}}  \SCsym{,}  \SCnt{Y}  \SCsym{,}  \Psi_{{\mathrm{2}}}  \vdash_\mathcal{C}  \SCnt{Z}}
                \end{array}
              }
            }{\Psi_{{\mathrm{1}}}  \SCsym{,}  \Phi_{{\mathrm{2}}}  \SCsym{,}  \SCnt{X_{{\mathrm{2}}}}  \SCsym{,}  \Phi_{{\mathrm{3}}}  \SCsym{,}  \Psi_{{\mathrm{2}}}  \vdash_\mathcal{C}  \SCnt{Z}}
          }{\Psi_{{\mathrm{1}}}  \SCsym{,}  \Phi_{{\mathrm{2}}}  \SCsym{,}  \SCnt{X_{{\mathrm{1}}}}  \multimap  \SCnt{X_{{\mathrm{2}}}}  \SCsym{,}  \Phi_{{\mathrm{1}}}  \SCsym{,}  \Phi_{{\mathrm{3}}}  \SCsym{,}  \Psi_{{\mathrm{2}}}  \vdash_\mathcal{C}  \SCnt{Z}}
        \end{math}
      \end{center}

\item Case 2:
      \begin{center}
        \scriptsize
        $\Pi_1$:
        \begin{math}
          $$\mprset{flushleft}
          \inferrule* [right={\tiny impL}] {
            {
              \begin{array}{cc}
                \pi_1 & \pi_2 \\
                {\Phi_{{\mathrm{1}}}  \vdash_\mathcal{C}  \SCnt{X_{{\mathrm{1}}}}} & {\Phi_{{\mathrm{2}}}  \SCsym{,}  \SCnt{X_{{\mathrm{2}}}}  \SCsym{,}  \Phi_{{\mathrm{3}}}  \vdash_\mathcal{C}  \SCnt{Y}}
              \end{array}
            }
          }{\Phi_{{\mathrm{2}}}  \SCsym{,}  \SCnt{X_{{\mathrm{1}}}}  \multimap  \SCnt{X_{{\mathrm{2}}}}  \SCsym{,}  \Phi_{{\mathrm{1}}}  \SCsym{,}  \Phi_{{\mathrm{3}}}  \vdash_\mathcal{C}  \SCnt{Y}}
        \end{math}
        \qquad\qquad
        \begin{math}
          \begin{array}{c}
            \Pi_2 \\
            {\Gamma_{{\mathrm{1}}}  \SCsym{;}  \SCnt{Y}  \SCsym{;}  \Gamma_{{\mathrm{2}}}  \vdash_\mathcal{L}  \SCnt{A}}
          \end{array}
        \end{math}
      \end{center}
      By assumption, $c(\Pi_1),c(\Pi_2)\leq |Y|$. By induction, there is a
      proof $\Pi'$ from $\pi_2$ and $\Pi_2$ for sequent
      $\Gamma_{{\mathrm{1}}}  \SCsym{;}  \Phi_{{\mathrm{2}}}  \SCsym{;}  \SCnt{X_{{\mathrm{2}}}}  \SCsym{;}  \Phi_{{\mathrm{3}}}  \SCsym{;}  \Gamma_{{\mathrm{2}}}  \vdash_\mathcal{L}  \SCnt{A}$ s.t. $c(\Pi')\leq |Y|$. Therefore, the
      proof $\Pi$ can be constructed as follows with $c(\Pi)\leq |Y|$.
      \begin{center}
        \scriptsize
        \begin{math}
          $$\mprset{flushleft}
          \inferrule* [right={\tiny impL}] {
            {
              \begin{array}{c}
                \pi_1 \\
                {\Phi_{{\mathrm{1}}}  \vdash_\mathcal{C}  \SCnt{X_{{\mathrm{1}}}}}
              \end{array}
            }
            $$\mprset{flushleft}
            \inferrule* [right={\tiny cut}] {
              {
                \begin{array}{cc}
                  \pi_2 & \Pi_2 \\
                  {\Phi_{{\mathrm{2}}}  \SCsym{,}  \SCnt{X_{{\mathrm{2}}}}  \SCsym{,}  \Phi_{{\mathrm{3}}}  \vdash_\mathcal{C}  \SCnt{Y}} & {\Gamma_{{\mathrm{1}}}  \SCsym{;}  \SCnt{Y}  \SCsym{;}  \Gamma_{{\mathrm{2}}}  \vdash_\mathcal{L}  \SCnt{A}}
                \end{array}
              }
            }{\Gamma_{{\mathrm{1}}}  \SCsym{;}  \Phi_{{\mathrm{2}}}  \SCsym{;}  \SCnt{X_{{\mathrm{2}}}}  \SCsym{;}  \Phi_{{\mathrm{3}}}  \SCsym{;}  \Gamma_{{\mathrm{2}}}  \vdash_\mathcal{L}  \SCnt{A}}
          }{\Gamma_{{\mathrm{1}}}  \SCsym{;}  \Phi_{{\mathrm{2}}}  \SCsym{;}  \SCnt{X_{{\mathrm{1}}}}  \multimap  \SCnt{X_{{\mathrm{2}}}}  \SCsym{;}  \Phi_{{\mathrm{1}}}  \SCsym{;}  \Phi_{{\mathrm{3}}}  \SCsym{;}  \Gamma_{{\mathrm{2}}}  \vdash_\mathcal{L}  \SCnt{A}}
        \end{math}
      \end{center}
\end{itemize}

\subsubsection{Left introduction of the non-commutative left implication $\lto$}
\begin{center}
\scriptsize
  $\Pi_1$:
  \begin{math}
    $$\mprset{flushleft}
    \inferrule* [right={\tiny impL}] {
      {
        \begin{array}{cc}
          \pi_1 & \pi_2 \\
          {\Gamma_{{\mathrm{1}}}  \vdash_\mathcal{L}  \SCnt{A_{{\mathrm{1}}}}} & {\Gamma_{{\mathrm{2}}}  \SCsym{;}  \SCnt{A_{{\mathrm{2}}}}  \SCsym{;}  \Gamma_{{\mathrm{3}}}  \vdash_\mathcal{L}  \SCnt{B}}
        \end{array}
      }
    }{\Gamma_{{\mathrm{2}}}  \SCsym{;}  \SCnt{A_{{\mathrm{1}}}}  \rightharpoonup  \SCnt{A_{{\mathrm{2}}}}  \SCsym{;}  \Gamma_{{\mathrm{1}}}  \SCsym{;}  \Gamma_{{\mathrm{3}}}  \vdash_\mathcal{L}  \SCnt{B}}
  \end{math}
  \qquad\qquad
  \begin{math}
    \begin{array}{c}
      \Pi_2 \\
      {\Delta_{{\mathrm{1}}}  \SCsym{;}  \SCnt{B}  \SCsym{;}  \Delta_{{\mathrm{2}}}  \vdash_\mathcal{L}  \SCnt{C}}
    \end{array}
  \end{math}
\end{center}
By assumption, $c(\Pi_1),c(\Pi_2)\leq |B|$. By induction, there is a
proof $\Pi'$ from $\pi_2$ and $\Pi_2$ for sequent
$\Delta_{{\mathrm{1}}}  \SCsym{;}  \Gamma_{{\mathrm{2}}}  \SCsym{;}  \SCnt{A_{{\mathrm{2}}}}  \SCsym{;}  \Gamma_{{\mathrm{3}}}  \SCsym{;}  \Delta_{{\mathrm{2}}}  \vdash_\mathcal{L}  \SCnt{C}$ s.t. $c(\Pi')\leq |B|$.
Therefore, the proof $\Pi$ can be constructed as follows with
$c(\Pi)\leq |B|$.
\begin{center}
  \scriptsize
  \begin{math}
    $$\mprset{flushleft}
    \inferrule* [right={\tiny impL}] {
      {
        \begin{array}{c}
          \pi_1 \\
          {\Gamma_{{\mathrm{1}}}  \vdash_\mathcal{L}  \SCnt{A_{{\mathrm{1}}}}}
        \end{array}
      }
      $$\mprset{flushleft}
      \inferrule* [right={\tiny cut}] {
        {
          \begin{array}{cc}
            \pi_2 & \Pi_2 \\
            {\Gamma_{{\mathrm{2}}}  \SCsym{;}  \SCnt{A_{{\mathrm{2}}}}  \SCsym{;}  \Gamma_{{\mathrm{3}}}  \vdash_\mathcal{L}  \SCnt{B}} & {\Delta_{{\mathrm{1}}}  \SCsym{;}  \SCnt{B}  \SCsym{;}  \Delta_{{\mathrm{2}}}  \vdash_\mathcal{L}  \SCnt{C}}
          \end{array}
        }
      }{\Delta_{{\mathrm{1}}}  \SCsym{;}  \Gamma_{{\mathrm{2}}}  \SCsym{;}  \SCnt{A_{{\mathrm{2}}}}  \SCsym{;}  \Gamma_{{\mathrm{3}}}  \SCsym{;}  \Delta_{{\mathrm{2}}}  \vdash_\mathcal{L}  \SCnt{C}}
    }{\Delta_{{\mathrm{1}}}  \SCsym{;}  \Gamma_{{\mathrm{2}}}  \SCsym{;}  \SCnt{A_{{\mathrm{1}}}}  \rightharpoonup  \SCnt{A_{{\mathrm{2}}}}  \SCsym{;}  \Gamma_{{\mathrm{1}}}  \SCsym{;}  \Gamma_{{\mathrm{3}}}  \SCsym{;}  \Delta_{{\mathrm{2}}}  \vdash_\mathcal{L}  \SCnt{C}}
  \end{math}
\end{center}

\subsubsection{Left introduction of the non-commutative right implication $\rto$}
\begin{center}
  \scriptsize
  $\Pi_1$:
  \begin{math}
    $$\mprset{flushleft}
    \inferrule* [right={\tiny impL}] {
      {
        \begin{array}{cc}
          \pi_1 & \pi_2 \\
          {\Gamma_{{\mathrm{1}}}  \vdash_\mathcal{L}  \SCnt{A_{{\mathrm{1}}}}} & {\Gamma_{{\mathrm{2}}}  \SCsym{;}  \SCnt{A_{{\mathrm{2}}}}  \SCsym{;}  \Gamma_{{\mathrm{3}}}  \vdash_\mathcal{L}  \SCnt{B}}
        \end{array}
      }
    }{\Gamma_{{\mathrm{2}}}  \SCsym{;}  \Gamma_{{\mathrm{1}}}  \SCsym{;}  \SCnt{A_{{\mathrm{2}}}}  \leftharpoonup  \SCnt{A_{{\mathrm{1}}}}  \SCsym{;}  \Gamma_{{\mathrm{3}}}  \vdash_\mathcal{L}  \SCnt{B}}
  \end{math}
  \qquad\qquad
  \begin{math}
    \begin{array}{c}
      \Pi_2 \\
      {\Delta_{{\mathrm{1}}}  \SCsym{;}  \SCnt{B}  \SCsym{;}  \Delta_{{\mathrm{2}}}  \vdash_\mathcal{L}  \SCnt{C}}
    \end{array}
  \end{math}
\end{center}
By assumption, $c(\Pi_1),c(\Pi_2)\leq |B|$. By induction, there is a
proof $\Pi'$ from $\pi_2$ and $\Pi_2$ for sequent
$\Delta_{{\mathrm{1}}}  \SCsym{;}  \Gamma_{{\mathrm{2}}}  \SCsym{;}  \SCnt{A_{{\mathrm{2}}}}  \SCsym{;}  \Gamma_{{\mathrm{3}}}  \SCsym{;}  \Delta_{{\mathrm{2}}}  \vdash_\mathcal{L}  \SCnt{C}$ s.t. $c(\Pi')\leq |B|$. Therefore, the
proof $\Pi$ can be constructed as follows with $c(\Pi)\leq |B|$.
\begin{center}
  \scriptsize
  \begin{math}
    $$\mprset{flushleft}
    \inferrule* [right={\tiny impL}] {
      {
        \begin{array}{c}
          \pi_1 \\
          {\Gamma_{{\mathrm{1}}}  \vdash_\mathcal{L}  \SCnt{A_{{\mathrm{1}}}}}
        \end{array}
      }
      $$\mprset{flushleft}
      \inferrule* [right={\tiny cut}] {
        {
          \begin{array}{cc}
            \pi_2 & \Pi_2 \\
            {\Gamma_{{\mathrm{2}}}  \SCsym{;}  \SCnt{A_{{\mathrm{2}}}}  \SCsym{;}  \Gamma_{{\mathrm{3}}}  \vdash_\mathcal{L}  \SCnt{B}} & {\Delta_{{\mathrm{1}}}  \SCsym{;}  \SCnt{B}  \SCsym{;}  \Delta_{{\mathrm{2}}}  \vdash_\mathcal{L}  \SCnt{C}}
          \end{array}
        }
      }{\Delta_{{\mathrm{1}}}  \SCsym{;}  \Gamma_{{\mathrm{2}}}  \SCsym{;}  \SCnt{A_{{\mathrm{2}}}}  \SCsym{;}  \Gamma_{{\mathrm{3}}}  \SCsym{;}  \Delta_{{\mathrm{2}}}  \vdash_\mathcal{L}  \SCnt{C}}
    }{\Delta_{{\mathrm{1}}}  \SCsym{;}  \Gamma_{{\mathrm{2}}}  \SCsym{;}  \Gamma_{{\mathrm{1}}}  \SCsym{;}  \SCnt{A_{{\mathrm{2}}}}  \leftharpoonup  \SCnt{A_{{\mathrm{1}}}}  \SCsym{;}  \Gamma_{{\mathrm{3}}}  \SCsym{;}  \Delta_{{\mathrm{2}}}  \vdash_\mathcal{L}  \SCnt{C}}
  \end{math}
\end{center}

\subsubsection{$\SCdruleTXXexName$}
\begin{itemize}
\item Case 1:
      \begin{center}
        \scriptsize
        $\Pi_1$:
        \begin{math}
          $$\mprset{flushleft}
          \inferrule* [right={\tiny ex}] {
            {
              \begin{array}{c}
                \pi \\
                {\Phi_{{\mathrm{1}}}  \SCsym{,}  \SCnt{X_{{\mathrm{1}}}}  \SCsym{,}  \SCnt{X_{{\mathrm{2}}}}  \SCsym{,}  \Phi_{{\mathrm{2}}}  \vdash_\mathcal{C}  \SCnt{Y}}
              \end{array}
            }
          }{\Phi_{{\mathrm{1}}}  \SCsym{,}  \SCnt{X_{{\mathrm{2}}}}  \SCsym{,}  \SCnt{X_{{\mathrm{1}}}}  \SCsym{,}  \Phi_{{\mathrm{2}}}  \vdash_\mathcal{C}  \SCnt{Y}}
        \end{math}
        \qquad\qquad
        \begin{math}
          \begin{array}{c}
            \Pi_2 \\
            {\Psi_{{\mathrm{1}}}  \SCsym{,}  \SCnt{Y}  \SCsym{,}  \Psi_{{\mathrm{2}}}  \vdash_\mathcal{C}  \SCnt{Z}}
          \end{array}
        \end{math}
      \end{center}
      By assumption, $c(\Pi_1),c(\Pi_2)\leq |Y|$. By induction on $\pi$
      and $\Pi_2$, there is a proof $\Pi'$ for sequent
      $\Psi_{{\mathrm{1}}}  \SCsym{,}  \Phi_{{\mathrm{1}}}  \SCsym{,}  \SCnt{X_{{\mathrm{1}}}}  \SCsym{,}  \SCnt{X_{{\mathrm{2}}}}  \SCsym{,}  \Phi_{{\mathrm{2}}}  \SCsym{,}  \Psi_{{\mathrm{2}}}  \vdash_\mathcal{C}  \SCnt{Z}$ s.t. $c(\Pi')\leq|Y|$. Therefore,
      the proof $\Pi$ can be constructed as follows, and
      $c(\Pi)=c(\Pi')\leq|Y|$.
      \begin{center}
        \scriptsize
        \begin{math}
          $$\mprset{flushleft}
          \inferrule* [right={\tiny ex}] {
            {
              \begin{array}{c}
                \Pi' \\
                {\Psi_{{\mathrm{1}}}  \SCsym{,}  \Phi_{{\mathrm{1}}}  \SCsym{,}  \SCnt{X_{{\mathrm{1}}}}  \SCsym{,}  \SCnt{X_{{\mathrm{2}}}}  \SCsym{,}  \Phi_{{\mathrm{2}}}  \SCsym{,}  \Psi_{{\mathrm{2}}}  \vdash_\mathcal{C}  \SCnt{Z}}
              \end{array}
            }
          }{\Psi_{{\mathrm{1}}}  \SCsym{,}  \Phi_{{\mathrm{1}}}  \SCsym{,}  \SCnt{X_{{\mathrm{2}}}}  \SCsym{,}  \SCnt{X_{{\mathrm{1}}}}  \SCsym{,}  \Phi_{{\mathrm{2}}}  \SCsym{,}  \Psi_{{\mathrm{2}}}  \vdash_\mathcal{C}  \SCnt{Z}}
        \end{math}
      \end{center}

\item Case 2:
      \begin{center}
        \scriptsize
        $\Pi_1$:
        \begin{math}
          $$\mprset{flushleft}
          \inferrule* [right={\tiny beta}] {
            {
              \begin{array}{c}
                \pi \\
                {\Phi_{{\mathrm{1}}}  \SCsym{,}  \SCnt{X}  \SCsym{,}  \SCnt{Y}  \SCsym{,}  \Phi_{{\mathrm{2}}}  \vdash_\mathcal{C}  \SCnt{Z}}
              \end{array}
            }
          }{\Phi_{{\mathrm{1}}}  \SCsym{,}  \SCnt{Y}  \SCsym{,}  \SCnt{X}  \SCsym{,}  \Phi_{{\mathrm{2}}}  \vdash_\mathcal{C}  \SCnt{Z}}
        \end{math}
        \qquad\qquad
        \begin{math}
          \begin{array}{c}
            \Pi_2 \\
            {\Gamma_{{\mathrm{1}}}  \SCsym{;}  \SCnt{Z}  \SCsym{;}  \Gamma_{{\mathrm{2}}}  \vdash_\mathcal{L}  \SCnt{A}}
          \end{array}
        \end{math}
      \end{center}
      By assumption, $c(\Pi_1),c(\Pi_2)\leq |Z|$. Similar as above, there
      is a proof $\Pi'$ constructed from $\pi$ and $\Pi_2$ for 
      $\Gamma_{{\mathrm{1}}}  \SCsym{;}  \Phi_{{\mathrm{1}}}  \SCsym{;}  \SCnt{X}  \SCsym{;}  \SCnt{Y}  \SCsym{;}  \Phi_{{\mathrm{2}}}  \SCsym{;}  \Gamma_{{\mathrm{2}}}  \vdash_\mathcal{L}  \SCnt{A}$ s.t. $c(\Pi')\leq|Z|$. Therefore,
      the proof $\Pi$ can be constructed as follows, and
      $c(\Pi)=c(\Pi')\leq|Z|$.
      \begin{center}
        \scriptsize
        \begin{math}
          $$\mprset{flushleft}
          \inferrule* [right={\tiny beta}] {
            {
              \begin{array}{c}
                \Pi' \\
                {\Gamma_{{\mathrm{1}}}  \SCsym{;}  \Phi_{{\mathrm{1}}}  \SCsym{;}  \SCnt{X}  \SCsym{;}  \SCnt{Y}  \SCsym{;}  \Phi_{{\mathrm{2}}}  \SCsym{;}  \Gamma_{{\mathrm{2}}}  \vdash_\mathcal{L}  \SCnt{A}}
              \end{array}
            }
          }{\Gamma_{{\mathrm{1}}}  \SCsym{;}  \Phi_{{\mathrm{1}}}  \SCsym{;}  \SCnt{Y}  \SCsym{;}  \SCnt{X}  \SCsym{;}  \Phi_{{\mathrm{2}}}  \SCsym{;}  \Gamma_{{\mathrm{2}}}  \vdash_\mathcal{L}  \SCnt{A}}
        \end{math}
      \end{center}
\end{itemize}

\subsubsection{$\SCdruleSXXexName$}
\begin{center}
  \scriptsize
  $\Pi_1$:
  \begin{math}
    $$\mprset{flushleft}
    \inferrule* [right={\tiny beta}] {
      {
        \begin{array}{c}
          \pi \\
          {\Gamma_{{\mathrm{1}}}  \SCsym{;}  \SCnt{X}  \SCsym{;}  \SCnt{Y}  \SCsym{;}  \Gamma_{{\mathrm{2}}}  \vdash_\mathcal{L}  \SCnt{A}}
        \end{array}
      }
    }{\Gamma_{{\mathrm{1}}}  \SCsym{;}  \SCnt{Y}  \SCsym{;}  \SCnt{X}  \SCsym{;}  \Gamma_{{\mathrm{2}}}  \vdash_\mathcal{L}  \SCnt{A}}
  \end{math}
  \qquad\qquad
  \begin{math}
    \begin{array}{c}
      \Pi_2 \\
      {\Delta_{{\mathrm{1}}}  \SCsym{;}  \SCnt{A}  \SCsym{;}  \Delta_{{\mathrm{2}}}  \vdash_\mathcal{L}  \SCnt{B}}
    \end{array}
  \end{math}
\end{center}
By assumption, $c(\Pi_1),c(\Pi_2)\leq |A|$. Similar as above, there
is a proof $\Pi'$ constructed from $\pi$ and $\Pi_2$ for sequent
$\Delta_{{\mathrm{1}}}  \SCsym{;}  \Gamma_{{\mathrm{1}}}  \SCsym{;}  \SCnt{X}  \SCsym{;}  \SCnt{Y}  \SCsym{;}  \Gamma_{{\mathrm{2}}}  \SCsym{;}  \Delta_{{\mathrm{2}}}  \vdash_\mathcal{L}  \SCnt{B}$ s.t. $c(\Pi')\leq|A|$. Therefore,
the proof $\Pi$ can be constructed as follows, and
$c(\Pi)=c(\Pi')\leq|A|$.
\begin{center}
  \scriptsize
  \begin{math}
    $$\mprset{flushleft}
    \inferrule* [right={\tiny beta}] {
      {
        \begin{array}{cc}
          \Pi' \\
          {\Delta_{{\mathrm{1}}}  \SCsym{;}  \Gamma_{{\mathrm{1}}}  \SCsym{;}  \SCnt{X}  \SCsym{;}  \SCnt{Y}  \SCsym{;}  \Gamma_{{\mathrm{2}}}  \SCsym{;}  \Delta_{{\mathrm{2}}}  \vdash_\mathcal{L}  \SCnt{B}}
        \end{array}
      }
    }{\Delta_{{\mathrm{1}}}  \SCsym{;}  \Gamma_{{\mathrm{1}}}  \SCsym{;}  \SCnt{Y}  \SCsym{;}  \SCnt{X}  \SCsym{;}  \Gamma_{{\mathrm{2}}}  \SCsym{;}  \Delta_{{\mathrm{2}}}  \vdash_\mathcal{L}  \SCnt{B}}
  \end{math}
\end{center}

\subsubsection{Left introduction of the commutative tensor product $\otimes$}
\begin{itemize}
\item Case 1:
      \begin{center}
        \scriptsize
        $\Pi_1$:
        \begin{math}
          $$\mprset{flushleft}
          \inferrule* [right={\tiny tenL}] {
            {
              \begin{array}{c}
                \pi \\
                {\Phi_{{\mathrm{1}}}  \SCsym{,}  \SCnt{X_{{\mathrm{1}}}}  \SCsym{,}  \SCnt{X_{{\mathrm{2}}}}  \SCsym{,}  \Phi_{{\mathrm{2}}}  \vdash_\mathcal{C}  \SCnt{Y}}
              \end{array}
            }
          }{\Phi_{{\mathrm{1}}}  \SCsym{,}  \SCnt{X_{{\mathrm{1}}}}  \otimes  \SCnt{X_{{\mathrm{2}}}}  \SCsym{,}  \Phi_{{\mathrm{2}}}  \vdash_\mathcal{C}  \SCnt{Y}}
        \end{math}
        \qquad\qquad
        \begin{math}
          \begin{array}{c}
            \Pi_2 \\
            {\Psi_{{\mathrm{1}}}  \SCsym{,}  \SCnt{Y}  \SCsym{,}  \Psi_{{\mathrm{2}}}  \vdash_\mathcal{C}  \SCnt{Z}}
          \end{array}
        \end{math}
      \end{center}
      By assumption, $c(\Pi_1),c(\Pi_2)\leq |Y|$. By induction, there is a
      proof $\Pi'$ from $\pi$ and $\Pi_2$ for sequent
      $\Psi_{{\mathrm{1}}}  \SCsym{,}  \Phi_{{\mathrm{1}}}  \SCsym{,}  \SCnt{X_{{\mathrm{1}}}}  \SCsym{,}  \SCnt{X_{{\mathrm{2}}}}  \SCsym{,}  \Phi_{{\mathrm{2}}}  \SCsym{,}  \Psi_{{\mathrm{2}}}  \vdash_\mathcal{C}  \SCnt{Z}$ s.t. $c(\Pi')\leq |Y|$. Therefore,
      the proof $\Pi$ can be constructed as follows with $c(\Pi)\leq |Y|$.
      \begin{center}
        \scriptsize
        \begin{math}
          $$\mprset{flushleft}
          \inferrule* [right={\tiny tenL}] {
            $$\mprset{flushleft}
            \inferrule* [right={\tiny cut}] {
              {
                \begin{array}{cc}
                  \pi & \Pi_2 \\
                  {\Phi_{{\mathrm{1}}}  \SCsym{,}  \SCnt{X_{{\mathrm{1}}}}  \SCsym{,}  \SCnt{X_{{\mathrm{2}}}}  \SCsym{,}  \Phi_{{\mathrm{2}}}  \vdash_\mathcal{C}  \SCnt{Y}} & {\Psi_{{\mathrm{1}}}  \SCsym{,}  \SCnt{Y}  \SCsym{,}  \Psi_{{\mathrm{2}}}  \vdash_\mathcal{C}  \SCnt{Z}}
                \end{array}
              }
            }{\Psi_{{\mathrm{1}}}  \SCsym{,}  \Phi_{{\mathrm{1}}}  \SCsym{,}  \SCnt{X_{{\mathrm{1}}}}  \SCsym{,}  \SCnt{X_{{\mathrm{2}}}}  \SCsym{,}  \Phi_{{\mathrm{2}}}  \SCsym{,}  \Psi_{{\mathrm{2}}}  \vdash_\mathcal{C}  \SCnt{Z}}
          }{\Psi_{{\mathrm{1}}}  \SCsym{,}  \Phi_{{\mathrm{1}}}  \SCsym{,}  \SCnt{X_{{\mathrm{1}}}}  \otimes  \SCnt{X_{{\mathrm{2}}}}  \SCsym{,}  \Phi_{{\mathrm{2}}}  \SCsym{,}  \Psi_{{\mathrm{2}}}  \vdash_\mathcal{C}  \SCnt{Z}}
        \end{math}
      \end{center}

\item Case 2:
      \begin{center}
        \scriptsize
        $\Pi_1$:
        \begin{math}
          $$\mprset{flushleft}
          \inferrule* [right={\tiny tenL}] {
            {
              \begin{array}{c}
                \pi \\
                {\Phi_{{\mathrm{1}}}  \SCsym{,}  \SCnt{X_{{\mathrm{1}}}}  \SCsym{,}  \SCnt{X_{{\mathrm{2}}}}  \SCsym{,}  \Phi_{{\mathrm{2}}}  \vdash_\mathcal{C}  \SCnt{Y}}
              \end{array}
            }
          }{\Phi_{{\mathrm{1}}}  \SCsym{,}  \SCnt{X_{{\mathrm{1}}}}  \otimes  \SCnt{X_{{\mathrm{2}}}}  \SCsym{,}  \Phi_{{\mathrm{2}}}  \vdash_\mathcal{C}  \SCnt{Y}}
        \end{math}
        \qquad\qquad
        \begin{math}
          \begin{array}{c}
            \Pi_2 \\
            {\Gamma_{{\mathrm{1}}}  \SCsym{;}  \SCnt{Y}  \SCsym{;}  \Gamma_{{\mathrm{2}}}  \vdash_\mathcal{L}  \SCnt{A}}
          \end{array}
        \end{math}
      \end{center}
      By assumption, $c(\Pi_1),c(\Pi_2)\leq |Y|$. By induction, there is a
      proof $\Pi'$ from $\pi$ and $\Pi_2$ for sequent
      $\Gamma_{{\mathrm{1}}}  \SCsym{;}  \Phi_{{\mathrm{1}}}  \SCsym{;}  \SCnt{X_{{\mathrm{1}}}}  \SCsym{;}  \SCnt{X_{{\mathrm{2}}}}  \SCsym{;}  \Phi_{{\mathrm{2}}}  \SCsym{;}  \Gamma_{{\mathrm{2}}}  \vdash_\mathcal{L}  \SCnt{A}$ s.t. $c(\Pi')\leq |Y|$. Therefore,
      the proof $\Pi$ can be constructed as follows with $c(\Pi)\leq |Y|$.
      \begin{center}
        \scriptsize
        \begin{math}
          $$\mprset{flushleft}
          \inferrule* [right={\tiny tenL1}] {
            $$\mprset{flushleft}
            \inferrule* [right={\tiny cut1}] {
              {
                \begin{array}{cc}
                  \pi & \Pi_2 \\
                  {\Phi_{{\mathrm{1}}}  \SCsym{,}  \SCnt{X_{{\mathrm{1}}}}  \SCsym{,}  \SCnt{X_{{\mathrm{2}}}}  \SCsym{,}  \Phi_{{\mathrm{2}}}  \vdash_\mathcal{C}  \SCnt{Y}} & {\Gamma_{{\mathrm{1}}}  \SCsym{;}  \SCnt{Y}  \SCsym{;}  \Gamma_{{\mathrm{2}}}  \vdash_\mathcal{L}  \SCnt{A}}
                \end{array}
              }
            }{\Gamma_{{\mathrm{1}}}  \SCsym{;}  \Phi_{{\mathrm{1}}}  \SCsym{;}  \SCnt{X_{{\mathrm{1}}}}  \SCsym{;}  \SCnt{X_{{\mathrm{2}}}}  \SCsym{;}  \Phi_{{\mathrm{2}}}  \SCsym{;}  \Gamma_{{\mathrm{2}}}  \vdash_\mathcal{L}  \SCnt{A}}
          }{\Gamma_{{\mathrm{1}}}  \SCsym{;}  \Phi_{{\mathrm{1}}}  \SCsym{;}  \SCnt{X_{{\mathrm{1}}}}  \otimes  \SCnt{X_{{\mathrm{2}}}}  \SCsym{;}  \Phi_{{\mathrm{2}}}  \SCsym{;}  \Gamma_{{\mathrm{2}}}  \vdash_\mathcal{L}  \SCnt{A}}
        \end{math}
      \end{center}

\item Case 3:
      \begin{center}
        \scriptsize
        $\Pi_1$:
        \begin{math}
          $$\mprset{flushleft}
          \inferrule* [right={\tiny tenL}] {
            {
              \begin{array}{c}
                \pi \\
                {\Gamma_{{\mathrm{1}}}  \SCsym{;}  \SCnt{X}  \SCsym{;}  \SCnt{Y}  \SCsym{;}  \Gamma_{{\mathrm{2}}}  \vdash_\mathcal{L}  \SCnt{A}}
              \end{array}
            }
          }{\Gamma_{{\mathrm{1}}}  \SCsym{;}  \SCnt{X}  \otimes  \SCnt{Y}  \SCsym{;}  \Gamma_{{\mathrm{2}}}  \vdash_\mathcal{L}  \SCnt{A}}
        \end{math}
        \qquad\qquad
        \begin{math}
          \begin{array}{c}
            \Pi_2 \\
            {\Delta_{{\mathrm{1}}}  \SCsym{;}  \SCnt{A}  \SCsym{;}  \Delta_{{\mathrm{2}}}  \vdash_\mathcal{L}  \SCnt{B}}
          \end{array}
        \end{math}
      \end{center}
      By assumption, $c(\Pi_1),c(\Pi_2)\leq |A|$. By induction, there is a
      proof $\Pi'$ from $\pi$ and $\Pi_2$ for sequent
      $\Delta_{{\mathrm{1}}}  \SCsym{;}  \SCnt{X}  \SCsym{;}  \SCnt{Y}  \SCsym{;}  \Gamma_{{\mathrm{2}}}  \SCsym{;}  \Delta_{{\mathrm{2}}}  \vdash_\mathcal{L}  \SCnt{B}$ s.t. $c(\Pi')\leq |A|$. Therefore, the
      proof $\Pi$ can be constructed as follows with $c(\Pi)\leq |A|$.
      \begin{center}
        \scriptsize
        \begin{math}
          $$\mprset{flushleft}
          \inferrule* [right={\tiny tenL1}] {
            $$\mprset{flushleft}
            \inferrule* [right={\tiny cut2}] {
              {
                \begin{array}{cc}
                  \pi & \Pi_2 \\
                  {\Gamma_{{\mathrm{1}}}  \SCsym{;}  \SCnt{X}  \SCsym{;}  \SCnt{Y}  \SCsym{;}  \Gamma_{{\mathrm{2}}}  \vdash_\mathcal{L}  \SCnt{A}} & {\Delta_{{\mathrm{1}}}  \SCsym{;}  \SCnt{A}  \SCsym{;}  \Delta_{{\mathrm{2}}}  \vdash_\mathcal{L}  \SCnt{B}}
                \end{array}
              }
            }{\Delta_{{\mathrm{1}}}  \SCsym{;}  \Gamma_{{\mathrm{1}}}  \SCsym{;}  \SCnt{X}  \SCsym{;}  \SCnt{Y}  \SCsym{;}  \Gamma_{{\mathrm{2}}}  \SCsym{;}  \Delta_{{\mathrm{2}}}  \vdash_\mathcal{L}  \SCnt{B}}
          }{\Delta_{{\mathrm{1}}}  \SCsym{;}  \Gamma_{{\mathrm{1}}}  \SCsym{;}  \SCnt{X}  \otimes  \SCnt{Y}  \SCsym{;}  \Gamma_{{\mathrm{2}}}  \SCsym{;}  \Delta_{{\mathrm{2}}}  \vdash_\mathcal{L}  \SCnt{B}}
        \end{math}
      \end{center}
\end{itemize}

\subsubsection{Left introduction of the non-commutative tensor products $\tri$}
\begin{center}
  \scriptsize
  $\Pi_1$:
  \begin{math}
    $$\mprset{flushleft}
    \inferrule* [right={\tiny tenL2}] {
      {
        \begin{array}{c}
          \pi \\
          {\Gamma_{{\mathrm{1}}}  \SCsym{;}  \SCnt{A_{{\mathrm{1}}}}  \SCsym{;}  \SCnt{A_{{\mathrm{2}}}}  \SCsym{;}  \Gamma_{{\mathrm{2}}}  \vdash_\mathcal{L}  \SCnt{B}}
        \end{array}
      }
    }{\Gamma_{{\mathrm{1}}}  \SCsym{;}  \SCnt{A_{{\mathrm{1}}}}  \triangleright  \SCnt{A_{{\mathrm{2}}}}  \SCsym{;}  \Gamma_{{\mathrm{2}}}  \vdash_\mathcal{L}  \SCnt{B}}
  \end{math}
  \qquad\qquad
  \begin{math}
    \begin{array}{c}
      \Pi_2 \\
      {\Delta_{{\mathrm{1}}}  \SCsym{;}  \SCnt{B}  \SCsym{;}  \Delta_{{\mathrm{2}}}  \vdash_\mathcal{L}  \SCnt{C}}
    \end{array}
  \end{math}
\end{center}
By assumption, $c(\Pi_1),c(\Pi_2)\leq |B|$. By induction, there is a
proof $\Pi'$ from $\pi$ and $\Pi_2$ for sequent \\
$\Delta_{{\mathrm{1}}}  \SCsym{;}  \Gamma_{{\mathrm{1}}}  \SCsym{;}  \SCnt{A_{{\mathrm{1}}}}  \SCsym{;}  \SCnt{A_{{\mathrm{2}}}}  \SCsym{;}  \Gamma_{{\mathrm{2}}}  \SCsym{;}  \Delta_{{\mathrm{2}}}  \vdash_\mathcal{L}  \SCnt{C}$ s.t. $c(\Pi')\leq |B|$.
Therefore, the proof $\Pi$ can be constructed as follows with
$c(\Pi)\leq |B|$.
\begin{center}
  \scriptsize
  \begin{math}
    $$\mprset{flushleft}
    \inferrule* [right={\tiny tenL2}] {
      $$\mprset{flushleft}
      \inferrule* [right={\tiny cut2}] {
        {
          \begin{array}{cc}
            \pi & \Pi_2 \\
            {\Gamma_{{\mathrm{1}}}  \SCsym{;}  \SCnt{A_{{\mathrm{1}}}}  \SCsym{;}  \SCnt{A_{{\mathrm{2}}}}  \SCsym{;}  \Gamma_{{\mathrm{2}}}  \vdash_\mathcal{L}  \SCnt{B}} & {\Delta_{{\mathrm{1}}}  \SCsym{;}  \SCnt{B}  \SCsym{;}  \Delta_{{\mathrm{2}}}  \vdash_\mathcal{L}  \SCnt{C}}
          \end{array}
        }
      }{\Delta_{{\mathrm{1}}}  \SCsym{;}  \Gamma_{{\mathrm{1}}}  \SCsym{;}  \SCnt{A_{{\mathrm{1}}}}  \SCsym{;}  \SCnt{A_{{\mathrm{2}}}}  \SCsym{;}  \Gamma_{{\mathrm{2}}}  \SCsym{;}  \Delta_{{\mathrm{2}}}  \vdash_\mathcal{L}  \SCnt{C}}
    }{\Delta_{{\mathrm{1}}}  \SCsym{;}  \Gamma_{{\mathrm{1}}}  \SCsym{;}  \SCnt{A_{{\mathrm{1}}}}  \triangleright  \SCnt{A_{{\mathrm{2}}}}  \SCsym{;}  \Gamma_{{\mathrm{2}}}  \SCsym{;}  \Delta_{{\mathrm{2}}}  \vdash_\mathcal{L}  \SCnt{C}}
  \end{math}
\end{center}

\subsubsection{Left introduction of the commutative unit $ \mathsf{Unit} $}
\begin{itemize}
\item Case 1:
      \begin{center}
        \scriptsize
        $\Pi_1$:
        \begin{math}
          $$\mprset{flushleft}
          \inferrule* [right={\tiny unitL}] {
            {
              \begin{array}{c}
                \pi \\
                {\Phi_{{\mathrm{1}}}  \SCsym{,}  \Phi_{{\mathrm{2}}}  \vdash_\mathcal{C}  \SCnt{X}}
              \end{array}
            }
          }{\Phi_{{\mathrm{1}}}  \SCsym{,}   \mathsf{Unit}   \SCsym{,}  \Phi_{{\mathrm{2}}}  \vdash_\mathcal{C}  \SCnt{X}}
        \end{math}
        \qquad\qquad
        \begin{math}
          \begin{array}{c}
            \Pi_2 \\
            {\Psi_{{\mathrm{1}}}  \SCsym{,}  \SCnt{X}  \SCsym{,}  \Psi_{{\mathrm{2}}}  \vdash_\mathcal{C}  \SCnt{Y}}
          \end{array}
        \end{math}
      \end{center}
      By assumption, $c(\Pi_1),c(\Pi_2)\leq |X|$. By induction, there is a
      proof $\Pi'$ from $\pi$ and $\Pi_2$ for sequent
      $\Psi_{{\mathrm{1}}}  \SCsym{,}  \Phi_{{\mathrm{1}}}  \SCsym{,}  \Phi_{{\mathrm{2}}}  \SCsym{,}  \Psi_{{\mathrm{2}}}  \vdash_\mathcal{C}  \SCnt{Y}$
      s.t. $c(\Pi')\leq |X|$. Therefore, the proof $\Pi$ can be constructed
      as follows, and $c(\Pi)=c(\Pi')\leq |X|$.
      \begin{center}
        \scriptsize
        \begin{math}
          $$\mprset{flushleft}
          \inferrule* [right={\tiny unitL}] {
            {
              \begin{array}{c}
                \Pi' \\
                {\Psi_{{\mathrm{1}}}  \SCsym{,}  \Phi_{{\mathrm{1}}}  \SCsym{,}  \Phi_{{\mathrm{2}}}  \SCsym{,}  \Psi_{{\mathrm{2}}}  \vdash_\mathcal{C}  \SCnt{Y}}
              \end{array}
            }
          }{\Psi_{{\mathrm{1}}}  \SCsym{,}  \Phi_{{\mathrm{1}}}  \SCsym{,}   \mathsf{Unit}   \SCsym{,}  \Phi_{{\mathrm{2}}}  \SCsym{,}  \Psi_{{\mathrm{2}}}  \vdash_\mathcal{C}  \SCnt{Y}}
        \end{math}
      \end{center}

\item Case 2:
      \begin{center}
        \scriptsize
        $\Pi_1$:
        \begin{math}
          $$\mprset{flushleft}
          \inferrule* [right={\tiny unitL}] {
            {
              \begin{array}{c}
                \pi \\
                {\Phi_{{\mathrm{1}}}  \SCsym{,}  \Phi_{{\mathrm{2}}}  \vdash_\mathcal{C}  \SCnt{X}}
              \end{array}
            }
          }{\Phi_{{\mathrm{1}}}  \SCsym{,}   \mathsf{Unit}   \SCsym{,}  \Phi_{{\mathrm{2}}}  \vdash_\mathcal{C}  \SCnt{X}}
        \end{math}
        \qquad\qquad
        \begin{math}
          \begin{array}{c}
            \Pi_2 \\
            {\Gamma_{{\mathrm{1}}}  \SCsym{;}  \SCnt{X}  \SCsym{;}  \Gamma_{{\mathrm{2}}}  \vdash_\mathcal{L}  \SCnt{A}}
          \end{array}
        \end{math}
      \end{center}
      By assumption, $c(\Pi_1),c(\Pi_2)\leq |X|$. By induction, there is a
      proof $\Pi'$ from $\pi$ and $\Pi_2$ for sequent
      $\Gamma_{{\mathrm{1}}}  \SCsym{;}  \Phi_{{\mathrm{1}}}  \SCsym{;}  \Phi_{{\mathrm{2}}}  \SCsym{;}  \Gamma_{{\mathrm{2}}}  \vdash_\mathcal{L}  \SCnt{A}$
      s.t. $c(\Pi')\leq |X|$. Therefore, the proof $\Pi$ can be constructed
      as follows, and $c(\Pi)=c(\Pi')\leq |X|$.
      \begin{center}
        \scriptsize
        \begin{math}
          $$\mprset{flushleft}
          \inferrule* [right={\tiny unitL}] {
            {
              \begin{array}{c}
                \Pi' \\
                {\Gamma_{{\mathrm{1}}}  \SCsym{;}  \Phi_{{\mathrm{1}}}  \SCsym{;}  \Phi_{{\mathrm{2}}}  \SCsym{;}  \Gamma_{{\mathrm{2}}}  \vdash_\mathcal{L}  \SCnt{A}}
              \end{array}
            }
          }{\Gamma_{{\mathrm{1}}}  \SCsym{;}  \Phi_{{\mathrm{1}}}  \SCsym{;}   \mathsf{Unit}   \SCsym{;}  \Phi_{{\mathrm{2}}}  \SCsym{;}  \Gamma_{{\mathrm{2}}}  \vdash_\mathcal{L}  \SCnt{A}}
        \end{math}
      \end{center}

\item Case 3:
      \begin{center}
        \scriptsize
        $\Pi_1$:
        \begin{math}
          $$\mprset{flushleft}
          \inferrule* [right={\tiny unitL}] {
            {
              \begin{array}{c}
                \pi \\
                {\Delta_{{\mathrm{1}}}  \SCsym{;}  \Delta_{{\mathrm{2}}}  \vdash_\mathcal{L}  \SCnt{A}}
              \end{array}
            }
          }{\Delta_{{\mathrm{1}}}  \SCsym{;}   \mathsf{Unit}   \SCsym{;}  \Delta_{{\mathrm{2}}}  \vdash_\mathcal{L}  \SCnt{A}}
        \end{math}
        \qquad\qquad
        \begin{math}
          \begin{array}{c}
            \Pi_2 \\
            {\Gamma_{{\mathrm{1}}}  \SCsym{;}  \SCnt{A}  \SCsym{;}  \Gamma_{{\mathrm{2}}}  \vdash_\mathcal{L}  \SCnt{B}}
          \end{array}
        \end{math}
      \end{center}
      By assumption, $c(\Pi_1),c(\Pi_2)\leq |X|$. By induction, there is a
      proof $\Pi'$ from $\pi$ and $\Pi_2$ for sequent
      $\Gamma_{{\mathrm{1}}}  \SCsym{;}  \Delta_{{\mathrm{1}}}  \SCsym{;}  \Delta_{{\mathrm{2}}}  \SCsym{;}  \Gamma_{{\mathrm{2}}}  \vdash_\mathcal{L}  \SCnt{B}$
      s.t. $c(\Pi')\leq |X|$. Therefore, the proof $\Pi$ can be constructed
      as follows, and $c(\Pi)=c(\Pi')\leq |X|$.
      \begin{center}
        \scriptsize
        \begin{math}
          $$\mprset{flushleft}
          \inferrule* [right={\tiny unitL}] {
            {
              \begin{array}{c}
                \Pi' \\
                {\Gamma_{{\mathrm{1}}}  \SCsym{;}  \Delta_{{\mathrm{1}}}  \SCsym{;}  \Delta_{{\mathrm{2}}}  \SCsym{;}  \Gamma_{{\mathrm{2}}}  \vdash_\mathcal{L}  \SCnt{B}}
              \end{array}
            }
          }{\Gamma_{{\mathrm{1}}}  \SCsym{;}  \Delta_{{\mathrm{1}}}  \SCsym{;}   \mathsf{Unit}   \SCsym{;}  \Delta_{{\mathrm{2}}}  \SCsym{;}  \Gamma_{{\mathrm{2}}}  \vdash_\mathcal{L}  \SCnt{B}}
        \end{math}
      \end{center}
\end{itemize}

\subsubsection{Left introduction of the non-commutative unit $ \mathsf{Unit} $}
\begin{center}
  \scriptsize
  $\Pi_1$:
  \begin{math}
    $$\mprset{flushleft}
    \inferrule* [right={\tiny unitL}] {
      {
        \begin{array}{c}
          \pi \\
          {\Delta_{{\mathrm{1}}}  \SCsym{;}  \Delta_{{\mathrm{2}}}  \vdash_\mathcal{L}  \SCnt{A}}
        \end{array}
      }
    }{\Delta_{{\mathrm{1}}}  \SCsym{;}   \mathsf{Unit}   \SCsym{;}  \Delta_{{\mathrm{2}}}  \vdash_\mathcal{L}  \SCnt{A}}
  \end{math}
  \qquad\qquad
  \begin{math}
    \begin{array}{c}
      \Pi_2 \\
      {\Gamma_{{\mathrm{1}}}  \SCsym{;}  \SCnt{A}  \SCsym{;}  \Gamma_{{\mathrm{2}}}  \vdash_\mathcal{L}  \SCnt{B}}
    \end{array}
  \end{math}
\end{center}
By assumption, $c(\Pi_1),c(\Pi_2)\leq |X|$. By induction, there is a
proof $\Pi'$ from $\pi$ and $\Pi_2$ for sequent
\begin{center}
  \begin{math}
    \Gamma_{{\mathrm{1}}}  \SCsym{;}  \Delta_{{\mathrm{1}}}  \SCsym{;}  \Delta_{{\mathrm{2}}}  \SCsym{;}  \Gamma_{{\mathrm{2}}}  \vdash_\mathcal{L}  \SCnt{B}
  \end{math}
\end{center}
s.t. $c(\Pi')\leq |X|$. Therefore, the proof $\Pi$ can be constructed
as follows, and $c(\Pi)=c(\Pi')\leq |X|$.
\begin{center}
  \scriptsize
  \begin{math}
    $$\mprset{flushleft}
    \inferrule* [right={\tiny unitL}] {
      {
        \begin{array}{c}
          \Pi' \\
          {\Gamma_{{\mathrm{1}}}  \SCsym{;}  \Delta_{{\mathrm{1}}}  \SCsym{;}  \Delta_{{\mathrm{2}}}  \SCsym{;}  \Gamma_{{\mathrm{2}}}  \vdash_\mathcal{L}  \SCnt{B}}
        \end{array}
      }
    }{\Gamma_{{\mathrm{1}}}  \SCsym{;}  \Delta_{{\mathrm{1}}}  \SCsym{;}   \mathsf{Unit}   \SCsym{;}  \Delta_{{\mathrm{2}}}  \SCsym{;}  \Gamma_{{\mathrm{2}}}  \vdash_\mathcal{L}  \SCnt{B}}
  \end{math}
\end{center}

\subsubsection{Left introduction of the functor $F$}
\begin{center}
  \scriptsize
  $\Pi_1$:
  \begin{math}
    $$\mprset{flushleft}
    \inferrule* [right={\tiny FL}] {
      {
        \begin{array}{c}
          \pi_1 \\
          {\Gamma_{{\mathrm{1}}}  \SCsym{;}  \SCnt{X}  \SCsym{;}  \Gamma_{{\mathrm{2}}}  \vdash_\mathcal{L}  \SCnt{A}}
        \end{array}
      }
    }{\Gamma_{{\mathrm{1}}}  \SCsym{;}   \mathsf{F} \SCnt{X}   \SCsym{;}  \Gamma_{{\mathrm{2}}}  \vdash_\mathcal{L}  \SCnt{A}}
  \end{math}
  \qquad\qquad
  \begin{math}
    \begin{array}{c}
      \Pi_2 \\
      {\Delta_{{\mathrm{1}}}  \SCsym{;}  \SCnt{A}  \SCsym{;}  \Delta_{{\mathrm{2}}}  \vdash_\mathcal{L}  \SCnt{B}}
    \end{array}
  \end{math}
\end{center}
By assumption, $c(\Pi_1),c(\Pi_2)\leq |A|$. By induction, there is a
proof $\Pi'$ from $\pi_2$ and $\Pi_2$ for sequent
$\Delta_{{\mathrm{1}}}  \SCsym{;}  \Gamma_{{\mathrm{1}}}  \SCsym{;}  \SCnt{X}  \SCsym{;}  \Gamma_{{\mathrm{2}}}  \SCsym{;}  \Delta_{{\mathrm{2}}}  \vdash_\mathcal{L}  \SCnt{B}$ s.t. $c(\Pi')\leq |A|$. Therefore, the
proof $\Pi$ can be constructed as follows with $c(\Pi)\leq |A|$.
\begin{center}
  \scriptsize
  \begin{math}
    $$\mprset{flushleft}
    \inferrule* [right={\tiny FL}] {
      $$\mprset{flushleft}
      \inferrule* [right={\tiny cut2}] {
        {
          \begin{array}{cc}
            \pi_2 & \Pi_2 \\
            {\Gamma_{{\mathrm{1}}}  \SCsym{;}  \SCnt{X}  \SCsym{;}  \Gamma_{{\mathrm{2}}}  \vdash_\mathcal{L}  \SCnt{A}} & {\Delta_{{\mathrm{1}}}  \SCsym{;}  \SCnt{A}  \SCsym{;}  \Delta_{{\mathrm{2}}}  \vdash_\mathcal{L}  \SCnt{B}}
          \end{array}
        }
      }{\Delta_{{\mathrm{1}}}  \SCsym{;}  \Gamma_{{\mathrm{1}}}  \SCsym{;}  \SCnt{X}  \SCsym{;}  \Gamma_{{\mathrm{2}}}  \SCsym{;}  \Delta_{{\mathrm{2}}}  \vdash_\mathcal{L}  \SCnt{B}}
    }{\Delta_{{\mathrm{1}}}  \SCsym{;}  \Gamma_{{\mathrm{1}}}  \SCsym{;}   \mathsf{F} \SCnt{X}   \SCsym{;}  \Gamma_{{\mathrm{2}}}  \SCsym{;}  \Delta_{{\mathrm{2}}}  \vdash_\mathcal{L}  \SCnt{B}}
  \end{math}
\end{center}

\subsubsection{Left introduction of the functor $G$}
\begin{center}
  \scriptsize
  $\Pi_1$:
  \begin{math}
    $$\mprset{flushleft}
    \inferrule* [right={\tiny GL}] {
      {
        \begin{array}{c}
          \pi_1 \\
          {\Gamma_{{\mathrm{1}}}  \SCsym{;}  \SCnt{A}  \SCsym{;}  \Gamma_{{\mathrm{2}}}  \vdash_\mathcal{L}  \SCnt{B}}
        \end{array}
      }
    }{\Gamma_{{\mathrm{1}}}  \SCsym{;}   \mathsf{G} \SCnt{A}   \SCsym{;}  \Gamma_{{\mathrm{2}}}  \vdash_\mathcal{L}  \SCnt{B}}
  \end{math}
  \qquad\qquad
  \begin{math}
    \begin{array}{c}
      \Pi_2 \\
      {\Delta_{{\mathrm{1}}}  \SCsym{;}  \SCnt{B}  \SCsym{;}  \Delta_{{\mathrm{2}}}  \vdash_\mathcal{L}  \SCnt{C}}
    \end{array}
  \end{math}
\end{center}
By assumption, $c(\Pi_1),c(\Pi_2)\leq |B|$. By induction, there is a
proof $\Pi'$ from $\pi_2$ and $\Pi_2$ for sequent
$\Delta_{{\mathrm{1}}}  \SCsym{;}  \Gamma_{{\mathrm{1}}}  \SCsym{;}  \SCnt{A}  \SCsym{;}  \Gamma_{{\mathrm{2}}}  \SCsym{;}  \Delta_{{\mathrm{2}}}  \vdash_\mathcal{L}  \SCnt{C}$ s.t. $c(\Pi')\leq |B|$. Therefore, the
proof $\Pi$ can be constructed as follows with $c(\Pi)\leq |B|$.
\begin{center}
  \scriptsize
  \begin{math}
    $$\mprset{flushleft}
    \inferrule* [right={\tiny GL}] {
      $$\mprset{flushleft}
      \inferrule* [right={\tiny cut2}] {
        {
          \begin{array}{cc}
            \pi_2 & \Pi_2 \\
            {\Gamma_{{\mathrm{1}}}  \SCsym{;}  \SCnt{A}  \SCsym{;}  \Gamma_{{\mathrm{2}}}  \vdash_\mathcal{L}  \SCnt{B}} & {\Delta_{{\mathrm{1}}}  \SCsym{;}  \SCnt{B}  \SCsym{;}  \Delta_{{\mathrm{2}}}  \vdash_\mathcal{L}  \SCnt{C}}
          \end{array}
        }
      }{\Delta_{{\mathrm{1}}}  \SCsym{;}  \Gamma_{{\mathrm{1}}}  \SCsym{;}  \SCnt{A}  \SCsym{;}  \Gamma_{{\mathrm{2}}}  \SCsym{;}  \Delta_{{\mathrm{2}}}  \vdash_\mathcal{L}  \SCnt{C}}
    }{\Delta_{{\mathrm{1}}}  \SCsym{;}  \Gamma_{{\mathrm{1}}}  \SCsym{;}   \mathsf{G} \SCnt{A}   \SCsym{;}  \Gamma_{{\mathrm{2}}}  \SCsym{;}  \Delta_{{\mathrm{2}}}  \vdash_\mathcal{L}  \SCnt{C}}
  \end{math}
\end{center}

\subsection{Secondary Hypothesis}

\subsubsection{Right introduction of the commutative tensor product $\otimes$}
\begin{itemize}
\item Case 1:
      \begin{center}
        \scriptsize
        \begin{math}
          \begin{array}{c}
            \Pi_1 \\
            {\Phi_{{\mathrm{2}}}  \vdash_\mathcal{C}  \SCnt{X}}
          \end{array}
        \end{math}
        \qquad\qquad
        $\Pi_2$:
        \begin{math}
          $$\mprset{flushleft}
          \inferrule* [right={\tiny tenR}] {
            {
              \begin{array}{cc}
                \pi_1 & \pi_2 \\
                {\Psi_{{\mathrm{1}}}  \SCsym{,}  \SCnt{X}  \SCsym{,}  \Psi_{{\mathrm{2}}}  \vdash_\mathcal{C}  \SCnt{Y_{{\mathrm{1}}}}} & {\Phi_{{\mathrm{1}}}  \vdash_\mathcal{C}  \SCnt{Y_{{\mathrm{2}}}}}
              \end{array}
            }
          }{\Psi_{{\mathrm{1}}}  \SCsym{,}  \SCnt{X}  \SCsym{,}  \Psi_{{\mathrm{2}}}  \SCsym{,}  \Phi_{{\mathrm{1}}}  \vdash_\mathcal{C}  \SCnt{Y_{{\mathrm{1}}}}  \otimes  \SCnt{Y_{{\mathrm{2}}}}}
        \end{math}
      \end{center}
      By assumption, $c(\Pi_1),c(\Pi_2)\leq |X|$. By induction on $\Pi_1$
      and $\pi_1$, there is a proof $\Pi'$ for sequent
      $\Psi_{{\mathrm{1}}}  \SCsym{,}  \Phi_{{\mathrm{2}}}  \SCsym{,}  \Psi_{{\mathrm{2}}}  \vdash_\mathcal{C}  \SCnt{Y_{{\mathrm{1}}}}$ s.t. $c(\Pi') \leq |X|$. Therefore, the proof
      $\Pi$ can be constructed as follows with $c(\Pi) = c(\Pi') \leq |X|$.
      \begin{center}
        \scriptsize
        \begin{math}
          $$\mprset{flushleft}
          \inferrule* [right={\tiny tenR}] {
            {
              \begin{array}{cc}
                \Pi' & \pi_1 \\
                {\Psi_{{\mathrm{1}}}  \SCsym{,}  \Phi_{{\mathrm{2}}}  \SCsym{,}  \Psi_{{\mathrm{2}}}  \vdash_\mathcal{C}  \SCnt{Y_{{\mathrm{1}}}}} & {\Phi_{{\mathrm{1}}}  \vdash_\mathcal{C}  \SCnt{Y_{{\mathrm{2}}}}}
              \end{array}
            }
          }{\Psi_{{\mathrm{1}}}  \SCsym{,}  \Phi_{{\mathrm{2}}}  \SCsym{,}  \Psi_{{\mathrm{2}}}  \SCsym{,}  \Phi_{{\mathrm{1}}}  \vdash_\mathcal{C}  \SCnt{Y_{{\mathrm{1}}}}  \otimes  \SCnt{Y_{{\mathrm{2}}}}}
        \end{math}
      \end{center}

\item Case 2:
      \begin{center}
        \scriptsize
        \begin{math}
          \begin{array}{c}
            \Pi_1 \\
            {\Phi_{{\mathrm{2}}}  \vdash_\mathcal{C}  \SCnt{X}}
          \end{array}
        \end{math}
        \qquad\qquad
        $\Pi_2$:
        \begin{math}
          $$\mprset{flushleft}
          \inferrule* [right={\tiny tenR}] {
            {
              \begin{array}{cc}
                \pi_1 & \pi_2 \\
                {\Phi_{{\mathrm{1}}}  \vdash_\mathcal{C}  \SCnt{Y_{{\mathrm{1}}}}} & {\Psi_{{\mathrm{1}}}  \SCsym{,}  \SCnt{X}  \SCsym{,}  \Psi_{{\mathrm{2}}}  \vdash_\mathcal{C}  \SCnt{Y_{{\mathrm{2}}}}}
              \end{array}
            }
          }{\Phi_{{\mathrm{1}}}  \SCsym{,}  \Psi_{{\mathrm{1}}}  \SCsym{,}  \SCnt{X}  \SCsym{,}  \Psi_{{\mathrm{2}}}  \vdash_\mathcal{C}  \SCnt{Y_{{\mathrm{1}}}}  \otimes  \SCnt{Y_{{\mathrm{2}}}}}
        \end{math}
      \end{center}
      By assumption, $c(\Pi_1),c(\Pi_2)\leq |X|$. By induction on $\Pi_1$
      and $\pi_2$, there is a proof $\Pi'$ for sequent
      $\Psi_{{\mathrm{1}}}  \SCsym{,}  \Phi_{{\mathrm{2}}}  \SCsym{,}  \Psi_{{\mathrm{2}}}  \vdash_\mathcal{C}  \SCnt{Y_{{\mathrm{2}}}}$ s.t. $c(\Pi') \leq |X|$. Therefore, the proof
      $\Pi$ can be constructed as follows with $c(\Pi) = c(\Pi') \leq |X|$.
      \begin{center}
        \scriptsize
        \begin{math}
          $$\mprset{flushleft}
          \inferrule* [right={\tiny tenR}] {
            {
              \begin{array}{cc}
                \pi_1 & \Pi' \\
                {\Phi_{{\mathrm{1}}}  \vdash_\mathcal{C}  \SCnt{Y_{{\mathrm{1}}}}} & {\Psi_{{\mathrm{1}}}  \SCsym{,}  \Phi_{{\mathrm{2}}}  \SCsym{,}  \Psi_{{\mathrm{2}}}  \vdash_\mathcal{C}  \SCnt{Y_{{\mathrm{2}}}}}
              \end{array}
            }
          }{\Phi_{{\mathrm{1}}}  \SCsym{,}  \Psi_{{\mathrm{1}}}  \SCsym{,}  \Phi_{{\mathrm{2}}}  \SCsym{,}  \Psi_{{\mathrm{2}}}  \vdash_\mathcal{C}  \SCnt{Y_{{\mathrm{1}}}}  \otimes  \SCnt{Y_{{\mathrm{2}}}}}
        \end{math}
      \end{center}
\end{itemize}

\subsubsection{Right introduction of the non-commutative tensor product $\tri$}
\begin{itemize}
\item Case 1:
      \begin{center}
        \scriptsize
        \begin{math}
          \begin{array}{c}
            \Pi_1 \\
            {\Phi  \vdash_\mathcal{C}  \SCnt{X}}
          \end{array}
        \end{math}
        \qquad\qquad
        $\Pi_2$:
        \begin{math}
          $$\mprset{flushleft}
          \inferrule* [right={\tiny tenR}] {
            {
              \begin{array}{cc}
                \pi_1 & \pi_2 \\
                {\Gamma_{{\mathrm{1}}}  \SCsym{;}  \SCnt{X}  \SCsym{;}  \Gamma_{{\mathrm{2}}}  \vdash_\mathcal{L}  \SCnt{A}} & {\Gamma_{{\mathrm{3}}}  \vdash_\mathcal{L}  \SCnt{B}}
              \end{array}
            }
          }{\Gamma_{{\mathrm{1}}}  \SCsym{;}  \SCnt{X}  \SCsym{;}  \Gamma_{{\mathrm{2}}}  \SCsym{;}  \Gamma_{{\mathrm{3}}}  \vdash_\mathcal{L}  \SCnt{A}  \triangleright  \SCnt{B}}
        \end{math}
      \end{center}
      By assumption, $c(\Pi_1),c(\Pi_2)\leq |X|$. By induction on $\Pi_1$
      and $\pi_1$, there is a proof $\Pi'$ for sequent
      $\Gamma_{{\mathrm{1}}}  \SCsym{;}  \Phi  \SCsym{;}  \Gamma_{{\mathrm{2}}}  \vdash_\mathcal{L}  \SCnt{A}$ s.t. $c(\Pi') \leq |X|$. Therefore, the proof
      $\Pi$ can be constructed as follows with $c(\Pi) = c(\Pi') \leq |X|$.
      \begin{center}
        \scriptsize
        \begin{math}
          $$\mprset{flushleft}
          \inferrule* [right={\tiny tenR}] {
            {
              \begin{array}{cc}
                \Pi' & \pi_1 \\
                {\Gamma_{{\mathrm{1}}}  \SCsym{;}  \Phi  \SCsym{;}  \Gamma_{{\mathrm{2}}}  \vdash_\mathcal{L}  \SCnt{A}} & {\Gamma_{{\mathrm{3}}}  \vdash_\mathcal{L}  \SCnt{B}}
              \end{array}
            }
          }{\Gamma_{{\mathrm{1}}}  \SCsym{;}  \Phi  \SCsym{;}  \Gamma_{{\mathrm{2}}}  \SCsym{;}  \Gamma_{{\mathrm{3}}}  \vdash_\mathcal{L}  \SCnt{A}  \triangleright  \SCnt{B}}
        \end{math}
      \end{center}

\item Case 2:
      \begin{center}
        \scriptsize
        \begin{math}
          \begin{array}{c}
            \Pi_1 \\
            {\Delta  \vdash_\mathcal{L}  \SCnt{C}}
          \end{array}
        \end{math}
        \qquad\qquad
        $\Pi_2$:
        \begin{math}
          $$\mprset{flushleft}
          \inferrule* [right={\tiny tenR}] {
            {
              \begin{array}{cc}
                \pi_1 & \pi_2 \\
                {\Gamma_{{\mathrm{1}}}  \SCsym{;}  \SCnt{C}  \SCsym{;}  \Gamma_{{\mathrm{2}}}  \vdash_\mathcal{L}  \SCnt{A}} & {\Gamma_{{\mathrm{3}}}  \vdash_\mathcal{L}  \SCnt{B}}
              \end{array}
            }
          }{\Gamma_{{\mathrm{1}}}  \SCsym{;}  \SCnt{C}  \SCsym{;}  \Gamma_{{\mathrm{2}}}  \SCsym{;}  \Gamma_{{\mathrm{3}}}  \vdash_\mathcal{L}  \SCnt{A}  \triangleright  \SCnt{B}}
        \end{math}
      \end{center}
      By assumption, $c(\Pi_1),c(\Pi_2)\leq |C|$. By induction on $\Pi_1$
      and $\pi_1$, there is a proof $\Pi'$ for sequent
      $\Gamma_{{\mathrm{1}}}  \SCsym{;}  \Delta  \SCsym{;}  \Gamma_{{\mathrm{2}}}  \vdash_\mathcal{L}  \SCnt{A}$ s.t. $c(\Pi') \leq |C|$. Therefore, the proof
      $\Pi$ can be constructed as follows with $c(\Pi) = c(\Pi') \leq |C|$.
      \begin{center}
        \scriptsize
        \begin{math}
          $$\mprset{flushleft}
          \inferrule* [right={\tiny tenR}] {
            {
              \begin{array}{cc}
                \Pi' & \pi_1 \\
                {\Gamma_{{\mathrm{1}}}  \SCsym{;}  \Delta  \SCsym{;}  \Gamma_{{\mathrm{2}}}  \vdash_\mathcal{L}  \SCnt{A}} & {\Gamma_{{\mathrm{3}}}  \vdash_\mathcal{L}  \SCnt{B}}
              \end{array}
            }
          }{\Gamma_{{\mathrm{1}}}  \SCsym{;}  \Delta  \SCsym{;}  \Gamma_{{\mathrm{2}}}  \SCsym{;}  \Gamma_{{\mathrm{3}}}  \vdash_\mathcal{L}  \SCnt{A}  \triangleright  \SCnt{B}}
        \end{math}
      \end{center}

\item Case 3:
      \begin{center}
        \scriptsize
        \begin{math}
          \begin{array}{c}
            \Pi_1 \\
            {\Phi  \vdash_\mathcal{C}  \SCnt{X}}
          \end{array}
        \end{math}
        \qquad\qquad
        $\Pi_2$:
        \begin{math}
          $$\mprset{flushleft}
          \inferrule* [right={\tiny tenR}] {
            {
              \begin{array}{cc}
                \pi_1 & \pi_2 \\
                {\Gamma_{{\mathrm{1}}}  \vdash_\mathcal{L}  \SCnt{A}} & {\Gamma_{{\mathrm{2}}}  \SCsym{;}  \SCnt{X}  \SCsym{;}  \Gamma_{{\mathrm{3}}}  \vdash_\mathcal{L}  \SCnt{B}}
              \end{array}
            }
          }{\Gamma_{{\mathrm{1}}}  \SCsym{;}  \Gamma_{{\mathrm{2}}}  \SCsym{;}  \SCnt{X}  \SCsym{;}  \Gamma_{{\mathrm{3}}}  \vdash_\mathcal{L}  \SCnt{A}  \triangleright  \SCnt{B}}
        \end{math}
      \end{center}
      By assumption, $c(\Pi_1),c(\Pi_2)\leq |X|$. By induction on $\Pi_1$
      and $\pi_2$, there is a proof $\Pi'$ for sequent
      $\Gamma_{{\mathrm{2}}}  \SCsym{;}  \Phi  \SCsym{;}  \Gamma_{{\mathrm{3}}}  \vdash_\mathcal{L}  \SCnt{B}$ s.t. $c(\Pi') \leq |X|$. Therefore, the proof
      $\Pi$ can be constructed as follows with $c(\Pi) = c(\Pi') \leq |X|$.
      \begin{center}
        \scriptsize
        \begin{math}
          $$\mprset{flushleft}
          \inferrule* [right={\tiny tenR}] {
            {
              \begin{array}{cc}
                \pi_1 & \Pi' \\
                {\Gamma_{{\mathrm{1}}}  \vdash_\mathcal{L}  \SCnt{A}} & {\Gamma_{{\mathrm{2}}}  \SCsym{;}  \Phi  \SCsym{;}  \Gamma_{{\mathrm{3}}}  \vdash_\mathcal{L}  \SCnt{B}}
              \end{array}
            }
          }{\Gamma_{{\mathrm{1}}}  \SCsym{;}  \Gamma_{{\mathrm{2}}}  \SCsym{;}  \Phi  \SCsym{;}  \Gamma_{{\mathrm{3}}}  \vdash_\mathcal{L}  \SCnt{A}  \triangleright  \SCnt{B}}
        \end{math}
      \end{center}

\item Case 4:
      \begin{center}
        \scriptsize
        \begin{math}
          \begin{array}{c}
            \Pi_1 \\
            {\Delta  \vdash_\mathcal{L}  \SCnt{C}}
          \end{array}
        \end{math}
        \qquad\qquad
        $\Pi_2$:
        \begin{math}
          $$\mprset{flushleft}
          \inferrule* [right={\tiny tenR}] {
            {
              \begin{array}{cc}
                \pi_1 & \pi_2 \\
                {\Gamma_{{\mathrm{1}}}  \vdash_\mathcal{L}  \SCnt{A}} & {\Gamma_{{\mathrm{2}}}  \SCsym{;}  \SCnt{C}  \SCsym{;}  \Gamma_{{\mathrm{3}}}  \vdash_\mathcal{L}  \SCnt{B}}
              \end{array}
            }
          }{\Gamma_{{\mathrm{1}}}  \SCsym{;}  \Gamma_{{\mathrm{2}}}  \SCsym{;}  \SCnt{C}  \SCsym{;}  \Gamma_{{\mathrm{3}}}  \vdash_\mathcal{L}  \SCnt{A}  \triangleright  \SCnt{B}}
        \end{math}
      \end{center}
      By assumption, $c(\Pi_1),c(\Pi_2)\leq |C|$. By induction on $\Pi_1$
      and $\pi_2$, there is a proof $\Pi'$ for sequent
      $\Gamma_{{\mathrm{2}}}  \SCsym{;}  \Delta  \SCsym{;}  \Gamma_{{\mathrm{3}}}  \vdash_\mathcal{L}  \SCnt{B}$ s.t. $c(\Pi') \leq |C|$. Therefore, the proof
      $\Pi$ can be constructed as follows with $c(\Pi) = c(\Pi') \leq |C|$.
      \begin{center}
        \scriptsize
        \begin{math}
          $$\mprset{flushleft}
          \inferrule* [right={\tiny tenR}] {
            {
              \begin{array}{cc}
                \pi_1 & \Pi' \\
                {\Gamma_{{\mathrm{1}}}  \vdash_\mathcal{L}  \SCnt{A}} & {\Gamma_{{\mathrm{2}}}  \SCsym{;}  \Delta  \SCsym{;}  \Gamma_{{\mathrm{3}}}  \vdash_\mathcal{L}  \SCnt{B}}
              \end{array}
            }
          }{\Gamma_{{\mathrm{1}}}  \SCsym{;}  \Gamma_{{\mathrm{2}}}  \SCsym{;}  \Delta  \SCsym{;}  \Gamma_{{\mathrm{3}}}  \vdash_\mathcal{L}  \SCnt{A}  \triangleright  \SCnt{B}}
        \end{math}
      \end{center}
\end{itemize}

\subsubsection{Left introduction of the commutative implication $\multimap$}
\begin{itemize}
\item Case 1:
      \begin{center}
        \scriptsize
        \begin{math}
          \begin{array}{c}
            \Pi_1 \\
            {\Phi  \vdash_\mathcal{C}  \SCnt{X}}
          \end{array}
        \end{math}
        \qquad\qquad
        $\Pi_2$:
        \begin{math}
          $$\mprset{flushleft}
          \inferrule* [right={\tiny impL}] {
            {
              \begin{array}{cc}
                \pi_1 & \pi_2 \\
                {\Psi_{{\mathrm{2}}}  \SCsym{,}  \SCnt{X}  \SCsym{,}  \Psi_{{\mathrm{3}}}  \vdash_\mathcal{C}  \SCnt{Y_{{\mathrm{1}}}}} & {\Psi_{{\mathrm{1}}}  \SCsym{,}  \SCnt{Y_{{\mathrm{2}}}}  \SCsym{,}  \Psi_{{\mathrm{4}}}  \vdash_\mathcal{C}  \SCnt{Z}}
              \end{array}
            }
          }{\Psi_{{\mathrm{1}}}  \SCsym{,}  \SCnt{Y_{{\mathrm{1}}}}  \multimap  \SCnt{Y_{{\mathrm{2}}}}  \SCsym{,}  \Psi_{{\mathrm{2}}}  \SCsym{,}  \SCnt{X}  \SCsym{,}  \Psi_{{\mathrm{3}}}  \SCsym{,}  \Psi_{{\mathrm{4}}}  \vdash_\mathcal{C}  \SCnt{Z}}
        \end{math}
      \end{center}
      By assumption, $c(\Pi_1),c(\Pi_2)\leq |X|$. By induction on $\Pi_1$ and $\pi_1$, there is
      a proof $\Pi'$ for sequent $\Psi_{{\mathrm{2}}}  \SCsym{,}  \Phi  \SCsym{,}  \Psi_{{\mathrm{3}}}  \vdash_\mathcal{C}  \SCnt{Y_{{\mathrm{1}}}}$ s.t. $c(\Pi') \leq |X|$. Therefore, the
      proof $\Pi$ can be constructed as follows with $c(\Pi) = c(\Pi') \leq |X|$.
      \begin{center}
        \scriptsize
        \begin{math}
          $$\mprset{flushleft}
          \inferrule* [right={\tiny impL}] {
            {
              \begin{array}{cc}
                \Pi' & \pi_2 \\
                {\Psi_{{\mathrm{2}}}  \SCsym{,}  \Phi  \SCsym{,}  \Psi_{{\mathrm{3}}}  \vdash_\mathcal{C}  \SCnt{Y_{{\mathrm{1}}}}} & {\Psi_{{\mathrm{1}}}  \SCsym{,}  \SCnt{Y_{{\mathrm{2}}}}  \SCsym{,}  \Psi_{{\mathrm{4}}}  \vdash_\mathcal{C}  \SCnt{Z}}
              \end{array}
            }
          }{\Psi_{{\mathrm{1}}}  \SCsym{,}  \SCnt{Y_{{\mathrm{1}}}}  \multimap  \SCnt{Y_{{\mathrm{2}}}}  \SCsym{,}  \Psi_{{\mathrm{2}}}  \SCsym{,}  \Phi  \SCsym{,}  \Psi_{{\mathrm{3}}}  \SCsym{,}  \Psi_{{\mathrm{4}}}  \vdash_\mathcal{C}  \SCnt{Z}}
        \end{math}
      \end{center}

\item Case 2:
      \begin{center}
        \scriptsize
        \begin{math}
          \begin{array}{c}
            \Pi_1 \\
            {\Phi  \vdash_\mathcal{C}  \SCnt{X}}
          \end{array}
        \end{math}
        \qquad\qquad
        $\Pi_2$:
        \begin{math}
          $$\mprset{flushleft}
          \inferrule* [right={\tiny impL}] {
            {
              \begin{array}{cc}
                \pi_1 & \pi_2 \\
                {\Psi_{{\mathrm{3}}}  \vdash_\mathcal{C}  \SCnt{Y_{{\mathrm{1}}}}} & {\Psi_{{\mathrm{1}}}  \SCsym{,}  \SCnt{X}  \SCsym{,}  \Psi_{{\mathrm{2}}}  \SCsym{,}  \SCnt{Y_{{\mathrm{2}}}}  \SCsym{,}  \Psi_{{\mathrm{4}}}  \vdash_\mathcal{C}  \SCnt{Z}}
              \end{array}
            }
          }{\Psi_{{\mathrm{1}}}  \SCsym{,}  \SCnt{X}  \SCsym{,}  \Psi_{{\mathrm{2}}}  \SCsym{,}  \SCnt{Y_{{\mathrm{1}}}}  \multimap  \SCnt{Y_{{\mathrm{2}}}}  \SCsym{,}  \Psi_{{\mathrm{3}}}  \SCsym{,}  \Psi_{{\mathrm{4}}}  \vdash_\mathcal{C}  \SCnt{Z}}
        \end{math}
      \end{center}
      By assumption, $c(\Pi_1),c(\Pi_2)\leq |X|$. By induction on $\Pi_1$ and $\pi_2$, there is
      a proof $\Pi'$ for sequent $\Psi_{{\mathrm{1}}}  \SCsym{,}  \Phi  \SCsym{,}  \Psi_{{\mathrm{2}}}  \SCsym{,}  \SCnt{Y_{{\mathrm{2}}}}  \SCsym{,}  \Psi_{{\mathrm{4}}}  \vdash_\mathcal{C}  \SCnt{Z}$ s.t. $c(\Pi') \leq |X|$.
      Therefore, the proof $\Pi$ can be constructed as follows with
      $c(\Pi) = c(\Pi') \leq |X|$.
      \begin{center}
        \scriptsize
        \begin{math}
          $$\mprset{flushleft}
          \inferrule* [right={\tiny impL}] {
            {
              \begin{array}{cc}
                \pi_1 & \Pi' \\
                {\Psi_{{\mathrm{3}}}  \vdash_\mathcal{C}  \SCnt{Y_{{\mathrm{1}}}}} & {\Psi_{{\mathrm{1}}}  \SCsym{,}  \Phi  \SCsym{,}  \Psi_{{\mathrm{2}}}  \SCsym{,}  \SCnt{Y_{{\mathrm{2}}}}  \SCsym{,}  \Psi_{{\mathrm{4}}}  \vdash_\mathcal{C}  \SCnt{Z}}
              \end{array}
            }
          }{\Psi_{{\mathrm{1}}}  \SCsym{,}  \Phi_{{\mathrm{1}}}  \SCsym{,}  \Psi_{{\mathrm{2}}}  \SCsym{,}  \SCnt{Y_{{\mathrm{1}}}}  \multimap  \SCnt{Y_{{\mathrm{2}}}}  \SCsym{,}  \Psi_{{\mathrm{3}}}  \SCsym{,}  \Psi_{{\mathrm{4}}}  \vdash_\mathcal{C}  \SCnt{Z}}
        \end{math}
      \end{center}

\item Case 3:
      \begin{center}
        \scriptsize
        \begin{math}
          \begin{array}{c}
            \Pi_1 \\
            {\Phi  \vdash_\mathcal{C}  \SCnt{X}}
          \end{array}
        \end{math}
        \qquad\qquad
        $\Pi_2$:
        \begin{math}
          $$\mprset{flushleft}
          \inferrule* [right={\tiny impL}] {
            {
              \begin{array}{cc}
                \pi_1 & \pi_2 \\
                {\Psi_{{\mathrm{2}}}  \vdash_\mathcal{C}  \SCnt{Y_{{\mathrm{1}}}}} & {\Psi_{{\mathrm{1}}}  \SCsym{,}  \SCnt{Y_{{\mathrm{2}}}}  \SCsym{,}  \Psi_{{\mathrm{3}}}  \SCsym{,}  \SCnt{X}  \SCsym{,}  \Psi_{{\mathrm{4}}}  \vdash_\mathcal{C}  \SCnt{Z}}
              \end{array}
            }
          }{\Psi_{{\mathrm{1}}}  \SCsym{,}  \SCnt{Y_{{\mathrm{1}}}}  \multimap  \SCnt{Y_{{\mathrm{2}}}}  \SCsym{,}  \Psi_{{\mathrm{2}}}  \SCsym{,}  \Psi_{{\mathrm{3}}}  \SCsym{,}  \SCnt{X}  \SCsym{,}  \Psi_{{\mathrm{4}}}  \vdash_\mathcal{C}  \SCnt{Z}}
        \end{math}
      \end{center}
      By assumption, $c(\Pi_1),c(\Pi_2)\leq |X|$. By induction on $\Pi_1$
      and $\pi_2$, there is a proof $\Pi'$ for sequent
      $\Psi_{{\mathrm{1}}}  \SCsym{,}  \Phi  \SCsym{,}  \Psi_{{\mathrm{2}}}  \SCsym{,}  \SCnt{Y_{{\mathrm{2}}}}  \SCsym{,}  \Psi_{{\mathrm{4}}}  \vdash_\mathcal{C}  \SCnt{Z}$ s.t. $c(\Pi') \leq |X|$. Therefore,
      the proof $\Pi$ can be constructed as follows with
      $c(\Pi) = c(\Pi') \leq |X|$.
      \begin{center}
        \scriptsize
        \begin{math}
          $$\mprset{flushleft}
          \inferrule* [right={\tiny impL}] {
            {
              \begin{array}{cc}
                \pi_1 & \Pi' \\
                {\Psi_{{\mathrm{2}}}  \vdash_\mathcal{C}  \SCnt{Y_{{\mathrm{1}}}}} & {\Psi_{{\mathrm{1}}}  \SCsym{,}  \SCnt{Y_{{\mathrm{2}}}}  \SCsym{,}  \Psi_{{\mathrm{3}}}  \SCsym{,}  \Phi  \SCsym{,}  \Psi_{{\mathrm{4}}}  \vdash_\mathcal{C}  \SCnt{Z}}
              \end{array}
            }
          }{\Psi_{{\mathrm{1}}}  \SCsym{,}  \SCnt{Y_{{\mathrm{1}}}}  \multimap  \SCnt{Y_{{\mathrm{2}}}}  \SCsym{,}  \Psi_{{\mathrm{2}}}  \SCsym{,}  \Psi_{{\mathrm{3}}}  \SCsym{,}  \Phi  \SCsym{,}  \Psi_{{\mathrm{4}}}  \vdash_\mathcal{C}  \SCnt{Z}}
        \end{math}
      \end{center}

\item Case 4:
      \begin{center}
        \scriptsize
        \begin{math}
          \begin{array}{c}
            \Pi_1 \\
            {\Phi  \vdash_\mathcal{C}  \SCnt{X}}
          \end{array}
        \end{math}
        \qquad\qquad
        $\Pi_2$:
        \begin{math}
          $$\mprset{flushleft}
          \inferrule* [right={\tiny impL}] {
            {
              \begin{array}{cc}
                \pi_1 & \pi_2 \\
                {\Psi_{{\mathrm{1}}}  \SCsym{,}  \SCnt{X}  \SCsym{,}  \Psi_{{\mathrm{2}}}  \vdash_\mathcal{C}  \SCnt{Y_{{\mathrm{1}}}}} & {\Gamma_{{\mathrm{1}}}  \SCsym{;}  \SCnt{Y_{{\mathrm{2}}}}  \SCsym{;}  \Gamma_{{\mathrm{2}}}  \vdash_\mathcal{L}  \SCnt{A}}
              \end{array}
            }
          }{\Gamma_{{\mathrm{1}}}  \SCsym{;}  \SCnt{Y_{{\mathrm{1}}}}  \multimap  \SCnt{Y_{{\mathrm{2}}}}  \SCsym{;}  \Psi_{{\mathrm{1}}}  \SCsym{;}  \SCnt{X}  \SCsym{;}  \Psi_{{\mathrm{2}}}  \SCsym{;}  \Gamma_{{\mathrm{2}}}  \vdash_\mathcal{L}  \SCnt{A}}
        \end{math}
      \end{center}
      By assumption, $c(\Pi_1),c(\Pi_2)\leq |X|$. By induction on $\Pi_1$
      and $\pi_1$, there is a proof $\Pi'$ for sequent
      $\Psi_{{\mathrm{1}}}  \SCsym{,}  \Phi  \SCsym{,}  \Psi_{{\mathrm{2}}}  \vdash_\mathcal{C}  \SCnt{Y_{{\mathrm{1}}}}$ s.t. $c(\Pi') \leq |X|$. Therefore, the proof
      $\Pi$ can be constructed as follows with $c(\Pi) = c(\Pi') \leq |X|$.
      \begin{center}
        \scriptsize
        \begin{math}
          $$\mprset{flushleft}
          \inferrule* [right={\tiny impL}] {
            {
              \begin{array}{cc}
                \Pi' & \pi_2 \\
                {\Psi_{{\mathrm{1}}}  \SCsym{,}  \Phi  \SCsym{,}  \Psi_{{\mathrm{2}}}  \vdash_\mathcal{C}  \SCnt{Y_{{\mathrm{1}}}}} & {\Gamma_{{\mathrm{1}}}  \SCsym{;}  \SCnt{Y_{{\mathrm{2}}}}  \SCsym{;}  \Gamma_{{\mathrm{2}}}  \vdash_\mathcal{L}  \SCnt{A}}
              \end{array}
            }
          }{\Gamma_{{\mathrm{1}}}  \SCsym{;}  \SCnt{Y_{{\mathrm{1}}}}  \multimap  \SCnt{Y_{{\mathrm{2}}}}  \SCsym{;}  \Psi_{{\mathrm{1}}}  \SCsym{;}  \Phi  \SCsym{;}  \Psi_{{\mathrm{2}}}  \SCsym{;}  \Gamma_{{\mathrm{2}}}  \vdash_\mathcal{L}  \SCnt{A}}
        \end{math}
      \end{center}

\item Case 5:
      \begin{center}
        \scriptsize
        \begin{math}
          \begin{array}{c}
            \Pi_1 \\
            {\Phi  \vdash_\mathcal{C}  \SCnt{X}}
          \end{array}
        \end{math}
        \qquad\qquad
        $\Pi_2$:
        \begin{math}
          $$\mprset{flushleft}
          \inferrule* [right={\tiny impL}] {
            {
              \begin{array}{cc}
                \pi_1 & \pi_2 \\
                {\Psi  \vdash_\mathcal{C}  \SCnt{Y_{{\mathrm{1}}}}} & {\Gamma_{{\mathrm{1}}}  \SCsym{;}  \SCnt{X}  \SCsym{;}  \Gamma_{{\mathrm{2}}}  \SCsym{;}  \SCnt{Y_{{\mathrm{2}}}}  \SCsym{;}  \Gamma_{{\mathrm{3}}}  \vdash_\mathcal{L}  \SCnt{A}}
              \end{array}
            }
          }{\Gamma_{{\mathrm{1}}}  \SCsym{;}  \SCnt{X}  \SCsym{;}  \Gamma_{{\mathrm{2}}}  \SCsym{;}  \SCnt{Y_{{\mathrm{1}}}}  \multimap  \SCnt{Y_{{\mathrm{2}}}}  \SCsym{;}  \Psi  \SCsym{;}  \Gamma_{{\mathrm{3}}}  \vdash_\mathcal{L}  \SCnt{A}}
        \end{math}
      \end{center}
      By assumption, $c(\Pi_1),c(\Pi_2)\leq |X|$. By induction on $\Pi_1$
      and $\pi_2$, there is a proof $\Pi'$ for sequent
      $\Gamma_{{\mathrm{1}}}  \SCsym{;}  \Phi  \SCsym{;}  \Gamma_{{\mathrm{2}}}  \SCsym{;}  \SCnt{Y_{{\mathrm{2}}}}  \SCsym{;}  \Gamma_{{\mathrm{3}}}  \vdash_\mathcal{L}  \SCnt{A}$ s.t. $c(\Pi') \leq |X|$. Therefore, the
      proof $\Pi$ can be constructed as follows with
      $c(\Pi) = c(\Pi') \leq |X|$.
      \begin{center}
        \scriptsize
        \begin{math}
          $$\mprset{flushleft}
          \inferrule* [right={\tiny impL}] {
            {
              \begin{array}{cc}
                \pi_1 & \Pi' \\
                {\Psi  \vdash_\mathcal{C}  \SCnt{Y_{{\mathrm{1}}}}} & {\Gamma_{{\mathrm{1}}}  \SCsym{;}  \Phi  \SCsym{;}  \Gamma_{{\mathrm{2}}}  \SCsym{;}  \SCnt{Y_{{\mathrm{2}}}}  \SCsym{;}  \Gamma_{{\mathrm{3}}}  \vdash_\mathcal{L}  \SCnt{A}}
              \end{array}
            }
          }{\Gamma_{{\mathrm{1}}}  \SCsym{;}  \Phi  \SCsym{;}  \Gamma_{{\mathrm{2}}}  \SCsym{;}  \SCnt{Y_{{\mathrm{1}}}}  \multimap  \SCnt{Y_{{\mathrm{2}}}}  \SCsym{;}  \Psi  \SCsym{;}  \Gamma_{{\mathrm{3}}}  \vdash_\mathcal{L}  \SCnt{A}}
        \end{math}
      \end{center}

\item Case 6:
      \begin{center}
        \scriptsize
        \begin{math}
          \begin{array}{c}
            \Pi_1 \\
            {\Delta  \vdash_\mathcal{L}  \SCnt{B}}
          \end{array}
        \end{math}
        \qquad\qquad
        $\Pi_2$:
        \begin{math}
          $$\mprset{flushleft}
          \inferrule* [right={\tiny impL}] {
            {
              \begin{array}{cc}
                \pi_1 & \pi_2 \\
                {\Psi  \vdash_\mathcal{C}  \SCnt{Y_{{\mathrm{1}}}}} & {\Gamma_{{\mathrm{1}}}  \SCsym{;}  \SCnt{B}  \SCsym{;}  \Gamma_{{\mathrm{2}}}  \SCsym{;}  \SCnt{Y_{{\mathrm{2}}}}  \SCsym{;}  \Gamma_{{\mathrm{3}}}  \vdash_\mathcal{L}  \SCnt{A}}
              \end{array}
            }
          }{\Gamma_{{\mathrm{1}}}  \SCsym{;}  \SCnt{B}  \SCsym{;}  \Gamma_{{\mathrm{2}}}  \SCsym{;}  \SCnt{Y_{{\mathrm{1}}}}  \multimap  \SCnt{Y_{{\mathrm{2}}}}  \SCsym{;}  \Psi  \SCsym{;}  \Gamma_{{\mathrm{3}}}  \vdash_\mathcal{L}  \SCnt{A}}
        \end{math}
      \end{center}
      By assumption, $c(\Pi_1),c(\Pi_2)\leq |B|$. By induction on $\Pi_1$
      and $\pi_2$, there is a proof $\Pi'$ for sequent
      $\Gamma_{{\mathrm{1}}}  \SCsym{;}  \Delta  \SCsym{;}  \Gamma_{{\mathrm{2}}}  \SCsym{;}  \SCnt{Y_{{\mathrm{2}}}}  \SCsym{;}  \Gamma_{{\mathrm{3}}}  \vdash_\mathcal{L}  \SCnt{A}$ s.t. $c(\Pi') \leq |B|$. Therefore, the
      proof $\Pi$ can be constructed as follows with
      $c(\Pi) = c(\Pi') \leq |B|$.
      \begin{center}
        \scriptsize
        \begin{math}
          $$\mprset{flushleft}
          \inferrule* [right={\tiny impL}] {
            {
              \begin{array}{cc}
                \pi_1 & \Pi' \\
                {\Psi  \vdash_\mathcal{C}  \SCnt{Y_{{\mathrm{1}}}}} & {\Gamma_{{\mathrm{1}}}  \SCsym{;}  \Delta  \SCsym{;}  \Gamma_{{\mathrm{2}}}  \SCsym{;}  \SCnt{Y_{{\mathrm{2}}}}  \SCsym{;}  \Gamma_{{\mathrm{3}}}  \vdash_\mathcal{L}  \SCnt{A}}
              \end{array}
            }
          }{\Gamma_{{\mathrm{1}}}  \SCsym{;}  \Delta  \SCsym{;}  \Gamma_{{\mathrm{2}}}  \SCsym{;}  \SCnt{Y_{{\mathrm{1}}}}  \multimap  \SCnt{Y_{{\mathrm{2}}}}  \SCsym{;}  \Psi  \SCsym{;}  \Gamma_{{\mathrm{3}}}  \vdash_\mathcal{L}  \SCnt{A}}
        \end{math}
      \end{center}

\item Case 7:
      \begin{center}
        \scriptsize
        \begin{math}
          \begin{array}{c}
            \Pi_1 \\
            {\Phi  \vdash_\mathcal{C}  \SCnt{X}}
          \end{array}
        \end{math}
        \qquad\qquad
        $\Pi_2$:
        \begin{math}
          $$\mprset{flushleft}
          \inferrule* [right={\tiny impL}] {
            {
              \begin{array}{cc}
                \pi_1 & \pi_2 \\
                {\Psi  \vdash_\mathcal{C}  \SCnt{Y_{{\mathrm{1}}}}} & {\Gamma_{{\mathrm{1}}}  \SCsym{;}  \SCnt{Y_{{\mathrm{2}}}}  \SCsym{;}  \Gamma_{{\mathrm{2}}}  \SCsym{;}  \SCnt{X}  \SCsym{;}  \Gamma_{{\mathrm{3}}}  \vdash_\mathcal{L}  \SCnt{A}}
              \end{array}
            }
          }{\Gamma_{{\mathrm{1}}}  \SCsym{;}  \SCnt{Y_{{\mathrm{1}}}}  \multimap  \SCnt{Y_{{\mathrm{2}}}}  \SCsym{;}  \Psi  \SCsym{;}  \Gamma_{{\mathrm{2}}}  \SCsym{;}  \SCnt{X}  \SCsym{;}  \Gamma_{{\mathrm{3}}}  \vdash_\mathcal{L}  \SCnt{A}}
        \end{math}
      \end{center}
      By assumption, $c(\Pi_1),c(\Pi_2)\leq |X|$. By induction on $\Pi_1$
      and $\pi_2$, there is a proof $\Pi'$ for sequent
      $\Gamma_{{\mathrm{1}}}  \SCsym{;}  \SCnt{Y_{{\mathrm{2}}}}  \SCsym{;}  \Gamma_{{\mathrm{2}}}  \SCsym{;}  \Phi  \SCsym{;}  \Gamma_{{\mathrm{3}}}  \vdash_\mathcal{L}  \SCnt{A}$ s.t. $c(\Pi') \leq |X|$. Therefore, the
      proof $\Pi$ can be constructed as follows with
      $c(\Pi) = c(\Pi') \leq |X|$.
      \begin{center}
        \scriptsize
        \begin{math}
          $$\mprset{flushleft}
          \inferrule* [right={\tiny impL}] {
            {
              \begin{array}{cc}
                \pi_1 & \Pi' \\
                {\Psi  \vdash_\mathcal{C}  \SCnt{Y_{{\mathrm{1}}}}} & {\Gamma_{{\mathrm{1}}}  \SCsym{;}  \SCnt{Y_{{\mathrm{2}}}}  \SCsym{;}  \Gamma_{{\mathrm{2}}}  \SCsym{;}  \Phi  \SCsym{;}  \Gamma_{{\mathrm{3}}}  \vdash_\mathcal{L}  \SCnt{A}}
              \end{array}
            }
          }{\Gamma_{{\mathrm{1}}}  \SCsym{;}  \SCnt{Y_{{\mathrm{1}}}}  \multimap  \SCnt{Y_{{\mathrm{2}}}}  \SCsym{;}  \Psi  \SCsym{;}  \Gamma_{{\mathrm{2}}}  \SCsym{;}  \Phi  \SCsym{;}  \Gamma_{{\mathrm{3}}}  \vdash_\mathcal{L}  \SCnt{A}}
        \end{math}
      \end{center}

\item Case 8:
      \begin{center}
        \scriptsize
        \begin{math}
          \begin{array}{c}
            \Pi_1 \\
            {\Delta  \vdash_\mathcal{L}  \SCnt{B}}
          \end{array}
        \end{math}
        \qquad\qquad
        $\Pi_2$:
        \begin{math}
          $$\mprset{flushleft}
          \inferrule* [right={\tiny impL}] {
            {
              \begin{array}{cc}
                \pi_1 & \pi_2 \\
                {\Psi  \vdash_\mathcal{C}  \SCnt{Y_{{\mathrm{1}}}}} & {\Gamma_{{\mathrm{1}}}  \SCsym{;}  \SCnt{Y_{{\mathrm{2}}}}  \SCsym{;}  \Gamma_{{\mathrm{2}}}  \SCsym{;}  \SCnt{B}  \SCsym{;}  \Gamma_{{\mathrm{3}}}  \vdash_\mathcal{L}  \SCnt{A}}
              \end{array}
            }
          }{\Gamma_{{\mathrm{1}}}  \SCsym{;}  \SCnt{Y_{{\mathrm{1}}}}  \multimap  \SCnt{Y_{{\mathrm{2}}}}  \SCsym{;}  \Psi  \SCsym{;}  \Gamma_{{\mathrm{2}}}  \SCsym{;}  \SCnt{B}  \SCsym{;}  \Gamma_{{\mathrm{3}}}  \vdash_\mathcal{L}  \SCnt{A}}
        \end{math}
      \end{center}
      By assumption, $c(\Pi_1),c(\Pi_2)\leq |B|$. By induction on $\Pi_1$
      and $\pi_2$, there is a proof $\Pi'$ for sequent
      $\Gamma_{{\mathrm{1}}}  \SCsym{;}  \SCnt{Y_{{\mathrm{2}}}}  \SCsym{;}  \Gamma_{{\mathrm{2}}}  \SCsym{;}  \Delta  \SCsym{;}  \Gamma_{{\mathrm{3}}}  \vdash_\mathcal{L}  \SCnt{A}$ s.t. $c(\Pi') \leq |B|$. Therefore,
      the proof $\Pi$ can be constructed as follows with
      $c(\Pi) = c(\Pi') \leq |B|$.
      \begin{center}
        \scriptsize
        \begin{math}
          $$\mprset{flushleft}
          \inferrule* [right={\tiny impL}] {
            {
              \begin{array}{cc}
                \pi_1 & \Pi' \\
                {\Psi  \vdash_\mathcal{C}  \SCnt{Y_{{\mathrm{1}}}}} & {\Gamma_{{\mathrm{1}}}  \SCsym{;}  \SCnt{Y_{{\mathrm{2}}}}  \SCsym{;}  \Gamma_{{\mathrm{2}}}  \SCsym{;}  \Delta  \SCsym{;}  \Gamma_{{\mathrm{3}}}  \vdash_\mathcal{L}  \SCnt{A}}
              \end{array}
            }
          }{\Gamma_{{\mathrm{1}}}  \SCsym{;}  \SCnt{Y_{{\mathrm{1}}}}  \multimap  \SCnt{Y_{{\mathrm{2}}}}  \SCsym{;}  \Psi  \SCsym{;}  \Gamma_{{\mathrm{2}}}  \SCsym{;}  \Delta  \SCsym{;}  \Gamma_{{\mathrm{3}}}  \vdash_\mathcal{L}  \SCnt{A}}
        \end{math}
      \end{center}
\end{itemize}

\subsubsection{Left introduction of the non-commutative left implication $\lto$}
\begin{itemize}
\item Case 1:
      \begin{center}
        \scriptsize
        \begin{math}
          \begin{array}{c}
            \Pi_1 \\
            {\Phi  \vdash_\mathcal{C}  \SCnt{X}}
          \end{array}
        \end{math}
        \qquad\qquad
        $\Pi_2$:
        \begin{math}
          $$\mprset{flushleft}
          \inferrule* [right={\tiny imprL}] {
            {
              \begin{array}{cc}
                \pi_1 & \pi_2 \\
                {\Delta_{{\mathrm{1}}}  \SCsym{;}  \SCnt{X}  \SCsym{;}  \Delta_{{\mathrm{2}}}  \vdash_\mathcal{L}  \SCnt{A_{{\mathrm{1}}}}} & {\Gamma_{{\mathrm{1}}}  \SCsym{;}  \SCnt{A_{{\mathrm{2}}}}  \SCsym{;}  \Gamma_{{\mathrm{2}}}  \vdash_\mathcal{L}  \SCnt{B}}
              \end{array}
            }
          }{\Gamma_{{\mathrm{1}}}  \SCsym{;}  \SCnt{A_{{\mathrm{1}}}}  \rightharpoonup  \SCnt{A_{{\mathrm{2}}}}  \SCsym{;}  \Delta_{{\mathrm{1}}}  \SCsym{;}  \SCnt{X}  \SCsym{;}  \Delta_{{\mathrm{2}}}  \SCsym{;}  \Gamma_{{\mathrm{2}}}  \vdash_\mathcal{L}  \SCnt{B}}
        \end{math}
      \end{center}
      By assumption, $c(\Pi_1),c(\Pi_2)\leq |X|$. By induction on $\Pi_1$
      and $\pi_1$, there is a proof $\Pi'$ for sequent
      $\Delta_{{\mathrm{1}}}  \SCsym{;}  \Phi  \SCsym{;}  \Delta_{{\mathrm{2}}}  \vdash_\mathcal{L}  \SCnt{A_{{\mathrm{1}}}}$ s.t. $c(\Pi') \leq |X|$. Therefore, the proof
      $\Pi$ can be constructed as follows with $c(\Pi) = c(\Pi') \leq |X|$.
      \begin{center}
        \scriptsize
        \begin{math}
          $$\mprset{flushleft}
          \inferrule* [right={\tiny impL}] {
            {
              \begin{array}{cc}
                \Pi' & \pi_2 \\
                {\Delta_{{\mathrm{1}}}  \SCsym{;}  \Phi  \SCsym{;}  \Delta_{{\mathrm{2}}}  \vdash_\mathcal{L}  \SCnt{A_{{\mathrm{1}}}}} & {\Gamma_{{\mathrm{1}}}  \SCsym{;}  \SCnt{A_{{\mathrm{2}}}}  \SCsym{;}  \Gamma_{{\mathrm{2}}}  \vdash_\mathcal{L}  \SCnt{B}}
              \end{array}
            }
          }{\Gamma_{{\mathrm{1}}}  \SCsym{;}  \SCnt{A_{{\mathrm{1}}}}  \rightharpoonup  \SCnt{A_{{\mathrm{2}}}}  \SCsym{;}  \Delta_{{\mathrm{1}}}  \SCsym{;}  \Phi  \SCsym{;}  \Delta_{{\mathrm{2}}}  \SCsym{;}  \Gamma_{{\mathrm{2}}}  \vdash_\mathcal{L}  \SCnt{B}}
        \end{math}
      \end{center}

\item Case 2:
      \begin{center}
        \scriptsize
        \begin{math}
          \begin{array}{c}
            \Pi_1 \\
            {\Gamma  \vdash_\mathcal{L}  \SCnt{C}}
          \end{array}
        \end{math}
        \qquad\qquad
        $\Pi_2$:
        \begin{math}
          $$\mprset{flushleft}
          \inferrule* [right={\tiny imprL}] {
            {
              \begin{array}{cc}
                \pi_1 & \pi_2 \\
                {\Delta_{{\mathrm{1}}}  \SCsym{;}  \SCnt{C}  \SCsym{;}  \Delta_{{\mathrm{2}}}  \vdash_\mathcal{L}  \SCnt{A_{{\mathrm{1}}}}} & {\Gamma_{{\mathrm{1}}}  \SCsym{;}  \SCnt{A_{{\mathrm{2}}}}  \SCsym{;}  \Gamma_{{\mathrm{2}}}  \vdash_\mathcal{L}  \SCnt{B}}
              \end{array}
            }
          }{\Gamma_{{\mathrm{1}}}  \SCsym{;}  \SCnt{A_{{\mathrm{1}}}}  \rightharpoonup  \SCnt{A_{{\mathrm{2}}}}  \SCsym{;}  \Delta_{{\mathrm{1}}}  \SCsym{;}  \SCnt{C}  \SCsym{;}  \Delta_{{\mathrm{2}}}  \SCsym{;}  \Gamma_{{\mathrm{2}}}  \vdash_\mathcal{L}  \SCnt{B}}
        \end{math}
      \end{center}
      By assumption, $c(\Pi_1),c(\Pi_2)\leq |C|$. By induction on $\Pi_1$
      and $\pi_1$, there is a proof $\Pi'$ for sequent
      $\Delta_{{\mathrm{1}}}  \SCsym{;}  \Gamma  \SCsym{;}  \Delta_{{\mathrm{2}}}  \vdash_\mathcal{L}  \SCnt{A_{{\mathrm{1}}}}$ s.t. $c(\Pi') \leq |C|$. Therefore, the proof
      $\Pi$ can be constructed as follows with $c(\Pi) = c(\Pi') \leq |C|$.
      \begin{center}
        \scriptsize
        \begin{math}
          $$\mprset{flushleft}
          \inferrule* [right={\tiny imprL}] {
            {
              \begin{array}{cc}
                \Pi' & \pi_2 \\
                {\Delta_{{\mathrm{1}}}  \SCsym{;}  \Gamma  \SCsym{;}  \Delta_{{\mathrm{2}}}  \vdash_\mathcal{L}  \SCnt{A_{{\mathrm{1}}}}} & {\Gamma_{{\mathrm{1}}}  \SCsym{;}  \SCnt{A_{{\mathrm{2}}}}  \SCsym{;}  \Gamma_{{\mathrm{2}}}  \vdash_\mathcal{L}  \SCnt{B}}
              \end{array}
            }
          }{\Gamma_{{\mathrm{1}}}  \SCsym{;}  \SCnt{A_{{\mathrm{1}}}}  \rightharpoonup  \SCnt{A_{{\mathrm{2}}}}  \SCsym{;}  \Delta_{{\mathrm{1}}}  \SCsym{;}  \Gamma  \SCsym{;}  \Delta_{{\mathrm{2}}}  \SCsym{;}  \Gamma_{{\mathrm{2}}}  \vdash_\mathcal{L}  \SCnt{B}}
        \end{math}
      \end{center}

\item Case 3:
      \begin{center}
        \scriptsize
        \begin{math}
          \begin{array}{c}
            \Pi_1 \\
            {\Phi  \vdash_\mathcal{C}  \SCnt{X}}
          \end{array}
        \end{math}
        \qquad\qquad
        $\Pi_2$:
        \begin{math}
          $$\mprset{flushleft}
          \inferrule* [right={\tiny imprL}] {
            {
              \begin{array}{cc}
                \pi_1 & \pi_2 \\
                {\Delta  \vdash_\mathcal{L}  \SCnt{A_{{\mathrm{1}}}}} & {\Gamma_{{\mathrm{1}}}  \SCsym{;}  \SCnt{X}  \SCsym{;}  \Gamma_{{\mathrm{2}}}  \SCsym{;}  \SCnt{A_{{\mathrm{2}}}}  \SCsym{;}  \Gamma_{{\mathrm{3}}}  \vdash_\mathcal{L}  \SCnt{B}}
              \end{array}
            }
          }{\Gamma_{{\mathrm{1}}}  \SCsym{;}  \SCnt{X}  \SCsym{;}  \Gamma_{{\mathrm{2}}}  \SCsym{;}  \SCnt{A_{{\mathrm{1}}}}  \rightharpoonup  \SCnt{A_{{\mathrm{2}}}}  \SCsym{;}  \Delta  \SCsym{;}  \Gamma_{{\mathrm{3}}}  \vdash_\mathcal{L}  \SCnt{B}}
        \end{math}
      \end{center}
      By assumption, $c(\Pi_1),c(\Pi_2)\leq |X|$. By induction on $\Pi_1$
      and $\pi_2$, there is a proof $\Pi'$ for sequent
      $\Gamma_{{\mathrm{1}}}  \SCsym{;}  \Phi  \SCsym{;}  \Gamma_{{\mathrm{2}}}  \SCsym{;}  \SCnt{A_{{\mathrm{2}}}}  \SCsym{;}  \Gamma_{{\mathrm{3}}}  \vdash_\mathcal{L}  \SCnt{B}$ s.t. $c(\Pi') \leq |X|$. Therefore,
      the proof $\Pi$ can be constructed as follows with
      $c(\Pi) = c(\Pi') \leq |X|$.
      \begin{center}
        \scriptsize
        \begin{math}
          $$\mprset{flushleft}
          \inferrule* [right={\tiny imprL}] {
            {
              \begin{array}{cc}
                \pi_1 & \Pi' \\
                {\Delta  \vdash_\mathcal{L}  \SCnt{A_{{\mathrm{1}}}}} & {\Gamma_{{\mathrm{1}}}  \SCsym{;}  \Phi  \SCsym{;}  \Gamma_{{\mathrm{2}}}  \SCsym{;}  \SCnt{A_{{\mathrm{2}}}}  \SCsym{;}  \Gamma_{{\mathrm{3}}}  \vdash_\mathcal{L}  \SCnt{B}}
              \end{array}
            }
          }{\Gamma_{{\mathrm{1}}}  \SCsym{;}  \Phi  \SCsym{;}  \Gamma_{{\mathrm{2}}}  \SCsym{;}  \SCnt{A_{{\mathrm{1}}}}  \rightharpoonup  \SCnt{A_{{\mathrm{2}}}}  \SCsym{;}  \Delta  \SCsym{;}  \Gamma_{{\mathrm{3}}}  \vdash_\mathcal{L}  \SCnt{B}}
        \end{math}
      \end{center}

\item Case 4:
      \begin{center}
        \scriptsize
        \begin{math}
          \begin{array}{c}
            \Pi_1 \\
            {\Delta_{{\mathrm{1}}}  \vdash_\mathcal{L}  \SCnt{B}}
          \end{array}
        \end{math}
        \qquad\qquad
        $\Pi_2$:
        \begin{math}
          $$\mprset{flushleft}
          \inferrule* [right={\tiny imprL}] {
            {
              \begin{array}{cc}
                \pi_1 & \pi_2 \\
                {\Delta_{{\mathrm{2}}}  \vdash_\mathcal{L}  \SCnt{A_{{\mathrm{1}}}}} & {\Gamma_{{\mathrm{1}}}  \SCsym{;}  \SCnt{B}  \SCsym{;}  \Gamma_{{\mathrm{2}}}  \SCsym{;}  \SCnt{A_{{\mathrm{2}}}}  \SCsym{;}  \Gamma_{{\mathrm{3}}}  \vdash_\mathcal{L}  \SCnt{C}}
              \end{array}
            }
          }{\Gamma_{{\mathrm{1}}}  \SCsym{;}  \SCnt{B}  \SCsym{;}  \Gamma_{{\mathrm{2}}}  \SCsym{;}  \SCnt{A_{{\mathrm{1}}}}  \rightharpoonup  \SCnt{A_{{\mathrm{2}}}}  \SCsym{;}  \Delta_{{\mathrm{2}}}  \SCsym{;}  \Gamma_{{\mathrm{3}}}  \vdash_\mathcal{L}  \SCnt{C}}
        \end{math}
      \end{center}
      By assumption, $c(\Pi_1),c(\Pi_2)\leq |B|$. By induction on $\Pi_1$
      and $\pi_2$, there is a proof $\Pi'$ for sequent
      $\Gamma_{{\mathrm{1}}}  \SCsym{;}  \Delta_{{\mathrm{1}}}  \SCsym{;}  \Gamma_{{\mathrm{2}}}  \SCsym{;}  \SCnt{A_{{\mathrm{2}}}}  \SCsym{;}  \Gamma_{{\mathrm{3}}}  \vdash_\mathcal{L}  \SCnt{C}$ s.t. $c(\Pi') \leq |B|$. Therefore,
      the proof $\Pi$ can be constructed as follows with
      $c(\Pi) = c(\Pi') \leq |B|$.
      \begin{center}
        \scriptsize
        \begin{math}
          $$\mprset{flushleft}
          \inferrule* [right={\tiny imprL}] {
            {
              \begin{array}{cc}
                \pi_1 & \Pi' \\
                {\Delta_{{\mathrm{2}}}  \vdash_\mathcal{L}  \SCnt{A_{{\mathrm{1}}}}} & {\Gamma_{{\mathrm{1}}}  \SCsym{;}  \Delta_{{\mathrm{1}}}  \SCsym{;}  \Gamma_{{\mathrm{2}}}  \SCsym{;}  \SCnt{A_{{\mathrm{2}}}}  \SCsym{;}  \Gamma_{{\mathrm{3}}}  \vdash_\mathcal{L}  \SCnt{C}}
              \end{array}
            }
          }{\Gamma_{{\mathrm{1}}}  \SCsym{;}  \Delta_{{\mathrm{1}}}  \SCsym{;}  \Gamma_{{\mathrm{2}}}  \SCsym{;}  \SCnt{A_{{\mathrm{1}}}}  \rightharpoonup  \SCnt{A_{{\mathrm{2}}}}  \SCsym{;}  \Delta_{{\mathrm{2}}}  \SCsym{;}  \Gamma_{{\mathrm{3}}}  \vdash_\mathcal{L}  \SCnt{C}}
        \end{math}
      \end{center}

\item Case 5:
      \begin{center}
        \scriptsize
        \begin{math}
          \begin{array}{c}
            \Pi_1 \\
            {\Phi  \vdash_\mathcal{C}  \SCnt{X}}
          \end{array}
        \end{math}
        \qquad\qquad
        $\Pi_2$:
        \begin{math}
          $$\mprset{flushleft}
          \inferrule* [right={\tiny imprL}] {
            {
              \begin{array}{cc}
                \pi_1 & \pi_2 \\
                {\Delta  \vdash_\mathcal{L}  \SCnt{A_{{\mathrm{1}}}}} & {\Gamma_{{\mathrm{1}}}  \SCsym{;}  \SCnt{A_{{\mathrm{2}}}}  \SCsym{;}  \Gamma_{{\mathrm{2}}}  \SCsym{;}  \SCnt{X}  \SCsym{;}  \Gamma_{{\mathrm{3}}}  \vdash_\mathcal{L}  \SCnt{B}}
              \end{array}
            }
          }{\Gamma_{{\mathrm{1}}}  \SCsym{;}  \SCnt{A_{{\mathrm{1}}}}  \rightharpoonup  \SCnt{A_{{\mathrm{2}}}}  \SCsym{;}  \Delta  \SCsym{;}  \Gamma_{{\mathrm{2}}}  \SCsym{;}  \SCnt{X}  \SCsym{;}  \Gamma_{{\mathrm{3}}}  \vdash_\mathcal{L}  \SCnt{B}}
        \end{math}
      \end{center}
      By assumption, $c(\Pi_1),c(\Pi_2)\leq |X|$. By induction on $\Pi_1$
      and $\pi_2$, there is a proof $\Pi'$ for sequent
      $\Gamma_{{\mathrm{1}}}  \SCsym{;}  \SCnt{A_{{\mathrm{2}}}}  \SCsym{;}  \Gamma_{{\mathrm{2}}}  \SCsym{;}  \Phi  \SCsym{;}  \Gamma_{{\mathrm{3}}}  \vdash_\mathcal{L}  \SCnt{B}$ s.t. $c(\Pi') \leq |X|$. Therefore, the
      proof $\Pi$ can be constructed as follows with
      $c(\Pi) = c(\Pi') \leq |X|$.
      \begin{center}
        \scriptsize
        \begin{math}
          $$\mprset{flushleft}
          \inferrule* [right={\tiny imprL}] {
            {
              \begin{array}{cc}
                \pi_1 & \Pi' \\
                {\Delta  \vdash_\mathcal{L}  \SCnt{A_{{\mathrm{1}}}}} & {\Gamma_{{\mathrm{1}}}  \SCsym{;}  \SCnt{A_{{\mathrm{2}}}}  \SCsym{;}  \Gamma_{{\mathrm{2}}}  \SCsym{;}  \Phi  \SCsym{;}  \Gamma_{{\mathrm{3}}}  \vdash_\mathcal{L}  \SCnt{B}}
              \end{array}
            }
          }{\Gamma_{{\mathrm{1}}}  \SCsym{;}  \SCnt{A_{{\mathrm{1}}}}  \rightharpoonup  \SCnt{A_{{\mathrm{2}}}}  \SCsym{;}  \Delta  \SCsym{;}  \Gamma_{{\mathrm{2}}}  \SCsym{;}  \Phi  \SCsym{;}  \Gamma_{{\mathrm{3}}}  \vdash_\mathcal{L}  \SCnt{B}}
        \end{math}
      \end{center}

\item Case 6:
      \begin{center}
        \scriptsize
        \begin{math}
          \begin{array}{c}
            \Pi_1 \\
            {\Delta_{{\mathrm{1}}}  \vdash_\mathcal{L}  \SCnt{B}}
          \end{array}
        \end{math}
        \qquad\qquad
        $\Pi_2$:
        \begin{math}
          $$\mprset{flushleft}
          \inferrule* [right={\tiny imprL}] {
            {
              \begin{array}{cc}
                \pi_1 & \pi_2 \\
                {\Delta_{{\mathrm{2}}}  \vdash_\mathcal{L}  \SCnt{A_{{\mathrm{1}}}}} & {\Gamma_{{\mathrm{1}}}  \SCsym{;}  \SCnt{A_{{\mathrm{2}}}}  \SCsym{;}  \Gamma_{{\mathrm{2}}}  \SCsym{;}  \SCnt{B}  \SCsym{;}  \Gamma_{{\mathrm{3}}}  \vdash_\mathcal{L}  \SCnt{C}}
              \end{array}
            }
          }{\Gamma_{{\mathrm{1}}}  \SCsym{;}  \SCnt{A_{{\mathrm{1}}}}  \rightharpoonup  \SCnt{A_{{\mathrm{2}}}}  \SCsym{;}  \Delta_{{\mathrm{2}}}  \SCsym{;}  \Gamma_{{\mathrm{2}}}  \SCsym{;}  \SCnt{B}  \SCsym{;}  \Gamma_{{\mathrm{3}}}  \vdash_\mathcal{L}  \SCnt{C}}
        \end{math}
      \end{center}
      By assumption, $c(\Pi_1),c(\Pi_2)\leq |B|$. By induction on $\Pi_1$
      and $\pi_2$, there is a proof $\Pi'$ for sequent
      $\Gamma_{{\mathrm{1}}}  \SCsym{;}  \SCnt{A_{{\mathrm{2}}}}  \SCsym{;}  \Gamma_{{\mathrm{2}}}  \SCsym{;}  \Delta_{{\mathrm{1}}}  \SCsym{;}  \Gamma_{{\mathrm{3}}}  \vdash_\mathcal{L}  \SCnt{C}$ s.t. $c(\Pi') \leq |B|$. Therefore,
      the proof $\Pi$ can be constructed as follows with
      $c(\Pi) = c(\Pi') \leq |B|$.
      \begin{center}
        \scriptsize
        \begin{math}
          $$\mprset{flushleft}
          \inferrule* [right={\tiny imprL}] {
            {
              \begin{array}{cc}
                \pi_1 & \Pi' \\
                {\Delta_{{\mathrm{2}}}  \vdash_\mathcal{L}  \SCnt{A_{{\mathrm{1}}}}} & {\Gamma_{{\mathrm{1}}}  \SCsym{;}  \SCnt{A_{{\mathrm{2}}}}  \SCsym{;}  \Gamma_{{\mathrm{2}}}  \SCsym{;}  \Delta_{{\mathrm{1}}}  \SCsym{;}  \Gamma_{{\mathrm{3}}}  \vdash_\mathcal{L}  \SCnt{C}}
              \end{array}
            }
          }{\Gamma_{{\mathrm{1}}}  \SCsym{;}  \SCnt{A_{{\mathrm{1}}}}  \rightharpoonup  \SCnt{A_{{\mathrm{2}}}}  \SCsym{;}  \Delta_{{\mathrm{2}}}  \SCsym{;}  \Gamma_{{\mathrm{2}}}  \SCsym{;}  \Delta_{{\mathrm{1}}}  \SCsym{;}  \Gamma_{{\mathrm{3}}}  \vdash_\mathcal{L}  \SCnt{C}}
        \end{math}
      \end{center}
\end{itemize}

\subsubsection{Left introduction of the non-commutative right implication $\rto$}
\begin{itemize}
\item Case 1:
      \begin{center}
        \scriptsize
        \begin{math}
          \begin{array}{c}
            \Pi_1 \\
            {\Phi  \vdash_\mathcal{C}  \SCnt{X}}
          \end{array}
        \end{math}
        \qquad\qquad
        $\Pi_2$:
        \begin{math}
          $$\mprset{flushleft}
          \inferrule* [right={\tiny implL}] {
            {
              \begin{array}{cc}
                \pi_1 & \pi_2 \\
                {\Delta_{{\mathrm{1}}}  \SCsym{;}  \SCnt{X}  \SCsym{;}  \Delta_{{\mathrm{2}}}  \vdash_\mathcal{L}  \SCnt{A_{{\mathrm{1}}}}} & {\Gamma_{{\mathrm{1}}}  \SCsym{;}  \SCnt{A_{{\mathrm{2}}}}  \SCsym{;}  \Gamma_{{\mathrm{2}}}  \vdash_\mathcal{L}  \SCnt{B}}
              \end{array}
            }
          }{\Gamma_{{\mathrm{1}}}  \SCsym{;}  \Delta_{{\mathrm{1}}}  \SCsym{;}  \SCnt{A_{{\mathrm{2}}}}  \leftharpoonup  \SCnt{A_{{\mathrm{1}}}}  \SCsym{;}  \SCnt{X}  \SCsym{;}  \Delta_{{\mathrm{2}}}  \SCsym{;}  \Gamma_{{\mathrm{2}}}  \vdash_\mathcal{L}  \SCnt{B}}
        \end{math}
      \end{center}
      By assumption, $c(\Pi_1),c(\Pi_2)\leq |X|$. By induction on $\Pi_1$
      and $\pi_1$, there is a proof $\Pi'$ for sequent
      $\Delta_{{\mathrm{1}}}  \SCsym{;}  \Phi  \SCsym{;}  \Delta_{{\mathrm{2}}}  \vdash_\mathcal{L}  \SCnt{A_{{\mathrm{1}}}}$ s.t. $c(\Pi') \leq |X|$. Therefore, the proof
      $\Pi$ can be constructed as follows with $c(\Pi) = c(\Pi') \leq |X|$.
      \begin{center}
        \scriptsize
        \begin{math}
          $$\mprset{flushleft}
          \inferrule* [right={\tiny implL}] {
            {
              \begin{array}{cc}
                \Pi' & \pi_2 \\
                {\Delta_{{\mathrm{1}}}  \SCsym{;}  \Phi  \SCsym{;}  \Delta_{{\mathrm{2}}}  \vdash_\mathcal{L}  \SCnt{A_{{\mathrm{1}}}}} & {\Gamma_{{\mathrm{1}}}  \SCsym{;}  \SCnt{A_{{\mathrm{2}}}}  \SCsym{;}  \Gamma_{{\mathrm{2}}}  \vdash_\mathcal{L}  \SCnt{B}}
              \end{array}
            }
          }{\Gamma_{{\mathrm{1}}}  \SCsym{;}  \Delta_{{\mathrm{1}}}  \SCsym{;}  \SCnt{A_{{\mathrm{2}}}}  \leftharpoonup  \SCnt{A_{{\mathrm{1}}}}  \SCsym{;}  \Phi  \SCsym{;}  \Delta_{{\mathrm{2}}}  \SCsym{;}  \Gamma_{{\mathrm{2}}}  \vdash_\mathcal{L}  \SCnt{B}}
        \end{math}
      \end{center}

\item Case 2:
      \begin{center}
        \scriptsize
        \begin{math}
          \begin{array}{c}
            \Pi_1 \\
            {\Gamma  \vdash_\mathcal{L}  \SCnt{C}}
          \end{array}
        \end{math}
        \qquad\qquad
        $\Pi_2$:
        \begin{math}
          $$\mprset{flushleft}
          \inferrule* [right={\tiny implL}] {
            {
              \begin{array}{cc}
                \pi_1 & \pi_2 \\
                {\Delta_{{\mathrm{1}}}  \SCsym{;}  \SCnt{C}  \SCsym{;}  \Delta_{{\mathrm{2}}}  \vdash_\mathcal{L}  \SCnt{A_{{\mathrm{1}}}}} & {\Gamma_{{\mathrm{1}}}  \SCsym{;}  \SCnt{A_{{\mathrm{2}}}}  \SCsym{;}  \Gamma_{{\mathrm{2}}}  \vdash_\mathcal{L}  \SCnt{B}}
              \end{array}
            }
          }{\Gamma_{{\mathrm{1}}}  \SCsym{;}  \Delta_{{\mathrm{1}}}  \SCsym{;}  \SCnt{C}  \SCsym{;}  \Delta_{{\mathrm{2}}}  \SCsym{;}  \SCnt{A_{{\mathrm{2}}}}  \leftharpoonup  \SCnt{A_{{\mathrm{1}}}}  \SCsym{;}  \Gamma_{{\mathrm{2}}}  \vdash_\mathcal{L}  \SCnt{B}}
        \end{math}
      \end{center}
      By assumption, $c(\Pi_1),c(\Pi_2)\leq |C|$. By induction on $\Pi_1$
      and $\pi_1$, there is a proof $\Pi'$ for sequent
      $\Delta_{{\mathrm{1}}}  \SCsym{;}  \Gamma  \SCsym{;}  \Delta_{{\mathrm{2}}}  \vdash_\mathcal{L}  \SCnt{A_{{\mathrm{1}}}}$ s.t. $c(\Pi') \leq |C|$. Therefore, the proof
      $\Pi$ can be constructed as follows with $c(\Pi) = c(\Pi') \leq |C|$.
      \begin{center}
        \scriptsize
        \begin{math}
          $$\mprset{flushleft}
          \inferrule* [right={\tiny implL}] {
            {
              \begin{array}{cc}
                \Pi' & \pi_2 \\
                {\Delta_{{\mathrm{1}}}  \SCsym{;}  \Gamma  \SCsym{;}  \Delta_{{\mathrm{2}}}  \vdash_\mathcal{L}  \SCnt{A_{{\mathrm{1}}}}} & {\Gamma_{{\mathrm{1}}}  \SCsym{;}  \SCnt{A_{{\mathrm{2}}}}  \SCsym{;}  \Gamma_{{\mathrm{2}}}  \vdash_\mathcal{L}  \SCnt{B}}
              \end{array}
            }
          }{\Gamma_{{\mathrm{1}}}  \SCsym{;}  \Delta_{{\mathrm{1}}}  \SCsym{;}  \Gamma  \SCsym{;}  \Delta_{{\mathrm{2}}}  \SCsym{;}  \SCnt{A_{{\mathrm{2}}}}  \leftharpoonup  \SCnt{A_{{\mathrm{1}}}}  \SCsym{;}  \Gamma_{{\mathrm{2}}}  \vdash_\mathcal{L}  \SCnt{B}}
        \end{math}
      \end{center}

\item Case 3:
      \begin{center}
        \scriptsize
        \begin{math}
          \begin{array}{c}
            \Pi_1 \\
            {\Phi  \vdash_\mathcal{C}  \SCnt{X}}
          \end{array}
        \end{math}
        \qquad\qquad
        $\Pi_2$:
        \begin{math}
          $$\mprset{flushleft}
          \inferrule* [right={\tiny implL}] {
            {
              \begin{array}{cc}
                \pi_1 & \pi_2 \\
                {\Delta  \vdash_\mathcal{L}  \SCnt{A_{{\mathrm{1}}}}} & {\Gamma_{{\mathrm{1}}}  \SCsym{;}  \SCnt{X}  \SCsym{;}  \Gamma_{{\mathrm{2}}}  \SCsym{;}  \SCnt{A_{{\mathrm{2}}}}  \SCsym{;}  \Gamma_{{\mathrm{3}}}  \vdash_\mathcal{L}  \SCnt{B}}
              \end{array}
            }
          }{\Gamma_{{\mathrm{1}}}  \SCsym{;}  \SCnt{X}  \SCsym{;}  \Gamma_{{\mathrm{2}}}  \SCsym{;}  \Delta  \SCsym{;}  \SCnt{A_{{\mathrm{2}}}}  \leftharpoonup  \SCnt{A_{{\mathrm{1}}}}  \SCsym{;}  \Gamma_{{\mathrm{3}}}  \vdash_\mathcal{L}  \SCnt{B}}
        \end{math}
      \end{center}
      By assumption, $c(\Pi_1),c(\Pi_2)\leq |X|$. By induction on $\Pi_1$
      and $\pi_2$, there is a proof $\Pi'$ for sequent
      $\Gamma_{{\mathrm{1}}}  \SCsym{;}  \Phi  \SCsym{;}  \Gamma_{{\mathrm{2}}}  \SCsym{;}  \SCnt{A_{{\mathrm{2}}}}  \SCsym{;}  \Gamma_{{\mathrm{3}}}  \vdash_\mathcal{L}  \SCnt{B}$ s.t. $c(\Pi') \leq |X|$. Therefore, the
      proof $\Pi$ can be constructed as follows with
      $c(\Pi) = c(\Pi') \leq |X|$.
      \begin{center}
        \scriptsize
        \begin{math}
          $$\mprset{flushleft}
          \inferrule* [right={\tiny implL}] {
            {
              \begin{array}{cc}
                \pi_1 & \Pi' \\
                {\Delta  \vdash_\mathcal{L}  \SCnt{A_{{\mathrm{1}}}}} & {\Gamma_{{\mathrm{1}}}  \SCsym{;}  \Phi  \SCsym{;}  \Gamma_{{\mathrm{2}}}  \SCsym{;}  \SCnt{A_{{\mathrm{2}}}}  \SCsym{;}  \Gamma_{{\mathrm{3}}}  \vdash_\mathcal{L}  \SCnt{B}}
              \end{array}
            }
          }{\Gamma_{{\mathrm{1}}}  \SCsym{;}  \Phi  \SCsym{;}  \Gamma_{{\mathrm{2}}}  \SCsym{;}  \Delta  \SCsym{;}  \SCnt{A_{{\mathrm{2}}}}  \leftharpoonup  \SCnt{A_{{\mathrm{1}}}}  \SCsym{;}  \Gamma_{{\mathrm{3}}}  \vdash_\mathcal{L}  \SCnt{B}}
        \end{math}
      \end{center}

\item Case 4:
      \begin{center}
        \scriptsize
        \begin{math}
          \begin{array}{c}
            \Pi_1 \\
            {\Delta_{{\mathrm{1}}}  \vdash_\mathcal{L}  \SCnt{B}}
          \end{array}
        \end{math}
        \qquad\qquad
        $\Pi_2$:
        \begin{math}
          $$\mprset{flushleft}
          \inferrule* [right={\tiny implL}] {
            {
              \begin{array}{cc}
                \pi_1 & \pi_2 \\
                {\Delta_{{\mathrm{2}}}  \vdash_\mathcal{L}  \SCnt{A_{{\mathrm{1}}}}} & {\Gamma_{{\mathrm{1}}}  \SCsym{;}  \SCnt{B}  \SCsym{;}  \Gamma_{{\mathrm{2}}}  \SCsym{;}  \SCnt{A_{{\mathrm{2}}}}  \SCsym{;}  \Gamma_{{\mathrm{3}}}  \vdash_\mathcal{L}  \SCnt{C}}
              \end{array}
            }
          }{\Gamma_{{\mathrm{1}}}  \SCsym{;}  \SCnt{B}  \SCsym{;}  \Gamma_{{\mathrm{2}}}  \SCsym{;}  \Delta_{{\mathrm{2}}}  \SCsym{;}  \SCnt{A_{{\mathrm{2}}}}  \leftharpoonup  \SCnt{A_{{\mathrm{1}}}}  \SCsym{;}  \Gamma_{{\mathrm{3}}}  \vdash_\mathcal{L}  \SCnt{C}}
        \end{math}
      \end{center}
      By assumption, $c(\Pi_1),c(\Pi_2)\leq |B|$. By induction on $\Pi_1$
      and $\pi_2$, there is a proof $\Pi'$ for sequent
      $\Gamma_{{\mathrm{1}}}  \SCsym{;}  \Delta_{{\mathrm{1}}}  \SCsym{;}  \Gamma_{{\mathrm{2}}}  \SCsym{;}  \SCnt{A_{{\mathrm{2}}}}  \SCsym{;}  \Gamma_{{\mathrm{3}}}  \vdash_\mathcal{L}  \SCnt{C}$ s.t. $c(\Pi') \leq |B|$. Therefore,
      the proof $\Pi$ can be constructed as follows with
      $c(\Pi) = c(\Pi') \leq |B|$.
      \begin{center}
        \scriptsize
        \begin{math}
          $$\mprset{flushleft}
          \inferrule* [right={\tiny implL}] {
            {
              \begin{array}{cc}
                \pi_1 & \Pi' \\
                {\Delta_{{\mathrm{2}}}  \vdash_\mathcal{L}  \SCnt{A_{{\mathrm{1}}}}} & {\Gamma_{{\mathrm{1}}}  \SCsym{;}  \Delta_{{\mathrm{1}}}  \SCsym{;}  \Gamma_{{\mathrm{2}}}  \SCsym{;}  \SCnt{A_{{\mathrm{2}}}}  \SCsym{;}  \Gamma_{{\mathrm{3}}}  \vdash_\mathcal{L}  \SCnt{C}}
              \end{array}
            }
          }{\Gamma_{{\mathrm{1}}}  \SCsym{;}  \Delta_{{\mathrm{1}}}  \SCsym{;}  \Gamma_{{\mathrm{2}}}  \SCsym{;}  \Delta_{{\mathrm{2}}}  \SCsym{;}  \SCnt{A_{{\mathrm{2}}}}  \leftharpoonup  \SCnt{A_{{\mathrm{1}}}}  \SCsym{;}  \Gamma_{{\mathrm{3}}}  \vdash_\mathcal{L}  \SCnt{C}}
        \end{math}
      \end{center}

\item Case 5:
      \begin{center}
        \scriptsize
        \begin{math}
          \begin{array}{c}
            \Pi_1 \\
            {\Phi  \vdash_\mathcal{C}  \SCnt{X}}
          \end{array}
        \end{math}
        \qquad\qquad
        $\Pi_2$:
        \begin{math}
          $$\mprset{flushleft}
          \inferrule* [right={\tiny implL}] {
            {
              \begin{array}{cc}
                \pi_1 & \pi_2 \\
                {\Delta  \vdash_\mathcal{L}  \SCnt{A_{{\mathrm{1}}}}} & {\Gamma_{{\mathrm{1}}}  \SCsym{;}  \SCnt{A_{{\mathrm{2}}}}  \SCsym{;}  \Gamma_{{\mathrm{2}}}  \SCsym{;}  \SCnt{X}  \SCsym{;}  \Gamma_{{\mathrm{3}}}  \vdash_\mathcal{L}  \SCnt{B}}
              \end{array}
            }
          }{\Gamma_{{\mathrm{1}}}  \SCsym{;}  \Delta  \SCsym{;}  \SCnt{A_{{\mathrm{2}}}}  \leftharpoonup  \SCnt{A_{{\mathrm{1}}}}  \SCsym{;}  \Delta  \SCsym{;}  \Gamma_{{\mathrm{2}}}  \SCsym{;}  \SCnt{X}  \SCsym{;}  \Gamma_{{\mathrm{3}}}  \vdash_\mathcal{L}  \SCnt{B}}
        \end{math}
      \end{center}
      By assumption, $c(\Pi_1),c(\Pi_2)\leq |X|$. By induction on $\Pi_1$
      and $\pi_2$, there is a proof $\Pi'$ for sequent
      $\Gamma_{{\mathrm{1}}}  \SCsym{;}  \SCnt{A_{{\mathrm{2}}}}  \SCsym{;}  \Gamma_{{\mathrm{2}}}  \SCsym{;}  \Phi  \SCsym{;}  \Gamma_{{\mathrm{3}}}  \vdash_\mathcal{L}  \SCnt{B}$ s.t. $c(\Pi') \leq |X|$. Therefore, the
      proof $\Pi$ can be constructed as follows with
      $c(\Pi) = c(\Pi') \leq |X|$.
      \begin{center}
        \scriptsize
        \begin{math}
          $$\mprset{flushleft}
          \inferrule* [right={\tiny implL}] {
            {
              \begin{array}{cc}
                \pi_1 & \Pi' \\
                {\Delta  \vdash_\mathcal{L}  \SCnt{A_{{\mathrm{1}}}}} & {\Gamma_{{\mathrm{1}}}  \SCsym{;}  \SCnt{A_{{\mathrm{2}}}}  \SCsym{;}  \Gamma_{{\mathrm{2}}}  \SCsym{;}  \Phi  \SCsym{;}  \Gamma_{{\mathrm{3}}}  \vdash_\mathcal{L}  \SCnt{B}}
              \end{array}
            }
          }{\Gamma_{{\mathrm{1}}}  \SCsym{;}  \Delta  \SCsym{;}  \SCnt{A_{{\mathrm{2}}}}  \leftharpoonup  \SCnt{A_{{\mathrm{1}}}}  \SCsym{;}  \Gamma_{{\mathrm{2}}}  \SCsym{;}  \Phi  \SCsym{;}  \Gamma_{{\mathrm{3}}}  \vdash_\mathcal{L}  \SCnt{B}}
        \end{math}
      \end{center}

\item Case 6:
    \begin{center}
      \scriptsize
      \begin{math}
        \begin{array}{c}
          \Pi_1 \\
          {\Delta_{{\mathrm{1}}}  \vdash_\mathcal{L}  \SCnt{B}}
        \end{array}
      \end{math}
      \qquad\qquad
      $\Pi_2$:
      \begin{math}
        $$\mprset{flushleft}
        \inferrule* [right={\tiny implL}] {
          {
            \begin{array}{cc}
              \pi_1 & \pi_2 \\
              {\Delta_{{\mathrm{2}}}  \vdash_\mathcal{L}  \SCnt{A_{{\mathrm{1}}}}} & {\Gamma_{{\mathrm{1}}}  \SCsym{;}  \SCnt{A_{{\mathrm{2}}}}  \SCsym{;}  \Gamma_{{\mathrm{2}}}  \SCsym{;}  \SCnt{B}  \SCsym{;}  \Gamma_{{\mathrm{3}}}  \vdash_\mathcal{L}  \SCnt{C}}
            \end{array}
          }
        }{\Gamma_{{\mathrm{1}}}  \SCsym{;}  \Delta_{{\mathrm{2}}}  \SCsym{;}  \SCnt{A_{{\mathrm{2}}}}  \leftharpoonup  \SCnt{A_{{\mathrm{1}}}}  \SCsym{;}  \Gamma_{{\mathrm{2}}}  \SCsym{;}  \SCnt{B}  \SCsym{;}  \Gamma_{{\mathrm{3}}}  \vdash_\mathcal{L}  \SCnt{C}}
      \end{math}
    \end{center}
    By assumption, $c(\Pi_1),c(\Pi_2)\leq |B|$. By induction on $\Pi_1$ and
    $\pi_2$, there is a proof $\Pi'$ for sequent
    $\Gamma_{{\mathrm{1}}}  \SCsym{;}  \SCnt{A_{{\mathrm{2}}}}  \SCsym{;}  \Gamma_{{\mathrm{2}}}  \SCsym{;}  \Delta_{{\mathrm{1}}}  \SCsym{;}  \Gamma_{{\mathrm{3}}}  \vdash_\mathcal{L}  \SCnt{C}$ s.t. $c(\Pi') \leq |B|$. Therefore, the
    proof $\Pi$ can be constructed as follows with
    $c(\Pi) = c(\Pi') \leq |B|$.
    \begin{center}
      \scriptsize
      \begin{math}
        $$\mprset{flushleft}
        \inferrule* [right={\tiny implL}] {
          {
            \begin{array}{cc}
              \pi_1 & \Pi' \\
              {\Delta_{{\mathrm{2}}}  \vdash_\mathcal{L}  \SCnt{A_{{\mathrm{1}}}}} & {\Gamma_{{\mathrm{1}}}  \SCsym{;}  \SCnt{A_{{\mathrm{2}}}}  \SCsym{;}  \Gamma_{{\mathrm{2}}}  \SCsym{;}  \Delta_{{\mathrm{1}}}  \SCsym{;}  \Gamma_{{\mathrm{3}}}  \vdash_\mathcal{L}  \SCnt{C}}
            \end{array}
          }
        }{\Gamma_{{\mathrm{1}}}  \SCsym{;}  \Delta_{{\mathrm{2}}}  \SCsym{;}  \SCnt{A_{{\mathrm{2}}}}  \leftharpoonup  \SCnt{A_{{\mathrm{1}}}}  \SCsym{;}  \Gamma_{{\mathrm{2}}}  \SCsym{;}  \Delta_{{\mathrm{1}}}  \SCsym{;}  \Gamma_{{\mathrm{3}}}  \vdash_\mathcal{L}  \SCnt{C}}
      \end{math}
    \end{center}
\end{itemize}

\subsubsection{Left introduction of the commutative tensor $\otimes$ (with low priority)}
\begin{itemize}
\item Case 1:
      \begin{center}
        \scriptsize
        \begin{math}
          \begin{array}{c}
            \Pi_1 \\
            {\Phi  \vdash_\mathcal{C}  \SCnt{X}}
          \end{array}
        \end{math}
        \qquad\qquad
        $\Pi_2$:
        \begin{math}
          $$\mprset{flushleft}
          \inferrule* [right={\tiny tenL}] {
            {
              \begin{array}{c}
                \pi \\
                {\Psi_{{\mathrm{1}}}  \SCsym{,}  \SCnt{X}  \SCsym{,}  \Psi_{{\mathrm{2}}}  \SCsym{,}  \SCnt{Y_{{\mathrm{1}}}}  \SCsym{,}  \SCnt{Y_{{\mathrm{2}}}}  \SCsym{,}  \Psi_{{\mathrm{3}}}  \vdash_\mathcal{C}  \SCnt{Z}}
              \end{array}
            }
          }{\Psi_{{\mathrm{1}}}  \SCsym{,}  \SCnt{X}  \SCsym{,}  \Psi_{{\mathrm{2}}}  \SCsym{,}  \SCnt{Y_{{\mathrm{1}}}}  \otimes  \SCnt{Y_{{\mathrm{2}}}}  \SCsym{,}  \Psi_{{\mathrm{3}}}  \vdash_\mathcal{C}  \SCnt{Z}}
        \end{math}
      \end{center}
      By assumption, $c(\Pi_1),c(\Pi_2)\leq |X|$. By induction on $\Pi_1$
      and $\pi$, there is a proof $\Pi'$ for sequent
      $\Psi_{{\mathrm{1}}}  \SCsym{,}  \Phi  \SCsym{,}  \Psi_{{\mathrm{2}}}  \SCsym{,}  \SCnt{Y_{{\mathrm{1}}}}  \SCsym{,}  \SCnt{Y_{{\mathrm{2}}}}  \SCsym{,}  \Psi_{{\mathrm{3}}}  \vdash_\mathcal{C}  \SCnt{Z}$ s.t. $c(\Pi') \leq |X|$. Therefore,
      the proof $\Pi$ can be constructed as follows with
      $c(\Pi) = c(\Pi') \leq |X|$.
      \begin{center}
        \scriptsize
        \begin{math}
          $$\mprset{flushleft}
          \inferrule* [right={\tiny tenL}] {
            {
              \begin{array}{c}
                \Pi' \\
                {\Psi_{{\mathrm{1}}}  \SCsym{,}  \Phi  \SCsym{,}  \Psi_{{\mathrm{2}}}  \SCsym{,}  \SCnt{Y_{{\mathrm{1}}}}  \SCsym{,}  \SCnt{Y_{{\mathrm{2}}}}  \SCsym{,}  \Psi_{{\mathrm{3}}}  \vdash_\mathcal{C}  \SCnt{Z}}
              \end{array}
            }
          }{\Psi_{{\mathrm{1}}}  \SCsym{,}  \Phi  \SCsym{,}  \Psi_{{\mathrm{2}}}  \SCsym{,}  \SCnt{Y_{{\mathrm{1}}}}  \otimes  \SCnt{Y_{{\mathrm{2}}}}  \SCsym{,}  \Psi_{{\mathrm{3}}}  \vdash_\mathcal{C}  \SCnt{Z}}
        \end{math}
      \end{center}

\item Case 2:
      \begin{center}
        \scriptsize
        \begin{math}
          \begin{array}{c}
            \Pi_1 \\
            {\Phi  \vdash_\mathcal{C}  \SCnt{X}}
          \end{array}
        \end{math}
        \qquad\qquad
        $\Pi_2$:
        \begin{math}
          $$\mprset{flushleft}
          \inferrule* [right={\tiny tenL}] {
            {
              \begin{array}{c}
                \pi \\
                {\Psi_{{\mathrm{1}}}  \SCsym{,}  \SCnt{Y_{{\mathrm{1}}}}  \SCsym{,}  \SCnt{Y_{{\mathrm{2}}}}  \SCsym{,}  \Psi_{{\mathrm{2}}}  \SCsym{,}  \SCnt{X}  \SCsym{,}  \Psi_{{\mathrm{3}}}  \vdash_\mathcal{C}  \SCnt{Z}}
              \end{array}
            }
          }{\Psi_{{\mathrm{1}}}  \SCsym{,}  \SCnt{Y_{{\mathrm{1}}}}  \otimes  \SCnt{Y_{{\mathrm{2}}}}  \SCsym{,}  \Psi_{{\mathrm{2}}}  \SCsym{,}  \SCnt{X}  \SCsym{,}  \Psi_{{\mathrm{3}}}  \vdash_\mathcal{C}  \SCnt{Z}}
        \end{math}
      \end{center}
      By assumption, $c(\Pi_1),c(\Pi_2)\leq |X|$. By induction on $\Pi_1$
      and $\pi$, there is a proof $\Pi'$ for sequent
      $\Psi_{{\mathrm{1}}}  \SCsym{,}  \SCnt{Y_{{\mathrm{1}}}}  \SCsym{,}  \SCnt{Y_{{\mathrm{2}}}}  \SCsym{,}  \Psi_{{\mathrm{2}}}  \SCsym{,}  \Phi  \SCsym{,}  \Psi_{{\mathrm{3}}}  \vdash_\mathcal{C}  \SCnt{Z}$ s.t. $c(\Pi') \leq |X|$. Therefore,
      the proof $\Pi$ can be constructed as follows with
      $c(\Pi) = c(\Pi') \leq |X|$.
      \begin{center}
        \scriptsize
        \begin{math}
          $$\mprset{flushleft}
          \inferrule* [right={\tiny tenL}] {
            {
              \begin{array}{c}
                \Pi' \\
                {\Psi_{{\mathrm{1}}}  \SCsym{,}  \SCnt{Y_{{\mathrm{1}}}}  \SCsym{,}  \SCnt{Y_{{\mathrm{2}}}}  \SCsym{,}  \Psi_{{\mathrm{2}}}  \SCsym{,}  \Phi  \SCsym{,}  \Psi_{{\mathrm{3}}}  \vdash_\mathcal{C}  \SCnt{Z}}
              \end{array}
            }
          }{\Psi_{{\mathrm{1}}}  \SCsym{,}  \SCnt{Y_{{\mathrm{1}}}}  \otimes  \SCnt{Y_{{\mathrm{2}}}}  \SCsym{,}  \Psi_{{\mathrm{2}}}  \SCsym{,}  \Phi  \SCsym{,}  \Psi_{{\mathrm{3}}}  \vdash_\mathcal{C}  \SCnt{Z}}
        \end{math}
      \end{center}

\item Case 3:
      \begin{center}
        \scriptsize
        \begin{math}
          \begin{array}{c}
            \Pi_1 \\
            {\Phi  \vdash_\mathcal{C}  \SCnt{X}}
          \end{array}
        \end{math}
        \qquad\qquad
        $\Pi_2$:
        \begin{math}
          $$\mprset{flushleft}
          \inferrule* [right={\tiny tenL}] {
            {
              \begin{array}{c}
                \pi \\
                {\Gamma_{{\mathrm{1}}}  \SCsym{;}  \SCnt{X}  \SCsym{;}  \Gamma_{{\mathrm{2}}}  \SCsym{;}  \SCnt{Y_{{\mathrm{1}}}}  \SCsym{;}  \SCnt{Y_{{\mathrm{2}}}}  \SCsym{;}  \Gamma_{{\mathrm{3}}}  \vdash_\mathcal{L}  \SCnt{A}}
              \end{array}
            }
          }{\Gamma_{{\mathrm{1}}}  \SCsym{;}  \SCnt{X}  \SCsym{;}  \Gamma_{{\mathrm{2}}}  \SCsym{;}  \SCnt{Y_{{\mathrm{1}}}}  \otimes  \SCnt{Y_{{\mathrm{2}}}}  \SCsym{;}  \Gamma_{{\mathrm{3}}}  \vdash_\mathcal{L}  \SCnt{A}}
        \end{math}
      \end{center}
      By assumption, $c(\Pi_1),c(\Pi_2)\leq |X|$. By induction on $\Pi_1$
      and $\pi$, there is a proof $\Pi'$ for sequent
      $\Gamma_{{\mathrm{1}}}  \SCsym{;}  \Phi  \SCsym{;}  \Gamma_{{\mathrm{2}}}  \SCsym{;}  \SCnt{Y_{{\mathrm{1}}}}  \SCsym{;}  \SCnt{Y_{{\mathrm{2}}}}  \SCsym{;}  \Gamma_{{\mathrm{3}}}  \vdash_\mathcal{L}  \SCnt{A}$ s.t. $c(\Pi') \leq |X|$. Therefore,
      the proof $\Pi$ can be constructed as follows with
      $c(\Pi) = c(\Pi') \leq |X|$.
      \begin{center}
        \scriptsize
        \begin{math}
          $$\mprset{flushleft}
          \inferrule* [right={\tiny tenL}] {
            {
              \begin{array}{c}
                \Pi' \\
                {\Gamma_{{\mathrm{1}}}  \SCsym{;}  \Phi  \SCsym{;}  \Gamma_{{\mathrm{2}}}  \SCsym{;}  \SCnt{Y_{{\mathrm{1}}}}  \SCsym{;}  \SCnt{Y_{{\mathrm{2}}}}  \SCsym{;}  \Gamma_{{\mathrm{3}}}  \vdash_\mathcal{L}  \SCnt{A}}
              \end{array}
            }
          }{\Gamma_{{\mathrm{1}}}  \SCsym{;}  \Phi  \SCsym{;}  \Gamma_{{\mathrm{2}}}  \SCsym{;}  \SCnt{Y_{{\mathrm{1}}}}  \otimes  \SCnt{Y_{{\mathrm{2}}}}  \SCsym{;}  \Gamma_{{\mathrm{3}}}  \vdash_\mathcal{L}  \SCnt{A}}
        \end{math}
      \end{center}

\item Case 4:
      \begin{center}
        \scriptsize
        \begin{math}
          \begin{array}{c}
            \Pi_1 \\
            {\Delta  \vdash_\mathcal{L}  \SCnt{B}}
          \end{array}
        \end{math}
        \qquad\qquad
        $\Pi_2$:
        \begin{math}
          $$\mprset{flushleft}
          \inferrule* [right={\tiny tenL}] {
            {
              \begin{array}{c}
                \pi \\
                {\Gamma_{{\mathrm{1}}}  \SCsym{;}  \SCnt{B}  \SCsym{;}  \Gamma_{{\mathrm{2}}}  \SCsym{;}  \SCnt{Y_{{\mathrm{1}}}}  \SCsym{;}  \SCnt{Y_{{\mathrm{2}}}}  \SCsym{;}  \Gamma_{{\mathrm{3}}}  \vdash_\mathcal{L}  \SCnt{A}}
              \end{array}
            }
          }{\Gamma_{{\mathrm{1}}}  \SCsym{;}  \SCnt{B}  \SCsym{;}  \Gamma_{{\mathrm{2}}}  \SCsym{;}  \SCnt{Y_{{\mathrm{1}}}}  \otimes  \SCnt{Y_{{\mathrm{2}}}}  \SCsym{;}  \Gamma_{{\mathrm{3}}}  \vdash_\mathcal{L}  \SCnt{A}}
        \end{math}
      \end{center}
      By assumption, $c(\Pi_1),c(\Pi_2)\leq |B|$. By induction on $\Pi_1$
      and $\pi$, there is a proof $\Pi'$ for sequent
      $\Gamma_{{\mathrm{1}}}  \SCsym{;}  \SCnt{B}  \SCsym{;}  \Gamma_{{\mathrm{2}}}  \SCsym{;}  \SCnt{Y_{{\mathrm{1}}}}  \SCsym{;}  \SCnt{Y_{{\mathrm{2}}}}  \SCsym{;}  \Gamma_{{\mathrm{3}}}  \vdash_\mathcal{L}  \SCnt{A}$ s.t. $c(\Pi') \leq |B|$. Therefore,
      the proof $\Pi$ can be constructed as follows with
      $c(\Pi) = c(\Pi') \leq |B|$.
      \begin{center}
        \scriptsize
        \begin{math}
          $$\mprset{flushleft}
          \inferrule* [right={\tiny tenL}] {
            {
              \begin{array}{c}
                \Pi' \\
                {\Gamma_{{\mathrm{1}}}  \SCsym{;}  \Delta  \SCsym{;}  \Gamma_{{\mathrm{2}}}  \SCsym{;}  \SCnt{Y_{{\mathrm{1}}}}  \SCsym{;}  \SCnt{Y_{{\mathrm{2}}}}  \SCsym{;}  \Gamma_{{\mathrm{3}}}  \vdash_\mathcal{L}  \SCnt{A}}
              \end{array}
            }
          }{\Gamma_{{\mathrm{1}}}  \SCsym{;}  \Delta  \SCsym{;}  \Gamma_{{\mathrm{2}}}  \SCsym{;}  \SCnt{Y_{{\mathrm{1}}}}  \otimes  \SCnt{Y_{{\mathrm{2}}}}  \SCsym{;}  \Gamma_{{\mathrm{3}}}  \vdash_\mathcal{L}  \SCnt{A}}
        \end{math}
      \end{center}

\item Case 5:
      \begin{center}
        \scriptsize
        \begin{math}
          \begin{array}{c}
            \Pi_1 \\
            {\Phi  \vdash_\mathcal{C}  \SCnt{X}}
          \end{array}
        \end{math}
        \qquad\qquad
        $\Pi_2$:
        \begin{math}
          $$\mprset{flushleft}
          \inferrule* [right={\tiny tenL}] {
            {
              \begin{array}{c}
                \pi \\
                {\Gamma_{{\mathrm{1}}}  \SCsym{;}  \SCnt{Y_{{\mathrm{1}}}}  \SCsym{;}  \SCnt{Y_{{\mathrm{2}}}}  \SCsym{;}  \Gamma_{{\mathrm{2}}}  \SCsym{;}  \SCnt{X}  \SCsym{;}  \Gamma_{{\mathrm{3}}}  \vdash_\mathcal{L}  \SCnt{A}}
              \end{array}
            }
          }{\Gamma_{{\mathrm{1}}}  \SCsym{;}  \SCnt{Y_{{\mathrm{1}}}}  \otimes  \SCnt{Y_{{\mathrm{2}}}}  \SCsym{;}  \Gamma_{{\mathrm{2}}}  \SCsym{;}  \SCnt{X}  \SCsym{;}  \Gamma_{{\mathrm{3}}}  \vdash_\mathcal{L}  \SCnt{A}}
        \end{math}
      \end{center}
      By assumption, $c(\Pi_1),c(\Pi_2)\leq |X|$. By induction on $\Pi_1$
      and $\pi$, there is a proof $\Pi'$ for sequent
      $\Gamma_{{\mathrm{1}}}  \SCsym{;}  \SCnt{Y_{{\mathrm{1}}}}  \SCsym{;}  \SCnt{Y_{{\mathrm{2}}}}  \SCsym{;}  \Gamma_{{\mathrm{2}}}  \SCsym{;}  \Phi  \SCsym{;}  \Gamma_{{\mathrm{3}}}  \vdash_\mathcal{L}  \SCnt{A}$ s.t. $c(\Pi') \leq |X|$. Therefore,
      the proof $\Pi$ can be constructed as follows with
      $c(\Pi) = c(\Pi') \leq |X|$.
      \begin{center}
        \scriptsize
        \begin{math}
          $$\mprset{flushleft}
          \inferrule* [right={\tiny tenL}] {
            {
              \begin{array}{c}
                \Pi' \\
                {\Gamma_{{\mathrm{1}}}  \SCsym{;}  \SCnt{Y_{{\mathrm{1}}}}  \SCsym{;}  \SCnt{Y_{{\mathrm{2}}}}  \SCsym{;}  \Gamma_{{\mathrm{2}}}  \SCsym{;}  \Phi  \SCsym{;}  \Gamma_{{\mathrm{3}}}  \vdash_\mathcal{L}  \SCnt{A}}
              \end{array}
            }
          }{\Gamma_{{\mathrm{1}}}  \SCsym{;}  \SCnt{Y_{{\mathrm{1}}}}  \otimes  \SCnt{Y_{{\mathrm{2}}}}  \SCsym{;}  \Gamma_{{\mathrm{2}}}  \SCsym{;}  \Phi  \SCsym{;}  \Gamma_{{\mathrm{3}}}  \vdash_\mathcal{L}  \SCnt{A}}
        \end{math}
      \end{center}

\item Case 6:
      \begin{center}
        \scriptsize
        \begin{math}
          \begin{array}{c}
            \Pi_1 \\
            {\Delta  \vdash_\mathcal{L}  \SCnt{B}}
          \end{array}
        \end{math}
        \qquad\qquad
        $\Pi_2$:
        \begin{math}
          $$\mprset{flushleft}
          \inferrule* [right={\tiny tenL}] {
            {
              \begin{array}{c}
                \pi \\
                {\Gamma_{{\mathrm{1}}}  \SCsym{;}  \SCnt{Y_{{\mathrm{1}}}}  \SCsym{;}  \SCnt{Y_{{\mathrm{2}}}}  \SCsym{;}  \Gamma_{{\mathrm{2}}}  \SCsym{;}  \SCnt{B}  \SCsym{;}  \Gamma_{{\mathrm{3}}}  \vdash_\mathcal{L}  \SCnt{A}}
              \end{array}
            }
          }{\Gamma_{{\mathrm{1}}}  \SCsym{;}  \SCnt{Y_{{\mathrm{1}}}}  \otimes  \SCnt{Y_{{\mathrm{2}}}}  \SCsym{;}  \Gamma_{{\mathrm{2}}}  \SCsym{;}  \SCnt{B}  \SCsym{;}  \Gamma_{{\mathrm{3}}}  \vdash_\mathcal{L}  \SCnt{A}}
        \end{math}
      \end{center}
      By assumption, $c(\Pi_1),c(\Pi_2)\leq |B|$. By induction on $\Pi_1$
      and $\pi$, there is a proof $\Pi'$ for sequent
      $\Gamma_{{\mathrm{1}}}  \SCsym{;}  \SCnt{Y_{{\mathrm{1}}}}  \SCsym{;}  \SCnt{Y_{{\mathrm{2}}}}  \SCsym{;}  \Gamma_{{\mathrm{2}}}  \SCsym{;}  \Delta  \SCsym{;}  \Gamma_{{\mathrm{3}}}  \vdash_\mathcal{L}  \SCnt{A}$ s.t. $c(\Pi') \leq |B|$. Therefore,
      the proof $\Pi$ can be constructed as follows with
      $c(\Pi) = c(\Pi') \leq |B|$.
      \begin{center}
        \scriptsize
        \begin{math}
          $$\mprset{flushleft}
          \inferrule* [right={\tiny tenL}] {
            {
              \begin{array}{c}
                \Pi' \\
                {\Gamma_{{\mathrm{1}}}  \SCsym{;}  \SCnt{Y_{{\mathrm{1}}}}  \SCsym{;}  \SCnt{Y_{{\mathrm{2}}}}  \SCsym{;}  \Gamma_{{\mathrm{2}}}  \SCsym{;}  \Delta  \SCsym{;}  \Gamma_{{\mathrm{3}}}  \vdash_\mathcal{L}  \SCnt{A}}
              \end{array}
            }
          }{\Gamma_{{\mathrm{1}}}  \SCsym{;}  \SCnt{Y_{{\mathrm{1}}}}  \otimes  \SCnt{Y_{{\mathrm{2}}}}  \SCsym{;}  \Gamma_{{\mathrm{2}}}  \SCsym{;}  \Delta  \SCsym{;}  \Gamma_{{\mathrm{3}}}  \vdash_\mathcal{L}  \SCnt{A}}
        \end{math}
      \end{center}
\end{itemize}

\subsubsection{Left introduction of the non-commutative tensor $\tri$ (with low priority)}:
\begin{itemize}
\item Case 1:
      \begin{center}
        \scriptsize
        \begin{math}
          \begin{array}{c}
            \Pi_1 \\
            {\Phi  \vdash_\mathcal{C}  \SCnt{X}}
          \end{array}
        \end{math}
        \qquad\qquad
        $\Pi_2$:
        \begin{math}
          $$\mprset{flushleft}
          \inferrule* [right={\tiny tenL}] {
            {
              \begin{array}{c}
                \pi \\
                {\Gamma_{{\mathrm{1}}}  \SCsym{;}  \SCnt{X}  \SCsym{;}  \Gamma_{{\mathrm{2}}}  \SCsym{;}  \SCnt{A_{{\mathrm{1}}}}  \SCsym{;}  \SCnt{A_{{\mathrm{2}}}}  \SCsym{;}  \Gamma_{{\mathrm{3}}}  \vdash_\mathcal{L}  \SCnt{B}}
              \end{array}
            }
          }{\Gamma_{{\mathrm{1}}}  \SCsym{;}  \SCnt{X}  \SCsym{;}  \Gamma_{{\mathrm{2}}}  \SCsym{;}  \SCnt{A_{{\mathrm{1}}}}  \triangleright  \SCnt{A_{{\mathrm{2}}}}  \SCsym{;}  \Gamma_{{\mathrm{3}}}  \vdash_\mathcal{L}  \SCnt{B}}
        \end{math}
      \end{center}
      By assumption, $c(\Pi_1),c(\Pi_2)\leq |X|$. By induction on $\Pi_1$
      and $\pi$, there is a proof $\Pi'$ for sequent
      $\Gamma_{{\mathrm{1}}}  \SCsym{;}  \Phi  \SCsym{;}  \Gamma_{{\mathrm{2}}}  \SCsym{;}  \SCnt{A_{{\mathrm{1}}}}  \SCsym{;}  \SCnt{A_{{\mathrm{2}}}}  \SCsym{;}  \Gamma_{{\mathrm{3}}}  \vdash_\mathcal{L}  \SCnt{B}$ s.t. $c(\Pi') \leq |X|$. Therefore,
      the proof $\Pi$ can be constructed as follows with
      $c(\Pi) = c(\Pi') \leq |X|$.
      \begin{center}
        \scriptsize
        \begin{math}
          $$\mprset{flushleft}
          \inferrule* [right={\tiny tenL}] {
            {
              \begin{array}{c}
                \Pi' \\
                {\Gamma_{{\mathrm{1}}}  \SCsym{;}  \Phi  \SCsym{;}  \Gamma_{{\mathrm{2}}}  \SCsym{;}  \SCnt{A_{{\mathrm{1}}}}  \SCsym{;}  \SCnt{A_{{\mathrm{2}}}}  \SCsym{;}  \Gamma_{{\mathrm{3}}}  \vdash_\mathcal{L}  \SCnt{B}}
              \end{array}
            }
          }{\Gamma_{{\mathrm{1}}}  \SCsym{;}  \Phi  \SCsym{;}  \Gamma_{{\mathrm{2}}}  \SCsym{;}  \SCnt{A_{{\mathrm{1}}}}  \triangleright  \SCnt{A_{{\mathrm{2}}}}  \SCsym{;}  \Gamma_{{\mathrm{3}}}  \vdash_\mathcal{L}  \SCnt{B}}
        \end{math}
      \end{center}

\item Case 2:
      \begin{center}
        \scriptsize
        \begin{math}
          \begin{array}{c}
            \Pi_1 \\
            {\Delta  \vdash_\mathcal{L}  \SCnt{B}}
          \end{array}
        \end{math}
        \qquad\qquad
        $\Pi_2$:
        \begin{math}
          $$\mprset{flushleft}
          \inferrule* [right={\tiny tenL}] {
            {
              \begin{array}{c}
                \pi \\
                {\Gamma_{{\mathrm{1}}}  \SCsym{;}  \SCnt{B}  \SCsym{;}  \Gamma_{{\mathrm{2}}}  \SCsym{;}  \SCnt{A_{{\mathrm{1}}}}  \SCsym{;}  \SCnt{A_{{\mathrm{2}}}}  \SCsym{;}  \Gamma_{{\mathrm{3}}}  \vdash_\mathcal{L}  \SCnt{C}}
              \end{array}
            }
          }{\Gamma_{{\mathrm{1}}}  \SCsym{;}  \SCnt{B}  \SCsym{;}  \Gamma_{{\mathrm{2}}}  \SCsym{;}  \SCnt{A_{{\mathrm{1}}}}  \triangleright  \SCnt{A_{{\mathrm{2}}}}  \SCsym{;}  \Gamma_{{\mathrm{3}}}  \vdash_\mathcal{L}  \SCnt{C}}
        \end{math}
      \end{center}
      By assumption, $c(\Pi_1),c(\Pi_2)\leq |B|$. By induction on $\Pi_1$
      and $\pi$, there is a proof $\Pi'$ for sequent
      $\Gamma_{{\mathrm{1}}}  \SCsym{;}  \Delta  \SCsym{;}  \Gamma_{{\mathrm{2}}}  \SCsym{;}  \SCnt{A_{{\mathrm{1}}}}  \SCsym{;}  \SCnt{A_{{\mathrm{2}}}}  \SCsym{;}  \Gamma_{{\mathrm{3}}}  \vdash_\mathcal{L}  \SCnt{C}$ s.t. $c(\Pi') \leq |B|$. Therefore,
      the proof $\Pi$ can be constructed as follows with
      $c(\Pi) = c(\Pi') \leq |B|$.
      \begin{center}
        \scriptsize
        \begin{math}
          $$\mprset{flushleft}
          \inferrule* [right={\tiny tenL}] {
            {
              \begin{array}{c}
                \Pi' \\
                {\Gamma_{{\mathrm{1}}}  \SCsym{;}  \Delta  \SCsym{;}  \Gamma_{{\mathrm{2}}}  \SCsym{;}  \SCnt{A_{{\mathrm{1}}}}  \SCsym{;}  \SCnt{A_{{\mathrm{2}}}}  \SCsym{;}  \Gamma_{{\mathrm{3}}}  \vdash_\mathcal{L}  \SCnt{C}}
              \end{array}
            }
          }{\Gamma_{{\mathrm{1}}}  \SCsym{;}  \Delta  \SCsym{;}  \Gamma_{{\mathrm{2}}}  \SCsym{;}  \SCnt{A_{{\mathrm{1}}}}  \triangleright  \SCnt{A_{{\mathrm{2}}}}  \SCsym{;}  \Gamma_{{\mathrm{3}}}  \vdash_\mathcal{L}  \SCnt{C}}
        \end{math}
      \end{center}

\item Case 3:
      \begin{center}
        \scriptsize
        \begin{math}
          \begin{array}{c}
            \Pi_1 \\
            {\Phi  \vdash_\mathcal{C}  \SCnt{X}}
          \end{array}
        \end{math}
        \qquad\qquad
        $\Pi_2$:
        \begin{math}
          $$\mprset{flushleft}
          \inferrule* [right={\tiny tenL}] {
            {
              \begin{array}{c}
                \pi \\
                {\Gamma_{{\mathrm{1}}}  \SCsym{;}  \SCnt{A_{{\mathrm{1}}}}  \SCsym{;}  \SCnt{A_{{\mathrm{2}}}}  \SCsym{;}  \Gamma_{{\mathrm{2}}}  \SCsym{;}  \SCnt{X}  \SCsym{;}  \Gamma_{{\mathrm{3}}}  \vdash_\mathcal{L}  \SCnt{B}}
              \end{array}
            }
          }{\Gamma_{{\mathrm{1}}}  \SCsym{;}  \SCnt{A_{{\mathrm{1}}}}  \triangleright  \SCnt{A_{{\mathrm{2}}}}  \SCsym{;}  \Gamma_{{\mathrm{2}}}  \SCsym{;}  \SCnt{X}  \SCsym{;}  \Gamma_{{\mathrm{3}}}  \vdash_\mathcal{L}  \SCnt{B}}
        \end{math}
      \end{center}
      By assumption, $c(\Pi_1),c(\Pi_2)\leq |X|$. By induction on $\Pi_1$
      and $\pi$, there is a proof $\Pi'$ for sequent
      $\Gamma_{{\mathrm{1}}}  \SCsym{;}  \SCnt{A_{{\mathrm{1}}}}  \SCsym{;}  \SCnt{A_{{\mathrm{2}}}}  \SCsym{;}  \Gamma_{{\mathrm{2}}}  \SCsym{;}  \Phi  \SCsym{;}  \Gamma_{{\mathrm{3}}}  \vdash_\mathcal{L}  \SCnt{A}$ s.t. $c(\Pi') \leq |X|$. Therefore,
      the proof $\Pi$ can be constructed as follows with
      $c(\Pi) = c(\Pi') \leq |X|$.
      \begin{center}
        \scriptsize
        \begin{math}
          $$\mprset{flushleft}
          \inferrule* [right={\tiny tenL}] {
            {
              \begin{array}{c}
                \Pi' \\
                {\Gamma_{{\mathrm{1}}}  \SCsym{;}  \SCnt{A_{{\mathrm{1}}}}  \SCsym{;}  \SCnt{A_{{\mathrm{2}}}}  \SCsym{;}  \Gamma_{{\mathrm{2}}}  \SCsym{;}  \Phi  \SCsym{;}  \Gamma_{{\mathrm{3}}}  \vdash_\mathcal{L}  \SCnt{B}}
              \end{array}
            }
          }{\Gamma_{{\mathrm{1}}}  \SCsym{;}  \SCnt{A_{{\mathrm{1}}}}  \triangleright  \SCnt{A_{{\mathrm{2}}}}  \SCsym{;}  \Gamma_{{\mathrm{2}}}  \SCsym{;}  \Phi  \SCsym{;}  \Gamma_{{\mathrm{3}}}  \vdash_\mathcal{L}  \SCnt{B}}
        \end{math}
      \end{center}

\item Case 4:
      \begin{center}
        \scriptsize
        \begin{math}
          \begin{array}{c}
            \Pi_1 \\
            {\Delta  \vdash_\mathcal{L}  \SCnt{B}}
          \end{array}
        \end{math}
        \qquad\qquad
        $\Pi_2$:
        \begin{math}
          $$\mprset{flushleft}
          \inferrule* [right={\tiny tenL}] {
            {
              \begin{array}{c}
                \pi \\
                {\Gamma_{{\mathrm{1}}}  \SCsym{;}  \SCnt{A_{{\mathrm{1}}}}  \SCsym{;}  \SCnt{A_{{\mathrm{2}}}}  \SCsym{;}  \Gamma_{{\mathrm{2}}}  \SCsym{;}  \SCnt{B}  \SCsym{;}  \Gamma_{{\mathrm{3}}}  \vdash_\mathcal{L}  \SCnt{C}}
              \end{array}
            }
          }{\Gamma_{{\mathrm{1}}}  \SCsym{;}  \SCnt{A_{{\mathrm{1}}}}  \triangleright  \SCnt{A_{{\mathrm{2}}}}  \SCsym{;}  \Gamma_{{\mathrm{2}}}  \SCsym{;}  \SCnt{B}  \SCsym{;}  \Gamma_{{\mathrm{3}}}  \vdash_\mathcal{L}  \SCnt{C}}
        \end{math}
      \end{center}
      By assumption, $c(\Pi_1),c(\Pi_2)\leq |B|$. By induction on $\Pi_1$
      and $\pi$, there is a proof $\Pi'$ for sequent
      $\Gamma_{{\mathrm{1}}}  \SCsym{;}  \SCnt{A_{{\mathrm{1}}}}  \SCsym{;}  \SCnt{A_{{\mathrm{2}}}}  \SCsym{;}  \Gamma_{{\mathrm{2}}}  \SCsym{;}  \Delta  \SCsym{;}  \Gamma_{{\mathrm{3}}}  \vdash_\mathcal{L}  \SCnt{C}$ s.t. $c(\Pi') \leq |B|$. Therefore,
      the proof $\Pi$ can be constructed as follows with
      $c(\Pi) = c(\Pi') \leq |B|$.
      \begin{center}
        \scriptsize
        \begin{math}
          $$\mprset{flushleft}
          \inferrule* [right={\tiny tenL}] {
            {
              \begin{array}{c}
                \Pi' \\
                {\Gamma_{{\mathrm{1}}}  \SCsym{;}  \SCnt{A_{{\mathrm{1}}}}  \SCsym{;}  \SCnt{A_{{\mathrm{2}}}}  \SCsym{;}  \Gamma_{{\mathrm{2}}}  \SCsym{;}  \Delta  \SCsym{;}  \Gamma_{{\mathrm{3}}}  \vdash_\mathcal{L}  \SCnt{C}}
              \end{array}
            }
          }{\Gamma_{{\mathrm{1}}}  \SCsym{;}  \SCnt{A_{{\mathrm{1}}}}  \triangleright  \SCnt{A_{{\mathrm{2}}}}  \SCsym{;}  \Gamma_{{\mathrm{2}}}  \SCsym{;}  \Delta  \SCsym{;}  \Gamma_{{\mathrm{3}}}  \vdash_\mathcal{L}  \SCnt{C}}
        \end{math}
      \end{center}
\end{itemize}

\subsubsection{$\SCdruleTXXexName$}
\begin{itemize}
\item Case 1:
      \begin{center}
        \scriptsize
        \begin{math}
          \begin{array}{c}
            \Pi_1 \\
            {\Phi  \vdash_\mathcal{C}  \SCnt{X}}
          \end{array}
        \end{math}
        \qquad\qquad
        $\Pi_2$:
        \begin{math}
          $$\mprset{flushleft}
          \inferrule* [right={\tiny beta}] {
            {
              \begin{array}{c}
                \pi \\
                {\Psi_{{\mathrm{1}}}  \SCsym{,}  \SCnt{X}  \SCsym{,}  \Psi_{{\mathrm{2}}}  \SCsym{,}  \SCnt{Y_{{\mathrm{1}}}}  \SCsym{,}  \SCnt{Y_{{\mathrm{2}}}}  \SCsym{,}  \Psi_{{\mathrm{3}}}  \vdash_\mathcal{C}  \SCnt{Z}}
              \end{array}
            }
          }{\Psi_{{\mathrm{1}}}  \SCsym{,}  \SCnt{X}  \SCsym{,}  \Psi_{{\mathrm{2}}}  \SCsym{,}  \SCnt{Y_{{\mathrm{2}}}}  \SCsym{,}  \SCnt{Y_{{\mathrm{1}}}}  \SCsym{,}  \Psi_{{\mathrm{3}}}  \vdash_\mathcal{C}  \SCnt{Z}}
        \end{math}
      \end{center}
      By assumption, $c(\Pi_1),c(\Pi_2)\leq |X|$. By induction on $\Pi_1$
      and $\pi$, there is a proof $\Pi'$ for sequent
      $\Psi_{{\mathrm{1}}}  \SCsym{,}  \Phi  \SCsym{,}  \Psi_{{\mathrm{2}}}  \SCsym{,}  \SCnt{Y_{{\mathrm{1}}}}  \SCsym{,}  \SCnt{Y_{{\mathrm{2}}}}  \SCsym{,}  \Psi_{{\mathrm{3}}}  \vdash_\mathcal{C}  \SCnt{Z}$ s.t. $c(\Pi') \leq |X|$. Therefore,
      the proof $\Pi$ can be constructed as follows with
      $c(\Pi) = c(\Pi') \leq |X|$.
      \begin{center}
        \scriptsize
        \begin{math}
          $$\mprset{flushleft}
          \inferrule* [right={\tiny cut}] {
            {
              \begin{array}{cc}
                \Pi' \\
                {\Psi_{{\mathrm{1}}}  \SCsym{,}  \Phi  \SCsym{,}  \Psi_{{\mathrm{2}}}  \SCsym{,}  \SCnt{Y_{{\mathrm{1}}}}  \SCsym{,}  \SCnt{Y_{{\mathrm{2}}}}  \SCsym{,}  \Psi_{{\mathrm{3}}}  \vdash_\mathcal{C}  \SCnt{Z}}
              \end{array}
            }
          }{\Psi_{{\mathrm{1}}}  \SCsym{,}  \Phi  \SCsym{,}  \Psi_{{\mathrm{2}}}  \SCsym{,}  \SCnt{Y_{{\mathrm{2}}}}  \SCsym{,}  \SCnt{Y_{{\mathrm{1}}}}  \SCsym{,}  \Psi_{{\mathrm{3}}}  \vdash_\mathcal{C}  \SCnt{Z}}
        \end{math}
      \end{center}

\item Case 2:
      \begin{center}
        \scriptsize
        \begin{math}
          \begin{array}{c}
            \Pi_1 \\
            {\Phi  \vdash_\mathcal{C}  \SCnt{X}}
          \end{array}
        \end{math}
        \qquad\qquad
        $\Pi_2$:
        \begin{math}
          $$\mprset{flushleft}
          \inferrule* [right={\tiny beta}] {
            {
              \begin{array}{c}
                \pi \\
                {\Psi_{{\mathrm{1}}}  \SCsym{,}  \SCnt{Y_{{\mathrm{1}}}}  \SCsym{,}  \SCnt{Y_{{\mathrm{2}}}}  \SCsym{,}  \Psi_{{\mathrm{2}}}  \SCsym{,}  \SCnt{X}  \SCsym{,}  \Psi_{{\mathrm{3}}}  \vdash_\mathcal{C}  \SCnt{Z}}
              \end{array}
            }
          }{\Psi_{{\mathrm{1}}}  \SCsym{,}  \SCnt{X}  \SCsym{,}  \Psi_{{\mathrm{2}}}  \SCsym{,}  \SCnt{Y_{{\mathrm{2}}}}  \SCsym{,}  \SCnt{Y_{{\mathrm{1}}}}  \SCsym{,}  \Psi_{{\mathrm{3}}}  \vdash_\mathcal{C}  \SCnt{Z}}
        \end{math}
      \end{center}
      By assumption, $c(\Pi_1),c(\Pi_2)\leq |X|$. By induction on $\Pi_1$
      and $\pi$, there is a proof $\Pi'$ for sequent
      $\Psi_{{\mathrm{1}}}  \SCsym{,}  \SCnt{Y_{{\mathrm{1}}}}  \SCsym{,}  \SCnt{Y_{{\mathrm{2}}}}  \SCsym{,}  \Psi_{{\mathrm{2}}}  \SCsym{,}  \Phi  \SCsym{,}  \Psi_{{\mathrm{3}}}  \vdash_\mathcal{C}  \SCnt{Z}$ s.t. $c(\Pi') \leq |X|$. Therefore,
      the proof $\Pi$ can be constructed as follows with
      $c(\Pi) = c(\Pi') \leq |X|$.
      \begin{center}
        \scriptsize
        \begin{math}
          $$\mprset{flushleft}
          \inferrule* [right={\tiny cut}] {
            {
              \begin{array}{cc}
                \Pi' \\
                {\Psi_{{\mathrm{1}}}  \SCsym{,}  \SCnt{Y_{{\mathrm{1}}}}  \SCsym{,}  \SCnt{Y_{{\mathrm{2}}}}  \SCsym{,}  \Psi_{{\mathrm{2}}}  \SCsym{,}  \Phi  \SCsym{,}  \Psi_{{\mathrm{3}}}  \vdash_\mathcal{C}  \SCnt{Z}}
              \end{array}
            }
          }{\Psi_{{\mathrm{1}}}  \SCsym{,}  \SCnt{Y_{{\mathrm{2}}}}  \SCsym{,}  \SCnt{Y_{{\mathrm{1}}}}  \SCsym{,}  \Psi_{{\mathrm{2}}}  \SCsym{,}  \Phi  \SCsym{,}  \Psi_{{\mathrm{3}}}  \vdash_\mathcal{C}  \SCnt{Z}}
        \end{math}
      \end{center}
\end{itemize}

\subsubsection{$\SCdruleSXXexName$}
\begin{itemize}
\item Case 1:
      \begin{center}
        \scriptsize
        \begin{math}
          \begin{array}{c}
            \Pi_1 \\
            {\Phi  \vdash_\mathcal{C}  \SCnt{X}}
          \end{array}
        \end{math}
        \qquad\qquad
        $\Pi_2$:
        \begin{math}
          $$\mprset{flushleft}
          \inferrule* [right={\tiny beta}] {
            {
              \begin{array}{c}
                \pi \\
                {\Gamma_{{\mathrm{1}}}  \SCsym{;}  \SCnt{X}  \SCsym{;}  \Gamma_{{\mathrm{2}}}  \SCsym{;}  \SCnt{Y_{{\mathrm{1}}}}  \SCsym{;}  \SCnt{Y_{{\mathrm{2}}}}  \SCsym{;}  \Gamma_{{\mathrm{3}}}  \vdash_\mathcal{L}  \SCnt{A}}
              \end{array}
            }
          }{\Gamma_{{\mathrm{1}}}  \SCsym{;}  \SCnt{X}  \SCsym{;}  \Gamma_{{\mathrm{2}}}  \SCsym{;}  \SCnt{Y_{{\mathrm{2}}}}  \SCsym{;}  \SCnt{Y_{{\mathrm{1}}}}  \SCsym{;}  \Gamma_{{\mathrm{3}}}  \vdash_\mathcal{L}  \SCnt{A}}
        \end{math}
      \end{center}
      By assumption, $c(\Pi_1),c(\Pi_2)\leq |X|$. By induction on $\Pi_1$
      and $\pi$, there is a proof $\Pi'$ for sequent
      $\Gamma_{{\mathrm{1}}}  \SCsym{;}  \Phi  \SCsym{;}  \Gamma_{{\mathrm{2}}}  \SCsym{;}  \SCnt{Y_{{\mathrm{1}}}}  \SCsym{;}  \SCnt{Y_{{\mathrm{2}}}}  \SCsym{;}  \Gamma_{{\mathrm{3}}}  \vdash_\mathcal{L}  \SCnt{A}$ s.t. $c(\Pi') \leq |X|$. Therefore,
      the proof $\Pi$ can be constructed as follows with
      $c(\Pi) = c(\Pi') \leq |X|$.
      \begin{center}
        \scriptsize
        \begin{math}
          $$\mprset{flushleft}
          \inferrule* [right={\tiny cut}] {
            {
              \begin{array}{cc}
                \Pi' \\
                {\Gamma_{{\mathrm{1}}}  \SCsym{;}  \Phi  \SCsym{;}  \Gamma_{{\mathrm{2}}}  \SCsym{;}  \SCnt{Y_{{\mathrm{1}}}}  \SCsym{;}  \SCnt{Y_{{\mathrm{2}}}}  \SCsym{;}  \Gamma_{{\mathrm{3}}}  \vdash_\mathcal{L}  \SCnt{A}}
              \end{array}
            }
          }{\Gamma_{{\mathrm{1}}}  \SCsym{;}  \Phi  \SCsym{;}  \Gamma_{{\mathrm{2}}}  \SCsym{;}  \SCnt{Y_{{\mathrm{2}}}}  \SCsym{;}  \SCnt{Y_{{\mathrm{1}}}}  \SCsym{;}  \Gamma_{{\mathrm{3}}}  \vdash_\mathcal{L}  \SCnt{A}}
        \end{math}
      \end{center}

\item Case 2:
      \begin{center}
        \scriptsize
        \begin{math}
          \begin{array}{c}
            \Pi_1 \\
            {\Delta  \vdash_\mathcal{L}  \SCnt{B}}
          \end{array}
        \end{math}
        \qquad\qquad
        $\Pi_2$:
        \begin{math}
          $$\mprset{flushleft}
          \inferrule* [right={\tiny beta}] {
            {
              \begin{array}{c}
                \pi \\
                {\Gamma_{{\mathrm{1}}}  \SCsym{;}  \SCnt{B}  \SCsym{;}  \Gamma_{{\mathrm{2}}}  \SCsym{;}  \SCnt{Y_{{\mathrm{1}}}}  \SCsym{;}  \SCnt{Y_{{\mathrm{2}}}}  \SCsym{;}  \Gamma_{{\mathrm{3}}}  \vdash_\mathcal{L}  \SCnt{A}}
              \end{array}
            }
          }{\Gamma_{{\mathrm{1}}}  \SCsym{;}  \SCnt{B}  \SCsym{;}  \Gamma_{{\mathrm{2}}}  \SCsym{;}  \SCnt{Y_{{\mathrm{2}}}}  \SCsym{;}  \SCnt{Y_{{\mathrm{1}}}}  \SCsym{;}  \Gamma_{{\mathrm{3}}}  \vdash_\mathcal{L}  \SCnt{A}}
        \end{math}
      \end{center}
      By assumption, $c(\Pi_1),c(\Pi_2)\leq |X|$. By induction on $\Pi_1$
      and $\pi$, there is a proof $\Pi'$ for sequent
      $\Gamma_{{\mathrm{1}}}  \SCsym{;}  \Delta  \SCsym{;}  \Gamma_{{\mathrm{2}}}  \SCsym{;}  \SCnt{Y_{{\mathrm{1}}}}  \SCsym{;}  \SCnt{Y_{{\mathrm{2}}}}  \SCsym{;}  \Gamma_{{\mathrm{3}}}  \vdash_\mathcal{L}  \SCnt{A}$ s.t. $c(\Pi') \leq |X|$. Therefore,
      the proof $\Pi$ can be constructed as follows with
      $c(\Pi) = c(\Pi') \leq |X|$.
      \begin{center}
        \scriptsize
        \begin{math}
          $$\mprset{flushleft}
          \inferrule* [right={\tiny cut}] {
            {
              \begin{array}{cc}
                \Pi' \\
                {\Gamma_{{\mathrm{1}}}  \SCsym{;}  \Delta  \SCsym{;}  \Gamma_{{\mathrm{2}}}  \SCsym{;}  \SCnt{Y_{{\mathrm{1}}}}  \SCsym{;}  \SCnt{Y_{{\mathrm{2}}}}  \SCsym{;}  \Gamma_{{\mathrm{3}}}  \vdash_\mathcal{L}  \SCnt{A}}
              \end{array}
            }
          }{\Gamma_{{\mathrm{1}}}  \SCsym{;}  \Delta  \SCsym{;}  \Gamma_{{\mathrm{2}}}  \SCsym{;}  \SCnt{Y_{{\mathrm{2}}}}  \SCsym{;}  \SCnt{Y_{{\mathrm{1}}}}  \SCsym{;}  \Gamma_{{\mathrm{3}}}  \vdash_\mathcal{L}  \SCnt{A}}
        \end{math}
      \end{center}

\item Case 3:
      \begin{center}
        \scriptsize
        \begin{math}
          \begin{array}{c}
            \Pi_1 \\
            {\Phi  \vdash_\mathcal{C}  \SCnt{X}}
          \end{array}
        \end{math}
        \qquad\qquad
        $\Pi_2$:
        \begin{math}
          $$\mprset{flushleft}
          \inferrule* [right={\tiny beta}] {
            {
              \begin{array}{c}
                \pi \\
                {\Gamma_{{\mathrm{1}}}  \SCsym{;}  \SCnt{Y_{{\mathrm{1}}}}  \SCsym{;}  \SCnt{Y_{{\mathrm{2}}}}  \SCsym{;}  \Gamma_{{\mathrm{2}}}  \SCsym{;}  \SCnt{X}  \SCsym{;}  \Gamma_{{\mathrm{3}}}  \vdash_\mathcal{L}  \SCnt{A}}
              \end{array}
            }
          }{\Gamma_{{\mathrm{1}}}  \SCsym{;}  \SCnt{X}  \SCsym{;}  \Gamma_{{\mathrm{2}}}  \SCsym{;}  \SCnt{Y_{{\mathrm{2}}}}  \SCsym{;}  \SCnt{Y_{{\mathrm{1}}}}  \SCsym{;}  \Gamma_{{\mathrm{3}}}  \vdash_\mathcal{L}  \SCnt{A}}
        \end{math}
      \end{center}
      By assumption, $c(\Pi_1),c(\Pi_2)\leq |X|$. By induction on $\Pi_1$
      and $\pi$, there is a proof $\Pi'$ for sequent
      $\Gamma_{{\mathrm{1}}}  \SCsym{;}  \SCnt{Y_{{\mathrm{1}}}}  \SCsym{;}  \SCnt{Y_{{\mathrm{2}}}}  \SCsym{;}  \Gamma_{{\mathrm{2}}}  \SCsym{;}  \Phi  \SCsym{;}  \Gamma_{{\mathrm{3}}}  \vdash_\mathcal{L}  \SCnt{A}$ s.t. $c(\Pi') \leq |X|$. Therefore,
      the proof $\Pi$ can be constructed as follows with
      $c(\Pi) = c(\Pi') \leq |X|$.
      \begin{center}
        \scriptsize
        \begin{math}
          $$\mprset{flushleft}
          \inferrule* [right={\tiny cut}] {
            {
              \begin{array}{cc}
                \Pi' \\
                {\Gamma_{{\mathrm{1}}}  \SCsym{;}  \SCnt{Y_{{\mathrm{1}}}}  \SCsym{;}  \SCnt{Y_{{\mathrm{2}}}}  \SCsym{;}  \Gamma_{{\mathrm{2}}}  \SCsym{;}  \Phi  \SCsym{;}  \Gamma_{{\mathrm{3}}}  \vdash_\mathcal{L}  \SCnt{A}}
              \end{array}
            }
          }{\Gamma_{{\mathrm{1}}}  \SCsym{;}  \SCnt{Y_{{\mathrm{2}}}}  \SCsym{;}  \SCnt{Y_{{\mathrm{1}}}}  \SCsym{;}  \Gamma_{{\mathrm{2}}}  \SCsym{;}  \Phi  \SCsym{;}  \Gamma_{{\mathrm{3}}}  \vdash_\mathcal{L}  \SCnt{A}}
        \end{math}
      \end{center}

\item Case 4:
      \begin{center}
        \scriptsize
        \begin{math}
          \begin{array}{c}
            \Pi_1 \\
            {\Delta  \vdash_\mathcal{L}  \SCnt{B}}
          \end{array}
        \end{math}
        \qquad\qquad
        $\Pi_2$:
        \begin{math}
          $$\mprset{flushleft}
          \inferrule* [right={\tiny beta}] {
            {
              \begin{array}{c}
                \pi \\
                {\Gamma_{{\mathrm{1}}}  \SCsym{;}  \SCnt{Y_{{\mathrm{1}}}}  \SCsym{;}  \SCnt{Y_{{\mathrm{2}}}}  \SCsym{;}  \Gamma_{{\mathrm{2}}}  \SCsym{;}  \SCnt{B}  \SCsym{;}  \Gamma_{{\mathrm{3}}}  \vdash_\mathcal{L}  \SCnt{A}}
              \end{array}
            }
          }{\Gamma_{{\mathrm{1}}}  \SCsym{;}  \SCnt{Y_{{\mathrm{2}}}}  \SCsym{;}  \SCnt{Y_{{\mathrm{1}}}}  \SCsym{;}  \Gamma_{{\mathrm{2}}}  \SCsym{;}  \SCnt{B}  \SCsym{;}  \Gamma_{{\mathrm{3}}}  \vdash_\mathcal{L}  \SCnt{A}}
        \end{math}
      \end{center}
      By assumption, $c(\Pi_1),c(\Pi_2)\leq |X|$. By induction on $\Pi_1$
      and $\pi$, there is a proof $\Pi'$ for sequent
      $\Gamma_{{\mathrm{1}}}  \SCsym{;}  \SCnt{Y_{{\mathrm{1}}}}  \SCsym{;}  \SCnt{Y_{{\mathrm{2}}}}  \SCsym{;}  \Gamma_{{\mathrm{2}}}  \SCsym{;}  \Delta  \SCsym{;}  \Gamma_{{\mathrm{3}}}  \vdash_\mathcal{L}  \SCnt{A}$ s.t. $c(\Pi') \leq |X|$. Therefore,
      the proof $\Pi$ can be constructed as follows with
      $c(\Pi) = c(\Pi') \leq |X|$.
      \begin{center}
        \scriptsize
        \begin{math}
          $$\mprset{flushleft}
          \inferrule* [right={\tiny cut}] {
            {
              \begin{array}{cc}
                \Pi' \\
                {\Gamma_{{\mathrm{1}}}  \SCsym{;}  \SCnt{Y_{{\mathrm{1}}}}  \SCsym{;}  \SCnt{Y_{{\mathrm{2}}}}  \SCsym{;}  \Gamma_{{\mathrm{2}}}  \SCsym{;}  \Delta  \SCsym{;}  \Gamma_{{\mathrm{3}}}  \vdash_\mathcal{L}  \SCnt{A}}
              \end{array}
            }
          }{\Gamma_{{\mathrm{1}}}  \SCsym{;}  \SCnt{Y_{{\mathrm{2}}}}  \SCsym{;}  \SCnt{Y_{{\mathrm{1}}}}  \SCsym{;}  \Gamma_{{\mathrm{2}}}  \SCsym{;}  \Delta  \SCsym{;}  \Gamma_{{\mathrm{3}}}  \vdash_\mathcal{L}  \SCnt{A}}
        \end{math}
      \end{center}
\end{itemize}

\subsubsection{Left introduction of the commutative unit $ \mathsf{Unit} $ (with low priority)}
\begin{itemize}
\item Case 1:
      \begin{center}
        \scriptsize
        \begin{math}
          \begin{array}{c}
            \Pi_1 \\
            {\Psi  \vdash_\mathcal{C}  \SCnt{X}}
          \end{array}
        \end{math}
        \qquad\qquad
        $\Pi_2$:
        \begin{math}
          $$\mprset{flushleft}
          \inferrule* [right={\tiny unitL}] {
            {
              \begin{array}{c}
                \pi \\
                {\Phi_{{\mathrm{1}}}  \SCsym{,}  \Phi_{{\mathrm{2}}}  \SCsym{,}  \SCnt{X}  \SCsym{,}  \Phi_{{\mathrm{3}}}  \vdash_\mathcal{C}  \SCnt{Y}}
              \end{array}
            }
          }{\Phi_{{\mathrm{1}}}  \SCsym{,}   \mathsf{Unit}   \SCsym{,}  \Phi_{{\mathrm{2}}}  \SCsym{,}  \SCnt{X}  \SCsym{,}  \Phi_{{\mathrm{3}}}  \vdash_\mathcal{C}  \SCnt{Y}}
        \end{math}
      \end{center}
      By assumption, $c(\Pi_1),c(\Pi_2)\leq |X|$. By induction on $\Pi_1$
      and $\pi$, there is a proof $\Pi'$ for sequent
      $\Phi_{{\mathrm{1}}}  \SCsym{,}  \Phi_{{\mathrm{2}}}  \SCsym{,}  \Psi  \SCsym{,}  \Phi_{{\mathrm{3}}}  \vdash_\mathcal{C}  \SCnt{Y}$
      s.t. $c(\Pi') \leq |X|$. Therefore, the proof $\Pi$ can be
      constructed as follows with $c(\Pi) = c(\Pi') \leq |X|$.
      \begin{center}
        \scriptsize
        \begin{math}
          $$\mprset{flushleft}
          \inferrule* [right={\tiny unitL}] {
            {
              \begin{array}{c}
                \Pi' \\
                {\Phi_{{\mathrm{1}}}  \SCsym{,}  \Phi_{{\mathrm{2}}}  \SCsym{,}  \Psi  \SCsym{,}  \Phi_{{\mathrm{3}}}  \vdash_\mathcal{C}  \SCnt{Y}}
              \end{array}
            }
          }{\Phi_{{\mathrm{1}}}  \SCsym{,}   \mathsf{Unit}   \SCsym{,}  \Phi_{{\mathrm{2}}}  \SCsym{,}  \Psi  \SCsym{,}  \Phi_{{\mathrm{3}}}  \vdash_\mathcal{C}  \SCnt{Y}}
        \end{math}
      \end{center}

\item Case 2:
      \begin{center}
        \scriptsize
        \begin{math}
          \begin{array}{c}
            \Pi_1 \\
            {\Phi  \vdash_\mathcal{C}  \SCnt{X}}
          \end{array}
        \end{math}
        \qquad\qquad
        $\Pi_2$:
        \begin{math}
          $$\mprset{flushleft}
          \inferrule* [right={\tiny unitL1}] {
            {
              \begin{array}{c}
                \pi \\
                {\Gamma_{{\mathrm{1}}}  \SCsym{;}  \Gamma_{{\mathrm{2}}}  \SCsym{;}  \SCnt{X}  \SCsym{;}  \Gamma_{{\mathrm{3}}}  \vdash_\mathcal{L}  \SCnt{A}}
              \end{array}
            }
          }{\Gamma_{{\mathrm{1}}}  \SCsym{;}   \mathsf{Unit}   \SCsym{;}  \Gamma_{{\mathrm{2}}}  \SCsym{;}  \SCnt{X}  \SCsym{;}  \Gamma_{{\mathrm{3}}}  \vdash_\mathcal{L}  \SCnt{A}}
        \end{math}
      \end{center}
      By assumption, $c(\Pi_1),c(\Pi_2)\leq |X|$. By induction on $\Pi_1$
      and $\pi$, there is a proof $\Pi'$ for sequent
      $\Gamma_{{\mathrm{1}}}  \SCsym{;}  \Gamma_{{\mathrm{2}}}  \SCsym{;}  \Phi  \SCsym{;}  \Gamma_{{\mathrm{2}}}  \vdash_\mathcal{L}  \SCnt{A}$
      s.t. $c(\Pi') \leq |X|$. Therefore, the proof $\Pi$ can be
      constructed as follows with $c(\Pi) = c(\Pi') \leq |X|$.
      \begin{center}
        \scriptsize
        \begin{math}
          $$\mprset{flushleft}
          \inferrule* [right={\tiny unitL1}] {
            {
              \begin{array}{c}
                \Pi' \\
                {\Gamma_{{\mathrm{1}}}  \SCsym{;}  \Gamma_{{\mathrm{2}}}  \SCsym{;}  \Phi  \SCsym{;}  \Gamma_{{\mathrm{3}}}  \vdash_\mathcal{L}  \SCnt{A}}
              \end{array}
            }
          }{\Gamma_{{\mathrm{1}}}  \SCsym{;}   \mathsf{Unit}   \SCsym{;}  \Gamma_{{\mathrm{2}}}  \SCsym{;}  \Phi  \SCsym{;}  \Gamma_{{\mathrm{3}}}  \vdash_\mathcal{L}  \SCnt{A}}
        \end{math}
      \end{center}

\item Case 3:
      \begin{center}
        \scriptsize
        \begin{math}
          \begin{array}{c}
            \Pi_1 \\
            {\Delta  \vdash_\mathcal{L}  \SCnt{B}}
          \end{array}
        \end{math}
        \qquad\qquad
        $\Pi_2$:
        \begin{math}
          $$\mprset{flushleft}
          \inferrule* [right={\tiny unitL1}] {
            {
              \begin{array}{c}
                \pi \\
                {\Gamma_{{\mathrm{1}}}  \SCsym{;}  \Gamma_{{\mathrm{2}}}  \SCsym{;}  \SCnt{B}  \SCsym{;}  \Gamma_{{\mathrm{3}}}  \vdash_\mathcal{L}  \SCnt{A}}
              \end{array}
            }
          }{\Gamma_{{\mathrm{1}}}  \SCsym{;}   \mathsf{Unit}   \SCsym{;}  \Gamma_{{\mathrm{2}}}  \SCsym{;}  \SCnt{B}  \SCsym{;}  \Gamma_{{\mathrm{3}}}  \vdash_\mathcal{L}  \SCnt{A}}
        \end{math}
      \end{center}
      By assumption, $c(\Pi_1),c(\Pi_2)\leq |B|$. By induction on $\Pi_1$
      and $\pi$, there is a proof $\Pi'$ for sequent
      $\Gamma_{{\mathrm{1}}}  \SCsym{;}  \Gamma_{{\mathrm{2}}}  \SCsym{;}  \Delta  \SCsym{;}  \Gamma_{{\mathrm{3}}}  \vdash_\mathcal{L}  \SCnt{A}$
      s.t. $c(\Pi') \leq |B|$. Therefore, the proof $\Pi$ can be
      constructed as follows with $c(\Pi) = c(\Pi') \leq |B|$.
      \begin{center}
        \scriptsize
        \begin{math}
          $$\mprset{flushleft}
          \inferrule* [right={\tiny unitL1}] {
            {
              \begin{array}{c}
                \Pi' \\
                {\Gamma_{{\mathrm{1}}}  \SCsym{;}  \Gamma_{{\mathrm{2}}}  \SCsym{;}  \Delta  \SCsym{;}  \Gamma_{{\mathrm{3}}}  \vdash_\mathcal{L}  \SCnt{A}}
              \end{array}
            }
          }{\Gamma_{{\mathrm{1}}}  \SCsym{;}   \mathsf{Unit}   \SCsym{;}  \Gamma_{{\mathrm{2}}}  \SCsym{;}  \Delta  \SCsym{;}  \Gamma_{{\mathrm{3}}}  \vdash_\mathcal{L}  \SCnt{A}}
        \end{math}
      \end{center}
\end{itemize}

\subsubsection{Left introduction of the non-commutative unit $ \mathsf{Unit} $ (with low priority)}
\begin{itemize}
\item Case 1:
      \begin{center}
        \scriptsize
        \begin{math}
          \begin{array}{c}
            \Pi_1 \\
            {\Phi  \vdash_\mathcal{C}  \SCnt{X}}
          \end{array}
        \end{math}
        \qquad\qquad
        $\Pi_2$:
        \begin{math}
          $$\mprset{flushleft}
          \inferrule* [right={\tiny unitL2}] {
            {
              \begin{array}{c}
                \pi \\
                {\Gamma_{{\mathrm{1}}}  \SCsym{;}  \Gamma_{{\mathrm{2}}}  \SCsym{;}  \SCnt{X}  \SCsym{;}  \Gamma_{{\mathrm{3}}}  \vdash_\mathcal{L}  \SCnt{A}}
              \end{array}
            }
          }{\Gamma_{{\mathrm{1}}}  \SCsym{;}   \mathsf{Unit}   \SCsym{;}  \Gamma_{{\mathrm{2}}}  \SCsym{;}  \SCnt{X}  \SCsym{;}  \Gamma_{{\mathrm{3}}}  \vdash_\mathcal{L}  \SCnt{A}}
        \end{math}
      \end{center}
      By assumption, $c(\Pi_1),c(\Pi_2)\leq |X|$. By induction on $\Pi_1$
      and $\pi$, there is a proof $\Pi'$ for sequent
      $\Gamma_{{\mathrm{1}}}  \SCsym{;}  \Gamma_{{\mathrm{2}}}  \SCsym{;}  \Phi  \SCsym{;}  \Gamma_{{\mathrm{3}}}  \vdash_\mathcal{L}  \SCnt{A}$
      s.t. $c(\Pi') \leq |X|$. Therefore, the proof $\Pi$ can be
      constructed as follows with $c(\Pi) = c(\Pi') \leq |X|$.
      \begin{center}
        \scriptsize
        \begin{math}
          $$\mprset{flushleft}
          \inferrule* [right={\tiny unitL2}] {
            {
              \begin{array}{c}
                \Pi' \\
                {\Gamma_{{\mathrm{1}}}  \SCsym{;}  \Gamma_{{\mathrm{2}}}  \SCsym{;}  \Phi  \SCsym{;}  \Gamma_{{\mathrm{3}}}  \vdash_\mathcal{L}  \SCnt{A}}
              \end{array}
            }
          }{\Gamma_{{\mathrm{1}}}  \SCsym{;}   \mathsf{Unit}   \SCsym{;}  \Gamma_{{\mathrm{2}}}  \SCsym{;}  \Phi  \SCsym{;}  \Gamma_{{\mathrm{3}}}  \vdash_\mathcal{L}  \SCnt{A}}
        \end{math}
      \end{center}

\item Case 2:
      \begin{center}
        \scriptsize
        \begin{math}
          \begin{array}{c}
            \Pi_1 \\
            {\Delta  \vdash_\mathcal{L}  \SCnt{B}}
          \end{array}
        \end{math}
        \qquad\qquad
        $\Pi_2$:
        \begin{math}
          $$\mprset{flushleft}
          \inferrule* [right={\tiny unitL2}] {
            {
              \begin{array}{c}
                \pi \\
                {\Gamma_{{\mathrm{1}}}  \SCsym{;}  \Gamma_{{\mathrm{2}}}  \SCsym{;}  \SCnt{B}  \SCsym{;}  \Gamma_{{\mathrm{3}}}  \vdash_\mathcal{L}  \SCnt{A}}
              \end{array}
            }
          }{\Gamma_{{\mathrm{1}}}  \SCsym{;}   \mathsf{Unit}   \SCsym{;}  \Gamma_{{\mathrm{2}}}  \SCsym{;}  \SCnt{B}  \SCsym{;}  \Gamma_{{\mathrm{3}}}  \vdash_\mathcal{L}  \SCnt{A}}
        \end{math}
      \end{center}
      By assumption, $c(\Pi_1),c(\Pi_2)\leq |B|$. By induction on $\Pi_1$
      and $\pi$, there is a proof $\Pi'$ for sequent
      $\Gamma_{{\mathrm{1}}}  \SCsym{;}  \Gamma_{{\mathrm{2}}}  \SCsym{;}  \Delta  \SCsym{;}  \Gamma_{{\mathrm{3}}}  \vdash_\mathcal{L}  \SCnt{A}$
      s.t. $c(\Pi') \leq |B|$. Therefore, the proof $\Pi$ can be
      constructed as follows with $c(\Pi) = c(\Pi') \leq |B|$.
      \begin{center}
        \scriptsize
        \begin{math}
          $$\mprset{flushleft}
          \inferrule* [right={\tiny unitL2}] {
            {
              \begin{array}{c}
                \Pi' \\
                {\Gamma_{{\mathrm{1}}}  \SCsym{;}  \Gamma_{{\mathrm{2}}}  \SCsym{;}  \Delta  \SCsym{;}  \Gamma_{{\mathrm{3}}}  \vdash_\mathcal{L}  \SCnt{A}}
              \end{array}
            }
          }{\Gamma_{{\mathrm{1}}}  \SCsym{;}   \mathsf{Unit}   \SCsym{;}  \Gamma_{{\mathrm{2}}}  \SCsym{;}  \Delta  \SCsym{;}  \Gamma_{{\mathrm{3}}}  \vdash_\mathcal{L}  \SCnt{A}}
        \end{math}
      \end{center}
\end{itemize}

\subsubsection{Right introduction of the commutative implication $\multimap$ (with low priority)}
\begin{center}
  \scriptsize
  \begin{math}
    \begin{array}{c}
      \Pi_1 \\
      {\Phi  \vdash_\mathcal{C}  \SCnt{X}}
    \end{array}
  \end{math}
  \qquad\qquad
  $\Pi_2$:
  \begin{math}
    $$\mprset{flushleft}
    \inferrule* [right={\tiny impR}] {
      {
        \begin{array}{c}
          \pi \\
          {\Psi_{{\mathrm{1}}}  \SCsym{,}  \SCnt{X}  \SCsym{,}  \Psi_{{\mathrm{2}}}  \SCsym{,}  \SCnt{Y_{{\mathrm{1}}}}  \vdash_\mathcal{C}  \SCnt{Y_{{\mathrm{2}}}}}
        \end{array}
      }
    }{\Psi_{{\mathrm{1}}}  \SCsym{,}  \SCnt{X}  \SCsym{,}  \Psi_{{\mathrm{2}}}  \vdash_\mathcal{C}  \SCnt{Y_{{\mathrm{1}}}}  \multimap  \SCnt{Y_{{\mathrm{2}}}}}
  \end{math}
\end{center}
By assumption, $c(\Pi_1),c(\Pi_2)\leq |X|$. By induction on $\Pi_1$
and $\pi$, there is a proof $\Pi'$ for sequent
\begin{center}
  \begin{math}
    \Psi_{{\mathrm{1}}}  \SCsym{,}  \Phi  \SCsym{,}  \Psi_{{\mathrm{2}}}  \SCsym{,}  \SCnt{Y_{{\mathrm{1}}}}  \vdash_\mathcal{C}  \SCnt{Y_{{\mathrm{2}}}}$ s.t. $c(\Pi') \leq |X|
  \end{math}
\end{center}
Therefore, the proof $\Pi$ can be constructed as follows with $c(\Pi)
= c(\Pi') \leq |X|$.
\begin{center}
  \scriptsize
  \begin{math}
    $$\mprset{flushleft}
    \inferrule* [right={\tiny impR}] {
      {
        \begin{array}{c}
          \Pi' \\
          {\Psi_{{\mathrm{1}}}  \SCsym{,}  \Phi  \SCsym{,}  \Psi_{{\mathrm{2}}}  \SCsym{,}  \SCnt{Y_{{\mathrm{1}}}}  \vdash_\mathcal{C}  \SCnt{Y_{{\mathrm{2}}}}}
        \end{array}
      }
    }{\Psi_{{\mathrm{1}}}  \SCsym{,}  \Phi  \SCsym{,}  \Psi_{{\mathrm{2}}}  \vdash_\mathcal{C}  \SCnt{Y_{{\mathrm{1}}}}  \multimap  \SCnt{Y_{{\mathrm{2}}}}}
  \end{math}
\end{center}

\subsubsection{Right introduction of the non-commutative left implication $\lto$ (with low priority)}
\begin{itemize}
\item Case 1:
      \begin{center}
        \scriptsize
        \begin{math}
          \begin{array}{c}
            \Pi_1 \\
            {\Phi  \vdash_\mathcal{C}  \SCnt{X}}
          \end{array}
        \end{math}
        \qquad\qquad
        $\Pi_2$:
        \begin{math}
          $$\mprset{flushleft}
          \inferrule* [right={\tiny impR}] {
            {
              \begin{array}{c}
                \pi \\
                {\Gamma_{{\mathrm{1}}}  \SCsym{;}  \SCnt{X}  \SCsym{;}  \Gamma_{{\mathrm{2}}}  \SCsym{;}  \SCnt{A}  \vdash_\mathcal{L}  \SCnt{B}}
              \end{array}
            }
          }{\Gamma_{{\mathrm{1}}}  \SCsym{;}  \SCnt{X}  \SCsym{;}  \Gamma_{{\mathrm{2}}}  \vdash_\mathcal{L}  \SCnt{A}  \rightharpoonup  \SCnt{B}}
        \end{math}
      \end{center}
      By assumption, $c(\Pi_1),c(\Pi_2)\leq |X|$. By induction on $\Pi_1$
      and $\pi$, there is a proof $\Pi'$ for sequent
      $\Gamma_{{\mathrm{1}}}  \SCsym{;}  \Phi  \SCsym{;}  \Gamma_{{\mathrm{2}}}  \SCsym{;}  \SCnt{A}  \vdash_\mathcal{L}  \SCnt{B}$ s.t. $c(\Pi') \leq |X|$. Therefore, the
      proof $\Pi$ can be constructed as follows with
      $c(\Pi) = c(\Pi') \leq |X|$.
      \begin{center}
        \scriptsize
        \begin{math}
          $$\mprset{flushleft}
          \inferrule* [right={\tiny implR}] {
            {
              \begin{array}{c}
                \Pi' \\
                {\Gamma_{{\mathrm{1}}}  \SCsym{;}  \Phi  \SCsym{;}  \Gamma_{{\mathrm{2}}}  \SCsym{;}  \SCnt{A}  \vdash_\mathcal{L}  \SCnt{B}}
              \end{array}
            }
          }{\Gamma_{{\mathrm{1}}}  \SCsym{;}  \Phi  \SCsym{;}  \Gamma_{{\mathrm{2}}}  \vdash_\mathcal{L}  \SCnt{A}  \rightharpoonup  \SCnt{B}}
        \end{math}
      \end{center}

\item Case 2:
      \begin{center}
        \scriptsize
        \begin{math}
          \begin{array}{c}
            \Pi_1 \\
            {\Delta  \vdash_\mathcal{L}  \SCnt{C}}
          \end{array}
        \end{math}
        \qquad\qquad
        $\Pi_2$:
        \begin{math}
          $$\mprset{flushleft}
          \inferrule* [right={\tiny impR}] {
            {
              \begin{array}{c}
                \pi \\
                {\Gamma_{{\mathrm{1}}}  \SCsym{;}  \SCnt{C}  \SCsym{;}  \Gamma_{{\mathrm{2}}}  \SCsym{;}  \SCnt{A}  \vdash_\mathcal{L}  \SCnt{B}}
              \end{array}
            }
          }{\Gamma_{{\mathrm{1}}}  \SCsym{;}  \SCnt{C}  \SCsym{;}  \Gamma_{{\mathrm{2}}}  \vdash_\mathcal{L}  \SCnt{A}  \rightharpoonup  \SCnt{B}}
        \end{math}
      \end{center}
      By assumption, $c(\Pi_1),c(\Pi_2)\leq |C|$. By induction on $\Pi_1$
      and $\pi$, there is a proof $\Pi'$ for sequent
      $\Gamma_{{\mathrm{1}}}  \SCsym{;}  \Delta  \SCsym{;}  \Gamma_{{\mathrm{2}}}  \SCsym{;}  \SCnt{A}  \vdash_\mathcal{L}  \SCnt{B}$ s.t. $c(\Pi') \leq |C|$. Therefore, the
      proof $\Pi$ can be constructed as follows with
      $c(\Pi) = c(\Pi') \leq |C|$.
      \begin{center}
        \scriptsize
        \begin{math}
          $$\mprset{flushleft}
          \inferrule* [right={\tiny implR}] {
            {
              \begin{array}{c}
                \Pi' \\
                {\Gamma_{{\mathrm{1}}}  \SCsym{;}  \Delta  \SCsym{;}  \Gamma_{{\mathrm{2}}}  \SCsym{;}  \SCnt{A}  \vdash_\mathcal{L}  \SCnt{B}}
              \end{array}
            }
          }{\Gamma_{{\mathrm{1}}}  \SCsym{;}  \Delta  \SCsym{;}  \Gamma_{{\mathrm{2}}}  \vdash_\mathcal{L}  \SCnt{A}  \rightharpoonup  \SCnt{B}}
        \end{math}
      \end{center}
\end{itemize}

\subsubsection{Right introduction of the non-commutative right implication $\rto$ (with low priority)}
\begin{itemize}
\item Case 1:
      \begin{center}
        \scriptsize
        \begin{math}
          \begin{array}{c}
            \Pi_1 \\
            {\Phi  \vdash_\mathcal{C}  \SCnt{X}}
          \end{array}
        \end{math}
        \qquad\qquad
        $\Pi_2$:
        \begin{math}
          $$\mprset{flushleft}
          \inferrule* [right={\tiny impL}] {
            {
              \begin{array}{c}
                \pi \\
                {\SCnt{A}  \SCsym{;}  \Gamma_{{\mathrm{1}}}  \SCsym{;}  \SCnt{X}  \SCsym{;}  \Gamma_{{\mathrm{2}}}  \vdash_\mathcal{L}  \SCnt{B}}
              \end{array}
            }
          }{\Gamma_{{\mathrm{1}}}  \SCsym{;}  \SCnt{X}  \SCsym{;}  \Gamma_{{\mathrm{2}}}  \vdash_\mathcal{L}  \SCnt{B}  \leftharpoonup  \SCnt{A}}
        \end{math}
      \end{center}
      By assumption, $c(\Pi_1),c(\Pi_2)\leq |X|$. By induction on $\Pi_1$
      and $\pi$, there is a proof $\Pi'$ for sequent
      $\SCnt{A}  \SCsym{;}  \Gamma_{{\mathrm{1}}}  \SCsym{;}  \Phi  \SCsym{;}  \Gamma_{{\mathrm{2}}}  \vdash_\mathcal{L}  \SCnt{B}$ s.t. $c(\Pi') \leq |X|$. Therefore, the
      proof $\Pi$ can be constructed as follows with
      $c(\Pi) = c(\Pi') \leq |X|$.
      \begin{center}
        \scriptsize
        \begin{math}
          $$\mprset{flushleft}
          \inferrule* [right={\tiny impR}] {
            {
              \begin{array}{c}
                \Pi' \\
                {\SCnt{A}  \SCsym{;}  \Gamma_{{\mathrm{1}}}  \SCsym{;}  \Phi  \SCsym{;}  \Gamma_{{\mathrm{2}}}  \vdash_\mathcal{L}  \SCnt{B}}
              \end{array}
            }
          }{\Gamma_{{\mathrm{1}}}  \SCsym{;}  \Phi  \SCsym{;}  \Gamma_{{\mathrm{2}}}  \vdash_\mathcal{L}  \SCnt{B}  \leftharpoonup  \SCnt{A}}
        \end{math}
      \end{center}

\item Case 2:
      \begin{center}
        \scriptsize
        \begin{math}
          \begin{array}{c}
            \Pi_1 \\
            {\Delta  \vdash_\mathcal{L}  \SCnt{C}}
          \end{array}
        \end{math}
        \qquad\qquad
        $\Pi_2$:
        \begin{math}
          $$\mprset{flushleft}
          \inferrule* [right={\tiny impR}] {
            {
              \begin{array}{c}
                \pi \\
                {\SCnt{A}  \SCsym{;}  \Gamma_{{\mathrm{1}}}  \SCsym{;}  \SCnt{C}  \SCsym{;}  \Gamma_{{\mathrm{2}}}  \vdash_\mathcal{L}  \SCnt{B}}
              \end{array}
            }
          }{\Gamma_{{\mathrm{1}}}  \SCsym{;}  \SCnt{C}  \SCsym{;}  \Gamma_{{\mathrm{2}}}  \vdash_\mathcal{L}  \SCnt{B}  \leftharpoonup  \SCnt{A}}
        \end{math}
      \end{center}
      By assumption, $c(\Pi_1),c(\Pi_2)\leq |C|$. By induction on $\Pi_1$
      and $\pi$, there is a proof $\Pi'$ for sequent
      $\Gamma_{{\mathrm{1}}}  \SCsym{;}  \Delta  \SCsym{;}  \Gamma_{{\mathrm{2}}}  \SCsym{;}  \SCnt{A}  \vdash_\mathcal{L}  \SCnt{B}$ s.t. $c(\Pi') \leq |C|$. Therefore, the
      proof $\Pi$ can be constructed as follows with
      $c(\Pi) = c(\Pi') \leq |C|$.
      \begin{center}
        \scriptsize
        \begin{math}
          $$\mprset{flushleft}
          \inferrule* [right={\tiny impR}] {
            {
              \begin{array}{c}
                \Pi' \\
                {\SCnt{A}  \SCsym{;}  \Gamma_{{\mathrm{1}}}  \SCsym{;}  \Delta  \SCsym{;}  \Gamma_{{\mathrm{2}}}  \vdash_\mathcal{L}  \SCnt{B}}
              \end{array}
            }
          }{\Gamma_{{\mathrm{1}}}  \SCsym{;}  \Delta  \SCsym{;}  \Gamma_{{\mathrm{2}}}  \vdash_\mathcal{L}  \SCnt{B}  \leftharpoonup  \SCnt{A}}
        \end{math}
      \end{center}
\end{itemize}

\subsubsection{Right introduction of the functor $F$}
\begin{center}
  \scriptsize
  \begin{math}
    \begin{array}{c}
      \Pi_1 \\
      {\Phi  \vdash_\mathcal{C}  \SCnt{X}}
    \end{array}
  \end{math}
  \qquad\qquad
  $\Pi_2$:
  \begin{math}
    $$\mprset{flushleft}
    \inferrule* [right={\tiny Fr}] {
      {
        \begin{array}{c}
          \pi \\
          {\Psi_{{\mathrm{1}}}  \SCsym{,}  \SCnt{X}  \SCsym{,}  \Psi_{{\mathrm{2}}}  \vdash_\mathcal{C}  \SCnt{Y}}
        \end{array}
      }
    }{\Psi_{{\mathrm{1}}}  \SCsym{,}  \SCnt{X}  \SCsym{,}  \Psi_{{\mathrm{2}}}  \vdash_\mathcal{L}   \mathsf{F} \SCnt{Y} }
  \end{math}
\end{center}
By assumption, $c(\Pi_1),c(\Pi_2)\leq |X|$. By induction on $\Pi_1$
and $\pi$, there is a proof $\Pi'$ for sequent $\Psi_{{\mathrm{1}}}  \SCsym{,}  \Phi  \SCsym{,}  \Psi_{{\mathrm{2}}}  \vdash_\mathcal{C}  \SCnt{Y}$
s.t. $c(\Pi') \leq |X|$. Therefore, the proof $\Pi$ can be
constructed as follows with $c(\Pi) = c(\Pi') \leq |X|$.
\begin{center}
  \scriptsize
  \begin{math}
    $$\mprset{flushleft}
    \inferrule* [right={\tiny Fr}] {
      {
        \begin{array}{c}
          \Pi' \\
          {\Psi_{{\mathrm{1}}}  \SCsym{,}  \Phi  \SCsym{,}  \Psi_{{\mathrm{2}}}  \vdash_\mathcal{C}  \SCnt{Y}}
        \end{array}
      }
    }{\Psi_{{\mathrm{1}}}  \SCsym{,}  \Phi  \SCsym{,}  \Psi_{{\mathrm{2}}}  \vdash_\mathcal{L}   \mathsf{F} \SCnt{Y} }
  \end{math}
\end{center}

\subsubsection{Left introduction of the functor $F$ (with low priority)}
\begin{itemize}
\item Case 1:
      \begin{center}
        \scriptsize
        \begin{math}
          \begin{array}{c}
            \Pi_1 \\
            {\Phi  \vdash_\mathcal{C}  \SCnt{X}}
          \end{array}
        \end{math}
        \qquad\qquad
        $\Pi_2$:
        \begin{math}
          $$\mprset{flushleft}
          \inferrule* [right={\tiny Fl}] {
            {
              \begin{array}{c}
                \pi \\
                {\Gamma_{{\mathrm{1}}}  \SCsym{;}  \SCnt{X}  \SCsym{;}  \Gamma_{{\mathrm{2}}}  \SCsym{;}  \SCnt{Y}  \SCsym{;}  \Gamma_{{\mathrm{3}}}  \vdash_\mathcal{L}  \SCnt{A}}
              \end{array}
            }
          }{\Gamma_{{\mathrm{1}}}  \SCsym{;}  \SCnt{X}  \SCsym{;}  \Gamma_{{\mathrm{2}}}  \SCsym{;}   \mathsf{F} \SCnt{Y}   \SCsym{;}  \Gamma_{{\mathrm{3}}}  \vdash_\mathcal{L}  \SCnt{A}}
        \end{math}
      \end{center}
      By assumption, $c(\Pi_1),c(\Pi_2)\leq |X|$. By induction on $\Pi_1$
      and $\pi$, there is a proof $\Pi'$ for sequent
      $\Gamma_{{\mathrm{1}}}  \SCsym{;}  \Phi  \SCsym{;}  \Gamma_{{\mathrm{2}}}  \SCsym{;}  \SCnt{Y}  \SCsym{;}  \Gamma_{{\mathrm{3}}}  \vdash_\mathcal{L}  \SCnt{A}$ s.t. $c(\Pi') \leq |X|$. Therefore, the
      proof $\Pi$ can be constructed as follows with
      $c(\Pi) = c(\Pi') \leq |X|$.
      \begin{center}
        \scriptsize
        \begin{math}
          $$\mprset{flushleft}
          \inferrule* [right={\tiny Fl}] {
            {
              \begin{array}{c}
                \Pi' \\
                {\Gamma_{{\mathrm{1}}}  \SCsym{;}  \Phi  \SCsym{;}  \Gamma_{{\mathrm{2}}}  \SCsym{;}  \SCnt{Y}  \SCsym{;}  \Gamma_{{\mathrm{3}}}  \vdash_\mathcal{L}  \SCnt{A}}
              \end{array}
            }
          }{\Gamma_{{\mathrm{1}}}  \SCsym{;}  \Phi  \SCsym{;}  \Gamma_{{\mathrm{2}}}  \SCsym{;}   \mathsf{F} \SCnt{Y}   \SCsym{;}  \Gamma_{{\mathrm{3}}}  \vdash_\mathcal{L}  \SCnt{A}}
        \end{math}
      \end{center}

\item Case 2:
      \begin{center}
        \scriptsize
        \begin{math}
          \begin{array}{c}
            \Pi_1 \\
            {\Delta  \vdash_\mathcal{L}  \SCnt{B}}
          \end{array}
        \end{math}
        \qquad\qquad
        $\Pi_2$:
        \begin{math}
          $$\mprset{flushleft}
          \inferrule* [right={\tiny Fl}] {
            {
              \begin{array}{c}
                \pi \\
                {\Gamma_{{\mathrm{1}}}  \SCsym{;}  \SCnt{B}  \SCsym{;}  \Gamma_{{\mathrm{2}}}  \SCsym{;}  \SCnt{Y}  \SCsym{;}  \Gamma_{{\mathrm{3}}}  \vdash_\mathcal{L}  \SCnt{A}}
              \end{array}
            }
          }{\Gamma_{{\mathrm{1}}}  \SCsym{;}  \SCnt{B}  \SCsym{;}  \Gamma_{{\mathrm{2}}}  \SCsym{;}   \mathsf{F} \SCnt{Y}   \SCsym{;}  \Gamma_{{\mathrm{3}}}  \vdash_\mathcal{L}  \SCnt{A}}
        \end{math}
      \end{center}
      By assumption, $c(\Pi_1),c(\Pi_2)\leq |B|$. By induction on $\Pi_1$
      and $\pi$, there is a proof $\Pi'$ for sequent
      $\Gamma_{{\mathrm{1}}}  \SCsym{;}  \Delta  \SCsym{;}  \Gamma_{{\mathrm{2}}}  \SCsym{;}  \SCnt{Y}  \SCsym{;}  \Gamma_{{\mathrm{3}}}  \vdash_\mathcal{L}  \SCnt{A}$ s.t. $c(\Pi') \leq |B|$. Therefore, the
      proof $\Pi$ can be constructed as follows with
      $c(\Pi) = c(\Pi') \leq |B|$.
      \begin{center}
        \scriptsize
        \begin{math}
          $$\mprset{flushleft}
          \inferrule* [right={\tiny Fl}] {
            {
              \begin{array}{c}
                \Pi' \\
                {\Gamma_{{\mathrm{1}}}  \SCsym{;}  \Delta  \SCsym{;}  \Gamma_{{\mathrm{2}}}  \SCsym{;}  \SCnt{Y}  \SCsym{;}  \Gamma_{{\mathrm{3}}}  \vdash_\mathcal{L}  \SCnt{A}}
              \end{array}
            }
          }{\Gamma_{{\mathrm{1}}}  \SCsym{;}  \Delta  \SCsym{;}  \Gamma_{{\mathrm{2}}}  \SCsym{;}   \mathsf{F} \SCnt{Y}   \SCsym{;}  \Gamma_{{\mathrm{3}}}  \vdash_\mathcal{L}  \SCnt{A}}
        \end{math}
      \end{center}

\item Case 3:
      \begin{center}
        \scriptsize
        \begin{math}
          \begin{array}{c}
            \Pi_1 \\
            {\Phi  \vdash_\mathcal{C}  \SCnt{X}}
          \end{array}
        \end{math}
        \qquad\qquad
        $\Pi_2$:
        \begin{math}
          $$\mprset{flushleft}
          \inferrule* [right={\tiny Fl}] {
            {
              \begin{array}{c}
                \pi \\
                {\Gamma_{{\mathrm{1}}}  \SCsym{;}  \SCnt{Y}  \SCsym{;}  \Gamma_{{\mathrm{2}}}  \SCsym{;}  \SCnt{X}  \SCsym{;}  \Gamma_{{\mathrm{3}}}  \vdash_\mathcal{L}  \SCnt{A}}
              \end{array}
            }
          }{\Gamma_{{\mathrm{1}}}  \SCsym{;}   \mathsf{F} \SCnt{Y}   \SCsym{;}  \Gamma_{{\mathrm{2}}}  \SCsym{;}  \SCnt{X}  \SCsym{;}  \Gamma_{{\mathrm{3}}}  \vdash_\mathcal{L}  \SCnt{A}}
        \end{math}
      \end{center}
      By assumption, $c(\Pi_1),c(\Pi_2)\leq |X|$. By induction on $\Pi_1$
      and $\pi$, there is a proof $\Pi'$ for sequent
      $\Gamma_{{\mathrm{1}}}  \SCsym{;}  \SCnt{Y}  \SCsym{;}  \Gamma_{{\mathrm{2}}}  \SCsym{;}  \Phi  \SCsym{;}  \Gamma_{{\mathrm{3}}}  \vdash_\mathcal{L}  \SCnt{A}$ s.t. $c(\Pi') \leq |X|$. Therefore, the
      proof $\Pi$ can be constructed as follows with
      $c(\Pi) = c(\Pi') \leq |X|$.
      \begin{center}
        \scriptsize
        \begin{math}
          $$\mprset{flushleft}
          \inferrule* [right={\tiny Fl}] {
            {
              \begin{array}{c}
                \Pi' \\
                {\Gamma_{{\mathrm{1}}}  \SCsym{;}  \SCnt{Y}  \SCsym{;}  \Gamma_{{\mathrm{2}}}  \SCsym{;}  \Phi  \SCsym{;}  \Gamma_{{\mathrm{3}}}  \vdash_\mathcal{L}  \SCnt{A}}
              \end{array}
            }
          }{\Gamma_{{\mathrm{1}}}  \SCsym{;}   \mathsf{F} \SCnt{Y}   \SCsym{;}  \Gamma_{{\mathrm{2}}}  \SCsym{;}  \Phi  \SCsym{;}  \Gamma_{{\mathrm{3}}}  \vdash_\mathcal{L}  \SCnt{A}}
        \end{math}
      \end{center}

\item Case 4:
      \begin{center}
        \scriptsize
        \begin{math}
          \begin{array}{c}
            \Pi_1 \\
            {\Delta  \vdash_\mathcal{L}  \SCnt{B}}
          \end{array}
        \end{math}
        \qquad\qquad
        $\Pi_2$:
        \begin{math}
          $$\mprset{flushleft}
          \inferrule* [right={\tiny Fl}] {
            {
              \begin{array}{c}
                \pi \\
                {\Gamma_{{\mathrm{1}}}  \SCsym{;}  \SCnt{Y}  \SCsym{;}  \Gamma_{{\mathrm{2}}}  \SCsym{;}  \SCnt{B}  \SCsym{;}  \Gamma_{{\mathrm{3}}}  \vdash_\mathcal{L}  \SCnt{A}}
              \end{array}
            }
          }{\Gamma_{{\mathrm{1}}}  \SCsym{;}   \mathsf{F} \SCnt{Y}   \SCsym{;}  \Gamma_{{\mathrm{2}}}  \SCsym{;}  \Delta  \SCsym{;}  \Gamma_{{\mathrm{3}}}  \vdash_\mathcal{L}  \SCnt{A}}
        \end{math}
      \end{center}
      By assumption, $c(\Pi_1),c(\Pi_2)\leq |B|$. By induction on $\Pi_1$
      and $\pi$, there is a proof $\Pi'$ for sequent
      $\Gamma_{{\mathrm{1}}}  \SCsym{;}  \SCnt{Y}  \SCsym{;}  \Gamma_{{\mathrm{2}}}  \SCsym{;}  \Delta  \SCsym{;}  \Gamma_{{\mathrm{3}}}  \vdash_\mathcal{L}  \SCnt{A}$ s.t. $c(\Pi') \leq |B|$. Therefore, the
      proof $\Pi$ can be constructed as follows with
      $c(\Pi) = c(\Pi') \leq |B|$.
      \begin{center}
        \scriptsize
        \begin{math}
          $$\mprset{flushleft}
          \inferrule* [right={\tiny Fl}] {
            {
              \begin{array}{c}
                \Pi' \\
                {\Gamma_{{\mathrm{1}}}  \SCsym{;}  \SCnt{Y}  \SCsym{;}  \Gamma_{{\mathrm{2}}}  \SCsym{;}  \Delta  \SCsym{;}  \Gamma_{{\mathrm{3}}}  \vdash_\mathcal{L}  \SCnt{A}}
              \end{array}
            }
          }{\Gamma_{{\mathrm{1}}}  \SCsym{;}   \mathsf{F} \SCnt{Y}   \SCsym{;}  \Gamma_{{\mathrm{2}}}  \SCsym{;}  \Delta  \SCsym{;}  \Gamma_{{\mathrm{3}}}  \vdash_\mathcal{L}  \SCnt{A}}
        \end{math}
      \end{center}
\end{itemize}

\subsubsection{Right introduction of the functor $G$ (with low priority)}
\begin{center}
  \scriptsize
  \begin{math}
    \begin{array}{c}
      \Pi_1 \\
      {\Phi  \vdash_\mathcal{C}  \SCnt{X}}
    \end{array}
  \end{math}
  \qquad\qquad
  $\Pi_2$:
  \begin{math}
    $$\mprset{flushleft}
    \inferrule* [right={\tiny Gr}] {
      {
        \begin{array}{c}
          \pi \\
          {\Psi_{{\mathrm{1}}}  \SCsym{;}  \SCnt{X}  \SCsym{;}  \Psi_{{\mathrm{2}}}  \vdash_\mathcal{L}  \SCnt{A}}
        \end{array}
      }
    }{\Psi_{{\mathrm{1}}}  \SCsym{,}  \SCnt{X}  \SCsym{,}  \Psi_{{\mathrm{2}}}  \vdash_\mathcal{C}   \mathsf{G} \SCnt{A} }
  \end{math}
\end{center}
By assumption, $c(\Pi_1),c(\Pi_2)\leq |X|$. By induction on $\Pi_1$
and $\pi$, there is a proof $\Pi'$ for sequent $\Psi_{{\mathrm{1}}}  \SCsym{,}  \Phi  \SCsym{,}  \Psi_{{\mathrm{2}}}  \vdash_\mathcal{L}  \SCnt{A}$
s.t. $c(\Pi') \leq |X|$. Therefore, the proof $\Pi$ can be
constructed as follows with $c(\Pi) = c(\Pi') \leq |X|$.
\begin{center}
  \scriptsize
  \begin{math}
    $$\mprset{flushleft}
    \inferrule* [right={\tiny Gr}] {
      {
        \begin{array}{c}
          \Pi' \\
          {\Psi_{{\mathrm{1}}}  \SCsym{;}  \Phi  \SCsym{;}  \Psi_{{\mathrm{2}}}  \vdash_\mathcal{L}  \SCnt{A}}
        \end{array}
      }
    }{\Psi_{{\mathrm{1}}}  \SCsym{,}  \Phi  \SCsym{,}  \Psi_{{\mathrm{2}}}  \vdash_\mathcal{C}   \mathsf{G} \SCnt{A} }
  \end{math}
\end{center}

\subsubsection{Left introduction of the functor $G$ (with low priority)}
\begin{itemize}
\item Case 1:
      \begin{center}
        \scriptsize
        \begin{math}
          \begin{array}{c}
            \Pi_1 \\
            {\Phi  \vdash_\mathcal{C}  \SCnt{X}}
          \end{array}
        \end{math}
        \qquad\qquad
        $\Pi_2$:
        \begin{math}
          $$\mprset{flushleft}
          \inferrule* [right={\tiny Gl}] {
            {
              \begin{array}{c}
                \pi \\
                {\Gamma_{{\mathrm{1}}}  \SCsym{;}  \SCnt{X}  \SCsym{;}  \Gamma_{{\mathrm{2}}}  \SCsym{;}  \SCnt{B}  \SCsym{;}  \Gamma_{{\mathrm{3}}}  \vdash_\mathcal{L}  \SCnt{A}}
              \end{array}
            }
          }{\Gamma_{{\mathrm{1}}}  \SCsym{;}  \SCnt{X}  \SCsym{;}  \Gamma_{{\mathrm{2}}}  \SCsym{;}   \mathsf{G} \SCnt{B}   \SCsym{;}  \Gamma_{{\mathrm{3}}}  \vdash_\mathcal{L}  \SCnt{A}}
        \end{math}
      \end{center}
      By assumption, $c(\Pi_1),c(\Pi_2)\leq |X|$. By induction on $\Pi_1$
      and $\pi$, there is a proof $\Pi'$ for sequent
      $\Gamma_{{\mathrm{1}}}  \SCsym{;}  \Phi  \SCsym{;}  \Gamma_{{\mathrm{2}}}  \SCsym{;}  \SCnt{B}  \SCsym{;}  \Gamma_{{\mathrm{3}}}  \vdash_\mathcal{L}  \SCnt{A}$ s.t. $c(\Pi') \leq |X|$. Therefore, the
      proof $\Pi$ can be constructed as follows with
      $c(\Pi) = c(\Pi') \leq |X|$.
      \begin{center}
        \scriptsize
        \begin{math}
          $$\mprset{flushleft}
          \inferrule* [right={\tiny Gl}] {
            {
              \begin{array}{c}
                \Pi' \\
                {\Gamma_{{\mathrm{1}}}  \SCsym{;}  \Phi  \SCsym{;}  \Gamma_{{\mathrm{2}}}  \SCsym{;}  \SCnt{B}  \SCsym{;}  \Gamma_{{\mathrm{3}}}  \vdash_\mathcal{L}  \SCnt{A}}
              \end{array}
            }
          }{\Gamma_{{\mathrm{1}}}  \SCsym{;}  \Phi  \SCsym{;}  \Gamma_{{\mathrm{2}}}  \SCsym{;}   \mathsf{G} \SCnt{B}   \SCsym{;}  \Gamma_{{\mathrm{3}}}  \vdash_\mathcal{L}  \SCnt{A}}
        \end{math}
      \end{center}

\item Case 2:
      \begin{center}
        \scriptsize
        \begin{math}
          \begin{array}{c}
            \Pi_1 \\
            {\Delta  \vdash_\mathcal{L}  \SCnt{B}}
          \end{array}
        \end{math}
        \qquad\qquad
        $\Pi_2$:
        \begin{math}
          $$\mprset{flushleft}
          \inferrule* [right={\tiny Gl}] {
            {
              \begin{array}{c}
                \pi \\
                {\Gamma_{{\mathrm{1}}}  \SCsym{;}  \SCnt{B}  \SCsym{;}  \Gamma_{{\mathrm{2}}}  \SCsym{;}  \SCnt{C}  \SCsym{;}  \Gamma_{{\mathrm{3}}}  \vdash_\mathcal{L}  \SCnt{A}}
              \end{array}
            }
          }{\Gamma_{{\mathrm{1}}}  \SCsym{;}  \SCnt{B}  \SCsym{;}  \Gamma_{{\mathrm{2}}}  \SCsym{;}   \mathsf{G} \SCnt{C}   \SCsym{;}  \Gamma_{{\mathrm{3}}}  \vdash_\mathcal{L}  \SCnt{A}}
        \end{math}
      \end{center}
      By assumption, $c(\Pi_1),c(\Pi_2)\leq |B|$. By induction on $\Pi_1$
      and $\pi$, there is a proof $\Pi'$ for sequent
      $\Gamma_{{\mathrm{1}}}  \SCsym{;}  \Delta  \SCsym{;}  \Gamma_{{\mathrm{2}}}  \SCsym{;}  \SCnt{C}  \SCsym{;}  \Gamma_{{\mathrm{3}}}  \vdash_\mathcal{L}  \SCnt{A}$ s.t. $c(\Pi') \leq |B|$. Therefore, the
      proof $\Pi$ can be constructed as follows with
      $c(\Pi) = c(\Pi') \leq |B|$.
      \begin{center}
        \scriptsize
        \begin{math}
          $$\mprset{flushleft}
          \inferrule* [right={\tiny Gl}] {
            {
              \begin{array}{c}
                \Pi' \\
                {\Gamma_{{\mathrm{1}}}  \SCsym{;}  \Delta  \SCsym{;}  \Gamma_{{\mathrm{2}}}  \SCsym{;}  \SCnt{C}  \SCsym{;}  \Gamma_{{\mathrm{3}}}  \vdash_\mathcal{L}  \SCnt{A}}
              \end{array}
            }
          }{\Gamma_{{\mathrm{1}}}  \SCsym{;}  \Delta  \SCsym{;}  \Gamma_{{\mathrm{2}}}  \SCsym{;}   \mathsf{G} \SCnt{C}   \SCsym{;}  \Gamma_{{\mathrm{3}}}  \vdash_\mathcal{L}  \SCnt{A}}
        \end{math}
      \end{center}

\item Case 3:
      \begin{center}
        \scriptsize
        \begin{math}
          \begin{array}{c}
            \Pi_1 \\
            {\Phi  \vdash_\mathcal{C}  \SCnt{X}}
          \end{array}
        \end{math}
        \qquad\qquad
        $\Pi_2$:
        \begin{math}
          $$\mprset{flushleft}
          \inferrule* [right={\tiny Gl}] {
            {
              \begin{array}{c}
                \pi \\
                {\Gamma_{{\mathrm{1}}}  \SCsym{;}  \SCnt{B}  \SCsym{;}  \Gamma_{{\mathrm{2}}}  \SCsym{;}  \SCnt{X}  \SCsym{;}  \Gamma_{{\mathrm{3}}}  \vdash_\mathcal{L}  \SCnt{A}}
              \end{array}
            }
          }{\Gamma_{{\mathrm{1}}}  \SCsym{;}   \mathsf{G} \SCnt{B}   \SCsym{;}  \Gamma_{{\mathrm{2}}}  \SCsym{;}  \SCnt{X}  \SCsym{;}  \Gamma_{{\mathrm{3}}}  \vdash_\mathcal{L}  \SCnt{A}}
        \end{math}
      \end{center}
      By assumption, $c(\Pi_1),c(\Pi_2)\leq |X|$. By induction on $\Pi_1$
      and $\pi$, there is a proof $\Pi'$ for sequent
      $\Gamma_{{\mathrm{1}}}  \SCsym{;}  \SCnt{B}  \SCsym{;}  \Gamma_{{\mathrm{2}}}  \SCsym{;}  \Phi  \SCsym{;}  \Gamma_{{\mathrm{3}}}  \vdash_\mathcal{L}  \SCnt{A}$ s.t. $c(\Pi') \leq |X|$. Therefore, the
      proof $\Pi$ can be constructed as follows with
      $c(\Pi) = c(\Pi') \leq |X|$.
      \begin{center}
        \scriptsize
        \begin{math}
          $$\mprset{flushleft}
          \inferrule* [right={\tiny Gl}] {
            {
              \begin{array}{c}
                \Pi' \\
                {\Gamma_{{\mathrm{1}}}  \SCsym{;}  \SCnt{B}  \SCsym{;}  \Gamma_{{\mathrm{2}}}  \SCsym{;}  \Phi  \SCsym{;}  \Gamma_{{\mathrm{3}}}  \vdash_\mathcal{L}  \SCnt{A}}
              \end{array}
            }
          }{\Gamma_{{\mathrm{1}}}  \SCsym{;}   \mathsf{G} \SCnt{B}   \SCsym{;}  \Gamma_{{\mathrm{2}}}  \SCsym{;}  \Phi  \SCsym{;}  \Gamma_{{\mathrm{3}}}  \vdash_\mathcal{L}  \SCnt{A}}
        \end{math}
      \end{center}

\item Case 4:
      \begin{center}
        \scriptsize
        \begin{math}
          \begin{array}{c}
            \Pi_1 \\
            {\Delta  \vdash_\mathcal{L}  \SCnt{B}}
          \end{array}
        \end{math}
        \qquad\qquad
        $\Pi_2$:
        \begin{math}
          $$\mprset{flushleft}
          \inferrule* [right={\tiny Gl}] {
            {
              \begin{array}{c}
                \pi \\
                {\Gamma_{{\mathrm{1}}}  \SCsym{;}  \SCnt{C}  \SCsym{;}  \Gamma_{{\mathrm{2}}}  \SCsym{;}  \SCnt{B}  \SCsym{;}  \Gamma_{{\mathrm{3}}}  \vdash_\mathcal{L}  \SCnt{A}}
              \end{array}
            }
          }{\Gamma_{{\mathrm{1}}}  \SCsym{;}   \mathsf{G} \SCnt{C}   \SCsym{;}  \Gamma_{{\mathrm{2}}}  \SCsym{;}  \SCnt{B}  \SCsym{;}  \Gamma_{{\mathrm{3}}}  \vdash_\mathcal{L}  \SCnt{A}}
        \end{math}
      \end{center}
      By assumption, $c(\Pi_1),c(\Pi_2)\leq |B|$. By induction on $\Pi_1$
      and $\pi$, there is a proof $\Pi'$ for sequent
      $\Gamma_{{\mathrm{1}}}  \SCsym{;}  \SCnt{C}  \SCsym{;}  \Gamma_{{\mathrm{2}}}  \SCsym{;}  \Delta  \SCsym{;}  \Gamma_{{\mathrm{3}}}  \vdash_\mathcal{L}  \SCnt{A}$ s.t. $c(\Pi') \leq |B|$. Therefore, the
      proof $\Pi$ can be constructed as follows with
      $c(\Pi) = c(\Pi') \leq |B|$.
      \begin{center}
        \scriptsize
        \begin{math}
          $$\mprset{flushleft}
          \inferrule* [right={\tiny Gl}] {
            {
              \begin{array}{c}
                \Pi' \\
                {\Gamma_{{\mathrm{1}}}  \SCsym{;}  \SCnt{C}  \SCsym{;}  \Gamma_{{\mathrm{2}}}  \SCsym{;}  \Delta  \SCsym{;}  \Gamma_{{\mathrm{3}}}  \vdash_\mathcal{L}  \SCnt{A}}
              \end{array}
            }
          }{\Gamma_{{\mathrm{1}}}  \SCsym{;}   \mathsf{G} \SCnt{C}   \SCsym{;}  \Gamma_{{\mathrm{2}}}  \SCsym{;}  \Delta  \SCsym{;}  \Gamma_{{\mathrm{3}}}  \vdash_\mathcal{L}  \SCnt{A}}
        \end{math}
      \end{center}

\end{itemize}

\section{Proof For Lemma~\ref{lem:monoidal-monad}}
\label{app:monoidal-monad}

Let $(\cat{C},\cat{L},F,G,\eta,\varepsilon)$ be a LAM. We define the monad
$(T,\eta:id_\cat{C}\rightarrow T,\mu:T^2\rightarrow T)$ on $\cat{C}$ as
$T=GF$, $\eta_X:X\rightarrow GFX$, and
$\mu_X=G\varepsilon_{FX}:GFGFX\rightarrow GFX$. Since $(F,\m{})$ and
$(G,\n{})$ are monoidal functors, we have
$$\t{X,Y}=G\m{X,Y}\circ\n{FX,FY}:TX\otimes TY\rightarrow T(X\otimes Y) \qquad\mbox{and}\qquad\t{I}=G\m{I}\circ\n{I'}:I\rightarrow TI.$$
The monad $T$ being monoidal means:
\begin{enumerate}
\item $T$ is a monoidal functor, i.e. the following diagrams commute:
      \begin{mathpar}
      \bfig
        \hSquares/->`->`->``->`->`->/<400>[
          (TX\otimes TY)\otimes TZ`TX\otimes(TY\otimes TZ)`TX\otimes T(Y\otimes Z)`
          T(X\otimes Y)\otimes TZ`T((X\otimes Y)\otimes Z)`T(X\otimes(Y\otimes Z));
          \alpha_{TX,TY,TZ}`id_{TX}\otimes\t{Y,Z}`\t{X,Y}\otimes id_{TZ}``
          \t{X,Y\otimes Z}`\t{X\otimes Y,Z}`T\alpha_{X,Y,Z}]
        \morphism(1300,200)//<0,0>[`;(1)]
      \efig
      \and
      \bfig
        \square/->`->`<-`->/<600,400>[
          I\otimes TX`TX`TI\otimes TX`T(I\otimes X);
          \lambda_{TX}`\t{I}\otimes id_{TX}`T\lambda_X`\t{I,X}]
        \morphism(350,200)//<0,0>[`;(2)]
      \efig
      \and
      \bfig
        \square/->`->`<-`->/<600,400>[
          TX\otimes I`TX`TX\otimes TI`T(X\otimes I);
          \rho_{TX}`id_{TX}\otimes\t{I}`T\rho_X`\t{X,I}]
        \morphism(350,200)//<0,0>[`;(3)]
      \efig
      \end{mathpar}
      We write $GF$ instead of $T$ in the proof for clarity. \\
      By replacing $\t{X,Y}$ with its definition, diagram (1) above
      commutes by the following commutative diagram, in which the two
      hexagons commute because $G$ and $F$ are monoidal functors, and the
      two quadrilaterals commute by the naturality of $\n{}$.
      \begin{mathpar}
      \bfig
        \iiixiii/->`->`->``->```->`<-`->``/<1400,400>[
          (GFX\otimes GFY)\otimes GFZ`GFX\otimes(GFY\otimes GFZ)`GFX\otimes G(FY\tri FZ)`
          G(FX\tri FY)\otimes GFZ`G(FX\tri(FY\tri FZ))`GFX\otimes GF(Y\otimes Z)`
          GF(X\otimes Y)\otimes GFZ`G((FX\tri FY)\tri FZ)`G(FX\tri F(Y\otimes Z));
          \alpha_{GFX,GFY,GFZ}`id_{GFX}\otimes\n{FY,FZ}`\n{FX,FY}\otimes id_{GFZ}``
          id_{GFX}\otimes G\m{Y,Z}```G\m{X,Y}\otimes id_{GFZ}`G\alpha'_{FX,FY,FZ}`
          \n{FX,F(Y\otimes Z)}``]
        \morphism(2800,800)|m|<-1400,-400>[
          GFX\otimes G(FY\tri FZ)`G(FX\tri(FY\tri FZ));\n{FX,FY\tri FZ}]
        \morphism(0,400)|m|<1400,-400>[
          G(FX\tri FY)\otimes GFZ`G((FX\tri FY)\tri FZ);\n{FX\tri FY,FZ}]
        \morphism(1400,400)|m|<1400,-400>[
          G(FX\tri(FY\tri FZ))`G(FX\tri F(Y\otimes Z));G(id_{FX}\tri\m{Y,Z})]
        \ptriangle(0,-400)|mlm|/`->`->/<1400,400>[
          GF(X\otimes Y)\otimes GFZ`G((FX\tri FY)\tri FZ)`G(F(X\otimes Y)\tri FZ);
          `\n{F(X\otimes Y),FZ}`G(\m{X,Y}\otimes id_{FZ})]
        \morphism(0,-400)|b|<1400,0>[
          G(F(X\otimes Y)\tri FZ)`GF((X\otimes Y)\otimes Z);G\m{X\otimes Y,Z}]
        \dtriangle(1400,-400)|mrb|/`->`->/<1400,400>[
          G(FX\tri F(Y\otimes Z))`GF((X\otimes Y)\otimes Z)`GF(X\otimes(Y\otimes Z));
          `G\m{X,Y\otimes Z}`GF\alpha_{X,Y,Z}]
      \efig
      \end{mathpar}
      Diagram (2) commutes by the following commutative diagrams, in which
      the top quadrilateral commutes because $G$ is monoidal, the right
      quadrilateral commutes because $F$ is monoidal, and the left square
      commutes by the naturality of $\n{}$.
      \begin{mathpar}
      \bfig
        \ptriangle/->`->`/<1600,400>[
          I\otimes GFX`GFX`GI'\otimes GFX;\lambda_{GFX}`\n{I'}\otimes id_{GFX}`]
        \square(0,-400)|lmmb|<800,400>[
          GI'\otimes GFX`G(I'\tri FX)`GFI\otimes GFX`G(FI\tri FX);
          \n{I',FX}`G\m{I}\otimes id_{GFX}`G(\m{I}\tri id_{FX})`\n{FI,FX}]
        \morphism(800,0)|m|<800,400>[G(I'\tri FX)`GFX;G\lambda'_{FX}]
        \dtriangle(800,-400)/`<-`->/<800,800>[
          GFX`G(FI\tri FX)`GF(I\otimes X);
          `GF\lambda_X`G\m{I,X}]
      \efig
      \end{mathpar}
      Similarly, diagram (3) commutes as follows:
      \begin{mathpar}
      \bfig
        \ptriangle/->`->`/<1600,400>[
          GFX\otimes I`GFX`GFX\otimes GI';\rho_{GFX}`id_{GFX}\otimes\n{I'}`]
        \square(0,-400)|lmmb|<800,400>[
          GFX\otimes GI'`G(FX\tri I')`GFX\otimes GFI`G(FX\tri FI);
          \n{FX,I'}`id_{GFX}\otimes G\m{I}`G(id_{FX}\otimes\m{I})`\n{FX,FI}]
        \morphism(800,0)|m|<800,400>[G(FX\tri I')`GFX;G\rho'_{FX}]
        \dtriangle(800,-400)/`<-`->/<800,800>[
          GFX`G(FX\tri FI)`GF(X\otimes I);
          `GF\rho_X`G\m{X,I}]
      \efig
      \end{mathpar}
\item $\eta$ is a monoidal natural transformation. In fact, since $\eta$
      is the unit of the monoidal adjunction, $\eta$ is monoidal by
      definition and thus the following two diagrams commute.
      \begin{mathpar}
      \bfig
        \square/=`->`->`->/<600,400>[
          X\otimes Y`X\otimes Y`TX\otimes TY`T(X\otimes Y);
          `\eta_X\otimes\eta_Y`\eta_{X\otimes Y}`\t{X,Y}]
      \efig
      \and
      \bfig
        \Vtriangle/->`=`<-/<400,400>[I`TI`I;\eta_I``\t{I}]
      \efig
      \end{mathpar}
\item $\mu$ is a monoidal natural transformation. It is obvious that since
      $\mu=G\varepsilon_{FA}$ and $\varepsilon$ is monoidal, so is $\mu$.
      Thus the following diagrams commute.
      \begin{mathpar}
      \bfig
        \square/`->`->`->/<1500,400>[
          T^2X\otimes T^2Y`T^2(X\otimes Y)`TX\otimes TY`T(X\otimes Y);
          `\mu_X\otimes\mu_Y`\mu_{X\otimes Y}`\t{X,Y}]
        \morphism(0,400)<800,0>[T^2X\otimes T^2Y`T(TX\otimes TY);\t{TX,TY}]
        \morphism(800,400)<700,0>[T(TX\otimes TY)`T^2(X\otimes Y);T\t{X,Y}]
      \efig
      \and
      \bfig
        \square/->`<-`<-`<-/<400,400>[T^2I`TI`TI`I;\mu_I`T\t{I}`\t{I}`\t{I}]
      \efig
      \end{mathpar}
\end{enumerate}

\section{Proof For Lemma~\ref{lem:strong-monad}}
\label{app:strong-monad}

\begin{definition}
\label{def:strong-monad}
Let $(\cat{M},\tri,I,\alpha,\lambda,\rho)$ be a monoidal category and
$(T,\eta,\mu)$ be a monad on $\cat{M}$. $T$ is a \textbf{strong monad} if
there is natural transformation $\tau$, called the \textbf{tensorial
strength}, with components $\tau_{A,B}:A\tri TB\rightarrow T(A\tri B)$
such that the following diagrams commute:
\begin{mathpar}
\bfig
  \Vtriangle<400,400>[I\tri TA`T(I\tri A)`TA;\tau_{I,A}`\lambda_{TA}`T\lambda_A]
\efig
\and
\bfig
  \Vtriangle<400,400>[
    A\tri B`A\tri TB`T(A\tri B);id_A\tri\eta_B`\eta_{A\tri B}`\tau_{A,B}]
\efig
\and
\bfig
  \square/->`->`->`/<1800,400>[
    (A\tri B)\tri TC`T((A\tri B)\tri C)`
    A\tri(B\tri TC)`T(A\tri(B\tri C));
    \tau_{A\tri B,C}`\alpha_{A,B,TC}`T\alpha_{A,B,C}`]
  \morphism<900,0>[A\tri(B\tri TC)`A\tri T(B\tri C);id_A\tri\tau_{B,C}]
  \morphism(900,0)<900,0>[A\tri T(B\tri C)`T(A\tri(B\tri C));\tau_{A,B\tri C}]  \efig
\and
\bfig
  \square/`->`->`->/<1400,400>[
    A\tri T^2B`T^2(A\tri B)`A\tri TB`T(A\tri B);
    `id_A\tri\mu_B`\mu_{A\tri B}`\tau_{A,B}]
  \morphism(0,400)<700,0>[A\tri T^2B`T(A\tri TB);\tau_{A,TB}]
  \morphism(700,400)<700,0>[T(A\tri TB)`T^2(A\tri B);T\tau_{A,B}]
\efig
\end{mathpar}
\end{definition}
\noindent
The proof for Lemma~\ref{lem:strong-monad} goes as follows.
\noindent
Let $(\cat{C},\cat{L},F,G,\eta,\varepsilon)$ be a LAM, where
$(\cat{C},\otimes,I,\alpha,\lambda,\rho)$ is symmetric monoidal closed,
and \\ $(\cat{L},\tri,I',\alpha',\lambda',\rho')$ is Lambek. In
Lemma~\ref{lem:monoidal-monad}, we have proved that the monad
$(T=GF,\eta,\mu)$ is monoidal with the natural transformation
$\t{X,Y}:TX\otimes TY\rightarrow T(X\otimes Y)$ and the morphism
$\t{I}:I\rightarrow TI$.
\noindent
We define the tensorial strength
$\tau_{X,Y}:X\otimes TY\rightarrow T(X\otimes Y)$ as
$$\tau_{X,Y}=\t{X,Y}\circ(\eta_X\otimes id_{TY}).$$
Since $\eta$ is a monoidal natural transformation, we have
$\eta_I=G\m{I}\circ\n{I'}$, and thus $\eta_I=\t{I}$. The following diagram
commutes because $T$ is monoidal, where the composition
$\t{I,X}\circ(\t{I}\otimes id_{TX})$ is the definition of $\tau_{I,X}$. So
the first triangle in Definition~\ref{def:strong-monad} commutes.
\begin{mathpar}
\bfig
  \square/->`->`->`<-/<600,400>[
    I\otimes TX`TI\otimes TX`TX`T(I\otimes X);
    \t{I}\otimes id_{TX}`\lambda_{TX}`\t{I,X}`T\lambda_X]
\efig
\end{mathpar}
Similarly, by using the definition of $\tau$, the the second triangle in the definition is
equivalent to the following diagram, which commutes because $\eta$ is a monoidal natural
transformation:
\begin{mathpar}
\bfig
  \square/->`->`->`<-/<600,400>[
    X\otimes Y`X\otimes TY`T(X\otimes Y)`TX\otimes TY;
    id_X\otimes\eta_Y`\eta_{X\otimes Y}`\eta_X\otimes id_{TY}`\t{X,Y}]
  \morphism(0,400)|m|<600,-400>[X\otimes Y`TX\otimes TY;\eta_X\otimes\eta_Y]
\efig
\end{mathpar}
The first pentagon in the definition commutes by the following commutative diagrams, because
$\eta$ and $\alpha$ are natural transformations and $T$ is monoidal:
\begin{mathpar}
\bfig
  \qtriangle|amm|/->`->`<-/<1000,400>[
    (X\otimes Y)\otimes TZ`T(X\otimes Y)\otimes TZ`(TX\otimes TY)\otimes TZ;
    \eta_{X\otimes Y}\otimes id_{TZ}`
    (\eta_X\otimes\eta_Y)\otimes id_{TZ}`
    \t{X,Y}\otimes id_{TZ}]
  \morphism(0,400)<0,-400>[(X\otimes Y)\otimes TZ`X\otimes(Y\otimes TZ);\alpha_{X,Y,TZ}]
  \morphism(1000,0)|m|<0,-400>[
    (TX\otimes TY)\otimes TZ`TX\otimes(TY\otimes TZ);\alpha_{TX,TY,TZ}]
  \Dtriangle(0,-800)|lmm|/->`->`<-/<1000,400>[
    X\otimes(Y\otimes TZ)`TX\otimes(TY\otimes TZ)`X\otimes(TY\otimes TZ);
    id_X\otimes(\eta_Y\otimes id_{TZ})`
    \eta_X\otimes(\eta_Y\otimes id_{TZ})`
    \eta_X\otimes id_{TY\otimes TZ}]
  \morphism(0,-800)|b|<1000,0>[
    X\otimes(TY\otimes TZ)`X\otimes T(Y\otimes Z);id_X\otimes\t{Y,Z}]
  \qtriangle(1000,0)|amr|/->``->/<1000,400>[
    T(X\otimes Y)\otimes TZ`T((X\otimes Y)\otimes Z)`T(X\otimes(Y\otimes Z));
    \t{X\otimes Y,Z}``T\alpha_{X,Y,Z}]
  \morphism(2000,-800)<0,800>[
    TX\otimes T(Y\otimes Z)`T(X\otimes(Y\otimes Z));\t{X,Y\otimes Z}]
  \btriangle(1000,-800)|mmb|/`->`->/<1000,400>[
    TX\otimes(TY\otimes TZ)`X\otimes T(Y\otimes Z)`TX\otimes T(Y\otimes Z);
    `id_{TX}\otimes\t{Y,Z}`\eta_X\otimes id_{T(Y\otimes Z)}]
\efig
\end{mathpar}
The last diagram in the definition commutes by the following commutative diagram, because
$T$ is a monad, $\t{}$ is a natural transformation, and $\mu$ is a monoidal natural
transformation:
\begin{mathpar}
\bfig
  \ptriangle/->`->`/<700,400>[
    X\otimes T^2Y`TX\otimes T^2Y`X\otimes TY;\eta_X\otimes id_{T^2Y}`id_X\otimes\mu_Y`]
  \btriangle(0,-400)/->``->/<700,400>[
    X\otimes TY`TX\otimes TY`T(X\otimes Y);\eta_X\otimes id_{TY}``\t{X,Y}]
  \morphism(700,400)|m|<-700,-800>[TX\otimes T^2Y`TX\otimes TY;id_{TX}\otimes\mu_Y]
  \morphism(700,0)|m|<-700,-400>[TX\otimes T^2Y`TX\otimes TY;id_{TX}\otimes\mu_Y]
  \qtriangle(700,0)/->``->/<1800,400>[
    TX\otimes T^2Y`T(X\otimes TY)`T(TX\otimes TY);\t{X,TY}``T(\eta_X\otimes id_{TY})]
  \btriangle(700,0)|mmm|/=`->`<-/<900,400>[
    TX\otimes T^2Y`TX\otimes T^2Y`T^2X\otimes T^2Y;
    `T\eta_X\otimes id_{T^2Y}`\mu_X\otimes id_{T^2Y}]
  \morphism(1600,0)|m|<900,0>[T^2X\otimes T^2Y`T(TX\otimes TY);\t{TX,TY}]
  \morphism(1600,0)|m|<-1600,-400>[T^2X\otimes T^2Y`TX\otimes TY;\mu_X\otimes\mu_Y]
  \dtriangle(700,-400)/`->`<-/<1800,400>[
    T(TX\otimes TY)`T(X\otimes Y)`T^2(X\otimes Y);`T\t{X,Y}`\mu_{X\otimes Y}]
\efig
\end{mathpar}

\section{Equivalence of Sequent Calculus and Natural Deduction Formalizations}
\label{app:sc-nd-equiv}

We prove the equivalence of the sequence calculus formalization and the
natural deduction formalization given in the paper by defining two
mappings, one from the rules in natural deduction to proofs the sequent
calculus, and the other is from the rules in sequent calculus to proofs in
natural deduction.

\subsection{Mapping from Natural Deduction to Sequent Calculus}

Function $S:\mathit{ND}\rightarrow\mathit{SC}$ maps a rule in the natural
deduction formalization to a proof of the same sequent in the sequent
calculus. The function is defined as follows:

\begin{itemize}
\item The axioms map to axioms.
\item Introduction rules map to right rules.
\item Elimination rules map to combinations of left rules with cuts:
  \begin{itemize}
  \item $\NDdruleTXXunitEName$:
    \begin{center}
      \footnotesize
      $\NDdruleTXXunitE{}$
    \end{center}
    maps to
    \begin{center}
      \footnotesize
      \begin{math}
        $$\mprset{flushleft}
        \inferrule* [right=$\ElledruleTXXcutName$] {
          {\Phi  \vdash_\mathcal{C}  \NDnt{t_{{\mathrm{1}}}}  \NDsym{:}   \mathsf{Unit} } \\
          $$\mprset{flushleft}
          \inferrule* [right=$\ElledruleTXXunitLName$] {
            {\Psi  \vdash_\mathcal{C}  \NDnt{t_{{\mathrm{2}}}}  \NDsym{:}  \NDnt{Y}}
          }{\NDmv{x}  \NDsym{:}   \mathsf{Unit}   \NDsym{,}  \Psi  \vdash_\mathcal{C}   \mathsf{let}\, \NDmv{x}  :   \mathsf{Unit}  \,\mathsf{be}\,  \mathsf{triv}  \,\mathsf{in}\, \NDnt{t_{{\mathrm{2}}}}   \NDsym{:}  \NDnt{Y}}
        }{\Phi  \NDsym{,}  \Psi  \vdash_\mathcal{C}  \NDsym{[}  \NDnt{t_{{\mathrm{1}}}}  \NDsym{/}  \NDmv{x}  \NDsym{]}  \NDsym{(}   \mathsf{let}\, \NDmv{x}  :   \mathsf{Unit}  \,\mathsf{be}\,  \mathsf{triv}  \,\mathsf{in}\, \NDnt{t_{{\mathrm{2}}}}   \NDsym{)}  \NDsym{:}  \NDnt{Y}}
      \end{math}
    \end{center}

  \item $\NDdruleTXXtenEName$:
    \begin{center}
      \footnotesize
      $\NDdruleTXXtenE{}$
    \end{center}
    maps to
    \begin{center}
      \footnotesize
      \begin{math}
        $$\mprset{flushleft}
        \inferrule* [right=$\ElledruleTXXcutName$] {
          {\Phi  \vdash_\mathcal{C}  \NDnt{t_{{\mathrm{1}}}}  \NDsym{:}  \NDnt{X}  \otimes  \NDnt{Y}} \\
          $$\mprset{flushleft}
          \inferrule* [right=$\ElledruleTXXtenLName$] {
            {\Psi_{{\mathrm{1}}}  \NDsym{,}  \NDmv{x}  \NDsym{:}  \NDnt{X}  \NDsym{,}  \NDmv{y}  \NDsym{:}  \NDnt{Y}  \NDsym{,}  \Psi_{{\mathrm{2}}}  \vdash_\mathcal{C}  \NDnt{t_{{\mathrm{2}}}}  \NDsym{:}  \NDnt{Z}}
          }{\Psi_{{\mathrm{1}}}  \NDsym{,}  \NDmv{z}  \NDsym{:}  \NDnt{X}  \otimes  \NDnt{Y}  \NDsym{,}  \Psi_{{\mathrm{2}}}  \vdash_\mathcal{C}   \mathsf{let}\, \NDmv{z}  :  \NDnt{X}  \otimes  \NDnt{Y} \,\mathsf{be}\, \NDmv{x}  \otimes  \NDmv{y} \,\mathsf{in}\, \NDnt{t_{{\mathrm{2}}}}   \NDsym{:}  \NDnt{Z}}
        }{\Psi_{{\mathrm{1}}}  \NDsym{,}  \Phi  \NDsym{,}  \Psi_{{\mathrm{2}}}  \vdash_\mathcal{C}  \NDsym{[}  \NDnt{t_{{\mathrm{1}}}}  \NDsym{/}  \NDmv{z}  \NDsym{]}  \NDsym{(}   \mathsf{let}\, \NDmv{z}  :  \NDnt{X}  \otimes  \NDnt{Y} \,\mathsf{be}\, \NDmv{x}  \otimes  \NDmv{y} \,\mathsf{in}\, \NDnt{t_{{\mathrm{2}}}}   \NDsym{)}  \NDsym{:}  \NDnt{Z}}
      \end{math}
    \end{center}

  \item $\NDdruleTXXimpEName$:
    \begin{center}
      \footnotesize
      $\NDdruleTXXimpE{}$
    \end{center}
    maps to
    \begin{center}
      \footnotesize
      \begin{math}
        $$\mprset{flushleft}
        \inferrule* [right=$\ElledruleTXXcutName$] {
          {\Phi  \vdash_\mathcal{C}  \NDnt{t_{{\mathrm{1}}}}  \NDsym{:}  \NDnt{X}  \multimap  \NDnt{Y}} \\
          $$\mprset{flushleft}
          \inferrule* [right=$\ElledruleTXXimpLName$] {
            {\Psi  \vdash_\mathcal{C}  \NDnt{t_{{\mathrm{2}}}}  \NDsym{:}  \NDnt{X}} \\
            $$\mprset{flushleft}
            \inferrule* [right=$\ElledruleTXXaxName$] {
              \\
            }{\NDmv{x}  \NDsym{:}  \NDnt{Y}  \vdash_\mathcal{C}  \NDmv{x}  \NDsym{:}  \NDnt{Y}}
          }{\NDmv{y}  \NDsym{:}  \NDnt{X}  \multimap  \NDnt{Y}  \NDsym{,}  \Psi  \vdash_\mathcal{C}  \NDsym{[}   \NDmv{y}   \NDnt{t_{{\mathrm{2}}}}   \NDsym{/}  \NDmv{x}  \NDsym{]}  \NDmv{x}  \NDsym{:}  \NDnt{Y}}
        }{\Phi  \NDsym{,}  \Psi  \vdash_\mathcal{C}  \NDsym{[}  \NDnt{t_{{\mathrm{1}}}}  \NDsym{/}  \NDmv{y}  \NDsym{]}  \NDsym{[}   \NDmv{y}   \NDnt{t_{{\mathrm{2}}}}   \NDsym{/}  \NDmv{x}  \NDsym{]}  \NDmv{x}  \NDsym{:}  \NDnt{Y}}
      \end{math}
    \end{center}

  \item $\NDdruleSXXunitEOneName$:
    \begin{center}
      \footnotesize
      $\NDdruleSXXunitEOne{}$
    \end{center}
    maps to
    \begin{center}
      \footnotesize
      \begin{math}
        $$\mprset{flushleft}
        \inferrule* [right=$\ElledruleSXXcutOneName$] {
          {\Phi  \vdash_\mathcal{C}  \NDnt{t}  \NDsym{:}   \mathsf{Unit} } \\
          $$\mprset{flushleft}
          \inferrule* [right=$\ElledruleSXXunitLOneName$] {
            {\Gamma  \vdash_\mathcal{L}  \NDnt{s}  \NDsym{:}  \NDnt{A}}
          }{\NDmv{x}  \NDsym{:}   \mathsf{Unit}   \NDsym{;}  \Gamma  \vdash_\mathcal{L}   \mathsf{let}\, \NDmv{x}  :   \mathsf{Unit}  \,\mathsf{be}\,  \mathsf{triv}  \,\mathsf{in}\, \NDnt{s}   \NDsym{:}  \NDnt{A}}
        }{\Phi  \NDsym{;}  \Psi  \vdash_\mathcal{L}  \NDsym{[}  \NDnt{t}  \NDsym{/}  \NDmv{x}  \NDsym{]}  \NDsym{(}   \mathsf{let}\, \NDmv{x}  :   \mathsf{Unit}  \,\mathsf{be}\,  \mathsf{triv}  \,\mathsf{in}\, \NDnt{s}   \NDsym{)}  \NDsym{:}  \NDnt{A}}
      \end{math}
    \end{center}

  \item $\NDdruleSXXunitETwoName$:
    \begin{center}
      \footnotesize
      $\NDdruleSXXunitETwo{}$
    \end{center}
    maps to
    \begin{center}
      \footnotesize
      \begin{math}
        $$\mprset{flushleft}
        \inferrule* [right=$\ElledruleSXXcutTwoName$] {
          {\Gamma  \vdash_\mathcal{L}  \NDnt{s_{{\mathrm{1}}}}  \NDsym{:}   \mathsf{Unit} } \\
          $$\mprset{flushleft}
          \inferrule* [right=$\ElledruleSXXunitLTwoName$] {
            {\Delta  \vdash_\mathcal{L}  \NDnt{s_{{\mathrm{2}}}}  \NDsym{:}  \NDnt{A}}
          }{\NDmv{x}  \NDsym{:}   \mathsf{Unit}   \NDsym{;}  \Delta  \vdash_\mathcal{L}   \mathsf{let}\, \NDmv{x}  :   \mathsf{Unit}  \,\mathsf{be}\,  \mathsf{triv}  \,\mathsf{in}\, \NDnt{s_{{\mathrm{2}}}}   \NDsym{:}  \NDnt{A}}
        }{\Gamma  \NDsym{;}  \Delta  \vdash_\mathcal{L}  \NDsym{[}  \NDnt{s_{{\mathrm{1}}}}  \NDsym{/}  \NDmv{x}  \NDsym{]}  \NDsym{(}   \mathsf{let}\, \NDmv{x}  :   \mathsf{Unit}  \,\mathsf{be}\,  \mathsf{triv}  \,\mathsf{in}\, \NDnt{s_{{\mathrm{2}}}}   \NDsym{)}  \NDsym{:}  \NDnt{A}}
      \end{math}
    \end{center}

  \item $\NDdruleSXXtenEOneName$:
    \begin{center}
      \footnotesize
      $\NDdruleSXXtenEOne{}$
    \end{center}
    maps to
    \begin{center}
      \footnotesize
      \begin{math}
        $$\mprset{flushleft}
        \inferrule* [right=$\ElledruleSXXcutOneName$] {
          {\Phi  \vdash_\mathcal{C}  \NDnt{t}  \NDsym{:}  \NDnt{X}  \otimes  \NDnt{Y}} \\
          $$\mprset{flushleft}
          \inferrule* [right=$\ElledruleSXXtenLOneName$] {
            {\Gamma_{{\mathrm{1}}}  \NDsym{;}  \NDmv{x}  \NDsym{:}  \NDnt{X}  \NDsym{;}  \NDmv{y}  \NDsym{:}  \NDnt{Y}  \NDsym{;}  \Gamma_{{\mathrm{2}}}  \vdash_\mathcal{L}  \NDnt{s}  \NDsym{:}  \NDnt{A}}
          }{\Gamma_{{\mathrm{1}}}  \NDsym{;}  \NDmv{z}  \NDsym{:}  \NDnt{X}  \otimes  \NDnt{Y}  \NDsym{;}  \Gamma_{{\mathrm{2}}}  \vdash_\mathcal{L}   \mathsf{let}\, \NDmv{z}  :  \NDnt{X}  \otimes  \NDnt{Y} \,\mathsf{be}\, \NDmv{x}  \otimes  \NDmv{y} \,\mathsf{in}\, \NDnt{s}   \NDsym{:}  \NDnt{A}}
        }{\Gamma_{{\mathrm{1}}}  \NDsym{;}  \Phi  \NDsym{;}  \Gamma_{{\mathrm{2}}}  \vdash_\mathcal{L}  \NDsym{[}  \NDnt{t}  \NDsym{/}  \NDmv{z}  \NDsym{]}  \NDsym{(}   \mathsf{let}\, \NDmv{z}  :  \NDnt{X}  \otimes  \NDnt{Y} \,\mathsf{be}\, \NDmv{x}  \otimes  \NDmv{y} \,\mathsf{in}\, \NDnt{s}   \NDsym{)}  \NDsym{:}  \NDnt{A}}
      \end{math}
    \end{center}

  \item $\NDdruleSXXtenETwoName$:
    \begin{center}
      \footnotesize
      $\NDdruleSXXtenETwo{}$
    \end{center}
    maps to
    \begin{center}
      \footnotesize
      \begin{math}
        $$\mprset{flushleft}
        \inferrule* [right=$\ElledruleSXXcutTwoName$] {
          {\Gamma  \vdash_\mathcal{L}  \NDnt{s_{{\mathrm{1}}}}  \NDsym{:}  \NDnt{A}  \triangleright  \NDnt{B}} \\
          $$\mprset{flushleft}
          \inferrule* [right=$\ElledruleSXXtenLTwoName$] {
            {\Delta_{{\mathrm{1}}}  \NDsym{;}  \NDmv{x}  \NDsym{:}  \NDnt{A}  \NDsym{;}  \NDmv{y}  \NDsym{:}  \NDnt{B}  \NDsym{;}  \Delta_{{\mathrm{2}}}  \vdash_\mathcal{L}  \NDnt{s_{{\mathrm{2}}}}  \NDsym{:}  \NDnt{C}}
          }{\Delta_{{\mathrm{1}}}  \NDsym{;}  \NDmv{z}  \NDsym{:}  \NDnt{A}  \triangleright  \NDnt{B}  \NDsym{;}  \Delta_{{\mathrm{2}}}  \vdash_\mathcal{L}   \mathsf{let}\, \NDmv{z}  :  \NDnt{A}  \triangleright  \NDnt{B} \,\mathsf{be}\, \NDmv{x}  \triangleright  \NDmv{y} \,\mathsf{in}\, \NDnt{s_{{\mathrm{2}}}}   \NDsym{:}  \NDnt{C}}
        }{\Delta_{{\mathrm{1}}}  \NDsym{;}  \Gamma  \NDsym{;}  \Delta_{{\mathrm{2}}}  \vdash_\mathcal{L}  \NDsym{[}  \NDnt{s_{{\mathrm{1}}}}  \NDsym{/}  \NDmv{z}  \NDsym{]}  \NDsym{(}   \mathsf{let}\, \NDmv{z}  :  \NDnt{A}  \triangleright  \NDnt{B} \,\mathsf{be}\, \NDmv{x}  \triangleright  \NDmv{y} \,\mathsf{in}\, \NDnt{s_{{\mathrm{2}}}}   \NDsym{)}  \NDsym{:}  \NDnt{C}}
      \end{math}
    \end{center}

  \item $\NDdruleSXXimprEName:$
    \begin{center}
      \footnotesize
      $\NDdruleSXXimprE{}$
    \end{center}
    maps to
    \begin{center}
      \footnotesize
      \begin{math}
        $$\mprset{flushleft}
        \inferrule* [right=$\ElledruleSXXcutTwoName$] {
          {\Gamma  \vdash_\mathcal{L}  \NDnt{s_{{\mathrm{1}}}}  \NDsym{:}  \NDnt{A}  \rightharpoonup  \NDnt{B}} \\
          $$\mprset{flushleft}
          \inferrule* [right=$\ElledruleSXXimprLName$] {
            {\Delta  \vdash_\mathcal{L}  \NDnt{s_{{\mathrm{2}}}}  \NDsym{:}  \NDnt{A}} \\
            $$\mprset{flushleft}
            \inferrule* [right={\footnotesize ax}] {
              \\
            }{\NDmv{x}  \NDsym{:}  \NDnt{B}  \vdash_\mathcal{L}  \NDmv{x}  \NDsym{:}  \NDnt{B}}
          }{\NDmv{y}  \NDsym{:}  \NDnt{A}  \rightharpoonup  \NDnt{B}  \NDsym{;}  \Delta  \vdash_\mathcal{L}  \NDsym{[}   \mathsf{app}_r\, \NDmv{y} \, \NDnt{s_{{\mathrm{2}}}}   \NDsym{/}  \NDmv{x}  \NDsym{]}  \NDmv{x}  \NDsym{:}  \NDnt{B}}
        }{\Gamma  \NDsym{;}  \Delta  \vdash_\mathcal{L}  \NDsym{[}  \NDnt{s_{{\mathrm{1}}}}  \NDsym{/}  \NDmv{y}  \NDsym{]}  \NDsym{[}   \mathsf{app}_r\, \NDmv{y} \, \NDnt{s_{{\mathrm{2}}}}   \NDsym{/}  \NDmv{x}  \NDsym{]}  \NDmv{x}  \NDsym{:}  \NDnt{B}}
      \end{math}
    \end{center}

  \item $\NDdruleSXXimplEName:$
    \begin{center}
      \footnotesize
      $\NDdruleSXXimplE{}$
    \end{center}
    maps to
    \begin{center}
      \footnotesize
      \begin{math}
        $$\mprset{flushleft}
        \inferrule* [right=$\ElledruleSXXcutTwoName$] {
          {\Gamma  \vdash_\mathcal{L}  \NDnt{s_{{\mathrm{1}}}}  \NDsym{:}  \NDnt{B}  \leftharpoonup  \NDnt{A}} \\
          $$\mprset{flushleft}
          \inferrule* [right=$\ElledruleSXXimplLName$] {
            {\Delta  \vdash_\mathcal{L}  \NDnt{s_{{\mathrm{2}}}}  \NDsym{:}  \NDnt{A}} \\
            $$\mprset{flushleft}
            \inferrule* [right=$\ElledruleSXXaxName$] {
              \\
            }{\NDmv{x}  \NDsym{:}  \NDnt{B}  \vdash_\mathcal{L}  \NDmv{x}  \NDsym{:}  \NDnt{B}}
          }{\Delta  \NDsym{;}  \NDmv{y}  \NDsym{:}  \NDnt{B}  \leftharpoonup  \NDnt{A}  \vdash_\mathcal{L}  \NDsym{[}   \mathsf{app}_l\, \NDmv{y} \, \NDnt{s_{{\mathrm{2}}}}   \NDsym{/}  \NDmv{x}  \NDsym{]}  \NDmv{x}  \NDsym{:}  \NDnt{B}}
        }{\Delta  \NDsym{;}  \Gamma  \vdash_\mathcal{L}  \NDsym{[}  \NDnt{s_{{\mathrm{1}}}}  \NDsym{/}  \NDmv{y}  \NDsym{]}  \NDsym{[}   \mathsf{app}_l\, \NDmv{y} \, \NDnt{s_{{\mathrm{2}}}}   \NDsym{/}  \NDmv{x}  \NDsym{]}  \NDmv{x}  \NDsym{:}  \NDnt{B}}
      \end{math}
    \end{center}

  \item $\NDdruleSXXFEName$:
    \begin{center}
      \footnotesize
      $\NDdruleSXXFE{}$
    \end{center}
    maps to
    \begin{center}
      \footnotesize
      \begin{math}
        $$\mprset{flushleft}
        \inferrule* [right=$\ElledruleSXXcutTwoName$] {
          {\Gamma  \vdash_\mathcal{L}  \NDmv{y}  \NDsym{:}   \mathsf{F} \NDnt{X} } \\
          $$\mprset{flushleft}
          \inferrule* [right=$\ElledruleSXXFlName$] {
            {\Delta_{{\mathrm{1}}}  \NDsym{;}  \NDmv{x}  \NDsym{:}  \NDnt{X}  \NDsym{;}  \Delta_{{\mathrm{2}}}  \vdash_\mathcal{L}  \NDnt{s}  \NDsym{:}  \NDnt{A}}
          }{\Delta_{{\mathrm{1}}}  \NDsym{;}  \NDmv{z}  \NDsym{:}   \mathsf{F} \NDnt{X}   \NDsym{;}  \Delta_{{\mathrm{2}}}  \vdash_\mathcal{L}   \mathsf{let}\, \NDmv{z}  :   \mathsf{F} \NDnt{X}  \,\mathsf{be}\,  \mathsf{F}\, \NDmv{x}  \,\mathsf{in}\, \NDnt{s}   \NDsym{:}  \NDnt{A}}
        }{\Delta_{{\mathrm{1}}}  \NDsym{;}  \Gamma  \NDsym{;}  \Delta_{{\mathrm{2}}}  \vdash_\mathcal{L}  \NDsym{[}  \NDmv{y}  \NDsym{/}  \NDmv{z}  \NDsym{]}  \NDsym{(}   \mathsf{let}\, \NDmv{y}  :   \mathsf{F} \NDnt{X}  \,\mathsf{be}\,  \mathsf{F}\, \NDmv{x}  \,\mathsf{in}\, \NDnt{s}   \NDsym{)}  \NDsym{:}  \NDnt{A}}
      \end{math}
    \end{center}

  \item $\NDdruleSXXGEName$:
    \begin{center}
      \footnotesize
      $\NDdruleSXXGE{}$
    \end{center}
    maps to
    \begin{center}
      \footnotesize
      \begin{math}
        $$\mprset{flushleft}
        \inferrule* [right=$\ElledruleSXXcutOneName$] {
          $$\mprset{flushleft}
          \inferrule* [right=$\ElledruleSXXGlName$] {
            $$\mprset{flushleft}
            \inferrule* [right=$\ElledruleSXXaxName$] {
              \\
            }{\NDmv{x}  \NDsym{:}  \NDnt{A}  \vdash_\mathcal{L}  \NDmv{x}  \NDsym{:}  \NDnt{A}}
          }{\NDmv{y}  \NDsym{:}   \mathsf{G} \NDnt{A}   \vdash_\mathcal{L}   \mathsf{let}\, \NDmv{y}  :   \mathsf{G} \NDnt{A}  \,\mathsf{be}\,  \mathsf{G}\, \NDmv{x}  \,\mathsf{in}\, \NDmv{x}   \NDsym{:}  \NDnt{A}} \\
           {\Phi  \vdash_\mathcal{C}  \NDnt{t}  \NDsym{:}   \mathsf{G} \NDnt{A} }
        }{\Phi  \vdash_\mathcal{L}  \NDsym{[}  \NDnt{t}  \NDsym{/}  \NDmv{y}  \NDsym{]}  \NDsym{(}   \mathsf{let}\, \NDmv{y}  :   \mathsf{G} \NDnt{A}  \,\mathsf{be}\,  \mathsf{G}\, \NDmv{x}  \,\mathsf{in}\, \NDmv{x}   \NDsym{)}  \NDsym{:}  \NDnt{A}}
      \end{math}
    \end{center}
  \end{itemize}
\end{itemize}

\subsection{Mapping from Sequent Calculus to Natural Deduction}
Function $N:\mathit{SC}\rightarrow\mathit{ND}$ maps a rule in the sequent
calculus to a proof of the same sequent in the natural deduction. The
function is defined as follows:

\begin{itemize}
\item Axioms map to axioms.
\item Right rules map to introductions.
\item Left rules map to eliminations modulo some structural fiddling.
  \begin{itemize}
  \item $\ElledruleTXXunitLName$:
    \begin{center}
      \footnotesize
      $\ElledruleTXXunitL{}$
    \end{center}
    maps to
    \begin{center}
      \footnotesize
      \begin{math}
        $$\mprset{flushleft}
        \inferrule* [right=$\NDdruleTXXcutName$] {
          $$\mprset{flushleft}
          \inferrule* [right=$\NDdruleTXXidName$] {
            \\
          }{\NDmv{x}  \NDsym{:}   \mathsf{Unit}   \vdash_\mathcal{C}  \NDmv{x}  \NDsym{:}   \mathsf{Unit} } \\
          {\Phi  \NDsym{,}  \Psi  \vdash_\mathcal{C}  \NDnt{t}  \NDsym{:}  \NDnt{X}}
        }{\Phi  \NDsym{,}  \Psi  \vdash_\mathcal{C}   \mathsf{let}\, \NDmv{x}  :   \mathsf{Unit}  \,\mathsf{be}\,  \mathsf{triv}  \,\mathsf{in}\, \NDnt{t}   \NDsym{:}  \NDnt{X}}
      \end{math}
    \end{center}

  \item $\ElledruleTXXtenLName$:
    \begin{center}
      \footnotesize
      $\ElledruleTXXtenL{}$
    \end{center}
    maps to
    \begin{center}
      \footnotesize
      \begin{math}
        $$\mprset{flushleft}
        \inferrule* [right=$\NDdruleTXXtenEName$] {
          $$\mprset{flushleft}
          \inferrule* [right=$\NDdruleTXXidName$] {
            \\
          }{\NDmv{z}  \NDsym{:}  \NDnt{X}  \otimes  \NDnt{Y}  \vdash_\mathcal{C}  \NDmv{z}  \NDsym{:}  \NDnt{X}  \otimes  \NDnt{Y}} \\
          {\Phi  \NDsym{,}  \NDmv{x}  \NDsym{:}  \NDnt{X}  \NDsym{,}  \NDmv{y}  \NDsym{:}  \NDnt{Y}  \NDsym{,}  \Psi  \vdash_\mathcal{C}  \NDnt{t}  \NDsym{:}  \NDnt{Z}}
        }{\Phi  \NDsym{,}  \NDmv{z}  \NDsym{:}  \NDnt{X}  \otimes  \NDnt{Y}  \NDsym{,}  \Psi  \vdash_\mathcal{C}   \mathsf{let}\, \NDmv{z}  :  \NDnt{X}  \otimes  \NDnt{Y} \,\mathsf{be}\, \NDmv{x}  \otimes  \NDmv{y} \,\mathsf{in}\, \NDnt{t}   \NDsym{:}  \NDnt{Z}}
      \end{math}
    \end{center}

  \item $\ElledruleTXXimpLName$:
    \begin{center}
      \footnotesize
      $\ElledruleTXXimpL{}$
    \end{center}
    maps to
    \begin{center}
      \footnotesize
      \begin{math}
        $$\mprset{flushleft}
        \inferrule* [right=$\NDdruleTXXcutName$] {
          $$\mprset{flushleft}
          \inferrule* [right=$\NDdruleTXXimpEName$] {
            $$\mprset{flushleft}
            \inferrule* [right=$\NDdruleTXXidName$] {
              \\
            }{\NDmv{y}  \NDsym{:}  \NDnt{X}  \multimap  \NDnt{Y}  \vdash_\mathcal{C}  \NDmv{y}  \NDsym{:}  \NDnt{X}  \multimap  \NDnt{Y}} \\
            {\Phi  \vdash_\mathcal{C}  \NDnt{t_{{\mathrm{1}}}}  \NDsym{:}  \NDnt{X}}
          }{\NDmv{y}  \NDsym{:}  \NDnt{X}  \multimap  \NDnt{Y}  \NDsym{,}  \Phi  \vdash_\mathcal{C}   \NDmv{y}   \NDnt{t_{{\mathrm{1}}}}   \NDsym{:}  \NDnt{Y}} \\
           {\Psi_{{\mathrm{1}}}  \NDsym{,}  \NDmv{x}  \NDsym{:}  \NDnt{Y}  \NDsym{,}  \Psi_{{\mathrm{2}}}  \vdash_\mathcal{C}  \NDnt{t_{{\mathrm{2}}}}  \NDsym{:}  \NDnt{Z}}
        }{\Psi_{{\mathrm{1}}}  \NDsym{,}  \NDmv{y}  \NDsym{:}  \NDnt{X}  \multimap  \NDnt{Y}  \NDsym{,}  \Phi  \NDsym{,}  \Psi_{{\mathrm{2}}}  \vdash_\mathcal{C}  \NDsym{[}   \NDmv{y}   \NDnt{t_{{\mathrm{1}}}}   \NDsym{/}  \NDmv{x}  \NDsym{]}  \NDnt{t_{{\mathrm{2}}}}  \NDsym{:}  \NDnt{Z}}
      \end{math}
    \end{center}

  \item $\ElledruleSXXunitLOneName$:
    \begin{center}
      \footnotesize
      $\ElledruleSXXunitLOne{}$
    \end{center}
    maps to
    \begin{center}
      \footnotesize
      \begin{math}
        $$\mprset{flushleft}
        \inferrule* [right=$\NDdruleSXXunitEOneName$] {
          $$\mprset{flushleft}
          \inferrule* [right=$\NDdruleTXXidName$] {
            \\
          }{\NDmv{x}  \NDsym{:}   \mathsf{Unit}   \vdash_\mathcal{C}  \NDmv{x}  \NDsym{:}   \mathsf{Unit} } \\
          {\Gamma  \NDsym{;}  \Delta  \vdash_\mathcal{L}  \NDnt{s}  \NDsym{:}  \NDnt{A}}
        }{\Gamma  \NDsym{;}  \NDmv{x}  \NDsym{:}   \mathsf{Unit}   \NDsym{;}  \Delta  \vdash_\mathcal{L}   \mathsf{let}\, \NDmv{x}  :   \mathsf{Unit}  \,\mathsf{be}\,  \mathsf{triv}  \,\mathsf{in}\, \NDnt{s}   \NDsym{:}  \NDnt{A}}
      \end{math}
    \end{center}

  \item $\ElledruleSXXunitLTwoName$:
    \begin{center}
      \footnotesize
      $\ElledruleSXXunitLTwo{}$
    \end{center}
    maps to
    \begin{center}
      \footnotesize
      \begin{math}
        $$\mprset{flushleft}
        \inferrule* [right=$\NDdruleSXXunitETwoName$] {
          $$\mprset{flushleft}
          \inferrule* [right=$\NDdruleSXXidName$] {
            \\
          }{\NDmv{x}  \NDsym{:}   \mathsf{Unit}   \vdash_\mathcal{L}  \NDmv{x}  \NDsym{:}   \mathsf{Unit} } \\
          {\Gamma  \NDsym{;}  \Delta  \vdash_\mathcal{L}  \NDnt{s}  \NDsym{:}  \NDnt{A}}
        }{\Gamma  \NDsym{;}  \NDmv{x}  \NDsym{:}   \mathsf{Unit}   \NDsym{;}  \Delta  \vdash_\mathcal{L}   \mathsf{let}\, \NDmv{x}  :   \mathsf{Unit}  \,\mathsf{be}\,  \mathsf{triv}  \,\mathsf{in}\, \NDnt{s}   \NDsym{:}  \NDnt{A}}
      \end{math}
    \end{center}

  \item $\ElledruleSXXtenLOneName$:
    \begin{center}
      \footnotesize
      $\ElledruleSXXtenLOne{}$
    \end{center}
    maps to
    \begin{center}
      \footnotesize
      \begin{math}
        $$\mprset{flushleft}
        \inferrule* [right=$\NDdruleSXXtenEOneName$] {
          $$\mprset{flushleft}
          \inferrule* [right=$\NDdruleTXXidName$] {
            \\
          }{\NDmv{z}  \NDsym{:}  \NDnt{X}  \otimes  \NDnt{Y}  \vdash_\mathcal{C}  \NDmv{z}  \NDsym{:}  \NDnt{X}  \otimes  \NDnt{Y}} \\
          {\Gamma  \NDsym{;}  \NDmv{x}  \NDsym{:}  \NDnt{X}  \NDsym{;}  \NDmv{y}  \NDsym{:}  \NDnt{Y}  \NDsym{;}  \Delta  \vdash_\mathcal{L}  \NDnt{s}  \NDsym{:}  \NDnt{A}}
        }{\Gamma  \NDsym{;}  \NDmv{z}  \NDsym{:}  \NDnt{X}  \otimes  \NDnt{Y}  \NDsym{;}  \Delta  \vdash_\mathcal{L}   \mathsf{let}\, \NDmv{z}  :  \NDnt{X}  \otimes  \NDnt{Y} \,\mathsf{be}\, \NDmv{x}  \otimes  \NDmv{y} \,\mathsf{in}\, \NDnt{s}   \NDsym{:}  \NDnt{A}}
      \end{math}
    \end{center}

  \item $\ElledruleSXXtenLTwoName$:
    \begin{center}
      \footnotesize
      $\ElledruleSXXtenLTwo{}$
    \end{center}
    maps to
    \begin{center}
      \footnotesize
      \begin{math}
        $$\mprset{flushleft}
        \inferrule* [right=$\NDdruleSXXtenETwoName$] {
          $$\mprset{flushleft}
          \inferrule* [right=$\NDdruleSXXidName$] {
            \\
          }{\NDmv{z}  \NDsym{:}  \NDnt{A}  \triangleright  \NDnt{B}  \vdash_\mathcal{L}  \NDmv{z}  \NDsym{:}  \NDnt{A}  \triangleright  \NDnt{B}} \\
          {\Gamma  \NDsym{;}  \NDmv{x}  \NDsym{:}  \NDnt{A}  \NDsym{;}  \NDmv{y}  \NDsym{:}  \NDnt{B}  \NDsym{;}  \Delta  \vdash_\mathcal{L}  \NDnt{s}  \NDsym{:}  \NDnt{C}}
        }{\Gamma  \NDsym{;}  \NDmv{z}  \NDsym{:}  \NDnt{A}  \triangleright  \NDnt{B}  \NDsym{;}  \Delta  \vdash_\mathcal{L}   \mathsf{let}\, \NDmv{z}  :  \NDnt{A}  \triangleright  \NDnt{B} \,\mathsf{be}\, \NDmv{x}  \triangleright  \NDmv{y} \,\mathsf{in}\, \NDnt{s}   \NDsym{:}  \NDnt{C}}
      \end{math}
    \end{center}

  \item $\ElledruleSXXimpLName$:
    \begin{center}
      \footnotesize
      $\ElledruleSXXimpL{}$
    \end{center}
    maps to
    \begin{center}
      \footnotesize
      \begin{math}
        $$\mprset{flushleft}
        \inferrule* [right=$\NDdruleSXXcutOneName$] {
          $$\mprset{flushleft}
          \inferrule* [right=$\NDdruleTXXimpEName$] {
            $$\mprset{flushleft}
            \inferrule* [right=$\NDdruleSXXidName$] {
              \\
            }{\NDmv{y}  \NDsym{:}  \NDnt{X}  \multimap  \NDnt{Y}  \vdash_\mathcal{C}  \NDmv{y}  \NDsym{:}  \NDnt{X}  \multimap  \NDnt{Y}} \\
            {\Phi  \vdash_\mathcal{C}  \NDnt{t}  \NDsym{:}  \NDnt{X}}
          }{\NDmv{y}  \NDsym{:}  \NDnt{X}  \multimap  \NDnt{Y}  \NDsym{,}  \Phi  \vdash_\mathcal{C}   \NDmv{y}   \NDnt{t}   \NDsym{:}  \NDnt{Y}} \\
           {\Gamma  \NDsym{;}  \NDmv{x}  \NDsym{:}  \NDnt{Y}  \NDsym{;}  \Delta  \vdash_\mathcal{L}  \NDnt{s}  \NDsym{:}  \NDnt{A}}
        }{\Gamma  \NDsym{;}  \NDmv{y}  \NDsym{:}  \NDnt{X}  \multimap  \NDnt{Y}  \NDsym{;}  \Phi  \NDsym{;}  \Delta  \vdash_\mathcal{L}  \NDsym{[}   \NDmv{y}   \NDnt{t}   \NDsym{/}  \NDmv{x}  \NDsym{]}  \NDnt{s}  \NDsym{:}  \NDnt{A}}
      \end{math}
    \end{center}

  \item $\ElledruleSXXimprLName$:
    \begin{center}
      \footnotesize
      $\ElledruleSXXimprL{}$
    \end{center}
    maps to
    \begin{center}
      \footnotesize
      \begin{math}
        $$\mprset{flushleft}
        \inferrule* [right=$\NDdruleSXXcutTwoName$] {
          $$\mprset{flushleft}
          \inferrule* [right=$\NDdruleSXXimprEName$] {
            $$\mprset{flushleft}
            \inferrule* [right=$\NDdruleSXXidName$] {
              \\
            }{\NDmv{y}  \NDsym{:}  \NDnt{A}  \rightharpoonup  \NDnt{B}  \vdash_\mathcal{L}  \NDmv{y}  \NDsym{:}  \NDnt{A}  \rightharpoonup  \NDnt{B}} \\
            {\Gamma  \vdash_\mathcal{L}  \NDnt{s_{{\mathrm{1}}}}  \NDsym{:}  \NDnt{A}}
          }{\NDmv{y}  \NDsym{:}  \NDnt{A}  \rightharpoonup  \NDnt{B}  \NDsym{;}  \Gamma  \vdash_\mathcal{L}   \mathsf{app}_r\, \NDmv{y} \, \NDnt{s_{{\mathrm{1}}}}   \NDsym{:}  \NDnt{B}} \\
           {\Delta_{{\mathrm{1}}}  \NDsym{;}  \NDmv{x}  \NDsym{:}  \NDnt{B}  \NDsym{;}  \Delta_{{\mathrm{2}}}  \vdash_\mathcal{L}  \NDnt{s_{{\mathrm{2}}}}  \NDsym{:}  \NDnt{C}}
        }{\Delta_{{\mathrm{1}}}  \NDsym{;}  \NDmv{y}  \NDsym{:}  \NDnt{A}  \rightharpoonup  \NDnt{B}  \NDsym{;}  \Gamma  \NDsym{;}  \Delta_{{\mathrm{2}}}  \vdash_\mathcal{L}  \NDsym{[}   \mathsf{app}_r\, \NDmv{y} \, \NDnt{s_{{\mathrm{1}}}}   \NDsym{/}  \NDmv{x}  \NDsym{]}  \NDnt{s_{{\mathrm{2}}}}  \NDsym{:}  \NDnt{C}}
      \end{math}
    \end{center}

  \item $\ElledruleSXXimplLName$:
    \begin{center}
      \footnotesize
      $\ElledruleSXXimplL{}$
    \end{center}
    maps to
    \begin{center}
      \footnotesize
      \begin{math}
        $$\mprset{flushleft}
        \inferrule* [right=$\NDdruleSXXcutTwoName$] {
          $$\mprset{flushleft}
          \inferrule* [right=$\NDdruleSXXimplEName$] {
            $$\mprset{flushleft}
            \inferrule* [right=$\NDdruleSXXidName$] {
              \\
            }{\NDmv{y}  \NDsym{:}  \NDnt{B}  \leftharpoonup  \NDnt{A}  \vdash_\mathcal{L}  \NDmv{y}  \NDsym{:}  \NDnt{B}  \leftharpoonup  \NDnt{A}} \\
            {\Gamma  \vdash_\mathcal{L}  \NDnt{s_{{\mathrm{1}}}}  \NDsym{:}  \NDnt{A}}
          }{\Gamma  \NDsym{;}  \NDmv{y}  \NDsym{:}  \NDnt{B}  \leftharpoonup  \NDnt{A}  \vdash_\mathcal{L}   \mathsf{app}_l\, \NDmv{y} \, \NDnt{s_{{\mathrm{1}}}}   \NDsym{:}  \NDnt{B}} \\
           {\Delta_{{\mathrm{1}}}  \NDsym{;}  \NDmv{x}  \NDsym{:}  \NDnt{B}  \NDsym{;}  \Delta_{{\mathrm{2}}}  \vdash_\mathcal{L}  \NDnt{s_{{\mathrm{2}}}}  \NDsym{:}  \NDnt{C}}
        }{\Delta_{{\mathrm{1}}}  \NDsym{;}  \Gamma  \NDsym{;}  \NDmv{y}  \NDsym{:}  \NDnt{B}  \leftharpoonup  \NDnt{A}  \NDsym{;}  \Delta_{{\mathrm{2}}}  \vdash_\mathcal{L}  \NDsym{[}   \mathsf{app}_l\, \NDmv{y} \, \NDnt{s_{{\mathrm{1}}}}   \NDsym{/}  \NDmv{x}  \NDsym{]}  \NDnt{s_{{\mathrm{2}}}}  \NDsym{:}  \NDnt{C}}
      \end{math}
    \end{center}

  \item $\ElledruleSXXFlName$:
    \begin{center}
      \footnotesize
      $\ElledruleSXXFl{}$
    \end{center}
    maps to
    \begin{center}
      \footnotesize
      \begin{math}
        $$\mprset{flushleft}
        \inferrule* [right=$\NDdruleSXXFEName$] {
          $$\mprset{flushleft}
          \inferrule* [right=$\NDdruleSXXidName$] {
            \\
          }{\NDmv{y}  \NDsym{:}   \mathsf{F} \NDnt{X}   \vdash_\mathcal{L}  \NDmv{y}  \NDsym{:}   \mathsf{F} \NDnt{X} } \\
          {\Gamma  \NDsym{;}  \NDmv{x}  \NDsym{:}  \NDnt{X}  \NDsym{;}  \Delta  \vdash_\mathcal{L}  \NDnt{s}  \NDsym{:}  \NDnt{A}}
        }{\Gamma  \NDsym{;}  \NDmv{y}  \NDsym{:}   \mathsf{F} \NDnt{X}   \NDsym{;}  \Delta  \vdash_\mathcal{L}   \mathsf{let}\,  \mathsf{F} \NDmv{x}   :   \mathsf{F} \NDnt{X}  \,\mathsf{be}\, \NDmv{y} \,\mathsf{in}\, \NDnt{s}   \NDsym{:}  \NDnt{A}}
      \end{math}
    \end{center}

  \item $\ElledruleSXXGlName$:
    \begin{center}
      \footnotesize
      $\ElledruleSXXGl{}$
    \end{center}
    maps to
    \begin{center}
      \footnotesize
      \begin{math}
        $$\mprset{flushleft}
        \inferrule* [right=$\NDdruleSXXcutTwoName$] {
          $$\mprset{flushleft}
          \inferrule* [right=$\NDdruleSXXGEName$] {
            $$\mprset{flushleft}
            \inferrule* [right=$\NDdruleSXXidName$] {
              \\
            }{\NDmv{y}  \NDsym{:}   \mathsf{G} \NDnt{A}   \vdash_\mathcal{C}  \NDmv{y}  \NDsym{:}   \mathsf{G} \NDnt{A} }
          }{\NDmv{y}  \NDsym{:}   \mathsf{G} \NDnt{A}   \vdash_\mathcal{L}   \mathsf{derelict}\, \NDmv{y}   \NDsym{:}  \NDnt{A}} \\
           {\Gamma  \NDsym{;}  \NDmv{x}  \NDsym{:}  \NDnt{A}  \NDsym{;}  \Delta  \vdash_\mathcal{L}  \NDnt{s}  \NDsym{:}  \NDnt{B}}
        }{\Gamma  \NDsym{;}  \NDmv{y}  \NDsym{:}   \mathsf{G} \NDnt{A}   \NDsym{;}  \Delta  \vdash_\mathcal{L}  \NDsym{[}   \mathsf{derelict}\, \NDmv{y}   \NDsym{/}  \NDmv{x}  \NDsym{]}  \NDnt{s}  \NDsym{:}  \NDnt{B}}
      \end{math}
    \end{center}
    
  \end{itemize}
\end{itemize}

\end{document}